\documentclass[11pt,a4paper]{article}
\usepackage{geometry,a4wide}
\usepackage{subcaption}
\usepackage{tabularx}
\usepackage{amsmath,amssymb,amsthm}
\usepackage[english]{babel}
\usepackage{mathtools}
\usepackage{enumerate}
\usepackage[numbers]{natbib}
\usepackage[hidelinks]{hyperref}
\usepackage{xcolor}

\renewcommand{\arraystretch}{0.99}

\newtheorem{theorem}{Theorem}

\newtheorem{lemma}[theorem]{Lemma}
\newtheorem{proposition}[theorem]{Proposition}

\theoremstyle{definition}

\newtheorem{definition}[theorem]{Definition}
\theoremstyle{remark}
\newtheorem{remark}[theorem]{Remark}
\newtheorem{example}[theorem]{Example}

\newtheorem{exercise}[theorem]{Exercise}
\newtheorem{problem}[theorem]{Problem}
\numberwithin{equation}{section}

\makeindex    

\usepackage{algorithm}
\usepackage[noend]{algpseudocode}

\newcommand{\bigslant}[2]{{\raisebox{.2em}{$#1$}\left/\raisebox{-.2em}{$#2$}\right.}}

\DeclareMathOperator{\rk}{rk}

\DeclareMathOperator{\Id}{Id}
\DeclareMathOperator{\GL}{GL}
\DeclareMathOperator{\supp}{Supp}
\DeclareMathOperator{\wt}{wt}
\usepackage{comment}
\newcommand{\F}{\mathbb{F}}
\newcommand{\N}{\mathbb{N}}

\newcommand{\mC}{\mathcal{C}}

\usepackage{tocloft}
\newcommand{\bP}{\mathbf{P}} 
\newcommand{\bS}{\mathbf{S}} 
\newcommand{\bD}{\mathbf{D}} 
\newcommand{\bE}{\mathbf{E}} 

\newcommand{\bG}{\mathbf{G}}  
\newcommand{\bM}{\mathbf{M}} 
\newcommand{\bs}{\mathbf{s}} 
\newcommand{\bQ}{\mathbf{Q}}  
\newcommand{\bH}{\mathbf{H}}  
\newcommand{\bR}{\mathbf{R}}  
\newcommand{\bU}{\mathbf{U}} 
\newcommand{\bT}{\mathbf{T}} 
\newcommand{\bA}{\mathbf{A}} 
\newcommand{\bB}{\mathbf{B}}
\newcommand{\bC}{\mathbf{C}}
\newcommand{\be}{\mathbf{e}} 
\newcommand{\bX}{\mathbf{X}}  
\newcommand{\br}{\mathbf{r}} 
\newcommand{\bh}{\mathbf{h}} 
\newcommand{\bm}{\mathbf{m}} 
\newcommand{\bc}{\mathbf{c}} 
\newcommand{\bv}{\mathbf{v}} 
\newcommand{\bu}{\mathbf{u}} 
\newcommand{\bx}{\mathbf{x}}
\newcommand{\by}{\mathbf{y}}
\newcommand{\ba}{\mathbf{a}} 
\newcommand{\bb}{\mathbf{b}} 
\newcommand{\bz}{\mathbf{0}} 
\newcommand{\bg}{\mathbf{g}}

 \newcommand{\wtH}{\text{wt}_H}
\newcommand{\lcq}{\left\lceil \log_2(q) \right\rceil}

\newcommand{\gb}{\genfrac{[}{]}{0pt}{}}
\usepackage{extarrows}

 \usepackage{authblk}

\begin{document}

\title{A Survey on Code-based Cryptography}

 	\author[1]{Violetta Weger}
	\affil[1]{Department of Electrical and Computer Engineering\\
		Technical University of Munich\\
		Theresienstrasse 90\\
		80333 Munich, Germany\\ violetta.weger@tum.de
	}
	
	\author[2]{Niklas Gassner}
	\affil[2]{Institute of Mathematics\\
		University of Zurich\\
		Winterthurerstrasse 190\\
		8057 Zurich, Switzerland\\ $\{$niklas.gassner, rosenthal$\}$@math.uzh.ch
	}

	\author[2]{Joachim Rosenthal}



\maketitle

\begin{abstract}

The improvements on quantum technology are threatening our daily
cybersecurity, as a capable quantum computer can break all currently
employed asymmetric cryptosystems. In preparation for the quantum-era
the National Institute of Standards and Technology (NIST) has
initiated in 2016 a standardization process for public-key encryption (PKE)
schemes, key-encapsulation mechanisms (KEM) and digital signature schemes. In 2023, NIST made an additional call for post-quantum signatures. 
With this chapter we aim at providing a survey on code-based
cryptography, focusing on PKEs and signature schemes. We cover the main frameworks introduced in code-based cryptography and analyze their security assumptions. We provide
the mathematical background in a lecture notes style, with the
intention of reaching a wider audience.
\end{abstract}

\clearpage
 
\tableofcontents
 
\clearpage
 
\section{Introduction}\label{sec:intro}

Current public-key cryptosystems are based on integer factorization or the discrete logarithm problem over an elliptic curve or over a finite field. While there are no algorithms known for classical computers to solve these problems efficiently, Shor's algorithm allows a quantum computer to solve these problems in polynomial time \cite{shor}. 
As research on quantum computers advances, the cryptographic community is searching for cryptosystems that will survive attacks on quantum computers. This area of research is called \emph{post-quantum cryptography}.
\medskip

In 2016, the National Institute of Standards and Technology (NIST) has initiated a standardization process for post-quantum cryptosystems.
Such cryptosystems can be based on any hard problem, which cannot be solved by a capable quantum computer in polynomial time. Preferably, these are NP-complete problems, i.e., at least as hard as the hardest problems in NP.
\medskip

The main candidates for post-quantum cryptography are:

\begin{itemize}
\item \textbf{Code-based cryptography} (CBC): CBC is using hard problems from algebraic coding theory. Usually, this is the  NP-complete problem of decoding a random linear code. 
\item \textbf{Lattice-based cryptography}: Lattice-based cryptography is based on hard problems over lattices, such as the NP-complete problems of finding the shortest vector, respectively the closest vector to a given vector in a lattice.  For an overview see \cite{lattice}.
\item \textbf{Multivariate cryptography}: Multivariate cryptography is based on the NP-complete problem of solving multivariate (quadratic) equations defined over some finite field. For an overview see \cite{multivariate}.
\item  
\textbf{Isogeny-based cryptography}: Isogeny-based cryptography is based on finding the isogeny map between two supersingular elliptic curves \cite{isogeny}.
\item \textbf{Hash-based cryptography}: These cryptosystems base their security on the security of hash functions.
\end{itemize}
This survey only covers code-based cryptography, thus, we refer an interested reader to \cite{PQC}, for an overview on post-quantum cryptography in general. 
 \medskip

\emph{Code-based cryptography} denotes any cryptographic system, which bases its security on hard problems from algebraic coding theory. Classically, this problem is the decoding of a random linear code. This problem was shown to be NP-complete in 1978, by Berlekamp, McEliece and Van Tilborg in \cite{berlekamp}. In the same year, McEliece proposed the first code-based cryptosystem \cite{mceliece}, in which one picks a code with underlying algebraic structure that allows efficient decoding and then disguises this code as a seemingly random linear code. A message gets encrypted as corrupted codeword. With the knowledge of the secret code, one can recover the initial message, but an adversary faces the challenge of decoding a random linear code.
\medskip

In 2022, NIST selected 4 cryptographic systems to get standardized, namely  the lattice-based encryption scheme KYBER \cite{NISTKyber}, the lattice-based signature schemes DILITHIUM \cite{NISTdilithium} and FALCON \cite{NISTfalcon} and the hash-based signature scheme SPHINCS$^+$ \cite{NISTsphincs}. However, the standardization process of 2016 is not over yet, as three code-based schemes have moved to the fourth and final round, namely Classical McEliece \cite{NISTMcEliece}, HQC \cite{NISTHQC} and BIKE \cite{NISTBike}. 

The research in this area is, however, far from complete. In fact, in 2023, NIST has reopened the standardization call for signature schemes. 
Within this new call, we can find many code-based schemes and many new and interesting problems.
\medskip

In this chapter we give an extensive  survey on code-based cryptography, explaining the mathematical background of such systems and the difficulties of proposing secure and at the same time practical schemes. We cover the main proposals in the standardization call and the approaches to break such systems. With the reopened standardization process for digital signature schemes, we hope to reach  different research communities to tackle this new challenge   together. 
\medskip

 \subsection{Organization of the Chapter}
 This chapter is organized as follows. In Section \ref{sec:prelim}, we introduce some basics of algebraic coding theory as well as the basics of asymmetric cryptography, such as public-key encryption schemes and signature schemes. In particular, we aim at introducing all used coding-theoretic objects in Section \ref{sec:code} and to describe on a high-level the considered cryptographic schemes in \ref{sec:crypto}. This includes public-key encryption (PKE), key-encapsulation mechanism (KEM) and signature schemes. In particular, we show how to construct a signature scheme via the Fiat-Shamir transform on a Zero-Knowledge (ZK) protocol.
 We also cover the new methods, such as protocols with helpers and Multi-Party Computations (MPC).
 \medskip
 
The main focus of this chapter will lay on Section \ref{sec:pkeframework} where we introduce the public-key encryption frameworks by McEliece, Niederreiter, Alekhnovich as well as the quasi-cyclic scheme, the GPT cryptosystem and the Faure-Loidreau cryptosystem.
\medskip

In Section \ref{sec:sign}, we discuss some code-based signatures, starting with the first construction method, namely hash-and-sign in Section \ref{sec:hash}, then moving to some classic code-based ZK protocols in Section \ref{sec:ZKID} and describe  some new techniques, such as MPC-in-the-head.

In Section \ref{sec:security}, we analyze the security of these systems, where we first focus on the decoding problem of a random linear code: we present the proofs of NP-completeness in Section \ref{sec:NP} and the best-known solvers for the underlying problems in Section \ref{sec:ISD}.  
In the second part of the security analysis, namely Section \ref{sec:attack} we also present some algebraic attacks, which clearly depend on the chosen secret code. For this section, we focus on two of the most preferred codes, one being Reed-Solomon codes and the other being their rank metric analog, Gabidulin codes. Finally, we end the security analysis by shortly reporting on some other ways of attacking code-based systems, such as side-channel attacks, in Section \ref{sec:other}.
\medskip

In Section \ref{sec:overview}, we provide a historical overview on the main code-based PKE and signature scheme proposals, stating their differences, in the notion of the given frameworks, and whether they are  broken.  

In Section \ref{sec:nist}, we shortly cover the submissions to the NIST standardization process with a focus on the finalists in Section \ref{sec:finalist}: Classic McEliece, BIKE and HQC.

In Section \ref{sec:new}, we present the 11 code-based signature schemes submitted to the reopened standardization call and compare their performance in terms of signature and public key size and their running times.

\section{Preliminaries}\label{sec:prelim}

In order to make this chapter as self-contained as possible, we present here a rather long preliminary section, which hopefully makes this survey also accessible to non-experts. We start with the notation used throughout this chapter, followed by the basics of algebraic coding theory and defining all concepts and codes that will be used or mentioned and finally presenting the basics of the considered schemes on a very high-level and with specific examples.

\subsection{Notation}\label{sec:notation}
We denote by $\mathbb{F}_q$ the finite field with $q$ elements, where $q$ is a prime power and denote by $\mathbb{F}_q^\star$ its multiplicative group, i.e., $\mathbb{F}_q \setminus \{0\}$.
 Throughout this chapter, we denote by bold upper case or lower case letters matrices, respectively vectors, e.g. $\bx \in \mathbb{F}_q^n$ and $\bA \in \mathbb{F}_q^{k \times n}$. The identity matrix of size $k$ is denoted by  $\text{Id}_k$. Sets are denoted by upper case letters and for a set $S$, we denote by $\mid S \mid$ its cardinality. By $\text{GL}_n(q)$ we denote the $n \times n$ invertible matrices over $\mathbb{F}_q.$
Notation specific to only one part of this chapter will be defined right before they are used.
 
\subsection{Algebraic Coding Theory}\label{sec:code}

This section is designed to recall and/or introduce all definitions and coding theoretic objects required in this chapter. Most proofs will be omitted or left as an exercise. For interested readers that are completely new to algebraic coding theory we recommend the following books \cite{roth,berlekampbook,vanlint,macwilliams}. 
We also leave away the references to standard definitions and results, which can be found in any book on coding theory. For more specific results, we will give a proper reference.

\subsubsection{Basics on Hamming-Metric Codes}
In classical coding theory one considers the finite field $\mathbb{F}_q$ of $q$ elements, where $q$ is a prime power. 

\begin{definition}[Linear Code]
 Let $1 \leq k \leq n$ be integers. Then, an $[n,k]$ \emph{linear code} $\mC$ over $\F_q$ is a $k$-dimensional linear subspace of $\F_q^n$.
\end{definition}

Note that we emphasize the linearity, as a \emph{code} is simply any subset $\mathcal{C} \subseteq \mathbb{F}_q^n$.

The parameter $n$  is called the \emph{length} of the code,  the elements in the code are called \emph{codewords} and $R=k/n$ is called the \emph{rate} of the code.
In order to measure how far apart two vectors are, we endow $\F_q$ with a metric. Usually, this is the \emph{Hamming metric}.
\begin{definition}[Hamming Metric]
Let $n$ be a positive integer. For $\bx \in \F_q^n$, the \emph{Hamming weight} of $\bx$   is given by the size of its support, {i.e.},
$$\wtH(\bx) = \mid \{ i \in \{1, \ldots, n\} \mid x_i \neq 0 \} \mid.$$
For $\bx,\by \in \F_q^n$, the \emph{Hamming distance} between $\bx$ and $\by$ is given by the number of positions in which they differ, {i.e.},
$$d_H(\bx,\by) = \mid \{ i \in \{1, \ldots, n\} \mid x_i \neq y_i \} \mid.$$
\end{definition}
Note that the Hamming distance is induced from the Hamming weight, that is $d_H(\bx,\by) = \wtH(\bx-\by).$
Having defined a metric, one can also consider the minimum distance of a code, i.e., the smallest distance achieved by its distinct codewords. 
\begin{definition}[Minimum Distance]
Let $\mC$ be a  code over $\F_q$. The \emph{minimum Hamming distance} of $\mC$ is denoted by $d_H(\mC)$ and given by
$$d_H(\mC) = \min \{ d_H(\bx, \by) \mid \bx, \by \in \mC, \ \bx \neq \by\}.$$
\end{definition}
\begin{exercise}
Show that for a linear code $\mathcal{C}$, we have $$d_H(\mC) = \min\{ \wtH(\bx) \mid \bx \in \mC, \bx \neq \mathbf{0} \}.$$
\end{exercise}

\begin{exercise}
Give an example, where $$d_H(\mC) \neq \min\{ \wtH(\bx) \mid \bx \in \mC, \bx \neq \mathbf{0} \}.$$
\end{exercise}
We denote by $d_H(\bx, \mC)$ the minimal distance between $\bx \in \F_q^n$ and a codeword in $\mC$.

Let $r$ be a positive integer. We define the Hamming ball as all the vectors which have at most Hamming weight $r$, i.e.,
$$B_H(r,n,q) = \{ \bx \in \F_q^n \mid \wtH( \bx) \leq r\}.$$
\begin{exercise} Show that 
$$\mid B_H( r,n,q) \mid = \sum_{i=0}^r \binom{n}{i}(q-1)^i.$$
\end{exercise}
The minimum distance of a code is an important parameter, since it is connected to the error correction capability of the code. \\
 We say that a code can \emph{correct} up to $t$ errors, if for all $\bx \in \F_q^n$ with $d_H(\bx, \mC) \leq  t$, there exists exactly one $\by \in \mC$, such that $d_H(\bx, \by) \leq t$.
A  \emph{decoding algorithm} $\mathcal{D}$ is an algorithm that is given such a word $\bx \in \F_q^n$ and returns the closest codeword, $\by \in \mC$, such that $d_H(\bx, \by) \leq t$. 
The most interesting codes for applications  are codes with an efficient decoding algorithm, which clearly not every code possesses. 

\begin{exercise}
Let $\mC$ be a linear code over $\F_q$ of length $n$ and of minimum distance $d_H$. Show that the code can  correct  up to $t:=\left\lfloor \frac{d_H-1}{2} \right\rfloor$ errors.
\end{exercise}
One of the most important bounds in coding theory is the Singleton bound, which provides an upper bound on the minimum distance of a code.
\begin{theorem}[Singleton Bound \cite{sing}]
Let $k \leq n$ be positive integers and let $\mC$ be an $[n,k]$ linear code over $\F_q$. Then, 
$$d_H \leq n-k+1.$$
\end{theorem}
\begin{exercise}
Prove the Singleton Bound by showing that  deleting $d_H-1$ of the positions is an injective map.
\end{exercise}
A code that achieves the Singleton bound is called a \emph{maximum distance separable} (MDS) code. MDS codes   are of immense interest, since they can correct the maximal amount of errors for fixed code parameters.

Linear codes allow for an easy representation through their generator matrices, which have the code as an image.
\begin{definition}[Generator Matrix]
Let $k \leq n$ be positive integers and let $\mC$ be an $[n,k]$ linear code over $\F_q$. Then, a matrix $\bG \in \F_q^{k \times n}$ is called a \emph{generator matrix} of $\mC$ if
$$\mC = \left\{ \bx \bG \mid \bx \in \F_q^k\right\},$$
that is, the rows of $\bG$ form a basis of $\mC$.
\end{definition}
We will often write $\langle \bG \rangle$ to denote the code generated by $\bG.$

One can also represent the code through a matrix $\bH$, which has the code as kernel.
\begin{definition}[Parity-Check Matrix]
Let $k \leq n$ be positive integers and let $\mC$ be an $[n,k]$ linear code over $\F_q$. Then, a matrix $\bH \in \F_q^{(n-k) \times n}$ is called a \emph{parity-check matrix} of $\mC$, if
$$\mC = \left\{ \by \in \F_q^n \mid \bH\by^\top = \bz\right\}.$$
\end{definition}
For any $\bx \in \F_q^n$, we  call $\bx \bH^\top$ a \emph{syndrome}.
\begin{exercise}
Let $k \leq n$ be positive integers and let $\mC$ be an $[n,k]$ linear code over $\F_q$. Let $\bH$ be a parity-check matrix of $\mC$. Show that  $\mC$ has minimum distance $d_H$ if and only if every $d_H-1$ columns of $\bH$ are linearly independent and there exist $d_H$ columns, which are linearly dependent. 
\end{exercise}
For $\bx, \by \in \F_q^n$ let us denote by $\langle \bx, \by \rangle$ the standard inner product, {i.e.}, 
$$\langle \bx, \by \rangle = \sum_{i=1}^n x_iy_i.$$
Then, we can define the dual of an $[n,k]$ linear  code $\mC$ over $\F_q$ as the orthogonal space of $\mC$.
\begin{definition}[Dual Code]
Let  $k \leq n$ be positive integers and let  $\mC$ be an $[n,k]$ linear  code over $\F_q$. The \emph{dual code} $\mC^\perp$ is  an $[n,n-k]$ linear code over $\F_q$, defined as 
$$\mC^\perp = \{ \bx \in \F_q^n \mid \langle \bx, \by \rangle = 0 \ \forall \ \by \in \mC \}.$$
\end{definition}

\begin{exercise} Show that a parity-check matrix of $\mC$ is in fact a generator matrix of $\mC^\perp$. 
\end{exercise}
\begin{exercise}
Show that the dual of an MDS code is an MDS code.
\end{exercise}
For $\bx \in \F_q^n$ and $S \subseteq \{1, \ldots, n\}$ we denote by $\bx_S$ the vector consisting of the entries of $\bx$ indexed by $S$. While for $\bA \in \F_q^{k \times n}$, we denote by $\bA_S$ the matrix consisting of the columns of $\bA$ indexed by $S$. Similarly, we denote by $\mC_S$ the code consisting of the codewords $\bc_S$.

Observe that an $[n,k]$ linear code can  be completely defined by certain sets 
of $k$ positions. The following  concept characterizes such defining sets.
\begin{definition}[Information Set]
Let $k\leq n$ be positive integers and let $\mC$ be an $[n,k]$ linear code over $\F_q$. Then, a set $I \subset \{1, \ldots, n\}$ of size $k$ is called an \emph{information set} of $\mathcal{C}$ if $$\mid \mC \mid = \mid \mC_I \mid.$$
\end{definition}
\begin{exercise}
How many information sets can an $[n,k]$ linear code have at most?
\end{exercise} 
\begin{exercise}
Let $\mC$  be an $[n,k]$ linear code,  $I$ an information set and let $\bG$ be a generator matrix  and $\bH$ a parity-check matrix. Show that $\bG_I$ is an invertible matrix of size $k$.
If $I^C:=  \{1,\ldots ,n\}\setminus I$ is the complement set of $I$, then, 
$\bH_{I^C}$ is an invertible matrix of size $n-k$.
\end{exercise}
\begin{exercise}
Let $\mC$ be the code generated by $\bG \in \F_5^{2 \times 4}$, given as
$$\bG = \begin{pmatrix} 1 & 3 & 2 & 3 \\ 0 & 4 & 4 & 3 \end{pmatrix}.$$ 
Determine all information sets of this code.
\end{exercise}
\begin{definition}[Systematic Form]
Let $k \leq n$ be positive integers and $\mC$ be an $[n,k]$ linear code over $\F_q$. Then, there exist  some permutation matrix $\bP$ and some invertible matrix $\bU$ that bring $\bG$ in  \emph{systematic form}, i.e., 
$$\bU\bG\bP =\begin{pmatrix}
\Id_{k} & \bA
\end{pmatrix}, $$ where $\bA \in \F_q^{k \times (n-k)}$. Similarly, there exist  some permutation matrix $\bP'$ and some invertible matrix $\bU'$, that bring $\bH$ into systematic form as  $$\bU'\bH\bP' =\begin{pmatrix}
\bB & \Id_{n-k}
\end{pmatrix}, $$ where $\bB \in \F_q^{(n-k) \times k}$. 
\end{definition}

Let us denote by $V_H(r,n,q)$ the volume of a ball in the Hamming metric, i.e.,
$$V_H(r,n,q) = \mid B_H(r,n,q) \mid.$$
The Gilbert-Varshamov bound \cite{gilbert,varshamov,sacks} is one of the most prominent bounds in coding theory and widely used in code-based cryptography since it provides a sufficient condition for the existence of linear codes. 
\begin{theorem}[Gilbert-Varshamov bound]
Let $q$ be a prime power and let $k \leq n$ and $d_H$ be positive integers, such that 
\begin{equation*}
    V_H(d_H-2,n-1, q) <q^{n-k}.
  \end{equation*}
Then, there exists a  $[n,k]$ linear code over $\mathbb{F}_q$ with minimum Hamming distance at least $d_H$. 
\end{theorem}
The better known Gilbert-Varshamov bound is a statement on the maximal size of a code, that is:
let us denote by $A_H(n,d,q)$ the maximal size of a code in $\mathbb{F}_q^n$ having minimum Hamming distance $d$. 
\begin{theorem}[Gilbert-Varshamov Bound] Let $q$ be a prime power and $n,d$ be positive integers. Then,
\begin{equation*}
  A_H(n,d,q) \geq \frac{q^{n}}{V_H(d-1,n,q)}.  
\end{equation*}
\end{theorem}
It turns out that random codes attain the asymptotic Gilbert-Varshamov bound with high probability. This will be an important result for the asymptotic analysis of some algorithms. Let us first give some notation: let $0\leq \delta\leq 1$ denote the relative minimum distance, i.e., $\delta= d/n$ and let us denote by   $$\overline{R}(\delta) = \limsup\limits_{n \to \infty} \frac{1}{n} \log_{q} A_H(n, \delta n,q)$$
the asymptotic information rate.
\begin{definition}[Entropy Function]
 For a positive integer $q \geq 2$ the $q$-ary entropy function is defined as follows:
 \begin{align*}
     h_q: [0,1] &\to \mathbb{R}, \\
     x & \to x \log_q(q-1) - x \log_q(x) -(1-x)\log_q(1-x).
 \end{align*}
\end{definition}
\medskip

\begin{exercise}
Show that for $s \in [0,1-1/q]$ we have that 
\begin{enumerate}
    \item $V_H(sn,n,q) \leq q^{h_q(s)n},$
    \item $V_H(sn,n,q) \geq q^{h_q(s)n-o(n)},$
\end{enumerate}
using Stirling's formula.
\end{exercise}
\begin{theorem}[The Asymptotic Gilbert-Varshamov Bound]
For every prime power $q$ and $\delta \in [0, 1-1/q]$ there exists an infinite family $\mathcal{C}$ of codes with rate
$$\overline{R}(\delta) \geq 1-h_q(\delta). $$
\end{theorem}
Recall that in complexity theory we write $f(n) = \Omega(g(n))$, if $$\limsup\limits_{n \to \infty} \left | \frac{f(n)}{g(n)} \right | >0.$$ %
For example, $f(n) = \Omega(n)$ means that $f(n)$ grows at least polynomially in $n$.
\begin{theorem}\label{randomGv}
For every prime power $q, \delta \in [0,1-1/q)$ and $0 < \varepsilon < 1-h_q(s)$ and sufficiently large positive integer $n$. The following holds for 
$$k = \left\lceil (1-h_q(\delta)-\varepsilon)n \right\rceil.$$ If $\bG \in \mathbb{F}_q^{k \times n}$ is chosen uniformly at random, the linear code $\mathcal{C}$ generated by $\bG$ has rate at least $1-h_q(\delta)-\varepsilon$ and relative minimum distance at least $\delta$ with probability at least $1- e^{-\Omega(n)}.$
\end{theorem}
\begin{exercise}
Prove  Theorem \ref{randomGv} following these steps:
\begin{enumerate}
    \item What is the probability for $\bG$ to have full rank?
    \item For each non-zero $\bx \in \mathbb{F}_q^k$  show that $\bx\bG$ is a uniformly random element.
    \item Show that the probability that $\text{wt}_H(\bx\bG) \leq \delta n$ is at most $q^{(h_q(\delta)-1)n}.$
    \item Use the union bound over all non-zero $\bx$ and the choice of $k$ to get the claim.
\end{enumerate}
\end{exercise}
This was first proven in \cite{gv,pierce} and shows that for a random code with large length, we know what minimum Hamming distance to expect. 

These results hold also more generally over any finite chain ring and for any additive weight, see \cite{free}.

We also want to introduce the following two methods to get a new code from an old code: puncturing and shortening. When we puncture a code we essentially delete all coordinates indexed by a certain set in all codewords, while shortening can be regarded as the puncturing of a special subcode.

\begin{definition}
Let $\mC$ be an $[n,k]$ linear code over $\mathbb{F}_q$ and let $S \subseteq \{1, \ldots, n\}$ be a set of size $s$. Then, we define the \emph{punctured code} $\mC^S$ in $S$ as follows
$$\mC^S = \{ (c_i)_{i \not\in S} \mid c \in \mC\}.$$ Let us define $\mC(S)$ to be the subcode containing all codewords which are 0 in $S$, that is 
$$\mC(S) = \{c \in \mathcal{C} \mid c_i = 0 \  \forall i \in S\}. $$ Then, we define the \emph{shortened code} $\mC_S$ in $S$ to be $$\mC_S = \mC(S)^S.$$
\end{definition}

Clearly, the punctured code $\mC^S$ has now length $n-s$. What happens to its dimension?
\begin{exercise}
Show that if $s< d$, the minimum distance of $\mC$, then $\mC^S$ has dimension $k$. 
\end{exercise}

Shortening and puncturing of a code are heavily connected through the dual code:
\begin{theorem}\label{shortpunc}
Let $\mC$ be a linear $[n,k]$ code over $\mathbb{F}_q$ with dual code $\mC^\perp.$ Let $S \subseteq \{1, \ldots, n\}$ be a set of size $s$. Then
\begin{enumerate}
    \item $(\mC^\perp)_S = (\mC^S)^\perp$, 
    \item$(\mC^\perp)^S = (\mC_S)^\perp.$ 
\end{enumerate}
\end{theorem}

\begin{example}
Let us consider the binary code generated by $$\bG = \begin{pmatrix}
1 & 0 & 0 & 1 & 1 & 0 \\ 0 & 1 & 0 & 0 & 1 & 1 \\ 0 & 0 & 1 & 1 & 1 & 1 
\end{pmatrix}, $$ and $S=\{4,5\}.$ Then, the punctured code $\mC^S$ has generator matrix
$$ \bG^S = \begin{pmatrix}
1 & 0 & 0 & 0   \\ 0 & 1 & 0 & 1 \\ 0 & 0 & 1 & 1  
\end{pmatrix}.$$ Note that $\mC(S) =\{(1,0,1,0,0,1), (0,0,0,0,0,0) \}$, thus the generator matrix of $\mC_S$ is given by 
$$\bG_S= \begin{pmatrix}   1 & 0 & 1 & 1
\end{pmatrix}.$$
\end{example}
\begin{exercise} 
Show that Theorem \ref{shortpunc} holds for this example. \end{exercise}

\subsubsection{Matrix Codes}
Let us denote by $\mathbb{F}_q^{n \times m}$ the $n \times m$ matrices over $\mathbb{F}_q.$

Instead of considering subspaces in $\mathbb{F}_q^n$, we can also consider subspaces in $\mathbb{F}_q^{m \times n}$, referred to as \emph{matrix codes}.
\begin{definition}[Matrix  Codes]
An $\mathbb{F}_q$-linear subspace of $\mathbb{F}_q^{n \times m}$ is called a \emph{matrix  code}.
\end{definition}
Thus, instead of a $k \times n$ generator matrix $\bG \in \mathbb{F}_q^{k \times n}$, we generate the code with $k$ generating matrices $\bG_1, \ldots, \bG_k \in \mathbb{F}_q^{m \times n}$, then every codeword is of the form 
$$\bC= \lambda_1 \bG_1 + \cdots + \lambda_k \bG_k,$$ for some $\lambda_i \in \mathbb{F}_q$. Since these codes are only linear over $\mathbb{F}_q$, they are also called $\mathbb{F}_q$-linear codes.

One can define the Hamming metric on such matrices, by either considering the number of non-zero columns or the number of non-zero entries.

For a matrix $\bA \in \mathbb{F}_q^{m \times n}$ let us denote by $\bc_i \in \mathbb{F}_q^m$ its columns for $i \in \{1, \ldots, n\}$, by $\br_j \in \mathbb{F}_q^n$ its rows for $j \in \{1, \ldots, m\}$ and finally by $a_{i,j}$ its entries for $(i,j) \in \{1, \ldots, n\} \times \{1, \ldots, m\}.$ Given $\bA \in \mathbb{F}_q^{m \times n}$ we define
\begin{align*}
    \text{wt}_{H,c}(\bA) & =|\{i \in \{1, \ldots, n\} \mid \bc_i \neq \bz \}|, \\ 
    \text{wt}_{H,v}(\bA) &= | \{ (i,j) \in \{1, \ldots, n\} \times \{1, \ldots, m\} \mid a_{i,j} \neq 0\}|.
\end{align*}
We will specify which notion of Hamming metric we are using, whenever we use matrix codes. 

\begin{definition}
    Given a matrix $\bA \in \mathbb{F}_q^{m \times n}$ with rows $\ba_1, \ldots, \ba_m \in \mathbb{F}_q^n$ we define the \emph{vectorization} of $\bA$ to be $\text{vec}(\bA)= (\ba_1, \ldots, \ba_m) \in \mathbb{F}_q^{mn}.$
\end{definition}
The Hamming weight of $\text{vec}(\bA)$ coincides with the second notion of Hamming metric of matrices, i.e., 
$$\text{wt}_H(\text{vec}(\bA))=\text{wt}_{H,v}(\bA).$$

  Let $\Gamma=\{\gamma_1, \ldots, \gamma_m\}$ be a basis of $\mathbb{F}_{q^m}$ over $\mathbb{F}_q.$
  That is, we can write every element $a \in \mathbb{F}_{q^m}$ as 
  $$a=\sum_{i=1}^m a_i \gamma_i,$$ with $a_i \in \mathbb{F}_q.$
\begin{definition}
    Let $\Gamma=\{\gamma_1, \ldots, \gamma_m\}$ be a basis of $\mathbb{F}_{q^m}$ over $\mathbb{F}_q.$ Then, we can define the \emph{extension map}
    \begin{align*}
        \Gamma: \mathbb{F}_{q^m} &\to \mathbb{F}_q^m \\
        a=\sum_{i=1}^m a_i  \gamma_i &\mapsto (a_1, \ldots, a_m).
    \end{align*}
    By abuse of notation we will also use $\Gamma$ to denote the extension map 
    $\Gamma: \mathbb{F}_{q^m}^n \to \mathbb{F}_q^{m \times n},$ where each entry is extended to a column. 
\end{definition}
The Hamming weight of the vector $\Gamma^{-1}(\bA)=\ba \in \mathbb{F}_{q^m}^n$ coincides with the first notion of Hamming weight for matrices, i.e., 
\begin{equation}\label{extham} \text{wt}_H(\Gamma^{-1}(\bA)) = \text{wt}_{H,c}(\bA).
\end{equation}

\begin{exercise}
    Show that  Equation \eqref{extham} is independent of the choice of basis $\Gamma.$
\end{exercise}

The extension map can also be applied to a code itself, that is:
\begin{definition}
Let $\mathcal{C} \subseteq \mathbb{F}_{q^m}^n$ be a linear code and let $\Gamma$ be a basis of $\mathbb{F}_{q^m}$ over $\mathbb{F}_q$. The \emph{code associated with $\Gamma$} is given by
$$\Gamma(\mathcal{C})= \{ \Gamma(\bc) \mid \bc \in \mathcal{C}\}.$$
\end{definition}
Note that since $\mathcal{C}\subseteq \mathbb{F}_{q^m}^n$ was $\mathbb{F}_{q^m}$-linear, we get that $\Gamma(\mathcal{C}) \subseteq \mathbb{F}_q^{m \times n}$ is $\mathbb{F}_q$-linear.

The dual code of a matrix code, requires a new inner product, which extends the previous standard inner product. For this, recall that the \emph{trace} of a matrix is the sum of the entries on its diagonal.
\begin{definition}
    Let $\bA, \bB \in \mathbb{F}_q^{m \times n}$, then we define their trace product as
    $$\text{Tr}(\bA\bB^\top).$$
\end{definition}

\begin{definition}
    Let $\mathcal{C} \subseteq \mathbb{F}_q^{m \times n}$ be a linear matrix code, then its \emph{dual code} is given by
    $$\mathcal{C}^\perp = \{\bA \in \mathbb{F}_q^{m \times n} \mid \text{Tr}(\bA\bB^\top)=\bz \text{ for all } \bB \in \mathcal{C}\}.$$
\end{definition}
This product is compatible with the standard inner product on $\mathbb{F}_{q^m}^n$. 
For this we need the following definition.
\begin{definition}
  Let $\Gamma=\{\gamma_1, \ldots, \gamma_m\}, \Gamma'=\{\gamma_1', \ldots, \gamma_m'\}$ be bases of $\mathbb{F}_{q^m}$ over $\mathbb{F}_q$. We say that  $\Gamma$ and 
  $\Gamma'$ are orthogonal if  $$\text{Tr}_{\mathbb{F}_q}(\gamma_i \gamma_j') = \delta_{i,j},$$ where $\delta_{i,j}$ denotes the Kronecker delta function, i.e., it outputs 0 if $i \neq j$ and 1 if $i=j$, and $\text{Tr}_{\mathbb{F}_q}$ denotes the   field trace, i.e.,
  \begin{align*} \text{Tr}: \mathbb{F}_{q^m} &\to \mathbb{F}_q \\ 
  a & \mapsto \sum_{i=0}^{m-1} a^{q^i}.
  \end{align*}
\end{definition}
\begin{proposition}\label{prop:dualmatrix}
Let $\mathcal{C} \subseteq \mathbb{F}_{q^m}^n$ be a linear code.
    Let $\Gamma, \Gamma'$ be orthogonal bases of $\mathbb{F}_{q^m}$ over $\mathbb{F}_q$, then 
    $$\Gamma(\mathcal{C})^\perp= \Gamma'(\mathcal{C}^\perp).$$
\end{proposition}

\begin{exercise}
Show that Proposition \ref{prop:dualmatrix} holds for the example $$\mathcal{C}=\langle 1, \alpha \rangle \subseteq \mathbb{F}_8^2, $$
where $\mathbb{F}_8=\mathbb{F}_2[\alpha]$ and $\alpha^3=\alpha+1$, $\Gamma =\{1, \alpha, \alpha^2\}, \Gamma'=\{1, \alpha^2,\alpha\}.$
\end{exercise}

The new inner product is in fact also compatible with the vectorization:
\begin{proposition}
    Let $\bA,\bB \in \mathbb{F}_q^{m \times n}$, then $$\text{Tr}(\bA^\top\bB) = \langle \text{vec}(\bA) ,\text{vec}(\bB)\rangle.$$
\end{proposition}

\subsubsection{Generalized Reed-Solomon Codes}
In order to give a self-contained chapter, we also want to introduce some of the most prominent codes that are used in code-based cryptography. For this we  start with Generalized Reed-Solomon codes (GRS), \cite{rs}. 
\begin{definition}[Generalized Reed-Solomon Code]\label{def:reedsolomon}
Let  $k \leq n \leq q$ be positive integers. Let $\alpha \in \F_q^n$ be an $n$-tuple of distinct elements, i.e., $\alpha =(\alpha_1, \ldots, \alpha_n)$ with $\alpha_i \neq \alpha_j,$ for all $i \neq j \in \{1, \ldots, n\}$. Let $\beta \in \F_q^n$ be an $n$-tuple of nonzero elements, i.e., $\beta =(\beta_1, \ldots, \beta_n),$ with $ \beta_i \neq 0$ for all $i \in \{1, \ldots, n\}.$  The \emph{Generalized Reed-Solomon code} of length $n$ and dimension $k$, denoted by $\text{GRS}_{n,k}(\alpha, \beta)$ is defined as
\begin{equation*}
\text{GRS}_{n,k}(\alpha,\beta)= \left\{ (\beta_1 f(\alpha_1), \ldots, \beta_n f(\alpha_n)) \bigm| f \in \F_q[x], \ \text{deg}(f) <  k \right\}.
\end{equation*}
\end{definition}
In the case where $\beta= (1, \ldots,1 )$, we call the code $ \text{GRS}_{n,k}(\alpha, \beta)$ a \emph{Reed-Solomon} (RS) code and denote it by $\text{RS}_{n,k}(\alpha).$
\begin{exercise}
Show that the Vandermonde matrix 
$$\begin{pmatrix} 1 & \cdots & 1 \\
\alpha_1 & \cdots & \alpha_n \\
\vdots & & \vdots \\
\alpha_1^{k-1} & \cdots & \alpha_n^{k-1} \end{pmatrix}$$ is a generator matrix of a RS code.  Similarly, build a generator matrix of the $\text{GRS}_{n,k}(\alpha, \beta)$ code.
\end{exercise}
\begin{exercise}
Show that GRS codes are MDS codes, i.e., $$d_H(\text{GRS}_{n,k}(\alpha,\beta)) = n-k+1.$$ 
\end{exercise}
Observe that the dual code of a GRS code is again a GRS code.
\begin{proposition}\label{prop:dualGRS}
Let $k \leq n\leq q$ be positive integers. Then \begin{equation*}
\text{GRS}_{n,k}(\alpha,\beta)^{\perp} = \text{GRS}_{n,n-k}(\alpha,\gamma),
\end{equation*}
where 
\begin{equation*}
\gamma_i =  \beta_i^{-1}  \prod_{\substack{ j= 1 \\ j \neq i}}^n(\alpha_i-\alpha_j)^{-1}.
\end{equation*}
\end{proposition}

\subsubsection{Goppa Codes}
Another important family of codes in code-based cryptography  is the family of classical $q$-ary Goppa codes \cite{goppa1,goppa2,goppa3}. 

Let $m$ be a positive integer, $n=q^m$ and $\F_{q^m}$ be a finite field. Let $G \in \F_{q^m}[x]$. Then define the quotient ring
\begin{equation*}
S_m =\bigslant{ \F_{q^m}[x]}{\langle G \rangle}.
\end{equation*}
\begin{lemma}
Let $\alpha \in \F_q$ be such that $G(\alpha)\neq 0$. Then $(x-\alpha)$ is invertible in $S_m$ and 
\begin{equation*}
(x-\alpha)^{-1} = -\frac{1}{G(\alpha)}\frac{G(x)-G(\alpha)}{x-\alpha}.
\end{equation*}
\end{lemma}
\begin{definition}[Classical Goppa Code]\label{def:goppa} 
Let $\alpha = (\alpha_1, \ldots, \alpha_n) \in \F_{q^m}^n$, be such that  $\alpha_i \neq \alpha_j$ for all $i \neq j \in \{1, \ldots, n\}$, and $G(\alpha_i)\neq 0$ for all $i \in \{1, \ldots, n\}$. Then we can define the \emph{classical $q$-ary Goppa code} as
 \begin{equation*}
 \Gamma(\alpha,G)= \left\{c \in \F_q^n \ \bigg\vert \ \sum_{i=1}^n \frac{c_i}{x-\alpha_i}=0 \ \text{in} \ S_m\right\}.
 \end{equation*}
\end{definition}
\begin{proposition}
The Goppa code $\Gamma(\alpha,G)$ has minimum Hamming distance $d_H(\Gamma(\alpha,G)) \geq \deg(G) +1$ and dimension $k \geq n-m\deg(G)$.
\end{proposition}
In order to construct a parity-check matrix of a classical Goppa code, 
let us define $\beta = (G(\alpha_1)^{-1}, \ldots, G(\alpha_n)^{-1})$. The parity-check matrix of $\Gamma(\alpha, G)$ is then given by the weighted Vandermonde matrix
$$ \bH= \begin{pmatrix} 	\beta_1 & \cdots & \beta_n \\
	\beta_1 \alpha_1 & \cdots & \beta_n \alpha_n \\
	\vdots & & \vdots \\
	\beta_1 \alpha_1^{r-1} & \cdots & \beta_n \alpha_n^{r-1}
\end{pmatrix}.$$
Note that  $\bH \in \mathbb{F}_{q^m}^{(n-k) \times n}$, but the code $\Gamma(\alpha,G)$ is the  $\mathbb{F}_q$-kernel of $\bH$.

From this construction, we can already see that strong connection between classical Goppa codes and GRS codes. For this  we define subfield subcodes and alternant codes  in the following. 
\begin{definition}[Subfield Subcode] 
Let $\mathcal{C}$ be an $[n,k]$ linear code over $\mathbb{F}_{q^m}$. The \emph{subfield subcode} of $\mathcal{C}$  over $\mathbb{F}_q$ is then defined as 
$$\mathcal{C}_{\mathbb{F}_q}=\mathcal{C} \cap \mathbb{F}_q^n.$$
\end{definition}
\begin{proposition}\label{prop:subfield} 
Let  $\mathcal{C}$ be an $[n,k]$ linear code over $\mathbb{F}_{q^m}$ with minimum distance $d$. Then $\mathcal{C}_{\mathbb{F}_q}$ has dimension $\geq n-m(n-k)$ and minimum distance $\geq d$. 
\end{proposition}
\begin{exercise}
Prove Proposition \ref{prop:subfield} using the map 
\begin{align*}
    \phi: \mathbb{F}_{q^m}^n & \to \mathbb{F}_{q^m}^n, \\
    (x_1, \ldots, x_n) & \mapsto (x_1^q-x_1, \ldots, x_n^q-x_n).
\end{align*}
\end{exercise}
A special case of subfield subcodes are the alternant codes, where one takes subfield subcodes of GRS codes.
\begin{definition}[Alternant Code] 
Let $\alpha \in \mathbb{F}_{q^m}^n$ be pairwise distinct and $\beta \in (\mathbb{F}_{q^m}^\star)^n$. Then the \emph{alternant code} $\mathcal{A}_{m,n,k}(\alpha, \beta)$ is defined as 
$$\mathcal{A}_{m,n,k}(\alpha, \beta)=\text{GRS}_{m,n,k}(\alpha, \beta) \cap \mathbb{F}_q^n.$$
\end{definition}
\begin{proposition}\label{prop:alt} 
The alternant code $\mathcal{A}_{m,n,k}(\alpha, \beta)$ has dimension $\geq n-m(n-k)$ and minimum distance $\geq n-k+1$. 
\end{proposition}
\begin{exercise}
Prove Proposition \ref{prop:alt}.
\end{exercise}
Thus, classical Goppa codes are alternant codes, i.e., subfield subcodes of particular GRS codes, where the weights $\beta_i$ are the inverses of the evaluations $g(\alpha_i)$, for a polynomial $g$. 

\subsubsection{Cyclic Codes}
Another important family of codes is that of cyclic codes. They can be represented through only one vector.
Let $c=(c_1, \ldots, c_n) \in \mathbb{F}_q^n$, then we denote by $\sigma(c)$ its \emph{cyclic shift}, i.e.,
$$\sigma(c_1, \ldots, c_n) = (c_n, c_1, \ldots, c_{n-1}).$$
We call a code cyclic, if the cyclic shift of any codeword is also a codeword.
\begin{definition}[Cyclic Code] %
Let $\mathcal{C}$ be an $[n,k]$ linear code over $\mathbb{F}_{q}$. We say that $\mathcal{C}$ is \emph{cyclic} if $\sigma(\mathcal{C})= \mathcal{C}.$
\end{definition}
\begin{proposition}\label{prop:RScyc} %
Let $k \leq n =q-1$ be positive integers and let $\alpha \in \mathbb{F}_q^n$ be such that $\alpha_i=\gamma^{i-1}$, for $i \in\{1, \ldots, n\}$ and $\gamma$ a primitive element in $\mathbb{F}_q$. Then $\text{RS}_{n,k}(\alpha)$ is a cyclic code.
\end{proposition}
\begin{exercise}
Prove Proposition \ref{prop:RScyc}.
\end{exercise}
Note that any polynomial  $c(x) =\sum_{i=0}^{n-1} c_i x^i \in \mathbb{F}_q[x]$ of degree (at most) $n-1$ corresponds naturally to a vector $c =(c_0, \ldots, c_{n-1}) \in \mathbb{F}_q^n$. 
\begin{proposition}\label{prop:cyc} %
Cyclic codes over $\mathbb{F}_q$ of length $n$ correspond to ideals of $\mathbb{F}_q[x]/(x^n-1)$.
\end{proposition}
\begin{exercise}
Prove Proposition \ref{prop:cyc} using the map 
\begin{align*} \varphi: \mathbb{F}_{q}[x]/(x^n-1) & \to \mathbb{F}_q^n, \\ 
c(x) & \mapsto (c_0, \ldots, c_{n-1}). \end{align*} In particular, what is $\varphi(x \cdot c(x))?$
\end{exercise}
Since we can see cyclic codes as ideals in $\mathbb{F}_q[x]/(x^n-1)$, we can also consider the generator polynomial of a cyclic code. 
\begin{definition}[Generator Polynomial] The \emph{generator polynomial} of a cyclic code $\mathcal{C} \subset \mathbb{F}_q^n$ is the  unique monic generator of minimal degree of the corresponding ideal in $\mathbb{F}_q[x]/(x^n-1)$. 
\end{definition}
\begin{proposition}\label{prop:genpoly}
Let $\mathcal{C}$ be a cyclic code over $\mathbb{F}_q$ of length $n$ with generator polynomial $g(x) = \sum_{i=0}^{r} g_ix^i$, where $r$ is the degree of $g$. Then \begin{enumerate}
    \item $g(x) \mid x^n-1$.
    \item $\mathcal{C}$ has dimension $n-r.$
    \item A generator matrix $G \in \mathbb{F}_q^{(n-r) \times n}$ of $\mathcal{C}$ is given by 
    $$G= \begin{pmatrix}g_0 & \cdots & g_r & & \\
    & \ddots & & \ddots & \\
    & & g_0 & \cdots & g_r \end{pmatrix}.$$
    \item Let $h(x)$ be such that $g(x)h(x)= x^n-1$, then $\langle g(x) \rangle^{\perp} = \langle h(x) \rangle.$
\end{enumerate}
\end{proposition}
\begin{exercise}
Prove Proposition \ref{prop:genpoly}. 
\end{exercise}
\begin{exercise}
How many cyclic codes over $\F_3$ of length $4$ exist?
\end{exercise}
Note that the generator matrix in Proposition \ref{prop:genpoly} is in a special form, such a matrix is called a \emph{circulant matrix}. 
\begin{exercise}
Give the generator polynomial of $\text{RS}_{n,k}(\alpha).$
\end{exercise}
\begin{exercise}
Let us consider the code $\mC$ over $\F_3$ generated by
$$\bG= \begin{pmatrix} 1 & 0 & 1 & 0 \\ 0 & 1 & 0 & 1
\end{pmatrix}.$$
\begin{enumerate}
    \item Show that $\mC$ is cyclic.
    \item Find the generator polynomial of $\mC$.
    \item Find the generator polynomial of $\mC^\perp.$
\end{enumerate}
\end{exercise}

Finally, since we know how to compute the polynomial product $u(x) \cdot v(x) \in \mathbb{F}_q[x]/(x^n-1)$, we can define a new vector multiplication in $\mathbb{F}_q^n$. 
\begin{definition}[Rotation Matrix]
    Let $\bu,\bv \in \mathbb{F}_q^n$ and define the \emph{rotation matrix} as $$\text{rot}(\bu)= \begin{pmatrix} \bu \\ \sigma(\bu) \\ \vdots \\ \sigma^{n-1}(\bu) \end{pmatrix}.$$
Let us denote by $\bu \bv = \bu \text{rot}(\bv).$    
\end{definition}
~
\medskip

\begin{exercise} 
\begin{enumerate}
    \item Show that $\varphi(\bu \bv)= u(x)v(x)$.
    \item Show that $\bu\bv=\bv\bu.$
\end{enumerate}
   \end{exercise}

Finally, we introduce quasi-cyclic codes. For $\bx  =(x_1, \ldots, x_n)\in \F_q^n$ and some $\ell \in \{1, \ldots, n\}$ we denote by $\sigma_\ell(x)$ its \emph{$\ell$-cyclic shift}, i.e.,
$$\sigma_\ell(\bx) = (x_{1 + \ell}, \ldots, x_{n + \ell}),$$
where the indices $i +\ell$ should be considered modulo $n$. 
\begin{definition} 
An $[n,k]$ linear code $\mC$ is a \emph{quasi-cyclic} (QC) code, if there exists $\ell \in \mathbb{N}$, such that $\sigma_\ell(\mC)=\mC.$
\end{definition}
In addition, if $n=\ell a$, for some $a \in \mathbb{N},$ then it is convenient to write the generator matrix composed into $a \times a$ circulant matrices. 

   \subsubsection{LDPC Codes}
Another interesting family of codes for cryptography are the \emph{low-density parity-check} (LDPC) codes introduced by Gallager \cite{gallager}. The idea of LDPC codes is to have a parity-check matrix that is sparse. These codes are usually defined over the binary, although they can be generalized to arbitrary finite fields \cite{davey}, for the applications in cryptography the binary LDPC codes suffice.
In order to define LDPC codes we introduce the notation of \emph{row-weight}, respectively \emph{column-weight} of a matrix, which refers to the  Hamming weight of each row, respectively of each column. Thus, a matrix having row-weight $w$, asks for each row to have Hamming weight $w$.
Classically LDPC codes are defined as follows.
\begin{definition}
Let $\lambda, \rho \in \mathbb{N}.$ An $[n,k]$ linear code $\mC$
 over $\F_2$ is called a  \emph{$(\lambda, \rho)$-regular LDPC code}, if there exists a parity-check matrix $\bH \in \F_2^{(n-k) \times n}$ of $\mathcal{C}$ which has column-weight $\lambda$ and row-weight $\rho.$
 \end{definition}
A more common definition for cryptographic applications reads as follows.
\begin{definition}
Let $w \in \mathbb{N}$ be a constant. An $[n,k]$ linear code $\mC$ over $\F_2$ is called a  \emph{$w$-low-density parity-check code}, if there exists a parity-check matrix $\bH \in \F_2^{(n-k) \times n}$  of $\mathcal{C}$ having row-weight $w$. 
\end{definition}
\begin{exercise}
Show that the rate of an  $(\lambda, \rho)$-regular LDPC code is given by $1- \lambda/\rho.$ 
\end{exercise}
For a parity-check matrix $\bH \in \F_2^{(n-k) \times n}$ and a received vector $\bx \in \F_2^n$  we call the $(n-k)$ equations derived  from $\bH \bx^\top$  \emph{parity-checks}, i.e.,  $$\sum_{j=1}^n h_{ij} x_j $$ for all $i \in \{1, \ldots, n-k\}$. We say that a parity-check  is \emph{satisfied} if 
$$\sum_{j=1}^n h_{ij} x_j =0,$$ and else call it \emph{unsatisfied}. 
\medskip

LDPC codes are interesting from a coding-theoretic point of view, as they (essentially) achieve Shannon capacity in a practical way. From a cryptographic stand point, these codes are interesting as they have no algebraic structure, which might be detected by an attacker, but nevertheless have an efficient decoding algorithm.

One decoding algorithm  dates back to Gallager \cite{gallager} and is called Bit-Flipping algorithm. There have been many improvements (e.g. \cite{bf1, bf2,bf3})). The algorithm is iterative and its error correction capability increases with the code length. 
 The idea of the Bit-Flipping algorithm is that at each iteration the  number of unsatisfied parity-check equations associated to each bit of the received vector is computed. Each bit which has more than $b \in \mathbb{N}$
 (some threshold parameter) unsatisfied parity-check equations is flipped and the syndrome is updated accordingly. This process is repeated until either the syndrome becomes $\bz$, or until a maximal number of iteration $M \in \mathbb{N}$ is reached. In the later case we have a \emph{decoding failure}.
 The complexity of this algorithm is thus given by $\mathcal{O}(nwN)$, where $w$  is the row-weight of the parity-check matrix and $N$ is the average number of iterations.
 \medskip
 
 One can also relax the condition on the row-weight of LDPC codes, to get moderate-density parity-check (MDPC) codes \cite{mdpcp}.
 \begin{definition}\label{def:mdpc}
An $[n,k]$ linear code $\mC$ over $\F_2$ is called a  \emph{moderate-density parity-check code}, if there exists a parity-check matrix $\bH \in \F_2^{(n-k) \times n}$  having row-weight $\mathcal{O}(\sqrt{n \log(n)})$. 
 \end{definition}
 Thus, the only difference to LDPC codes is that we allow a larger row-weight in the parity-check matrix (for LDPC codes $w$ was chosen constant in $n$). This might however lead to an increase of iterations within the Bit-Flipping algorithm
 and decoding failures become increasingly likely.

\subsubsection{Reed-Muller Codes}
 
Next, we introduce a class of codes, the Reed-Muller codes, introduced in \cite{muller} in 1954. They are, similarly to Reed-Solomon codes, constructed as the evaluation of polynomials. While Reed-Solomon codes only consider polynomials in one variable, Reed-Muller codes use multivariate polynomials. For this part we follow \cite{guruswamiRM}.

Let $p$ be a prime, $q = p^n$ and $m,r$ be positive integers. Denote with $\mathbb{F}_q [x_1, \ldots, x_m]_{\leq r}$ the $\mathbb{F}_q$-vector space of polynomials in $m$ variables of degree at most $r$ and fix an order $\{ \alpha_1, \alpha_2, \ldots, \alpha_{q^m} \}$ of $\mathbb{F}_q^m$. 
\begin{definition}\label{def:reedmuller}
The \emph{Reed-Muller} code $\text{RM}_q(m,r)$ over $\mathbb{F}_q$ is defined as the image of the evaluation map
\begin{align*}
ev: \mathbb{F}_q [x_1, \ldots, x_m]_{\leq r}  &\to \mathbb{F}_q^{q^m}, \\
f &\mapsto (f(\alpha_1), f(\alpha_2), \ldots, f(\alpha_{q^m})).
\end{align*}
\end{definition}

We will note that there exist efficient decoding algorithms for Reed-Muller codes, the first efficient decoding algorithm was published in \cite{reeddecoding}.
\medskip

 For the case $q=2$, we can compute dimension and minimum distance of $\text{RM}_q(m,r)$.

 \begin{proposition}
 \label{RMdim}
 Let $r \leq m$. Then $\dim_{\mathbb{F}_2}(\text{RM}_2(m,r)) = \sum_{i=0}^r {m \choose i}$. 
 \end{proposition}
  \begin{proposition}
 Let $r \leq m$. The minimum distance of $RM_2 (m,r)$ is $2^{m-r}$.
 \end{proposition}

\subsubsection{Concatenated Codes}
Concatenated codes were first introduced by Forney \cite{forney}, and use the basic idea of a double encoding process through two codes.
\begin{definition}\label{def:concat}
Let $\mC_1$ be an $[n_1,k_1]$ linear code of minimum distance $d_1$ over $\mathbb{F}_q$, called \emph{inner code} and $\mC_2$ be an $[n_2,k_2]$ linear code of minimum distance $d_2$ over $\mathbb{F}_{q^{k_1}}$, called \emph{outer code}. Then, the \emph{concatenated} code $\mC= \mC_2 \circ \mC_1$ is an $[n_1n_2, k_1k_2]$ linear code over $\mathbb{F}_q$ of minimum distance at least $d_1d_2$.
\end{definition}
 The codewords of $\mC$ are built as follows: for any $\bu \in \mathbb{F}_{q^{k_1}}^{k_2}$, encode $\bu$ using a generator matrix $\bG_2$ of $\mC_2$, receiving the codeword $((\bu\bG_2)_1, \ldots, (\bu\bG_2)_{n_2})$. Let us denote for $a \in \mathbb{F}_{q^{k_1}}$ by $\overline{a}$ the corresponding vector in $\mathbb{F}_q^{k_1}$ having fixed a basis. As a next step we represent the entries of each codeword as a vector in $\mathbb{F}_q^{k_1}$ and encode them using a generator matrix $\bG_1$ of $\mC_1$. Then, the codewords of $\mC$ are of the form
$$(\overline{(\bu\bG_2)_1} \bG_1, \ldots, \overline{(\bu\bG_2)_{n_2}} \bG_1).$$

\subsubsection{$(U,U+V)$-Codes}

Given two codes $\mathcal{C}_1$ and $\mathcal{C}_2 \subseteq \mathbb{F}_q^n$, we can also construct new codes, for example using the \emph{$(U,U+V)$-construction}.
\begin{definition}
    Let $\mathcal{C}_1,\mathcal{C}_2 \subseteq \mathbb{F}_q^n$ with dimension $k_1,$ respectively $k_2.$ Then, the \emph{$(U,U+V)$-code} of $\mathcal{C}_1, \mathcal{C}_2$ is given by
    $$\mathcal{C}=\{(\bu,\bu+\bv) \mid \bu \in \mathcal{C}_1, \bv \in \mathcal{C}_2\}.$$
\end{definition}

\begin{proposition}\label{prop:uu+v}
  Let $\mathcal{C}_1,\mathcal{C}_2 \subseteq \mathbb{F}_q^n$ with dimension $k_1,$ respectively $k_2$ and minimum Hamming distance $d_1$, respectively $d_2.$ Then, the $(U,U+V)$-code $\mathcal{C} \subseteq \mathbb{F}_q^{2n}$ has dimension $k=k_1+k_2$ and minimum Hamming distance $d=\min\{2d_1,d_2\}.$ 
\end{proposition}
\begin{exercise}
Prove Proposition \ref{prop:uu+v}.  \emph{Hint:} Show first that if $\bG_1,\bG_2$ are generator matrices of $\mathcal{C}_1,$ respectively $\mathcal{C}_2$, then $\bG= \begin{pmatrix} \bG_1 & \bG_1 \\ \bz & \bG_2 \end{pmatrix}$ is a generator matrix of $\mathcal{C}.$ 
\end{exercise}

The encoding of a message $(\bm_1,\bm_2)$ gives then the codeword $(\bm_1\bG_1, \bm_1\bG_1+\bm_2\bG_2)$ and a received word can be assumed of the form $(\br_1, \br_2)=(\bm_1\bG_1+\be_1, \bm_1\bG_1+\bm_2\bG_2+\be_2)$ for some error vector $(\be_1, \be_2).$
Note that a decoder for $\mathcal{C}$ would first  decode $\br_1$ using the decoder of $\mathcal{C}_1$ to get $\bm_1$. One can then take $\bm_1\bG_1$ away from $\br_2$ and then  use the decoder of $\mathcal{C}_2,$ to recover $\bm_2.$ 
\begin{exercise}
Show that the Reed-Muller code $\text{RM}_2(m,r)$ is a $(U,U+V)$-code for the code $\mathcal{C}_1$ being a $\text{RM}_2(m-1,r)$ and $\mathcal{C}_2$ a $\text{RM}_2(m-1,r-1)$ code. 
    
\end{exercise}

\subsubsection{Product Codes}

Similar to concatenation of codes and the $(U,U+V)$-construction, we can also build the tensor product of two codes $\mathcal{C}_1, \mathcal{C}_2$. For a matrix $\bC \in \mathbb{F}_q^{k \times n}$ let us denote by $\bc_i \in \mathbb{F}_q^k$ for $i \in \{1, \ldots, n\}$ the columns of $\bC$, and similarly by $\br_i\in \mathbb{F}_q^n$ for $i \in \{1, \ldots, k\}$ the rows of $\bC.$
\begin{definition}
    Let $\mathcal{C}_1 \subseteq \mathbb{F}_q^{n_1}$ and $\mathcal{C}_2 \subseteq \mathbb{F}_q^{n_2}.$ Then, the \emph{product code} of $\mathcal{C}_1,\mathcal{C}_2$ is defined as
    $$\mathcal{C}=\mathcal{C}_1 \otimes \mathcal{C}_2 = \{\bC \in \mathbb{F}_q^{n_1 \times n_2} \mid \bc_i \in \mathcal{C}_1, \br_j \in \mathcal{C}_2, i \in \{1, \ldots, n_2\}, j \in \{1, \ldots, n_1\}\}.$$
\end{definition}

Let us define the Hamming weight of a matrix $\bA$ to be the number of non-zero entries in $\bA.$
\begin{proposition}\label{prop:prod}
 Let $\mathcal{C}_1 \subseteq \mathbb{F}_q^{n_1}$ and $\mathcal{C}_2 \subseteq \mathbb{F}_q^{n_2}$ of dimension $k_1$, respectively $k_2$ and minimum Hamming distance $d_1$, respectively $d_2$.
    Then, the tensor product code $\mathcal{C}_1 \otimes \mathcal{C}_2 \subseteq \mathbb{F}_q^{n_1 \times n_2}$ has dimension $k_1k_2$ and minimum Hamming distance $d_1d_2.$
\end{proposition}

\begin{exercise}
Show that every codeword of $\mathcal{C}_1 \otimes \mathcal{C}_2$ is given by 
$$\bG_1^\top \bA \bG_2,$$ for $\bG_1 \in \mathbb{F}_q^{k_1 \times n_1}$ a generator matrix of $\mathcal{C}_1,$ $\bG_2 \in \mathbb{F}_q^{k_2 \times n_2}$ a generator matrix of $\mathcal{C}_2$ and a matrix $\bA \in \mathbb{F}_q^{k_1 \times k_2}.$
\end{exercise}

\begin{exercise}
    Prove Proposition \ref{prop:prod}.
\end{exercise}

Note that this is very similar to the definition of concatenated codes, where the resulting code also had length $n_1n_2$ and dimension $k_1k_2$. However, for concatenated codes we only know that $d \geq d_1d_2,$ while for tensor product codes, we know that their minimum distance is exactly $d_1d_2.$

\subsubsection{Rank-Metric Codes}

Until now, we have considered classical coding theory, where the finite field is endowed with the Hamming metric. However, there exist many more metrics, for example the rank metric (introduced in \cite{delsarte,cover,gabidulin}).
In the following we 
introduce \emph{rank-metric codes}, for which we follow the notation of \cite{elisa}.

\begin{definition}[Rank Metric]
Let $\bx,\by \in \mathbb{F}_{q^m}^n$. The \emph{rank weight} of $\bx$ is defined as the dimension of the $\mathbb{F}_q$-vector space generated by its entries, i.e.,
$$\wt_R(\bx)= \dim_{\mathbb{F}_q}\left( \langle x_1, \ldots, x_n \rangle_{\mathbb{F}_q}\right)$$
and the \emph{rank distance} between $\bx$ and $\by$ is given by 
$$d_R(\bx,\by)= \wt_R(\bx-\by).$$
Let $\mathcal{C} \subseteq \mathbb{F}_{q^m}^n$ be a linear code, then its \emph{minimum rank distance} is given by 
$$d_R(\mathcal{C})=\min\{ \text{wt}_R(\bc) \mid \bc \neq \bz, \bc \in \mathcal{C}\}.$$
\end{definition}

The rank support of a vector $\bx \in \mathbb{F}_{q^m}^n$ is often given by 
$$\text{supp}(\bx)=\langle x_1, \ldots, x_n \rangle_{\mathbb{F}_q} \subset \mathbb{F}_{q^m}.$$
We will later see also two different notions of rank support. 

Let $\Gamma=\{\gamma_1, \ldots, \gamma_m\}$ be a basis of $\mathbb{F}_{q^m}$ over $\mathbb{F}_q$.
Using the extension map, i.e.,
\begin{align*} 
\Gamma: \mathbb{F}_{q^m} &\to \mathbb{F}_q^{m \times n}\\
\ba &\mapsto \Gamma(\ba),
\end{align*}
we can see that 
\begin{equation}\label{eq:indp}
    \text{wt}_R(\ba)=\text{rk}(\Gamma(\ba)).
\end{equation}
\begin{exercise}
    Show that Equation \eqref{eq:indp} is independent of the choice of basis $\Gamma.$
\end{exercise}
Thus, the extension map is a $\mathbb{F}_q$-linear isometry.

In fact, we can also endow $\mathbb{F}_q^{m \times n}$ with the rank metric.

\begin{definition}[Rank Metric]
Let $\bA,\bB \in \mathbb{F}_q^{n \times m}.$ The \emph{rank weight} of $\bA$ is given by the rank of $\bA$, denoted by $\rk(\bA)$ and the rank distance between $\bA$ and $\bB$ is given by 
$$d_R(\bA,\bB) =\rk(\bA-\bB).$$
Let $\mathcal{C} \subseteq \mathbb{F}_q^{m \times n}$ be a linear matrix code, then its \emph{minimum rank distance} is given by
$$d_R(\mathcal{C})=\min\{\text{rk}(\bC) \mid \bC \neq \bz, \bC \in \mathcal{C}\}.$$
\end{definition}

Recall, that for $\mathcal{C} \subseteq \mathbb{F}_{q^m}^n$ we defined the matrix code associated to $\Gamma$ as 
$$\Gamma(\mathcal{C})= \{ \Gamma(\bc) \mid \bc \in \mathcal{C}\} \subseteq \mathbb{F}_q^{m \times n}.$$

\begin{proposition}
Let $\Gamma$ be a basis of $\mathbb{F}_{q^m}$ over $\mathbb{F}_q.$
    Let $\mathcal{C} \subseteq \mathbb{F}_{q^m}^n$ be a linear code of dimension $k$ and minimum rank distance $d_R$, then the associated matrix code $\Gamma(\mathcal{C}) \subseteq \mathbb{F}_q^{m \times n}$ is a matrix code of dimension $km$ and minimum rank distance $d_R.$
\end{proposition}
Thus, using the extension map any $\mathbb{F}_{q^m}$-linear code can also be seen as $\mathbb{F}_q$-linear code, however, the opposite is not true. 

\begin{example}
    Let us consider 
    $\mathbb{F}_4=\mathbb{F}_2[\alpha]$ and $\alpha^2=\alpha+1$, and $\Gamma=\{1,\alpha\}$. 
    The code $\mathcal{C}=\langle(1,\alpha)\rangle \subseteq \mathbb{F}_4^2$ has dimension 1 and minimum rank distance 2.  
    Then $$\mathcal{C}=\left\{(0,0),(1,\alpha),(\alpha,\alpha+1)\right\}$$ and $$\Gamma(\mathcal{C})=\left\{\begin{pmatrix} 0 & 0 \\ 0 & 0 \end{pmatrix}, \begin{pmatrix} 1 & 0 \\ 0 & 1 \end{pmatrix}, \begin{pmatrix} 0 & 1 \\ 1 & 1 \end{pmatrix}, \begin{pmatrix}
        1 & 1 \\ 1 & 0
    \end{pmatrix}\right\}.$$
The code $\Gamma(\mathcal{C}) \subseteq \mathbb{F}_2^{2 \times 2}$ has dimension $2$ and minimum rank distance 2.
However, consider $\mathcal{C}' = \left\langle \begin{pmatrix} 1 & 0 \\ 1 & 1 \end{pmatrix}, \begin{pmatrix}
    0 & 1 \\ 1 & 1 
\end{pmatrix} \right\rangle \subseteq \mathbb{F}_2^{2 \times 2}$ has dimension 2 and minimum rank distance 1.
We have $$\mathcal{C}'=\left\{ \begin{pmatrix} 0 & 0 \\ 0 & 0 \end{pmatrix}, \begin{pmatrix} 1 & 0 \\ 1 & 1 \end{pmatrix}, \begin{pmatrix}
    0 & 1 \\ 1 & 1 
\end{pmatrix}, \begin{pmatrix}
    1 & 1 \\ 0 & 0
\end{pmatrix}\right\}$$ and $$\Gamma^{-1}(\mathcal{C}') = \{ (0,0),(1+\alpha, \alpha), (\alpha,\alpha+1), (1, 1)\} \subseteq \mathbb{F}_4^2.$$ This subset of vectors is not a $\mathbb{F}_4$-linear code as for example $\alpha (1,1) = (\alpha,\alpha) \not\in \Gamma^{-1}(\mathcal{C}').$
\end{example}

\begin{definition}
    The \emph{rank-metric ball} of radius $r$ is defined as $$B_R(r,n,m,q)=\{\bx \in \mathbb{F}_{q^m}^n \mid \text{wt}_R(\bx) \leq r\}.$$
\end{definition}

\begin{proposition}
    The size of the rank-metric ball is approximately
    $$|B_R(r,n,m,q)| \sim q^{r(n+m-r+1)},$$ for large $n,m$.
\end{proposition}

Given a vector $\bx \in \mathbb{F}_{q^m}^n$ of rank weight $t$, we can split the vector into
$$\bx= \bc \bR,$$ for $\bc \in \mathbb{F}_{q^m}^t$ and the entries $c_i$ are $\mathbb{F}_q$-linearly independent, and $\bR \in \mathbb{F}_q^{t \times n}$ of rank $t.$

\begin{definition}
    The \emph{column support} of a vector $\bx \in \mathbb{F}_{q^m}^n$ of rank weight $t$, with splitting $\bc\bR$, is given by 
    $$\text{supp}_C(\bx) = \langle \Gamma(\bc)^\top \rangle \subseteq \mathbb{F}_q^m$$ and has dimension $t$.
    
    The \emph{row support} of a vector $\bx \in \mathbb{F}_{q^m}^n$ of rank weight $t$ and splitting $\bc\bR$ is given by 
    $$\text{supp}_R(\bx)=\langle \bR \rangle \subseteq \mathbb{F}_q^n$$ and has dimension $t.$
\end{definition}
\begin{exercise}
    Show that the definition of row and column support are independent of the choice of splitting.
\end{exercise}

Recall that in the Hamming metric the support of $x \in \mathbb{F}_{q^m}^n$ is defined as the indices of non-zero entries of $\bx$, i.e.,
$$\text{supp}_H(\bx) = \{ i \in \{1, \ldots, n\} \mid x_i \neq 0\},$$ and the Hamming weight coincides with its size, i.e., $$\text{wt}_H(\bx) = |\text{supp}_H(\bx)|.$$
For the rank metric, whether we choose the row or column support, the rank weight of $\bx$ coincides with the dimension of the support, i.e.,
$$\text{wt}_R(\bx) = \dim(\text{supp}_R(\bx))=\dim(\text{supp}_C(\bx)).$$

For a vector $\bx \in \mathbb{F}_{q^m}^n$ of Hamming weight $t$ there are $\binom{n}{t}$ many possible Hamming supports of $\bx,$ whereas if the rank weight is $t$, there are $\gb{n}{t}_q$, respectively $\gb{m}{t}_q$ many possible row supports, respectively column supports. 

 \medskip

With the minimum rank distance we can also state a Singleton bound \cite{delsarte}:
\begin{theorem}[$\mathbb{F}_{q}$-linear Rank-Metric Singleton Bound]
Let $\mathcal{C} \subset \mathbb{F}_q^{n \times m}$ be a matrix code of dimension $k$ with minimum rank distance $d_R(\mathcal{C}).$ Then $$ k \leq \max\{n,m\} (\min\{n,m\} -d_R(\mathcal{C})+1).$$
\end{theorem}
 
\begin{theorem}[$\mathbb{F}_{q^m}$-linear Rank-Metric Singleton Bound]
Let $\mathcal{C} \subset \mathbb{F}_{q^m}^{n}$ be a linear code of dimension $k$ with minimum rank distance $d_R(\mathcal{C}).$ Then $$ k \leq n-d_R(\mathcal{C})+1.$$
\end{theorem}
Codes achieving these bounds are called Maximum Rank Distance (MRD) codes. 

Note that MDS codes have density 1 for $q$ going to infinity, and  density 0 for $n$ going to infinity. Similar results hold also for the rank metric: $\mathbb{F}_{q^m}$-linear MRD codes are dense for $q$ going to infinity by \cite{ale} and since $d_R(\mC)$ is bounded by $m$, have density 0 for $n$ going to infinity. It was shown in \cite{anina} that $\mathbb{F}_q$-linear MRD codes 
are sparse for all parameter sets as the field grows, with only very few exceptions. 
Unlike in the Hamming metric, we know that $\mathbb{F}_{q^m}$-linear MRD codes exist for any set of parameters (with $n \leq m)$, by the seminal work of Delsarte \cite{delsarte} and Gabidulin \cite{gabidulin}.

We also have a rank-analogue of the Gilbert-Varshamov bound, \cite{gad}. Let us denote by $A_R(n,d,m,q)$ the maximal size of a code in $\mathbb{F}_{q^m}^n$ having minimum rank distance $d.$
\begin{theorem}[Gilbert-Varshamov Bound in the Rank Metric]
Let $q$ be a prime power and $m,n,d$ be positive integers. Then,
$$A_R(n,d,m,q) \geq  \frac{q^{mn}}{|B_R(d-1,n,m,q)|}.$$
\end{theorem}
We can also give the asymptotic version of this bound, for which we first define the \emph{relative minimum rank distance} to be $\delta = d_\textnormal{R}(\mathcal{C})/n$ and when considering the extension degree $m$ as function in $n$, we can define $M=\lim_{n \to \infty} m(n)/n.$ Then, the rank-metric Gilbert-Varshamov bound states, that 
 $$\overline{R}(\delta) = \limsup\limits_{n \to \infty} \frac{1}{n} \log_{q^m} A_R(n, \delta n,m,q) \geq (1-\delta)(1-M).$$
As in the Hamming metric, we know by \cite{loidreau} that random codes attain the Gilbert-Varshamov bound with high probability. 

\begin{proposition}
Let $\mathcal{C} \subseteq \mathbb{F}_{q^m}^n$ be a random linear code of dimension $k$. For $n$ large enough, we have that $\mathcal{C}$ has the relative minimum distance $$\delta = d_\textnormal{R}/n = M/2+1/2 -\sqrt{RM +(M-1)^2/4}$$ with high probability.
\end{proposition}
Interestingly, this bound does not depend on the field size $q$, which is in contrast to its Hamming-metric counterpart. In particular, if $M=1$, which will often be the case for applications, we get $\delta= 1-\sqrt{R}.$

\subsubsection{Gabidulin Code}
  
In order to introduce the classical Gabidulin codes let us first recall the basics of $q$-polynomials.

A $q$-polynomial or linearized polynomial $f$ of $q$-degree $d$ over $\mathbb{F}_{q^m}$ is a polynomial of the form 
$$f(x)= \sum\limits_{i=0}^d f_i x^{q^i}.$$
Let us denote by $P_{\ell}$ the $q$-polynomials of $q$-degree up to $\ell$ over $\mathbb{F}_{q^m}.$

The classical Gabidulin code can now be defined in a similar fashion as the Reed-Solomon code, i.e., as evaluation code.
\begin{definition}[Classical Gabidulin Code]
Let $g_1, \ldots, g_n \in \mathbb{F}_{q^m}$ be linearly independent over $\mathbb{F}_q$  and let $k \leq n \leq m$. The classical Gabidulin code $\mathcal{C} \subset \mathbb{F}_{q^m}^n$ of dimension $k$ is defined as 
$$\mathcal{C} = \{ (f(g_1), \ldots, f(g_n) ) \mid f \in P_{k-1} \}.$$
\end{definition}

\begin{exercise}
Show that classical Gabidulin codes are $\mathbb{F}_{q^m}$-linear MRD codes, by taking a non-zero codeword $c=(f(g_1), \ldots, f(g_n))$ and considering the $\mathbb{F}_q$-dimension of the kernel of the $q$-polynomial $f$.
\end{exercise}
In order to introduce the generalized Gabidulin codes, we first have to define the rank analog of the Vandermonde matrix, i.e., the Moore matrix \cite{moore}.
\begin{definition}[Moore Matrix]
Let $(v_1, \ldots, v_n) \in \mathbb{F}_{q^m}^n$ and $v_i$ are $\mathbb{F}_q$-linearly independent. We denote  by $$M_{s,k}(v_1, \ldots, v_n) \in \mathbb{F}_{q^m}^{k \times n}$$ the $s$-Moore matrix:
$$M_{s,k}(v_1, \ldots, v_n) = \begin{pmatrix} v_1 & \cdots & v_n \\ v_1^{[s]} & \cdots & v_n^{[s]} \\ 
\vdots & & \vdots \\ v_1^{[s(k-1)]} & \cdots & v_n^{[s(k-1)]} \end{pmatrix},$$
where $[i] = q^i.$
\end{definition}

The definition of Gabidulin codes can also be generalized, e.g. \cite{gab2}:
\begin{definition}[Generalized Gabidulin Code]
Let $g_1, \ldots, g_n \in \mathbb{F}_{q^m}$ be linearly independent over $\mathbb{F}_q$ and let $s$ be coprime to $m$. The generalized Gabidulin code $\mathcal{C} \subset \mathbb{F}_{q^m}^n$ of dimension $k$ is defined as the rowspan of $M_{s,k}(g_1, \ldots, g_s).$
\end{definition}
For $s=1$, we can see that this coincides with the  classical Gabidulin codes, which have the generator matrix
$$M_{1,k}(g_1, \ldots, g_n) = \begin{pmatrix} g_1 & \cdots & g_n \\ g_1^{q} & \cdots & g_n^{q} \\ 
\vdots & & \vdots \\ g_1^{q^{k-1}} & \cdots & g_n^{q^{k-1}} \end{pmatrix}.$$

Since the Moore matrix can be seen as a rank analog of a Vandermonde matrix, a generalized Gabidulin code can be seen as a rank analog of a generalized Reed-Solomon code.

\begin{theorem}
The generalized Gabidulin code $\mathcal{C} \subset \mathbb{F}_{q^m}^n$ of dimension $k$ is a $\mathbb{F}_{q^m}$-linear MRD code.
\end{theorem}
In addition, as in the Hamming metric we have nice duality results.
\begin{proposition} 
Let $ \mathcal{C} \subset \mathbb{F}_{q^m}^n$ be a $k$ dimensional generalized  Gabidulin code, then $\mathcal{C}^\perp \subset \mathbb{F}_{q^m}^n$ is a $n-k$ dimensional generalized Gabidulin code. 
\end{proposition}

This duality result holds (as in the Hamming metric) also more in general; for all $\mathbb{F}_{q^m}$-linear MRD codes.

\begin{proposition}
Let $ \mathcal{C} \subset \mathbb{F}_{q^m}^n$ be a $k$-dimensional $\mathbb{F}_{q^m}$-linear MRD code, then $\mathcal{C}^\perp \subset \mathbb{F}_{q^m}^n$ is a $(n-k)$-dimensional $\mathbb{F}_{q^m}$-linear MRD code. 
\end{proposition}

The classical Gabidulin code has been the first rank-metric code introduced into code-based cryptography  in \cite{gpt}, which is known as the GPT system.

\subsubsection{LRPC Codes}
Other classes of rank-metric codes that are used in code-based cryptography are the rank analogues  of LDPC and MDPC codes, first defined in \cite{lrpc}. Instead of asking for a low (respectively moderate) number of non-zero entries within each row of the parity-check matrix, one now has to consider the $\mathbb{F}_q$-subspace  generated by the coefficients of the parity-check matrix.

\begin{definition}[Low Rank Parity-Check Code (LRPC)]
Let $\bH \in \mathbb{F}_{q^m}^{(n-k) \times n}$ be a full rank matrix, such that its coefficients $h_{i,j}$ generate an $\mathbb{F}_q$-subspace $F$ of small dimension $d$,
$$F= \langle (h_{i,j})_{i,j} \rangle_{\mathbb{F}_q}.$$
The code $\mathcal{C} \subset \mathbb{F}_{q^m}^n$ having parity-check matrix $\bH$ is called a \emph{Low Rank Parity-Check} (LRPC) code of dual weight $d$ and support $F$.
\end{definition}

\subsubsection{Code Equivalence}

For the newer problems used in code-based cryptography, we will also need the notion of code equivalence. 

\begin{definition}[Isometry]
Let us consider the space $V$ endowed with the distance $d$.
    A linear map $\varphi: V \to V$ is called \emph{isometry} if it keeps the distance invariant. That is, for all $\bx,\by \in V$ we have $d(\bx,\by)=d(\varphi(\bx), \varphi(\by)).$
\end{definition}

Let us denote the set of all isometries for a fixed distance $d$ by $I_d.$

\begin{proposition}
The linear isometries of the Hamming metric in $V= \mathbb{F}_{q}^n$ consist of monomial transformations and automorphisms on $\mathbb{F}_{q}.$
\end{proposition}

For cryptography, we mainly focus on a subset of the Hamming-metric isometries, namely the monomial transformations $M_{n,q}= (\mathbb{F}_{q}^\star)^n \rtimes S_n$. Any map $\varphi \in M_{n,q}$ can be seen as a matrix $\bM=\bP\bD$, where $\bP$ is a $n \times n$ permutation matrix and $\mathbf D= \text{diag}(\bv)$ for $\bv \in (\mathbb{F}_{q}^\star)^n$ is a diagonal matrix. 

\begin{proposition}
    The linear isometries of the rank metric in $V= \mathbb{F}_q^{m\times n}$ for $m \leq n$, are given by  $\text{GL}_m(q) \rtimes \text{GL}_n(q)$ and automorphisms of $\mathbb{F}_q$.
\end{proposition}
For applications in cryptography, we again only focus on $\varphi \in \text{GL}_m(q) \rtimes \text{GL}_n(q).$

\begin{definition}[Code Equivalence]
Let us consider $V$ endowed with the distance $d$.
    Let $\mathcal{C}_1, \mathcal{C}_2 \subseteq V$ be linear codes. We say $\mathcal{C}_1$ is equivalent to $\mathcal{C}_2$, if there exists $\varphi \in I_d$ such that $\varphi(\mathcal{C}_1)=\varphi(\mathcal{C}_2).$
\end{definition}

Since $I_H=  (\mathbb{F}_{q}^\star)^n \rtimes (\text{Aut}(\mathbb{F}_q) \times S_n)$, we get two subclasses of code equivalence in the Hamming metric.

In the lightest version, we have the \emph{permutation equivalence}.
\begin{definition}[Permutation Equivalence]
  We say that two codes $\mC_1,\mC_2 \subseteq \mathbb{F}_q^n$ are \emph{permutation equivalent}, if there exists a permutation of indices, which transforms $\mC_1$ into $\mC_2$, that is there exists $\sigma \in S_n$, such that $\sigma(\mC_1)=\mC_2$. 
\end{definition}

When considering any monomial transformation, we get the \emph{linear equivalence.}
\begin{definition}[Linear Equivalence]
    We say that two codes $\mC_1,\mC_2 \subseteq \mathbb{F}_q^n$ are \emph{linear equivalent}, if there exists a map  $\varphi \in (\mathbb{F}_q^\star)^n \rtimes S_n$, such that $\varphi(\mC_1)=\mC_2$. 
\end{definition}
Clearly, permutation equivalent codes are also linear equivalent codes. 

\begin{exercise}
    Consider the code $\mathcal{C}_1\subseteq \mathbb{F}_3^3$ generated by $\bG_1=\begin{pmatrix}
        1 & 0 & 2 \\ 0 & 1 & 1
    \end{pmatrix}$ and the code  $\mathcal{C}_2 \subseteq \mathbb{F}_3^3$ generated by $\bG_2=\begin{pmatrix}
        1 & 0 & 1 \\ 0 & 1 & 0
    \end{pmatrix}$.
    Are the two codes linear equivalent, permutation equivalent or not equivalent?
\end{exercise}

\begin{proposition}
    If $\mathcal{C}_1, \mathcal{C}_2 \subseteq \mathbb{F}_q^n$ are permutation equivalent codes, then for any generator matrix $\bG_1$ of $\mathcal{C}_1$ and $\bG_2$ of $\mathcal{C}_2$, there exists a $n\times n$ permutation matrix $\bP$  and a matrix $\bS \in \GL_k(q)$ such that $$\bS \bG_1 \bP = \bG_2.$$
    If $\mathcal{C}_1, \mathcal{C}_2$ are linear equivalent codes, then for any generator matrix $\bG_1$ of $\mathcal{C}_1$ and $\bG_2$ of $\mathcal{C}_2$, there exists a $n\times n$ permutation matrix $\bP$, a diagonal matrix $\text{diag}(\bv)$ for $\bv \in (\mathbb{F}_q^\star)^n$ and a matrix $\bS \in \GL_k(q)$  such that $$ \bS \bG_1 \bP \text{diag}(\bv)= \bG_2.$$
\end{proposition}

\begin{exercise}
Let $\mathcal{C}_1, \mathcal{C}_2$ be linear equivalent codes. Show that $\mathcal{C}_1^\perp$ is linear equivalent to $\mathcal{C}_2^\perp.$
    \emph{Hint:} Use the fact that $\bG_1 \bH_1^\top= \bz$ and $\bS \bG_1 \bP \text{diag}(\bv) = \bG_2$.
\end{exercise}

Note that linear equivalent codes have the same minimum distance. Even more is true.

\begin{definition}[Weight Enumerator]
Let $\mathcal{C} \subseteq \mathbb{F}_q^n$ be a linear code.
    For any $w \in \{1, \ldots, n\}$, let us denote by $A_w(\mathcal{C}) = |\{ \bc \in \mathcal{C} \mid \text{wt}_H(\bc) = w\}|$ the \emph{weight enumerator} of $\mathcal{C}.$
\end{definition}

\begin{proposition}
    Let $\mathcal{C}_1,\mathcal{C}_2 \subseteq \mathbb{F}_q^n$ be linear equivalent codes, then for all $w \in \{1, \ldots, n\}$ we have that $$A_w(\mathcal{C}_1)=A_w(\mathcal{C}_2).$$
\end{proposition}

\begin{definition}[Hull]
Let $\mathcal{C} \subseteq \mathbb{F}_q^n$ be a linear code. Then the (Euclidean) \emph{hull} of $\mathcal{C}$ is given by $$\mathcal{H}(\mathcal{C})= \mathcal{C} \cap \mathcal{C}^\perp.$$
\end{definition}
In \cite{sendrier} it was shown, that the hull of a random code is with high probability trivial, i.e., $\mathcal{C} \cap \mathcal{C}^\perp =\{\bz\}.$
\begin{proposition}
    Let $\mathcal{C} \subseteq \mathbb{F}_q^n$ be chosen uniform at random. Then, w.h.p. $\mathcal{H}(\mathcal{C})=\{\mathbf{0}\}.$
\end{proposition}
\begin{proof}
    For any $\bc \in \mathcal{H}(\mathcal{C}),$ we have that $\bc \in \mathcal{C},$ thus there exists $\bm \in \mathbb{F}_q^k$ such that $\bm \bG=\bc$. We also have that $\bc \in \mathcal{C}^\perp,$ hence $\bc\bG^\top=\mathbf{0}.$ 
    Thus, $\bm (\bG\bG^\top)=\mathbf{0}$ and counting the number of $\bc \in \mathcal{H}(\mathcal{C})$ is equivalent to counting $\bm \in \mathbb{F}_q^k$ with  $\bm (\bG\bG^\top)=\mathbf{0}$. Since $\bG \in \mathbb{F}_q^{k \times n}$ is a random matrix, also $\bG\bG^\top \in \mathbb{F}_q^{k \times k}$ is random and has with probability 
    $$\prod_{i=1}^k (1-q^{-i})$$ full rank. 
    Due to the rank nullity theorem, we have 
    $$\text{dim}(\text{ker}(\bG\bG^\top))=k-\rk(\bG\bG^\top)=0,$$ w.h.p.
\end{proof}

Another way to prove this, is to note that any $\bc \in \mathcal{H}(\mathcal{C})$ must satisfy 
$$\begin{pmatrix}
    \bG \\ \bH \end{pmatrix} \bc^\top = \mathbf{0}.
 $$
 Note that $\langle \begin{pmatrix}
    \bG \\ \bH \end{pmatrix} \rangle = \mathcal{C} + \mathcal{C}^\perp,$ which is the smallest code containing $\mathcal{C}(\mathcal{C}).$
    
 Again, we are interested in the dimension of the kernel of this matrix, and due to the rank-nullity theorem in its rank. 
 We can assume that $\bG, \bH$ are in systematic form and perform row operations to get
$$ \begin{pmatrix} \bG' \\ \bH' \end{pmatrix} = \begin{pmatrix} \text{Id}_k & \bA \\ \mathbf{0} & \bA\bA^\top+ \text{Id}_{n-k} \end{pmatrix}.$$
Hence its rank is given by $k+ \rk(\bA\bA^\top+ \text{Id}_{n-k}).$ Assuming $\bG$ was a random matrix, we also have that $\bA$ and $\bA\bA^\top+ \text{Id}_{n-k}$ are random matrices, which have with high probability full rank.

\begin{definition}[Automorphism Group]
Let $\mathcal{C} \subseteq \mathbb{F}_q^n$ be  a linear code. The \emph{automorphism group} of $\mathcal{C}$ is given by the linear isometries that map $\mathcal{C}$ to $\mathcal{C}.$
\end{definition}

Note that just like the hull, the automorphism group of a random linear code is w.h.p. trivial \cite{rigid}, i.e., $\text{Aut}(\mathcal{C})=\{\text{id}\}.$

\begin{exercise}
    Give the automorphism group of $\mathcal{C}=\langle (1,0,0), (0,1,1) \subseteq \mathbb{F}_2^3$.
\end{exercise}

\begin{exercise}
    Let $\varphi \in \text{Aut}(\mathcal{C}).$ Show that $\varphi \in \text{Aut}(\mathcal{C}^\perp).$
\end{exercise}

\begin{exercise}
    Let $\varphi \in \text{Aut}(\mathcal{C}).$ Show that $\varphi \in \text{Aut}(\mathcal{C} \cap \mathcal{C}^\perp).$
\end{exercise}

Another invariant of isometries is the support of a code. 

\begin{definition}[Support of a Code]
    Let $\mathcal{C} \subset \mathbb{F}_q^n$ be a linear code. Then we define its support to be 
    $$\text{Supp}(\mathcal{C})= \{i \in \{1, \ldots, n\} \mid \exists \bc \in \mathcal{C}: \bc_i \neq 0\}.$$
\end{definition}
Clearly, for a non-degenerate code, the support will be $\{1, \ldots, n\},$ however, as soon as we go to subcodes of $\mathcal{C},$ this will change. 

\begin{proposition}
Let $\mathcal{C}_1, \mathcal{C}_2 \subset \mathbb{F}_q^n$ be linear equivalent. 
    For any subcode $\mathcal{D}_1 < \mathcal{C}_1$ of dimension $r<k$ and support size $s$ there exists $\mathcal{D}_2 < \mathcal{C}_2$ of dimension $r$ with support size $s$.
\end{proposition}

\begin{definition}[Rank-metric Equivalence]
    Let $\mathcal{C}_1, \mathcal{C}_2 \subseteq \mathbb{F}_q^{m\times n}$. We say that $\mathcal{C}_1$ is equivalent to $\mathcal{C}_2$ if there exists $\varphi \in \text{GL}_m(q) \rtimes \text{GL}_n(q)$ such that $\varphi(\mathcal{C}_1)=\mathcal{C}_2$.
\end{definition}

\subsubsection{Lee Metric Codes}

Let us consider $\mathbb{F}_p$, for $p>3$ a prime. Then we can define a different metric, called \emph{Lee metric.}

\begin{definition}[Lee Metric]
    Let $x \in \mathbb{F}_p$, and represent $x \in \{0, \ldots, p-1\}$. The \emph{Lee weight} of $x$ is given by $$ \text{wt}_L(x)= \min\{x,|p-x|\}.$$
    The largest possible Lee weight is thus $M= (p-1)/2.$
    Let $\bx \in \mathbb{F}_p^n$. The Lee weight is then extended additively on the entries, that is
    $$\text{wt}_L(\bx)= \sum_{i=1}^n \text{wt}_L(x_i).$$
    Let $\bx,\by \in \mathbb{F}_p^n$. Their \emph{Lee distance} is induced by the Lee weight, that is 
    $$d_L(\bx,\by)= \text{wt}(\bx-\by).$$
    Let $\mathcal{C} \subseteq \mathbb{F}_p^n$ be a linear code. The \emph{minimum Lee distance} of $\mathcal{C}$ is given by
    $$d_L(\mathcal{C})= \min \{\text{wt}_L(\bc) \mid \bc \in \mathcal{C}, \bc \neq 0\}.$$
\end{definition}
Note that the Lee metric can be defined over any integer residue ring $\mathbb{Z}/m\mathbb{Z}$, for any integer $m$. However, for the cryptographic purposes it is enough to consider prime fields. Since the Lee metric coincides with the Hamming metric in $\mathbb{F}_2$ and $\mathbb{F}_3$, we only focus on primes $p>3.$

\medskip

Note that,  $\text{wt}_H(\bv)\leq \wt_L(\bv) \leq M \text{wt}_H(\bv)$ and the average Lee weight of the vectors in $\mathbb{F}_p^n$ is given by $(M/2)n.$
We, thus, also get that linear code $\mathcal{C} \subseteq \mathbb{F}_p^n$ can correct more errors in the Lee metric as in the Hamming metric, i.e.,
$$d_H(\mathcal{C}) \leq d_L(\mathcal{C}).$$

Using the other bound, i.e., $d_L(\mathcal{C}) \leq M d_H(\mathcal{C}),$ we can easily adapt the Singleton bound \cite{shiromoto}.
\begin{theorem}
    Let $\mathcal{C} \subseteq \mathbb{F}_p^n $ be a linear code of dimension $k.$ Then,
    $$d_L(\mathcal{C})\leq M(n-k+1).$$
\end{theorem}

Unfortunately, this bound in only tight in $p=5,n=2$, as shown in \cite{bounds}.

\begin{exercise}
    Consider the symmetric representation $\{-(p-1)/2, \ldots, (p-1)/2\}$. Show that $\text{wt}_L(x)= |x|.$
\end{exercise}

We denote by $\delta$ the \emph{relative minimum Lee distance}, that is $$\delta= \frac{d_L(\mathcal{C})}{nM}.$$
Let us denote by $V_L(p,n,r) $ the \textit{Lee sphere} of radius $t$
$$V_L(p,n,t) := \{ \bm{x} \in \mathbb{F}_p^n \mid \wt_L(\bm{x}) = t\},$$ 
and by $$F_L(p,T) = \lim\limits_{n \to \infty} \frac{1}{n} \log_p( \lvert V_L(p,n,Tn M) \rvert)$$ its asymptotic size. 
The exact formulas for the size of $V_L(p,n,t)$ and $F_L(p,T)$ can be found in \cite{leeNP,gardysole}.
 
Let us denote by $A_L(n,d,p)$ the maximal size of a code in $\mathbb{F}_p^n$ of minimum Lee distance $d$ and by $$R(\delta)= \limsup\limits_{n \to \infty} \frac{1}{n} \log_p(A(n,d/(Mn),p)).$$
We can then state the Gilbert-Varshamov bound in the Lee-metric \cite{astola}.

\begin{theorem} 
Let $p$ be a prime and $n,d$ positive integers. Then, $$ R(\delta) \geq 1- F_L(p,\delta).$$
\end{theorem}

In \cite{free}, it was shown that random Lee-metric codes attain with high probability the Lee-metric GV bound, i.e., a random code has with high probability a relative minimum Lee distance $\delta$ such that $R(\delta)= 1-F_L(p,\delta)$. 
\medskip
 
We define a function $\text{sgn}(x)$, that gives us the sign of an element in $\mathbb{F}_p$.
\begin{definition}[Signum]
For $x \in \mathbb{F}_p=\left\{- \frac{p-1}{2}, \ldots, 0, \ldots, \frac{p-1}{2}\right\}$ let 
$$ \text{sgn}(x)= \begin{cases} 0 & \text{if } x=0, \\
1 & \text{if } x > 0, \\  
-1 & \text{if } x < 0. \end{cases}$$ 
\end{definition}
For the symmetric representation of $\mathbb{F}_p$, this corresponds to the common signum function.

Let us also define a matching function $\text{mt}(\bx,\by)$ that compares $\bx$ and $\by$ and counts the number of symbols that hold the same sign.
\begin{definition}[Sign Matches]
Let $\bx,\by \in \mathbb{F}_p^n$ and consider the number of matches in their  sign such that
$$\text{mt}(\bx,\by) =\; \lvert \{ i \in \{1, \ldots, n\} \mid \text{sgn}(x_i)=\text{sgn}(y_i), x_i \neq 0, y_i \neq 0\} \rvert.$$
\end{definition}

Finally, we introduce a function calculating the probability that a vector and a uniformly random hash digest (in $\{\pm 1\}^n$) have $\mu$ sign matches. 
 
\begin{definition}[Logarithmic Matching Probability (LMP)]
    For a fixed $\bv \in \mathbb{F}_p^n$ and a randomly chosen $\by \in \{\pm 1\}^n$, the probability of $\by$ to have $\mu$ sign matches with $\bv$ is 
    \[
    B(\mu,\text{wt}_H(\bv),1/2),
    \]
    where $B(k,n,q)$ is the binomial distribution defined as
    \[B(k,n,q) = \binom{n}{k}q^k (1-q)^{n-k} \enspace .\]
    To ease notation, we write 
    $\text{LMP}(\bv,\by) = -\log_2(B(\mu,\text{wt}_H(\bv),1/2))$.
    \end{definition}

In \cite{bariffi}, the authors computed the marginal distribution of entries where vectors are uniformly distributed in $V_L(p,n,w).$
Let $E$ denote a random variable corresponding to the realization of an entry of $\bx \in \mathbb{F}_p^n$. As $n$ tends to infinity, we have the following result on the distribution of the elements in $\bx \in \mathbb{F}_p^n$.
\begin{lemma}[{\cite[Lemma 1]{bariffi}}]\label{lem:marginal}
    For any $x \in \mathbb{F}_p$, the probability that one entry of $\bx$ is equal to $x$ is given by
    \begin{align*}
       p_w(x) = \frac{1}{Z(\beta)} \exp(-\beta \text{wt}_L (x)),
    \end{align*}
    where $Z(\beta)=\sum_{i = 0}^{p-1}\exp(-\beta \text{wt}_L(x))$ denotes the normalization constant and $\beta$ is the unique solution to $w = \sum_{i = 0}^{p-1} \text{wt}_L(i) p_w(x)$.
\end{lemma}

\begin{definition}[Typical Lee Set] \label{def:typicalset}
For a fixed weight $w$, let $p_w(x)$ be the probability from Lemma \ref{lem:marginal} of the element $x \in \mathbb{F}_p$. 
Then, we define the typical Lee set as
    $$T(p,n,w)=\left\{ \bx \in \mathbb{F}_p^n \mid \bx_i= x \text{ for }  p_w(x)n  \text{ coordinates } i\in \{1,\ldots,n \}\right\} $$  
 That is  
       the set of vectors, for which the element $x$ occurs $p_w(x)n$ times.
\end{definition}

\subsubsection{Restricted Errors}

Instead of considering a different metric on the vectors in $\mathbb{F}_p^n$, we can also restrict their entries. 

\begin{definition}[Restriction]
Let us consider $g \in \mathbb{F}_p^*$ of prime order $z$ and the subgroup $\mathbb{E}=\{g^i \mid i \in \{1, \ldots, z \} \} \subset \mathbb{F}_p^*.$ We say $E$ is a \emph{restriction}.
\end{definition}

Let us denote by $\star$ the component-wise multiplication of vectors.

\begin{proposition}
 $(\mathbb{E}^n, \star) $ is a commutative, transitive group isomorphic to $(\mathbb{F}_z^n,+)$.
 \end{proposition}
 
The isomorphism is given by
\begin{align*} 
\ell: \mathbb{E}^n &\to \mathbb{F}_z^n, \\
\bx= (g^{\ell_1}, \ldots, g^{\ell_n}) &\mapsto \ell(\bx)= (\ell_1, \ldots, \ell_n).
\end{align*}

This representation of vectors in $\mathbb{E}^n$ as vectors in $\mathbb{F}_z^n$ is helpful to shorten the sizes of objects. For the opposite direction of the isomorphism,  we use the following abuse of  notation
 $$\ba=g^{\ell(\ba)}=(g^{\ell(\ba)_1}, \ldots, g^{\ell(\ba)_n}),$$
 for some $\ell(\ba) =(\ell(\ba)_1, \ldots, \ell(\ba)_n) \in \mathbb{F}_z^n$.

\medskip

\begin{proposition}\label{prop:trans}
Any linear map  $\varphi: \mathbb{E}^n \to \mathbb{E}^n$  which acts transitively on $\mathbb{E}^n$ is simply given by component-wise multiplication, i.e.,
$\varphi(\mathbf{b})=\mathbf{a} \star \mathbf{b}$, for some $\mathbf{a} \in \mathbb{E}^n$. 
\end{proposition}

\begin{exercise} Prove Proposition \ref{prop:trans} \end{exercise}

Let the map $\varphi$ be the component-wise multiplication with $\mathbf{a} \in \mathbb{E}^n$.
Then we can compactly represent $\varphi$ through the vector $\ell(\ba)\in \mathbb{F}_z^n.$
Additionally, the computation $\varphi(\mathbf{b})= \mathbf{a} \star \mathbf{b}$ is given by an addition in $\mathbb{F}_z^n$; namely 
$\ell(\ba) + \ell(\bb).$

Instead of the restriction $\mathbb{E}$, we can also consider a \emph{restricted subgroup}.

\begin{definition}[Restricted Subgroup]
Let $(G, \star) \leq (\mathbb{E}^n, \star)$ with
$$G= \langle \mathbf{a}_1, \ldots, \mathbf{a}_m\rangle = \left\{\star_{i=1}^m \mathbf{a}_i^{u_i} \mid u_i \in \{1, \ldots z\}\right\},$$ for some $m <n.$
Then, we call $G$ a \emph{restricted subgroup} of $\mathbb{E}.$
\end{definition}

To construct elements $\mathbf{e} \in G$, we can collect all the exponents of the generators $\mathbf{a}_i$ into a matrix. That is, we define the matrix $\mathbf{M}_G \in \mathbb{F}_z^{ m \times n}$ as 
$$\mathbf{M}_G = \begin{pmatrix} \ell(\ba_1)_1 & \cdots &  \ell(\ba_1)_n \\ 
\vdots & & \vdots \\ 
\ell(\ba_m)_1 & \cdots & \ell(\ba_m)_n \end{pmatrix} = \begin{pmatrix} \ell(\ba_1) \\ \vdots \\ \ell(\ba_m)
\end{pmatrix}.
$$
To check whether $| G | = z^m$, it is enough to verify $\text{rank}(\textbf{M}_G)= m.$ For the remainder, we assume that this is the case. Hence, we can think of $\textbf{M}_G \in \mathbb{F}_z^{m \times n}$ as a generator matrix of a $m$-dimensional code in 
$\mathbb{F}_z^n.$ Thus, each codeword $\mathbf{c} \in \langle \mathbf{M}_G \rangle$ can be represented using an information vector $\mathbf{u} \in \mathbf{F}_z^{m},$ that is 
$$\mathbf{c}= \mathbf{u}\mathbf{M}_G.$$
The corresponding $\mathbf{e} \in G$ has then the exponents $\ell(\mathbf{c}).$

\begin{proposition} Let $G$ be a restricted subgroup, where $\mathbf{M}_G$ has full rank $m$. Then,  $\ell_G$ is a group homomorphism, where
\begin{align*} 
\ell_G: G  &\to \mathbb{F}_z^m, \\ 
\be= \ba_1^{u_1} \star \cdots \star \ba_m^{u_m} &\mapsto \ell_G(\be)= (u_1, \ldots, u_m).
\end{align*}
\end{proposition}

 \begin{proposition}
The linear maps $\varphi: G \to G$, which act transitively on $G$, are still given by component-wise multiplication with another element in $G$, i.e., 
for $\mathbf{e} \in G$, $\varphi(\mathbf{e})= \mathbf{e}' \star \mathbf{e}$. 
\end{proposition}

\subsection{Cryptography}\label{sec:crypto}

        As coding theory is the art of \emph{reliable} communication, this goes hand in hand with cryptography, the art of \emph{secure} communication. In cryptography we differ between two main branches,   symmetric cryptography and asymmetric cryptography. 
        \medskip
        
        In \emph{symmetric} cryptography there are the two parties that want to communicate with each other and prior to communication have exchanged some key, that will enable them a secure communication. Such secret key exchange might be performed using protocols such as the Diffie-Hellman key exchange \cite{dh}, which itself lies in the realm of asymmetric cryptography.
        
        More mathematically involved is the branch of \emph{asymmetric} cryptography, where the two parties do not share the same key. In this survey we will focus on two main subjects of asymmetric cryptography, that were also promoted by the NIST standardization call \cite{nist}, namely public-key encryption (PKE) schemes  and digital signature schemes. 
        
Many of these cryptographic schemes seem very abstract  when discussed in generality.
To get a grasp of the many definitions and concepts, we will also provide some easy examples.
First of all, let us recall the definition of a hash function. A \emph{hash function} is a function that compresses the input value to a fixed length. In addition, we want that it is computationally hard to reverse a hash function and also to find a different input giving the same hash value. In this chapter, we denote a publicly known hash function by $\mathsf{Hash}.$
        
        \subsubsection{Public-Key Encryption}
        Let us start with  public-key encryption (PKE) schemes. 
        A PKE consists of three steps: 
        \begin{enumerate} 
        \item key generation, 
        \item encryption, 
        \item decryption.
        \end{enumerate}
        The main idea is that one party,  usually called Alice, constructs a \emph{secret key} $\mathcal{S}$ and a connected \emph{public key} $\mathcal{P}$. The public key, as the name suggests, is made publicly known, while the secret key is kept private.
        \medskip
        
        This allows another party, usually called Bob, to use the public key to encrypt a \emph{message} $m$ by applying the public key, gaining the so called \emph{cipher} $c$. 
        
        The cipher is now sent through the insecure channel to Alice, who can use her secret key $\mathcal{S}$ to decrypt the cipher and recover the message $m$. 
        \medskip
        
        An adversary, usually called Eve, can only see the cipher $c$ and the public key $\mathcal{P}.$ In order for a public-key encryption scheme to be considered secure, it should be infeasible for Eve to recover from $c$ and $\mathcal{P}$ the message $m$. This also implies that the public key should not reveal the secret key. 
        
        \renewcommand{\arraystretch}{1.5}
\begin{table}[h!]\small
\caption{Public-Key Encryption}
\centering
\begin{tabular}{p{6cm}p{1cm}p{4.8cm}}
\hline
\textsf{ALICE} & & \multicolumn{1}{r}{\textsf{BOB}}\\
\hline
\hline
\multicolumn{1}{l}{KEY GENERATION} & & \\
\hline
Construct a secret key $\mathcal{S}$&  & \\
Construct a connected public key $\mathcal{P}$  &  &\\
& $\xlongrightarrow{\mathcal{P}}$  &\\
\hline
 & & \multicolumn{1}{r}{ENCRYPTION} \\
 \hline 
 &&\multicolumn{1}{r}{Choose a message $m$}\\
 && \multicolumn{1}{r}{Encrypt the message $c= \mathcal{P}(m)$} \\
&$\xlongleftarrow{c}$&\\
\hline
\multicolumn{1}{l}{DECRYPTION} & & \\
\hline
Decrypt the cipher $m=\mathcal{S}(c)$ & & \\
\hline
\end{tabular}
\label{table:PKE}
\end{table}

What exactly does infeasible mean, however? This is the topic of \emph{security}. For a cryptographic scheme, we define its \emph{security level} to be the average number of binary operations needed for an adversary to break the cryptosystem, that means either to recover the message (called \emph{message recovery}) or the secret key (called \emph{key recovery}). 

\medskip

Usual security levels are $2^{80}, 2^{128}, 2^{256}$ or even $2^{512},$ meaning for example that an adversary is expected to need at least $2^{80}$ binary operations in order to reveal the message. These are referred to as 80 bit, 128 bit, 256 bit, or 512 bit security levels.
\medskip

Apart from the security of a PKE, one is also interested in the performance, including how fast the PKE can be executed and how much storage the keys require.
Important parameters of a public-key encryption are 
\begin{itemize}
    \item the public key size,
    \item the secret key size, 
    \item the ciphertext size, 
    \item the decryption time. 
\end{itemize} 
These values are considered to be the \emph{performance} of the public-key encryption.
With 'size' we intend the bits that have to be sent or stored for this key, respectively for the cipher.
Clearly, one prefers small sizes and a fast decryption.

\medskip

As an example for a PKE, we can choose one of the most currently used schemes, namely RSA \cite{rsa}.
\medskip

\begin{example}[RSA]  

 \begin{enumerate}
     \item Key Generation: Alice chooses two distinct primes $p,q$ and computes $n=pq$ and $\varphi(n)=(p-1)(q-1).$ She chooses a natural number $e < \varphi(n)$, which is coprime to $\varphi(n).$ The public key is $\mathcal{P}=(n,e)$ and the secret key is $\mathcal{S}=(p,q)$.
 \item Encryption: Bob chooses a message $m$ and encrypts it by computing $$c=m^e \mod n.$$
 \item Decryption: Alice can decrypt the cipher by first computing $d$ and $b$ such that 
 $$de + b \varphi(n)=1.$$ Since $$c^d = \left(m^e \right)^d = m^{1-b \varphi(n)} = m\left(m^{\varphi(n)}\right)^{-b} = m1^{-b} =m,$$ she can recover the message $m$.
 \end{enumerate}
 \end{example}
 \medskip

Eve sees $n$ but there is no feasible algorithm to compute $p$ and $q$.

\begin{exercise}
Assume that Alice has chosen $p$ and $q$ to have $100$ digits. How large is the public key size?
\end{exercise}

\begin{exercise}
Assume that the fastest known algorithm to factor $n$ into $p$ and $q$ costs $\sqrt{n}$ binary operations. In order to reach a security level of $2^{80}$ binary operations, how large should Alice choose $p$ and $q$?
\end{exercise}

\begin{exercise}
To give you also a feeling for cryptanalysis; why should we always choose two distinct primes? Or in other words; how can you attack RSA if $p=q$?
\end{exercise}

\subsubsection{Key-Encapsulation Mechanisms}

A \emph{key-encapsulation mechanism} (KEM) is a way to transmit a key for symmetric cryptography using an asymmetric cryptosystem.
\medskip

 Public-key systems are often not optimal to transmit longer messages. Instead, the two parties  use a public-key system to share a random $m$, usually a number or vector. Then both parties use an agreed-on function, called \emph{key derivation function}, to calculate a key $M$ from $m$. 
 \medskip

The function is usually chosen to be a one-way function, meaning that computing back $m$ with only the knowledge of the function and $M$ is not computationally feasible. With this key, the parties can then encrypt their message. 
\medskip

Most KEM schemes are based on Shoup's idea  \cite{shoup}. In Table \ref{keyencaps} we give an outline, in which we assume that a public-key system is given. For this, let $\mathsf{Hash}$ denote a hash function.

 \renewcommand{\arraystretch}{1.5}
\begin{table}[h!]\small
\caption{Key-Encapsulation Scheme}
\label{keyencaps}
\centering
\begin{tabular}{p{5.4cm}p{1cm}p{5cm}}
\hline
\textsf{ALICE} & & \multicolumn{1}{r}{\textsf{BOB}}\\
\hline
\hline
\multicolumn{1}{l}{KEY GENERATION} & & \\
\hline
Generate a secret key $\mathcal{S}$&  & \\
Construct a connected public key $\mathcal{P}$&  &\\

& $\xlongrightarrow{\mathcal{P}}$  &\\
\hline
 & & \multicolumn{1}{r}{ENCRYPTION} \\
 \hline 
 &&\multicolumn{1}{r}{\parbox{5cm}{\raggedleft Choose a random message $m$}}\\
 &&\multicolumn{1}{r}{\parbox{5cm}{\raggedleft Generate a key $M = \mathsf{Hash}(m)$}}\\
 &&\multicolumn{1}{r}{\parbox[t]{5cm}{\raggedleft Use the public key $\mathcal{P}$ to encrypt $m$ as cipher $c$}}\\
&$\xlongleftarrow{c}$&\\
\hline
\multicolumn{1}{l}{DECRYPTION} & & \\
\hline
Using the secret key $\mathcal{S}$, decrypt $c$ to get $m$& & \\
compute $\mathsf{Hash}(m)=M$ & & \\
\hline
\multicolumn{1}{l}{COMMUNICATION} & &\\
\hline
The parties may now communicate with each other since they both possess a key to encrypt and decrypt messages \\
\hline
\end{tabular}
\label{table:KEM}
\end{table}
~
\medskip

As mentioned before, it is often the case that instead of directly encrypting the key $M$, a random $m$ is encrypted. From this $m$, both parties can generate a key using the agreed-on key derivation function.

\begin{example} For an example of a KEM we again consider RSA.
\begin{enumerate}
\item Key generation: Alice choose two distinct primes $p,q$ and computes $n=pq$ and $\varphi(n)=(p-1)(q-1)$. Alice also chooses a positive integer $e < \varphi(n),$ which is coprime to $\varphi(n).$ The public key is given by $\mathcal{P}= (n,e)$ and the private key is given by $(p,q).$
\item Encryption: Bob chooses a random message $m$ and computes its hash $M=\mathsf{Hash}(m).$ He then performs the usual steps of RSA, that is: he encrypts $c= m^e \mod n$ and sends this to Alice.
\item Decryption: Alice can compute $d= e^{-1} \mod \varphi(n)$ and computes $c^d = m \mod n.$ Also Alice can now compute the shared key $M= \mathsf{Hash}(m).$
\end{enumerate}
\end{example}

\subsubsection{Digital Signature Schemes}

Digital Signature schemes aim at giving a guarantee of the legitimate origin of an object, such as a digital message, exactly as signing a letter to prove that the sender of this letter is really you. 

In this process we  speak of \emph{authentication}, meaning that a receiver of the message can (with some probability) be sure that the sender is legit, and of \emph{integrity}, meaning that the message has not been altered.

A digital signature scheme again consists of three steps: 
\begin{enumerate}
    \item key generation, 
    \item signing, 
    \item verification.
\end{enumerate}
In digital signature schemes we consider two parties, one is the \emph{prover}, that has to prove his identity to the second party called \emph{verifier}, that in turn,  verifies the identity of the prover.  

As a first step, the prover constructs a secret key $\mathcal{S}$, which he keeps private and a public key $\mathcal{P}$, which is made public. The prover then chooses a message $m$,  and creates a signature $s$ using his secret key $\mathcal{S}$ and the message $m$, getting a signed message $(m,s).$ 

The verifier can easily read the message $m$, but wants to be sure that the sender really is the prover. Thus, he uses the public key $\mathcal{P}$  and the knowledge of the message $m$ on the signature $s$ to get authentication.

        \renewcommand{\arraystretch}{1.5}
\begin{table}[h!]\small
\caption{Digital Signature Scheme}
\centering
\begin{tabular}{p{5.4cm}p{1cm}p{5cm}}
\hline
\textsf{PROVER} & & \multicolumn{1}{r}{\textsf{VERIFIER}}\\
\hline
\hline
\multicolumn{1}{l}{KEY GENERATION} & & \\
\hline
Construct a secret key $\mathcal{S}$&  & \\
Construct a connected public key $\mathcal{P}$  &  &\\
& $\xlongrightarrow{\mathcal{P}}$  &\\
\hline
\multicolumn{1}{l}{SIGNING} \\
 \hline 
Choose a message $m$ && \\
Construct a signature $s$ from $\mathcal{S}$ and $m$ & &  \\
&$\xlongrightarrow{m,s}$&\\
\hline
& & \multicolumn{1}{r}{VERIFICATION}  \\
\hline
& & \multicolumn{1}{r}{Verify the signature  $s$ using $\mathcal{P}$ and $m$} \\
\hline
\end{tabular}
\label{table:DSS}
\end{table}

The security of a digital signature scheme introduces a new person, the \emph{impersonator}. 
An impersonator, tries to cheat the verifier and acts as a prover, however without the knowledge of the secret key $\mathcal{S}.$ An impersonator wins if a verifier has verified a forged signature. This comes with a certain probability, called \emph{cheating probability} or \emph{soundness error}. In order to ensure integrity a digital signature should always involve a secret key as well as the message itself.

Clearly, the secret key should still be infeasible to recover from the publicly known private key, thus one still has the usual adversary, called Eve,  and a security level, as in a public-key encryption scheme. 

The performance of a digital signature scheme consists of
\begin{itemize} 
\item the \emph{communication cost}, that is the total number of bits, that have been exchanged within the process, 
\item the \emph{signature size}, 
\item the public key size, \item  the secret key size, \item the verification time.  \end{itemize}

An easy example for a signature scheme  is given by turning the RSA public-key encryption protocol into a signature scheme.
\begin{example}[RSA Signature Scheme]

 \begin{enumerate}
     \item Key Generation: Alice chooses two distinct primes $p,q$ and computes $n=pq$ and $\varphi(n)=(p-1)(q-1).$ She chooses a natural number $e < \varphi(n)$, which is coprime to $\varphi(n).$ She computes $d$ and $b$ such that $$de + b \varphi(n)=1.$$  The public key is $\mathcal{P}=(n,e)$ and the secret key is $\mathcal{S}=(p,q,d)$.
 \item Signing: Alice chooses a message $m$ and signs it by computing $$s=m^d \mod n.$$ She then sends $m,s$ to Bob.
 \item Verification: Bob can verify the signature $s$ by   checking if  $$s^e =m \mod n.$$
 \end{enumerate}
\end{example}

\begin{exercise}
How would an impersonator forge a signature provided that the impersonator does not care about the content of the message $m$?
\end{exercise}

\subsubsection{Zero-Knowledge Protocols}

Since digital signature schemes can be constructed using the Fiat-Shamir transform \cite{fiatshamirtransform} on \emph{Zero-Knowledge} (ZK) protocols, we will also introduce the concept of ZK protocols and then of the transform itself.
\medskip

The process and notation for a ZK protocols are similar to that of a digital signature scheme. We have two parties, a prover and a verifier. Different to a digital signature scheme, the prover does not want to prove his identity to the verifier, but rather convince the verifier of his knowledge of a secret object, without revealing said object. 
\medskip

A ZK protocol consists of two stages: key generation and verification.
The verification process can consist of several communication steps between the verifier and the prover, in particular, we are interested in the following scheme:
\begin{enumerate} 
\item The  prover  prepares  two  \emph{commitments} $c_0,c_1,$ and sends them to the verifier.
\item The verifier randomly picks a \emph{challenge} $b \in \{0,1\},$  and sends it to the prover.
\item The prover provides a \emph{response} $r_b$ that only allows to verify $c_b$.
\item The verifier checks the validity of $c_b$, usually by recovering $c_b$ using $r_b$ and the public key.
\end{enumerate}

        \renewcommand{\arraystretch}{1.5}
\begin{table}[h!]\small
\caption{ZK Protocol}
\centering
\begin{tabular}{p{4.9cm}p{2.5cm}p{4cm}}
\hline
\textsf{PROVER} & & \multicolumn{1}{r}{\textsf{VERIFIER}}\\
\hline
\hline
\multicolumn{1}{l}{KEY GENERATION} & & \\
\hline
Construct a secret key $\mathcal{S}$&  & \\
Construct a connected public key $\mathcal{P}$  &  &\\
& $\xlongrightarrow{\mathcal{P}}$  &\\
\hline
& VERIFICATION & \\
 \hline 
Construct commitments $c_0,c_1$  && \\
&$\xlongrightarrow{c_0,c_1}$&\\
& & \multicolumn{1}{r}{Choose $b \in \{0,1\}$} \\
& $\xlongleftarrow{b}$ & \\
Construct response $r_b$ & & \\
&$\xlongrightarrow{r_b}$&\\
& & \multicolumn{1}{r}{Verify $c_b$ using $r_b$}\\
\hline 
\end{tabular}
\label{table:ZKID}
\end{table}

A ZK protocol has three important attributes:
\begin{enumerate}
    \item \emph{Zero-knowledge:} this means that no information about the secret is revealed during the process.
    \item \emph{Completeness:} meaning that an honest prover will always get accepted.
    \item \emph{Soundness:} for this, we want that an impersonator has only a small cheating probability to get accepted.
\end{enumerate}

Again, for the performance of the protocol, we have 
\begin{itemize}
    \item the communication cost, 
    \item the secret key, 
    \item the public key size, \item the verification time.
\end{itemize}

In order to achieve an acceptable cheating probability, the protocols are often repeated several times (called \emph{rounds}) and only if each instance was verified will the prover be accepted. Thus, if the ZK protocol previously had cheating probability $\alpha$, after $N$ such rounds we have a cheating probability of $\alpha^N.$

There exist several techniques in order to compress the communication cost within $N$ rounds, for example the \emph{compression technique}, first introduced in \cite{ags}. 
Let us explain this method in detail. 

Before the first round, the prover generates the commitments  for all the $N$ rounds, that is $c_b^{i}$ for $i \in \{1, \ldots, N\}$ and $b \in \{0,1\}$.  The prover then sends the hash value $$c = \mathsf{Hash}\big( c_0^1, c_1^1,\ldots, c_0^{N}, c_1^{N} \big)$$ to the verifier. 

In the $i$-th round, after receiving the challenge $b$, the prover sets their response $r_b$ such that the verifier can compute $c^i_b$, and additionally includes $c_{1-b}^i$.  

At the end of each round, the verifier uses $r_b$ to compute $c^i_b$, and stores it together with $c_{1-b}^i$. 

After the final round $N$, the verifier is able to check validity of the initial commitment $c$, by computing the hash of all the stored $c^i_b$. 

This way, one hash is sent at the beginning of the protocol, and only one hash (instead of two) is transmitted in each round and thus, the number of exchanged hash values reduces from $2N$ to $N + 1$.

\renewcommand{\arraystretch}{1.5}
\begin{figure}[ht!]\small
\centering
\caption{Compression Technique for $N$ Rounds}\label{compression}
\begin{tabular}{p{6cm}p{1cm}p{4.8cm}}
\hline
\textsf{PROVER} & & \multicolumn{1}{r}{\textsf{VERIFIER}}\\
\hline
Generate $c_b^i$,   for $i\in \{1,\ldots,N\}$ and $b \in \{0,1\}$ &  & \\
Set $c = \mathsf{Hash}\big( c_0^1, c_1^1,\ldots, c_0^{N}, c_1^{N} \big)$ &   & \\
 & $\xlongrightarrow{c}$  &\\
& \makebox[\dimexpr(\width-10mm)][l]{\hspace{-21mm}$\xlongrightarrow{\xlongleftarrow[\text{\footnotesize Repeat single round for $N$ times}]{}}$} &\\
&&  \multicolumn{1}{r}{Check validity of $c$}\\
\hline
\hline
 & \mbox{GENERIC $i$-th ROUND} & \\
 \hline
 &&\\
& \makebox[\dimexpr(\width-10mm)][l]{\hspace{-18mm}$\xlongrightarrow{\xlongleftarrow[\text{\footnotesize Exchange additional messages}]{}}$} &\\
 &&\multicolumn{1}{r}{Choose $b \in \{0 , 1\}$}\\
&$\xlongleftarrow{b}$&\\
Construct response $r_b$ &&\\ 
& \makebox[\dimexpr(\width-10mm)][l]{\hspace{-4mm}$\xlongrightarrow{r_b, \hspace{1mm} c_{1-b}^i}$} & \\
&&\multicolumn{1}{r}{Store $c^i_{1-b}$, compute and store $c^i_b$}\\
\hline
\end{tabular}
\end{figure}

 An easy example is again provided using the hardness of integer factorization, namely the Feige-Fiat-Shamir protocol \cite{feige}. 

\begin{example}[Feige-Fiat-Shamir]
\begin{enumerate}
    \item Key generation: The prover chooses two distinct primes $p,q$ and computes $n=pq$ and some positive integer $k$. The prover chooses $s_1,\ldots, s_k$ coprime to $n$. The prover now computes 
    $$v_i \equiv s_i^{-2} \mod n.$$ The public key is given by $\mathcal{P}=(n, v_1, \ldots, v_k)$. The secret key is given by $\mathcal{S} =(p,q, s_1, \ldots, s_k).$
    \item Verification: The prover chooses a random integer $c $ and a random sign $\sigma \in \{-1,1\}$ and computes $$x \equiv \sigma c^2 \mod n$$ and sends this to the verifier. 
    The verifier chooses the challenge $b=(b_1, \ldots, b_k) \in \mathbb{F}_2^k$ and sends $b$ to the prover.
    The prover then computes the response $$r \equiv c \prod\limits_{b_j =1} s_j \mod n$$ and sends $r$ to the verifier. The verifier can now check whether 
    $$x \equiv \pm r^2\prod\limits_{b_j =1} v_j \mod n.$$
\end{enumerate}
\end{example}
Eve, the impersonator, can see the public $v_i$ but she does not know the $s_i$. She can pick a random $r$ and  $b=(b_1, \ldots, b_k) \in \mathbb{F}_2^k.$ She then computes $$x \equiv r^2 \prod\limits_{b_j =1} v_j \mod n$$ and sends $x$ to the verifier. The verifier will then challenge her with his $b'$, but Eve simply returns her $r$. If Eve has correctly chosen $b=b'$, she will be verified.

\begin{exercise}
        What is the cheating probability of this scheme? If you repeat this process $t$ times before accepting the prover, what is now your cheating probability?
\end{exercise}

\begin{exercise}
        Let us assume that $k=10.$ How many times should you repeat this process in order to reach a cheating probability of at least $2^{128}$?
\end{exercise}

\subsubsection{Fiat-Shamir Transform}\label{sec:fiat}
The \emph{Fiat-Shamir transform} allows us to build a signature scheme from a ZK protocol. To avoid the communication with the verifier  that randomly picks a challenge, the challenge is replaced with the seemingly random hash of the commitment and message.

The following table follows the general description of the Fiat-Shamir transform from \cite{fiatshamirtransform}.
We assume that we are given a zero-knowledge identification scheme  and a public hash function $\mathsf{Hash}$.  
\begin{table}[h!]\small
\caption{Fiat-Shamir Transform}
\centering
\begin{tabular}{p{5.4cm}p{1cm}p{5cm}}
\hline
\textsf{PROVER} & & \multicolumn{1}{r}{\textsf{VERIFIER}}\\
\hline
\hline
\multicolumn{1}{l}{KEY GENERATION} & & \\
\hline
Given the public key $\mathcal{P}$ and the secret key $\mathcal{S}$ of some ZK protocol and  a message $m$ \\
Choose a commitment $c$\\
Compute $a = \mathsf{Hash}(m,c)$ \\
Compute a response $r$ to the challenge $a$\\
The signature is the pair $s=(a, r)$ \\
&$\xlongrightarrow{m,s}$&\\
\hline
 & & \multicolumn{1}{r}{VERIFICATION} \\
 \hline 
 &&\multicolumn{1}{r}{\parbox[t]{5cm}{\raggedleft Use the response $r$ and the public key $\mathcal{P}$ to construct the commitment $c$ \strut}}\\
 &&\multicolumn{1}{r}{Check if $\mathsf{Hash}(m,c)=a$}\\
\hline
\end{tabular}
\label{table:FST}
\end{table}

Using the Fiat-Shamir transform we can turn the Feige-Fiat-Shamir ZK protocol into a signature scheme.
\begin{example}[Fiat-Shamir digital signature scheme]
\begin{enumerate}
    \item Key Generation: Let $\mathsf{Hash}$ be a publicly known hash function. The prover chooses a positive integer $k$ and two distinct primes $p,q$ and computes $n=pq$. The prover chooses $s_1, \ldots, s_k$ integers coprime to $n$ and computes $v_i \equiv s_i^{-2} \mod n$ for all $i \in \{1, \ldots, k\}$. The secret key is given by $\mathcal{S} = (p,q, s_1, \ldots, s_k)$ and the public key is given by $(n, v_1, \ldots, v_k).$
    \item Verification: the prover chooses randomly $c_1, \ldots, c_t <n$ and computes $x_i \equiv c_i^2 \mod n$ for all $i \in \{1, \ldots, t\}$. In order to bypass the communication with the verifier from before, the prover computes the first $kt$ bits of  $$\mathsf{Hash}(m, x_1, \ldots, x_t)=(a_{1,1}, \ldots, a_{t,k})=a.$$ The prover now computes $r_i \equiv c_i \prod\limits_{a_{ij}=1}s_j \mod n$ for all $i \in \{1, \ldots, t\}$  and sends $(m,a,r_1, \ldots, r_t)$ to the verifier.
 The verifier computes $$z_i \equiv r_i^2 \prod\limits_{a_{i,j}=1}v_j \mod n$$ for all $i \in \{1, \ldots, t\}$ and checks if $$\mathsf{Hash}(m, z_1, \ldots, z_t)=a.$$
\end{enumerate}
\end{example}

\subsubsection{Multi-Party-Computations-in-the-Head}\label{sec:mpc}

Recall that any ZK protocol can be turned into a signature scheme via the Fiat-Shamir transform.
Assume that the used ZK protocol has a cheating probability of $\alpha$ and recall that this probability might be quite large. 
In order to get a resulting signature scheme attaining the security level $2^\lambda$, we require $N$ rounds of the ZK protocol, such that $\alpha^N < 2^{-\lambda}.$

Since the final signature is given by the communication cost within all $N$ rounds, such signature schemes usually suffer from large signature sizes. 

One very prominent technique in order to reduce the signature size was introduced in \cite{ikos}
and uses the idea of Multi-Party-Computations (MPC). 

In an MPC we have $N$ parties, called $p_1, \ldots, p_N$, each party is secretly provided a share $s_i$. The parties wish to collectively  compute a certain function of their shares, say $f(s_1, \ldots, s_N)$, in such a way that the shares $s_i$ remain only known to the party $p_i$ and an such that all shares are required. 

We say that an MPC protocol is 
\begin{itemize}
    \item correct, if the parties can correctly compute $f(s_1, \ldots, s_N)$,
    \item $t$-private, if any $t$ shares (or less) do not reveal any information on $f(s_1, \ldots, s_N)$,
    \item secure, if $f(s_1,\ldots, s_N)$ does not reveal any information on $s_i.$
\end{itemize}
An easy way to achieve an MPC protocol is to use Secret Sharing (SS) schemes. The whole theory of MPC and SS schemes is highly involved and we refer the interested reader to \cite{ss}.  

In a SS scheme, we have a \emph{dealer}, who wants to share a secret message with the parties, $p_1, \ldots, p_N$ and again each party $p_i$ is provided with a share $s_i.$
We introduce two parameters; $k\leq n$ the decoding threshold and $z<k$ the confidentiality threshold. These parameters take care of the following two constraint:
\begin{enumerate}
    \item A group of $k\leq n$ parties can decode the secret message using their shares. 
    \item A group of $z<k$ parties do not gain any information about the secret from their shares.  
\end{enumerate}

The security goal for such a scheme is thus \emph{confidentiality,} that is: no information about the secret should be leaked from any $z$ shares.

Important for this survey, will be additive sharing schemes. Let us, thus, start with a toy example.

\begin{example}\label{ex:SS}
    Let us consider $n=4, k=2,z=k-1=1$ and $q=5$. The secret message is some $m \in \mathbb{F}_5$, we choose a random value $r \in \mathbb{F}_q$ and we use an \emph{encoding polynomial}  $p(x)=m+rx$. The secret shares are then given by $s_i = p(i)= m +ir.$
\end{example}
\begin{exercise}
Consider the SS scheme in Example \ref{ex:SS}.
    \begin{enumerate}
        \item Show that the SS scheme attains privacy, i.e., an individual party with share $s_i$ does not gain information about $m.$
        \item Show that the SS scheme is decodable, i.e., any two parties can recover $m.$
    \end{enumerate}
\end{exercise}

The more general construction, is called \emph{Shamir's secret sharing scheme} \cite{shamir}. Given the integers $z=k-1, k \leq n < q$ and a polynomial $p(x) \in \mathbb{F}_q[x]$ of degree $z$, given by
$$p(x)= m + \sum_{i=1}^z r_i x^i,$$ where $r_i$ are chosen uniform at random from $\mathbb{F}_q.$
The secret shares are then given by $s_i=p(i).$

\begin{exercise}
    Show that Shamir's SS scheme is attains privacy and is decodable. 
\end{exercise}

One can also construct a SS scheme with $z<k-1,$ e.g. using McEliece-Sarwate's construction \cite{sarwate}. 
 Given the integers $z< k \leq n < q$ and a polynomial $p(x) \in \mathbb{F}_q[x]$ of degree $k-1$, given by
$$p(x)= \sum_{i=1}^z r_i x^i+ \sum_{i=1}^{k-z} m_ix^{z-1+i},$$ where $r_i$ are chosen uniform at random from $\mathbb{F}_q.$
The secret shares are then given by $s_i=p(i).$

\begin{exercise}
    Show that McEliece-Sarwate's SS scheme is attains privacy and is decodable. \\
    \emph{Hint:} Use the property of a $k \times n$ Vandermonde matrix, that each $k \times k$ submatrix is invertible. 
\end{exercise}

For a more sophisticated SS scheme, we assume that the secret is given by $\bs \in \mathbb{F}_q^n$ and all $N$ parties are provided with random $\bs_i \in \mathbb{F}_q^n$, such that $\sum_{i=1}^N \bs_i=\bs.$
Clearly, only if all $N$ parties open their shares, they can compute collectively $$f(\bs_1,\ldots, \bs_N)=\sum_{i=1}^N \bs_i=\bs,$$ while any $k <N$ parties cannot compute $\bs$.

For a secret $s$, we will use the notation $[[s]]$ to denote a possible splitting into $N$ additive shares $$[[s]]=(s_1, \ldots, s_N).$$

The idea of MPC-in-the-head (MPCitH), introduced in  \cite{ikos} is to use MPC protocols to build ZK protocols. 

For this, assume we are given an MPC protocol
in which $N$ parties $P_1, \ldots, P_N$ securely and correctly evaluate a function $f$ on a secret input $s$. Additionally, we require 
\begin{itemize} 
\item the secret $s$ has a sharing $[[s]]=(s_1, \ldots, s_N)$ and each party $P_i$ gets the input $s_i$, 
\item for some functions $\varphi_i$, the party $P_i$ computes the broadcast $\alpha_i=\varphi(s_i),$
\item  the function $f$, such that $f(\alpha_1, \ldots, \alpha_N)=1$, and anything else evaluates to 0,
\item  if $N-1$ parties reveal their shares $s_i$, or their broadcasts $\alpha_i$, they do not reveal anything on the secret $s$.
\end{itemize}
The resulting ZK protocol, requires a trapdoor function $F$, which is easy to compute and hard to invert. 
In code-based cryptography, this is usually the syndrome decoding problem. Namely, 
\begin{align*}
    F: B_H(t,n,q) &\to \mathbb{F}_q^{n-k}, \\ 
    \be &\mapsto \be\bH^\top.
\end{align*}
That is, we send vectors of Hamming weight at most $t$, to their syndromes for a fixed parity-check matrix $\bH.$ While this is easy to compute given $\bH$ and $\be$, it is hard to invert, that is: given $\bH,\bs$ find $\be.$
We say that the trapdoor function $F$ has \emph{target} $y$ if, for the sought solution $x$ we have $F(x)=y.$

In the previous example of syndrome decoding, the target would be the syndrome of the sought-after solution $\be$ and $F$ is completely determined by $\bH.$
 
Assuming such a trapdoor function $F$ and an MPC protocol, the resulting ZK protocol works as follows. 
        \renewcommand{\arraystretch}{1.5}
\begin{table}[h!]\small
\caption{ZK Protocol from MPC}
\centering
\begin{tabular}{p{4.9cm}p{2.5cm}p{4cm}}
\hline
\textsf{PROVER} & & \multicolumn{1}{r}{\textsf{VERIFIER}}\\
\hline
\hline
\multicolumn{1}{l}{KEY GENERATION} & & \\
\hline
Given MPC with secret $s$ and function $f$ and $\varphi_i$& & \\
Given trapdoor function $F$ with target $y$.&& \\
Secret key $\mathcal{S}=s$&  & \\
Public key $\mathcal{P}=\{f,F,y\}$  &  &\\
& $\xlongrightarrow{\mathcal{P}}$  &\\
\hline
& VERIFICATION & \\
 \hline 
For $i \in \{1, \ldots, N\}$:   && \\
\quad Compute $\alpha_i =\varphi_i(s_i)$ & & \\ 
\quad Compute  $c_i=\mathsf{Hash}(s_i,\rho_i)$  & & \\ \quad for some random $\rho_i$. && \\
&$\xlongrightarrow{c_1, \ldots, c_N}$&\\
Check if $f(\alpha_1, \ldots, \alpha_N)=1$
&$\xlongrightarrow{\alpha_1, \ldots, \alpha_N}$&\\
& & \multicolumn{1}{r}{Choose $b \in \{1, \ldots, N\}$} \\
& $\xlongleftarrow{b}$ & \\
Response $r_b=\{(s_i,\rho_i) \mid i \neq b\}$ & & \\
&$\xlongrightarrow{r_b}$&\\
& & \multicolumn{1}{r}{For all $i\neq b$:}\\
& & \multicolumn{1}{r}{\quad Check $c_i=\mathsf{Hash}(s_i,\rho_i)$}\\
& & \multicolumn{1}{r}{\quad Check $\alpha_i=\varphi_i(s_i)$}\\
& & \multicolumn{1}{r}{Check $F(\alpha_1,\ldots, \alpha_N)=y.$}\\
\hline 
\end{tabular}
\label{table:ZKMPC}
\end{table}
The main idea of the ZK protocol using MPCitH, is to run a MPC protocol in the prover's head, i.e., the prover simulates locally  all the parties of the MPC protocol and sends commitments to each party’s   share. To check that the MPC protocol runs correctly, the prover also sends the broadcasts $\alpha_i.$

The main benefit of MPCitH lies in the cheating probability. Since the MPC protocol is $N-1$-private, an impersonator not knowing the secret $s$, can guess any $N-1$ many shares and compute broadcasts and commitments to those. However, the last share $s_f$ is chosen at random, and is not such that $\sum_{i=1}^N s_i =s.$ In fact, finding the last $s_f$ would require the impersonator to invert the trapdoor function.

The verifier accepts the impersonator, only  if the verifier challenges exactly the random $s_f$, i.e., $b=f$. 
In any other case, the impersonator is required so send $s_f$ in the response and the verifier can check that the target $y$ is not reached. 

Thus, the new cheating probability is $\frac{1}{N}.$ This allows us to reduce the number of rounds required to achieve a certain security level and thus, in turn, the signature size. However, the broadcast computation has to be performed in each such round $N$ times, by the prover and the verifier.

\section{Code-Based Public-Key Encryption Frameworks}\label{sec:pkeframework}

Code-based cryptography and in particular code-based PKEs first came up with the seminal work of Robert J. McEliece in 1978 \cite{mceliece}. The main idea of code-based cryptography is to base the security of the cryptosystem on the hardness of decoding a random linear code. 
Since this problem is NP-hard, code-based cryptography is considered to be one of the most promising candidates for post-quantum cryptography. 
\medskip

In a nutshell, McEliece's idea as follows: the private key is given by  a linear code $\mathcal{C}$, which can efficiently correct $t$ errors. The public key is $\mathcal{C}'$ a disguised version of the linear code, which should not reveal the secret code, in fact, should behave randomly.
\medskip

While anyone with $\mathcal{C}'$, the publicly known code, can encode their message and possibly add some intentional errors, an attacker would only see a random code and in order to recover the message would need to decode it. 
\medskip

The constructor of the secret code however,  can transform the encoded message to a codeword of $\mathcal{C}$, which is efficiently decodable. 
\medskip

The first code-based cryptosystem by McEliece uses the generator matrix $\bG$ as a representation of the secret code and in order to disguise the generator matrix, one computes $\bG'=\bS \bG \bP$, where $\bS$ is an invertible matrix, thus only changing the basis of the code, and $\bP$ is a permutation matrix, thus giving a permutation equivalent code. In the encryption step the message is then encoded through $\bG'$ and an intentional error vector $\be$ is added.
\medskip

An equivalent \cite{equivalent} cryptosystem was proposed by Niederreiter in \cite{niederreiter}, where one uses the parity-check matrix $\bH$, and the disguised parity-check matrix $\bH'= \bS\bH\bP$, instead of the generator matrix and the cipher is given by the syndrome of an error vector, i.e., $\bs=\bH'\be^\top$. 
\medskip
 
The code-based system proposed by Alekhnovich uses the initial idea of McEliece, but twists the disguising of the code, by adding a row to the parity-check matrix, which is a erroneous codeword, thus making the error vector the main part of the secret key.  
The idea of Alekhnovich, which is not considered practical has been the starting point of a new framework, the quasi-cyclic scheme. 
\medskip

A different idea has been proposed by Augot and Finiasz in \cite{af}. Here the secret is given by the  support of an error vector, which allows to insert an error of weight beyond the error correction capacity. Thus, again the code can be made completely public.
\medskip

Finally, the McEliece framework has also been introduced for the rank metric by Gabidulin, Paramonov and Tretjakov and is usually denoted by the GPT system.
\medskip

Clearly, all of these cryptosystems (except for Alekhnovich's, which uses a random code) have been originally proposed for a specific code. In the following we will introduce the idea behind the systems as frameworks, thus without considering a specific code.  
        \newpage
\subsection{McEliece Framework}\label{sec:mcframework}

Although McEliece originally proposed in \cite{mceliece} to use a binary Goppa code as secret code, one usually denotes by the McEliece framework the following. 
Alice, the constructor of the system, chooses an $[n,k]$ linear code $\mathcal{C}$ over $\mathbb{F}_q$, which can efficiently decode $t$ errors through the decoding algorithm $\mathcal{D}.$ Instead of publishing a generator matrix $\bG$ of this code, which would then reveal to everyone the algebraic structure of $\mathcal{C}$ and especially how to decode, one hides $\bG$ through some scrambling: we compute $\bG' = \bS\bG\bP$, for some invertible matrix $\bS \in \text{GL}_k(q)$ and an $n \times n$ permutation matrix $\bP$. Hoping that the new matrix $\bG'$ and the code it generates $\mathcal{C}'$ seem random (although $\mathcal{C}'$ is permutation equivalent to $\mathcal{C})$,  Alice then publishes this disguised matrix $\bG'$ and the error correction capacity $t$ of $\mathcal{C}$. 
\medskip

Bob who wants to send a message $\bm \in \mathbb{F}_q^k$ to Alice can then use the public generator matrix $\bG'$ to encode his message, i.e., $\bm\bG'$, and then adds a random error vector $\be \in \mathbb{F}_q^n$ of Hamming weight up to $t$ to it, i.e., the cipher is given by $\bc= \bm \bG' +\be.$

An eavesdropper, Eve, only knows $\bG',t$ and the cipher $\bc$. In order to break the cryptosystem and to reveal the message $\bm$, she would need to decode $\mathcal{C}'$, which seems random to her. Thus, she is facing an NP-complete problem and the best known solvers have an exponential cost. 

\medskip

However, Alice can reverse the disguising by computing $\bc\bP^{-1}$, which results in a codeword of $\mathcal{C}$ added to some error vector of weight up to $t$. That is
$$\bc\bP^{-1} = \bm\bS\bG + \be\bP^{-1}.$$
Through the decoding algorithm $\mathcal{D}$ Alice gets $\bm\bS$ and thus by multiplying with $\bS^{-1}$, she recovers the message $\bm.$

       \renewcommand{\arraystretch}{1.5}
\begin{table}[h!]\small
\caption{McEliece Framework}
\centering
\begin{tabular}{p{5.4cm}p{1cm}p{5cm}}
\hline
\textsf{ALICE} & & \multicolumn{1}{r}{\textsf{BOB}}\\
\hline
\hline
\multicolumn{1}{l}{KEY GENERATION} & & \\
\hline
Choose a linear code $\mathcal{C} \subseteq \mathbb{F}_q^n$ of dimension $k$ and error correction capacity $t$. Let $\bG$ be a $k \times n$ generator matrix of $\mC$. &  & \\
Choose randomly $\bS \in \text{GL}_k(q)$ and an $n \times n$ permutation matrix $\bP$. Compute $\bG'=\bS\bG\bP$. & & \\
The public key is given by $\mathcal{P}= (t, \bG')$  and $\mathcal{S}=(\bG,\bS,\bP)$ &  &\\
& $\xlongrightarrow{\mathcal{P}}$  &\\
\hline
 & & \multicolumn{1}{r}{ENCRYPTION} \\
 \hline 
 &&\multicolumn{1}{r}{\parbox[t]{5cm}{\raggedleft Choose a message $\bm  \in \mathbb{F}_q^k$ and a random error vector $\be \in \mathbb{F}_q^n$ of weight at most $t$}} \\
 && \multicolumn{1}{r}{Encrypt the message $\bc= \bm\bG'+\be$} \\
&$\xlongleftarrow{\bc}$&\\
\hline
\multicolumn{1}{l}{DECRYPTION} & & \\
\hline
Decrypt the cipher, by decoding $\bc\bP^{-1}=\bm\bS\bG+\be\bP^{-1}$ to get $\bm\bS$, and finally recover the message as $\bm = (\bm \bS) \bS^{-1}$  & & \\
\hline
\end{tabular}
\label{table:mc}
\end{table}

\begin{exercise}
    Consider re-encryption: Given the public generator matrix $\bG$. Bob encrypts the message $\bm$ getting $\bc_1= \bm\bG+\be_1$ and later with the same $\bG$ the same message $\bm$ again getting $\bc_2=\bm\bG+\be_2$. Is this safe? 
\end{exercise}
Since this is the key part of this survey, we will provide a toy example explained in full detail. 

\begin{example}\label{hamming}
Let $\mathcal{C}$ be the $[7,4]$ binary Hamming code, which can efficiently correct 1 error. We take as generator matrix 
$$\bG = \begin{pmatrix} 
1 & 0 &0 & 0 & 1 & 1 & 0 \\
0 & 1 & 0 & 0 & 1 & 0 &  1 \\
0 & 0 & 1 & 0 & 0 & 1 & 1 \\
0 & 0 & 0 & 1 & 1 & 1 & 1 
\end{pmatrix}.$$
We choose $\bS \in \text{GL}_4(2)$ to be 
$$\bS= \begin{pmatrix} 
0 & 1 & 1 & 1  \\
1 & 0 & 1 & 1 \\
 1 & 0 & 1 & 0 \\
 0 & 0 & 1 & 1
\end{pmatrix}$$
 and the permutation matrix $\bP$ to be 
 $$\bP = \begin{pmatrix}
 0 & 1 & 0 &0 & 0 & 0 & 0 \\
 0 & 0 & 0 &1 & 0 & 0 & 0 \\
  0 & 0 & 0 &0 & 0 & 1 & 0 \\
   1 & 0 & 0 &0 & 0 & 0 & 0 \\
    0 & 0 & 0 &0 & 1 & 0 & 0 \\
     0 & 0 & 1 &0 & 0 & 0 & 0 \\
      0 & 0 & 0 &0 & 0 & 0 & 1 
 \end{pmatrix}.$$
 We thus compute 
 $$ \bG' = \begin{pmatrix}
 1 & 0 & 0 & 1 & 0 & 1 & 1 \\
 1 & 1 & 1 & 0 & 0 & 1 & 0 \\
 0 & 1 & 0 & 0 & 1 & 1 & 1 \\
 1 & 0 & 0 & 1 & 1 & 0 & 0 
 \end{pmatrix} $$
 and publish $(\bG', 1)$, since $t=1$. 
 The message we want to send is $\bm= (1,0,1,1) \in \mathbb{F}_2^4$ and thus we compute
 $$\bm\bG' = (0,1,0,1,0,1,0).$$
 Now, we choose an error vector $\be \in \mathbb{F}_2^7$ of Hamming weight 1, e.g.,
 $$\be= (1,0,0,0,0,0,0).$$
 Thus, the cipher is given by 
 $$\bc= (1,1,0,1,0,1,0).$$
 The constructor, who possesses $\bP$ can compute 
 $$\bc\bP^{-1}=\bc\bP^\top=(1,1,1,1,0,0,0).$$
 We can now use the decoding algorithm of Hamming codes to recover $\bm\bS=(1,1,1,0)$ and by multiplying with $$\bS^{-1}= \begin{pmatrix} 0 & 1 & 0 & 1 \\ 1 & 0 & 0 & 1 \\ 0 & 1 & 1 & 1 \\ 0 &1 &1 & 0 \end{pmatrix}$$
 we recover the message $\bm= ( 1,0,1,1).$
 
 In this toy example, an attacker which sees $\bG', t,\bc$ has two possibilities: \begin{enumerate}
     \item recover the message directly,
     \item recover the secret key.
 \end{enumerate}
 The first type of attack could work as follows:
 \begin{enumerate}
     \item We bring $\bG'$ into a row-reduced form, that is for $\bG' = [ \bA \ \mid \ \bB ]$ we compute $\bA^{-1} \bG',$ giving
     $$\overline{\bG}= \begin{pmatrix} 
     1 & 0 & 0 & 0 & 1 & 1 & 0 \\
     0 & 1 & 0 & 0 & 1 & 1 & 1 \\
     0 & 0 & 1 & 0 & 0 & 1 & 1 \\
     0 & 0 & 0 & 1 & 1 & 0 & 1
     \end{pmatrix}.$$
     With $\overline{\bG}= [ \Id_4 \ \mid \ \bC]$ we can also compute the parity-check matrix as $\overline{\bH}= [ \bC^\top \ \mid \ \Id_3],$ that is 
     $$\overline{\bH} = \begin{pmatrix}
     1 & 1 & 0 & 1 & 1 & 0 & 0 \\
     1 & 1 & 1 & 0 & 0 & 1 &  0 \\
     0 & 1 & 1 & 1 & 0 & 0 & 1 
     \end{pmatrix}.$$
     \item We can now compute the syndrome of $\bc$ through $\overline{\bH}$, i.e.,
     $$\bs = \bc\overline{\bH}^\top = (1,1,0).$$
     Note that this is also the syndrome of the error vector $\be$, i.e., 
     $\be \overline{\bH}^\top= \bs .$ Since there is only one entry of $\be$ that is non-zero, we must have that the syndrome $\bs$ is equal to the column $\bh_i$ where $\be_i \neq 0$. And in fact, $\bs= \bh_1$, thus we have found 
     $$\be = (1,0,0,0,0,0,0)$$ and $$\bc-\be=\bm\bG' = (0,1,0,1,0,1,0).$$
     Note that the moment we know the error vector, we can use linear algebra to recover the message. Since this is a toy example, we will also execute this step. 
     \item Denote by $\overline{\bm}= \bm \bA$, then 
     $$(0,1,0,1,0,1,0) = \bm \bG' = \bm \bA \bA^{-1} \bG' = \overline{\bm} \overline{\bG}.$$ 
     Since $\overline{\bG}= [ \Id_4 \ \mid \bC ]$, we have that $$\overline{\bm} \overline{\bG}= (\overline{\bm}, \overline{\bm} \bC).$$
     Hence, we can directly read off that $\overline{\bm}=(0,1,0,1)$ and by multiplying with $\bA^{-1}$, we recover $\bm=(1,0,1,1).$
 \end{enumerate}
The second type of attack, namely a key-recovery attack, is in nature more algebraic. Knowing that the secret code is a $[7,4]$ binary code that can correct one error, the suspicion that the secret code is a Hamming code is natural. If not, one could proceed as follows.  
 \begin{enumerate}
     \item We choose a set $I \subset \{1, \ldots, n\}$ of size $k$, which is a possible information set. Let us denote by $\bG'_I$ the matrix consisting of the columns of $\bG'$ indexed by $I$.
     \item We compute $(\bG'_I)^{-1}\bG'$ to get an identity matrix in the columns indexed by $I$.
     \item Choose the permutation matrix $\bP'$ which brings the identity matrix in the columns indexed by $I$ to the first $k$ positions. 
 \end{enumerate}
 With this, if one chose $I = \{4,2,6,3\}$ (the order matters here only for the permutation matrix) we  recover $\bG$ and will now finally be able to read off the secret code and thus also know its decoding algorithm. With this, we can compute from $\bG$ and $\bG'$ the matrices $\bS$ and $\bP.$ 
\end{example}

Although this example for the McEliece framework is clearly using a code that should not be used in practice,   it shows in a few easy steps the main ideas of the attacks. For example, the minimum distance of a code should be large enough, since else an easy search for the error vector will reveal the message, and also the public code parameters should not reveal anything on the structure of the secret code, meaning that there should be many codes having such parameters.

These two different kind of attacks aim at solving two different problems the security of the McEliece system is based upon: 
\begin{enumerate}
    \item decoding the erroneous codeword, assuming that the code is random, should be hard,
    \item the public code, which is permutation equivalent to the secret code, should not reveal the algebraic structure of the secret code.
\end{enumerate}
Only if both of these points are fulfilled is the security of the cryptosystem guaranteed. 

We will see more on this in Section \ref{sec:security}.

\subsection{Niederreiter Framework}\label{sec:niedframework}

The Niederreiter framework \cite{niederreiter} uses the parity-check matrix instead of the generator matrix, resulting  in an equivalently secure system \cite{equivalent}. Niederreiter originally proposed to use GRS codes as secret codes, however, we will consider with the Niederreiter framework the more general scheme.

Alice again chooses an $[n,k]$ linear code $\mathcal{C}$ over $\mathbb{F}_q$ which can efficiently decode up to $t$ errors. She then scrambles a parity-check matrix $\bH$ of $\mathcal{C}$ by computing $\bH' = \bS\bH\bP$, for some invertible matrix $\bS \in \GL_{n-k}(\mathbb{F}_q)$ and an $n \times n$ permutation matrix $\bP.$ She publishes the seemingly random parity-check matrix $\bH'$ together with the error correction capacity $t$. 

Bob can then encrypt a message $\bm \in \mathbb{F}_q^n$ of Hamming weight up to $t$, simply by computing the syndrome of $\bm$ through the parity-check matrix $\bH'$, i.e., the cipher is given by $\bc=\bm\bH'^\top.$

While Eve would only have access to $\bH'$, which looks random to her, $t$ and $\bc$, she faces an NP-hard problem and can only apply exponential time algorithms in order to recover $\bm.$ 

Alice, on the other hand, can recover the message by computing $\bS^{-1}\bc$, which results in a syndrome of her code $\mathcal{C},$ which she knows how to decode. That is 

$$\bS^{-1}\bc^\top= \bH\bP\bm^\top,$$ where $\bP\bm^\top$ still has Hamming weight up to $t$. Thus, she recovers $\bP\bm^\top$ and by multiplication with $\bP^{-1}$, she recovers the message $\bm.$

       \renewcommand{\arraystretch}{1.5}
\begin{table}[h!]\small
\caption{Niederreiter Framework}
\centering
\begin{tabular}{p{5.4cm}p{1cm}p{5cm}}
\hline
\textsf{ALICE} & & \multicolumn{1}{r}{\textsf{BOB}}\\
\hline
\hline
\multicolumn{1}{l}{KEY GENERATION} & & \\
\hline
Choose a linear code $\mathcal{C} \subseteq \mathbb{F}_q^n$ of dimension $k$ that can efficiently correct $t$ errors. Let $\bH$ be a $(n-k) \times n$ parity-check matrix of $\mC$ &  & \\
Choose randomly $\bS \in \text{GL}_{n-k}(q)$ and an $n \times n$ permutation matrix $\bP$. Compute $\bH'=\bS\bH\bP$ & & \\
The public key is given by $\mathcal{P}= (t, \bH')$  &  &\\
& $\xlongrightarrow{\mathcal{P}}$  &\\
\hline
 & & \multicolumn{1}{r}{ENCRYPTION} \\
 \hline 
 &&\multicolumn{1}{r}{\parbox[t]{5cm}{\raggedleft Choose a message $\bm  \in \mathbb{F}_q^n$  of weight at most $t$ \strut}} \\
 && \multicolumn{1}{r}{\parbox[t]{5cm}{\raggedleft Encrypt the message $\bc^\top= \bH'\bm^\top$}} \\
&$\xlongleftarrow{\bc}$&\\
\hline
\multicolumn{1}{l}{DECRYPTION} & & \\
\hline
Decrypt the cipher  by decoding $\bS^{-1}\bc^\top=\bH\bP\bm^\top$ to get $\bP\bm^\top$, and finally recover the message as $\bm^\top = \bP^{-1} (\bP \bm^\top)$  & & \\
\hline
\end{tabular}
\label{table:nied}
\end{table}

We provide the same toy example for the Niederreiter framework.
\begin{example}
This time, we start with a parity-check matrix $\bH$ of the $[7,4]$ binary Hamming code, given by
$$\bH= \begin{pmatrix} 
1 & 1 & 0 & 1 & 1 & 0 & 0 \\
1 & 0 & 1 & 1 & 0 & 1 & 0 \\
0 & 1 & 1 & 1 & 0 & 0 & 1 \\
\end{pmatrix}.$$
We choose as invertible matrix $\bS \in \text{G}_3(\mathbb{F}_2)$ the following
$$\bS= \begin{pmatrix}
1 & 1 & 0 \\ 0 & 1 & 1 \\ 0 & 0 & 1
\end{pmatrix}$$ and as permutation matrix we choose
$$\bP=\begin{pmatrix}
0 & 0 & 1 & 0 & 0 & 0 & 0 \\
0 & 0 & 0 & 0 & 1 & 0 & 0 \\
1 & 0 & 0 & 0 & 0 & 0 & 0 \\
0 & 0 & 0 & 0 & 0 & 0 & 1 \\
0 & 1 & 0 & 0 & 0 & 0 & 0 \\
0 & 0 & 0 & 1 & 0 & 0 & 0 \\
0 & 0 & 0 & 0 & 0 & 1 & 0 \\
\end{pmatrix}.$$
With this, we compute 
$$\bH' = \bS \bH \bP = \begin{pmatrix}
0 & 0 & 1 & 1 & 1 & 1 & 0 \\
0 & 1 & 1 & 1 & 0 & 0 & 1 \\
1 & 0 & 0 & 1 & 0 & 1 & 1 \\
\end{pmatrix}.$$
The public key is given by $\bH'$ and $t=1.$
Assume that we want to send the message $\bm= (0,0,1,0,0,0,0) \in \mathbb{F}_2^7$. For this, we compute the cipher as the syndrome of $\bm$ through $\bH'$, i.e., 
$$\bc = \bm (\bH')^\top = (1,1,0)$$ and send it to the constructor. 
The constructor which knows $\bS$ and $\bP$ first computes
$$\bS^{-1}\bc^\top =\bH\bP\bm^\top = (0,1,0)^\top,$$ and then uses the decoding algorithm of the Hamming code to  get $$\bm\bP^\top=(0,0,0,0,0,1,0).$$ Finally multiplying this with $\bP^{-1}$ we get the message $\bm=(0,0,1,0,0,0,0).$
\end{example}

The security is clearly equivalent to that of Example \ref{hamming}, due to the duality of $\bG$ and $\bH$ and the attacks form Example \ref{hamming} work here as well.

\subsection{Alekhnovich's Cryptosystems}\label{alekframework}

Alekhnovich's cryptosystem \cite{alekhnovich} marks the first code-based cryptosystem with a security proof, i.e., it relies solely on the decoding problem. 
This seminal work lays the foundations of modern code-based cryptography, where researchers try to construct code-based cryptosystems with a provable reduction to the problem of
decoding a random linear code.

There are two variants to this cryptosystem, both are relying on the following hard problem:
\begin{problem} 
Given a code $\mathcal{C}$, distinguish a random vector from an erroneous codeword of $\mathcal{C}.$
\end{problem}
Note that variations of these cryptosystem are used in \cite{NISTHQC, NISTBike}.
For the following description of the two variants we rely on the survey \cite{zemor} and for more details we also refer to \cite{zemor}.

\subsubsection{The First Variant}

The idea is not to keep the parity-check matrix or generator matrix of the code hidden, but a random error vector. Thus, a random matrix $\bA$ is chosen and to this one adds the row $\bx\bA +\be $, thus an erroneous codeword of the code generated by $\bA$ is added resulting in the augmented matrix $\bH$. Let us consider $\mathcal{C}$ to be $\text{Ker}(\bH)$, that is the code having $\bH$ as parity-check matrix. One then publishes $\bG$, a generator matrix of $\mathcal{C}.$
\medskip

In this variant one only encrypts a single bit. One either sends as cipher an erroneous codeword of $\mathcal{C}^\perp$ or a random vector, depending if 0 or 1 was encrypted. Finally, using the secret error vector $\be$,  one can compute the standard inner product of the cipher and $\be$ and will recover the message, with some decryption failure.  
\medskip

More in detail, if the cipher was given by $\ba\bG + \be'$, for a random $\ba \in \mathbb{F}_2^{n-k}$ and a random error vector $\be' \in \mathbb{F}_2^n$ of weight $t$, then 
$$ \langle \be, \ba\bG+\be'\rangle = \langle \be, \ba\bG\rangle + \langle \be, \be' \rangle.$$
Note that $\langle \be, \ba\bG \rangle=0$, since $\be \in \mathcal{C}^\perp$ by construction. In addition, since $\text{wt}_H(\be) = \text{wt}_H(\be')= t=o(\sqrt{n})$, we have that 
$\langle \be, \be' \rangle=0$ with high probability. If the cipher was given by a random vector $\bc \in \mathbb{F}_2^n$ instead, then with probability 1/2 we get $\langle \be, \bc \rangle=1.$ 
\medskip

Thus, there is a decryption failure in the case $m=1$ of probability 1/2. In order to get a reliable system one can encrypt the message multiple times. A systematic description of Alekhnovich's First Variant can be found in Table \ref{alekhno1tab}.

       \renewcommand{\arraystretch}{1.5}
\begin{table}[h!]\small
\caption{Alekhnovich First Variant}
\label{alekhno1tab}
\centering
\begin{tabular}{p{5.4cm}p{1cm}p{5cm}}
\hline
\textsf{ALICE} & & \multicolumn{1}{r}{\textsf{BOB}}\\
\hline
\hline
\multicolumn{1}{l}{KEY GENERATION} & & \\
\hline
Let $t \in o(\sqrt{n})$ and choose a random matrix $\bA \in \mathbb{F}_2^{k \times n}$ &  & \\
Let $\be \in \mathbb{F}_2^n$ be a random vector of weight $t$ and let $\bx \in \mathbb{F}_2^k$ be a random vector & & \\
Compute $\by= \bx\bA + \be$ and $\bH^\top = (\bA^\top, \by^\top)$ & & \\
Let $\mathcal{C} = \text{ker}(\bH)$ and choose a generator matrix $\bG \in  \mathbb{F}_2^{(n-k-1) \times n} of \mathcal{C}$ & & \\
The public key is given by $\mathcal{P}=  (\bG,t)$ and the secret key is $\mathcal{S}= \be$ &  &\\
& $\xlongrightarrow{\mathcal{P}}$   &\\
\hline
 & & \multicolumn{1}{r}{ENCRYPTION} \\
 \hline 
 &&\multicolumn{1}{r}{Choose a message $\bm  \in \mathbb{F}_2$ \strut} \\
 & & \multicolumn{1}{r}{\parbox[t]{5cm}{\raggedleft If $\bm=0$: choose $\ba \in \mathbb{F}_2^{n-k-1}$ and $\be' \in \mathbb{F}_2^n$ of weight $t$ at random, send $\bc = \ba\bG + \be'$ \strut}}\\
 && \multicolumn{1}{r}{\strut If $\bm=1$: choose a random vector $\bc \in \mathbb{F}_2^n$} \\
&$\xlongleftarrow{\bc}$&\\
\hline
\multicolumn{1}{l}{DECRYPTION} & & \\
\hline
Decrypt the cipher, by computing $\bb=\langle \be, \bc\rangle$. \\
If $\bm=0$:  $\bb=0$ with high probability & & \\
  If $\bm=1$:  $\bb=1$ with  probability $1/2$ & & \\
\hline
\end{tabular}
\label{table:A1}
\end{table}
~
\medskip

We give an example of the first variant. 
\begin{example}
Let 
$$\bA = \begin{pmatrix}
1 & 1 & 0 & 0 & 0 & 0 \\
1 & 0 & 1 & 0 & 0 & 0 \\
0 & 0 & 0 & 1 & 1 & 0 \\
0 & 0 & 0 & 1 & 0 & 1
\end{pmatrix}.$$

We choose $\bm = ( 0 , 1 , 0 , 1)$, $\be = ( 1 , 0 , 0 , 0 , 0 , 0 )$ and compute 
$$ \bm \bA + \be = ( 0 , 0 , 1 , 1 , 0 , 1 ).$$
If we append this to the matrix $\bA$, we get the matrix 
$$\bH = \begin{pmatrix}
1 & 1 & 0 & 0 & 0 & 0 \\
1 & 0 & 1 & 0 & 0 & 0 \\
0 & 0 & 0 & 1 & 1 & 0 \\
0 & 0 & 0 & 1 & 0 & 1 \\
0 & 0 & 1 & 1 & 0 & 1 \\
\end{pmatrix}.
$$
The dual code ${C}$ of $\bH$ has a generator matrix 
$$\bG = \begin{pmatrix} 0 & 0 & 0 & 1 & 1 & 1
\end{pmatrix}.$$

We encrypt $0$ as
$$\bc_{0} = ( 0 , 0 , 0 , 1 , 1 , 1 ) + ( 0 , 1 , 0 , 0 , 0 , 0 ) = ( 0, 1 , 0 , 1 , 1 , 1 ),$$
and $1$ as random vector 
$$\bc_{1} = ( 1 , 0 , 1 , 0 , 0 , 1 ).$$

To decrypt the cipher $\bc$, we compute $\langle \be,\bc \rangle.$
If we receive $\bc_0$, we compute that $\langle \be, \bc_0 \rangle = 0$.
If we receive $\bc_1$, we see that $\langle \be, \bc_1 \rangle = 1$.

\end{example}

\subsubsection{The Second Variant}

In this variant one generalizes the idea of the first variant and construct directly a matrix $\bM$ in which every row is an erroneous codeword. 
\medskip

This is achieved by choosing at random $\bA \in \mathbb{F}_2^{n/2 \times n}, \bX \in \mathbb{F}_2^{n \times n/2} $ and $\bE \in \mathbb{F}_2^{n \times n}$ having row weight $t$. Then one computes the matrix $\bM =\bX\bA+\bE.$
\medskip

Let $\mathcal{C}_0$ be a binary code of length $n$, that can correct codewords transmitted through a binary symmetric channel (BSC) with transition probability $t^2/n$. Let us consider \begin{align*}
    \varphi: \mathbb{F}_2^n & \to \mathbb{F}_2^n, \\  \bx &\mapsto \bM\bx.
\end{align*}
Define $$\mathcal{C}_1= \varphi^{-1}(\mathcal{C}_0)=\{ \bx \in \mathbb{F}_2^n \mid \varphi(\bx) \in \mathcal{C}_0\},$$
$\mathcal{C}_2= \text{Ker}(\bA)$ and finally $\mathcal{C}= \mathcal{C}_1 \cap \mathcal{C}_2.$
Let $\bG \in \mathbb{F}_2^{k \times n}$ be a generator matrix of $\mathcal{C}$. This generator matrix is made public, while the error vectors in $\bE$ are kept secret. 
\medskip

To encrypt a message $\bm \in \mathbb{F}_2^{k/2}$ we first append a random vector $\br \in \mathbb{F}_2^{k/2}$ to get $\bx =(\bm, \br) \in \mathbb{F}_2^k$ and then compute $$\bc = \bx\bG + \be,$$ for some random error vector $\be \in \mathbb{F}_2^n$ of weight $t$. 

To decrypt we now compute \begin{align*}
    \by^\top & = \bE \bc^\top = \bE (\bx\bG + \be)^\top\\
    &= \bE(\bx\bG)^\top + \bE\be^\top \\
    &= \bX\bA(\bx\bG)^\top + \bM(\bx\bG)^\top + \bE\be^\top \\
    &= \bM(\bx\bG)^\top + \bE\be^\top,
\end{align*}
where we have used that $\bA \ba^\top =0$ for all $\ba\in \mathcal{C}$, in particular also for $\bx\bG$. Note that $\mathbf{z}^\top=\bM(\bx\bG)^\top \in \mathcal{C}_0,$ since $\mathcal{C} \subseteq \mathcal{C}_1$ and $\varphi(\mathcal{C}_1)=\mathcal{C}_0.$ Finally, every row $\be_i$ of $\bE$ has weight $t$ and thus, $\langle \be_i, \be \rangle =1$ with probability at most $t^2/n.$
Thus, the decoding algorithm of $\mathcal{C}_0$ on $\by$ gives $\mathbf{z}$ with high probability. Finally, we can solve the linear system $$\bx\bG = \varphi^{-1}(\mathbf{z})$$ to get $\bx$ and the first $k/2$ bits reveal the message $\bm.$

       \renewcommand{\arraystretch}{1.5}
\begin{table}[h!]\small
\caption{Alekhnovich Second Variant}
\centering
\begin{tabular}{p{5.4cm}p{1cm}p{5cm}}
\hline
\textsf{ALICE} & & \multicolumn{1}{r}{\textsf{BOB}}\\
\hline
\hline
\multicolumn{1}{l}{KEY GENERATION} & & \\
\hline
Choose  random matrices $\bA \in \mathbb{F}_2^{n/2 \times n}, \bX \in \mathbb{F}_2^{n \times n/2}$ and $\bE \in \mathbb{F}_2^{n \times n}$ of row weight $t$ &  & \\
 Set $\bM=\bX\bA+\bE \in \text{GL}_n(2)$ &  & \\
Let $\mathcal{C}_0$ be a binary code of length $n$ that can efficiently correct codewords transmitted through a BSC of transition probability $t^2/n$ & & \\
Let $\varphi$ be the map $\bx \mapsto \bM\bx$ & & \\
Let $\mathcal{C} = \varphi^{-1}(\mathcal{C}_0) \cap \text{Ker}(\bA)$  & & \\
Let $\bG \in \mathbb{F}_2^{k \times n}$ be a generator matrix of $\mathcal{C}$ & & \\
The public key is given by $\mathcal{P}=  (\bG,t)$  and $\mathcal{S}=\bE$ &  &\\
& $\xlongrightarrow{\mathcal{P}}$  &\\
\hline
 & & \multicolumn{1}{r}{ENCRYPTION} \\
 \hline 
 &&\multicolumn{1}{r}{\parbox[t]{5cm}{\raggedleft Choose a message $\bm  \in \mathbb{F}_2^{k/2}$ and choose randomly $\br \in \mathbb{F}_2^{k/2}$ and $\be \in \mathbb{F}_2^n$ of weight $t$}} \\
 && \multicolumn{1}{r}{Compute $\bx= (\bm, \br) \in \mathbb{F}_2^k$ and $\bc=\bx\bG+\be$} \\
&$\xlongleftarrow{\bc}$&\\
\hline
\multicolumn{1}{l}{DECRYPTION} & & \\
\hline
Decrypt the cipher, by computing $\by^\top= \bE\bc^\top = \mathbf{z}^\top + \bE\be^\top$ & & \\
and use the decoding algorithm of $\mathcal{C}_0$ on $\by$ to get  $\mathbf{z}$ & & \\
Recover $\bx$ from the linear system $\bx\bG =\varphi^{-1}(\mathbf{z})$ and thus $\bm$& & \\
\hline
\end{tabular}
\label{table:A2}
\end{table}

\subsection{Quasi-Cyclic Scheme}\label{sec:quasicyclic}

The quasi-cyclic scheme is inspired by the  scheme of Alekhnovich, introduced in \cite{quasicyclicscheme} and used in \cite{NISTHQC}. Similarly to Alekhnovich's schemes, it is a probabilistic approach to encryption schemes and does not hide the initial code, which needs to be efficiently decodable. The message gets encrypted as codeword to which an error, too large to decode, gets added. With the knowledge of the private key parts of this error can be cancelled out resulting (with high probability) in a vector which can be decoded to recover the message.

We present the scheme in the  Hamming metric, but note that the scheme can also be adapted to the rank metric.
\medskip

Let $n$ be a positive integer, $q$ be a prime power and $\mathcal{R} = \mathbb{F}_q[x]/(x^n-1)$. Recall from Section \ref{sec:prelim} that we identify vector $\ba=(a_0,a_1 ,\ldots, a_{n-1}) \in \mathbb{F}_q^n$  with the polynomial $a(x)= \sum_{i=0}^{n-1} a_i x^i \in \mathcal{R}$ and vice versa.  

The quasi-cyclic framework uses two types of codes:
\begin{enumerate}
    \item An $[n,k]$ linear code $\mathcal{C}$ over $\mathbb{F}_q$, which can efficiently decode $\delta$ errors. A generator matrix $\bG \in \mathbb{F}_q^{k \times n}$ is made public.
    \item A random  quasi-cyclic $[2n,n]$ code presented through a parity-check matrix $$\bH= \begin{pmatrix} \text{Id}_n & \mid  \text{rot}(\bh) \end{pmatrix},$$ which does not require to be efficiently decodable and is also made public. 
\end{enumerate}
Recall that vector multiplication of any vector $\bv$ and $\bh$ is given by $\bv \text{rot}(\bh),$ as this corresponds to the polynomial multiplication $v(x)h(x) \in \mathcal{R}.$
 
Let $w$, $w_r$ and $w_e$ be positive integers  all in the range of $\sqrt{n}/2$. These are publicly known parameters.

The cryptosystem then proceeds as follows. Alice chooses a random $h \in \mathbb{F}_q^n$   and an $[n,k]$ linear code  $\mathcal{C}$ over $\mathbb{F}_q$, that can efficiently correct $t$ errors and  chooses a generator matrix $\bG$ of $\mathcal{C}$.

\begin{table}[h!]\small
\caption{Quasi-Cyclic Scheme}
\begin{tabular}{p{5.4cm}p{1cm}p{5cm}}
\hline
\textsf{ALICE} & & \multicolumn{1}{r}{\textsf{BOB}}\\
\hline
\hline
\multicolumn{1}{l}{KEY GENERATION} & & \\
\hline
Choose an $[n,k]$ linear code $\mathcal{C}$ over $\mathbb{F}_q$, which can efficiently decode $t$ errors with  \\
generator matrix $\bG \in \mathbb{F}_q^{k \times n}$  and choose $\bh \in \mathbb{F}_q^n$\\ Choose  $\by, \mathbf{z} \in \mathbb{F}_q^n$ of weight $\wtH(\by)=\wtH(\mathbf{z})=w$, compute $\bs=\by+\bh\mathbf{z}$ \\
The public key is  $\mathcal{P}= (\bG, \bh, \bs,w_e, w_r)$ and the secret key is $\mathcal{S}= (\by, \mathbf{z})$\\
& $\xlongrightarrow{\mathcal{P}}$  &\\
\hline
 & & \multicolumn{1}{r}{ENCRYPTION} \\
 \hline 
 &&\multicolumn{1}{r}{Choose a message $\bm \in \mathbb{F}_q^k$  } \\
 &&\multicolumn{1}{r}{Choose $\be \in \mathbb{F}_q^n$ such that $\wtH(\be) = w_e$ } \\
 &&\multicolumn{1}{r}{\parbox[t]{5cm}{\raggedleft Choose $\br_1, \br_2 \in \mathbb{F}_q^n$ such that $\wtH(\br_1) = \wtH(\br_2) = w_r$ \strut}} \\
 &&\multicolumn{1}{r}{\parbox[t]{5cm}{\raggedleft Compute $\bu=\br_1 + \bh \br_2$}}\\
  &&\multicolumn{1}{r}{Compute $\bv=\bm \bG +\bs\br_2 +\be$}\\
 &&\multicolumn{1}{r}{The cipher is $\bc=(\bu,\bv)$}\\
&$\xlongleftarrow{\bc}$&\\
\hline
\multicolumn{1}{l}{DECRYPTION} & & \\
\hline
Compute $\bc'=\bv-\bu\mathbf{z}$ and use the decoding algorithm of ${\mC}$ to recover $\bm$\\
\hline
\end{tabular}
\label{table:QCS}
\end{table}

Alice then also chooses two elements $\by,\mathbf{z} \in \mathbb{F}_q^n$, corresponding to the vector $\by, \mathbf{z}$ both of Hamming weight $w$. 

She publishes the generator matrix $\bG$, the random element $\bh$ and $\bs=\by+ \bh\mathbf{z},$ while $\by$ and $\mathbf{z}$ are kept secret and can clearly not be recovered from $\bs$ and $\bh$. In fact, we can write 
$$\bs= (\by,\mathbf{z}) \begin{pmatrix} \text{Id}_n \\ \text{rot}(\bh)\end{pmatrix},$$
thus $\bH= (\text{Id}_n, \text{rot}(\bh)^\top)$ acts as quasi-cyclic parity-check matrix and $(\by,\mathbf{z})$ as unknown error vector.

Bob, who wants to send a message $\bm \in \mathbb{F}_p^k$ to Alice, can choose $\be \in \mathbb{F}_q^n$    of Hamming weight $w_e$  and two elements $\br_1, \br_2 \in \mathbb{F}_q^n$,   both of Hamming weight $w_r$. He then computes $\bu=\br_1+\bh\br_2$  and $$\bv=\bm\bG+\bs\br_2 +\be.$$  The cipher is then given by $\bc=(\bu,\bv).$

The message $\bm$ is thus encoded through the public $\bG$ and an error vector $\bs\br_2+\be$ is added, where both $\br_2$ and $\be$ were randomly chosen by Bob. The only control Alice has on the error vector is in $\bs$. This knowledge and also the additional information of Bob on $\br_2$ provided through the vector $\bu$ will allow Alice to decrypt the cipher.

In fact, Alice can use the decoding algorithm of $\mathcal{C}$ on $\bv-\bu\mathbf{z}$, since 

\begin{align*}
        \bv-\bu\mathbf{z} &= \bm \bG + \bs \br_2 + \be - (\br_1 +\bh \br_2)\mathbf{z} \\
        &= \bm\bG + (\by+\bh \mathbf{z})\br_2 + \be - \br_1\mathbf{z} - \bh \br_2\mathbf{z} \\
        &= \bm \bG + (\by \br_2 - \br_1\mathbf{z} +\be).
\end{align*}
It follows that the decryption succeeds if $\wtH(\by\br_2 - \br_1\mathbf{z} +\be) \leq t$. Note that parameter sets should be chosen such that this happens with high probability, but clearly the framework does have a \emph{decoding failure rate} (DFR). 
\medskip

\begin{remark}
The reason why we can make the generator matrix of the efficiently decodable code public, lies in the random choice of $h$, which determines the parity-check matrix $\bH$ and in the fact that the error added to the codeword has a weight larger than the error correction capacity of the public code.  

In fact, $\bu$ and $\bs$ are two syndromes through $\bH$ of a vector with given weight, as $$\bu = (\br_1, \br_2)\bH^\top$$ and $\bs= (\by,\mathbf{z})\bH^\top$. In order to recover $(\br_1,\br_2)$ or $(\by,\mathbf{z})$, an attacker would need to solve the NP-hard syndrome decoding problem. In addition, since $\text{wt}_H(\bs\br_2+\be)>t$ even with the knowledge of $\bG$ and $\bv$ an attacker can not uniquely determine the message $\bm$. 
\end{remark}

\medskip

Since the algebraic code, which is efficiently decodable, is publicly known, the security of this framework is different to that of the McEliece framework and the Niederreiter framework, as it does not rely on the indistinguishability of the code. 
\medskip

\begin{remark}
However, we want to stress the fact, that the SDP is NP-hard for a completely random code. The code with the double circulant parity-check matrix $\bH$ is in fact not completely random, and thus the question arises, if also this new problem lies in the complexity class of NP-hard problems. 
\end{remark}
 
\begin{example}
We choose $R = \mathbb{F}_2 [x]/(x^7 +1)$ and as code the binary repetition code of length $7$, which can correct up to $3$ errors.  The generator matrix $\bG$ is given by
$$ \bG = \begin{pmatrix}
1 & 1 & 1 & 1 & 1 & 1 & 1
\end{pmatrix},$$
and codewords with more ones than zeroes are decoded to $(1,1,1,1,1,1,1)$, everything else to $(0,0,0,0,0,0,0)$.
Further, we choose
$$ h(x) = 1 + x + x^2 \in \mathcal{R},$$ or equivalently 
$\bh=(1,1,1,0,0,0,0)$
and $w = w_r = w_e = 1$. We pick $y(x) = 1$, $z(x) = x^3$, both in $\mathcal{R}$  of weight $w = 1$, or equivalently $\by=(1,0,0,0,0,0,0), \mathbf{z}=(0,0,0,1,0,0,0)$ and compute
$$ s(x) = y(x) + h(x) z(x) = 1 + x^3 + x^4 + x^5.$$ Equivalently one can compute 
$$\bs=\by+\mathbf{z}\text{rot}(\bh)=(1,0,0,1,1,1,0).$$
The public key is then given by
$$ \mathcal{P} = ( \bG, \bh, \bs, w_e, w_r ),$$
the secret key is the pair
$$ \mathcal{S} = (\by, \mathbf{z}).$$

For this example, the message is  $\bm = (1) \in \mathbb{F}_2^1$. We also pick $e(x) = x \in \mathcal{R}$, that is $\be=(0,1,0,0,0,0,0)$ of weight $w_e=1$ and $r_1(x) = r_2(x)=x^2$ in $\mathcal{R}$, that is $\br_1=\br_2=(0,0,1,0,0,0,0)$ of weight $w_r = 1$. We can then compute
$$ u(x) = r_1(x) + h(x) r_2(x) =  x^3 + x^4 ,$$ or equivalently
$$\bu=\br_1+\br_2\text{rot}(\bh),$$

hence $\bu = (0,0,0,1,1,0,0),$
and since $s(x)r_2(x)= 1+x^2+x^5+x^6$ of weight $5>t$ we get
\begin{align*}
\bv &= \bm \bG + \bs \br_2 + \be   = (1,1,1,1,1,1,1) + (1,0,1,0,0,1,1) + (0,1,0,0,0,0,0) \\
& = (0,0,0,1,1,0,0).
\end{align*} We can then send the cipher
$$ \bc = (\bu, \bv) =  ( (0,0,0,1,1,0,0) , (0,0,0,1,1,0,0) ).$$
To decrypt the cipher, we compute with the knowledge of the secret key $ \mathcal{S} = ( y,z ) = (1, x^3)$ that  $u(x)z(x)=1+x^6$ and compute 
\begin{align*} \bv - \bu \mathbf{z} &=  (0,0,0,1,1,0,0) - (1,0,0,0,0,0,1) \\
& = (1,0,0,1,1,0,1),
\end{align*}
which gets decoded to to the codeword $(1,1,1,1,1,1,1)$, from which we recover the message $\bm=(1)$.
\end{example}

\begin{exercise}
Repeat this example with the fixed public parameters $\bG=(1,1,1,1,1,1,1)$, $h(x)=1+x+x^2$, $s(x)=1+x^3+x^4+x^5$, $w_e=w_r=1$ and the secret key $\mathcal{S}=(1,x^3)$, but now Bob chooses $e(x)=x^4, r_1(x)=1,r_2(x)=x$. Is the decryption successful in this case? 
\end{exercise}

\subsection{Augot-Finiasz Cryptosystem}

In its original version the Augot-Finiasz (AF) cryptosystem uses polynomial reconstructions, for this survey, however, we translate it into an easier formulation.

Similar to the quasi-cyclic framework, one can choose a code $\mathcal{C}$ which can efficiently decode $t$ errors and can make it public. The system does not rely on any  hiding of the structured code. 
 The idea of the AF and the FL system is publish a structured code $\mathcal{C}= \langle \bG \rangle$ which can correct $w$ erasures and $t$ errors, usually this means that $d>2t+w$. One then also publishes a corrupted codeword $\mathbf{y}= \mathbf{m}' \mathbf{G} + \mathbf{e}'$, where the error vector $\mathbf{e}'$ has weight $w$, but keeps the support of $\mathbf{e}'$ secret. 
Without the knowledge of the support, and as long as $w > \lfloor (d-1)/2\rfloor,$ an attacker cannot recover $\mathbf{m}'$ or equivalently $\mathbf{e}'.$

To encrypt a message $\mathbf{m}$, one chooses at random a vector $\mathbf{e}$ of  weight $t$, a random $\alpha \in \mathbb{F}_{q}$, such that $$\text{supp}(\alpha \mathbf{e}')=\text{supp}(\mathbf{e}')$$ and computes the cipher as 
$$ \mathbf{c}= \mathbf{m}\bG + \alpha \mathbf{y} + \be.$$
Clearly, the cipher is still a corrupted codeword of $\mathcal{C},$ where the error vector is 
$$\tilde{\mathbf{e}}= \alpha \mathbf{e}' + \mathbf{e}.$$ 
If $\mathbf{e}'$ and $\mathbf{e}$ are chosen at random then 
$\text{wt}_H(\tilde{\mathbf{e}})\geq w-t.$
Thus, as long as $w-t>\frac{d-1}{2}$ an attacker can still not decode the cipher without knowing the secret error support. 

On the other hand, the constructor of the scheme knows $\text{supp}(\be)$ and can use an erasure decoder to get rid off $\text{supp}(\mathbf{e}').$
Being left with at most $t$ errors, the constructor of the system can use the error-decoder of the public code and compute the $\mathbf{m}'+\alpha \mathbf{m}$. Finally, knowing $\mathbf{m}'$ and ensuring that $\alpha$ is visible in the vector $\alpha \mathbf{m}$, one recovers the message $\mathbf{m}.$

\begin{table}[h!]\small
\caption{AF Cryptosystem}
\label{AF}
\begin{tabular}{p{5.4cm}p{1cm}p{5cm}}
\hline
\textsf{ALICE} & & \multicolumn{1}{r}{\textsf{BOB}}\\
\hline
\hline
\multicolumn{1}{l}{KEY GENERATION} & & \\
\hline
Choose a generator matrix $\bG \in \mathbb{F}_{q}^{k \times n}$ which can correct $t$ errors and $w$ erasures\\
Choose $\be' \in \mathbb{F}_q^n$ of weight $w$ having support in $S$ \\ 
Choose $(1,\bm') \in  \mathbb{F}_{q}^k$\\
Compute $\by=(1,\bm')\bG+\be'$\\
The public key is  $\mathcal{P}= (\bG, \by, t)$ and the secret key is $\mathcal{S}= (\be')$\\
& $\xlongrightarrow{\mathcal{P}}$  &\\
\hline
 & & \multicolumn{1}{r}{ENCRYPTION} \\
 \hline 
 &&\multicolumn{1}{r}{Choose $\be \in \mathbb{F}_{q}^{n}$ with $wt_H(\be) \leq t$} \\
  &&\multicolumn{1}{r}{Choose $\alpha \in \mathbb{F}_q$} \\ 
 &&\multicolumn{1}{r}{Encrypt $\bm  \in \mathbb{F}_{q}^{k}$ as $\bc=\bm\bG+\alpha\by +\be$}\\
&$\xlongleftarrow{\bc}$&\\
\hline
\multicolumn{1}{l}{DECRYPTION} & & \\
\hline
Puncture $\bc$ in the positions indexed by $S$ & & \\ 
Decode $\bc_{S^C}$ and recover $\alpha(1,\bm')+\bm$ and thus $\alpha$ as well as $\bm$.  & & \\
\hline
\end{tabular}
\label{table:AF}
\end{table}

The decryption works, as 
$$\bc=(\bm+\alpha(1, \bm'))\bG+ \alpha\be' + \be$$ and $\alpha\be'$ has support in $S$. Thus,
$$\bc_{S^C}=(\bm+\alpha(1,\bm'))\bG+\be,$$ and since $\text{wt}_H(\be)\leq t$, we can decode the public code $\langle \bG\rangle$ and recover the message $\bm+\alpha(1,\bm')$. Although, we do not know $\alpha$, we have chosen the message of $\by$ such that we can read $\alpha$ of the first entry, namely $(1,\bm').$
Thus, we can remove $\alpha(1,\bm')$ from the recovered message and recover $\bm.$

\begin{example}
Let us give a toy example also for the AF system.
Let us consider $\mathbb{F}_{16}=\mathbb{F}_2[\alpha],$ where $\alpha^4=\alpha+1$
and the Reed-Solomon code generated by 
$$\bG= \begin{pmatrix}  1 & \alpha & \alpha+1 & \alpha^2 & \alpha^2+1 & \alpha^3 & \alpha^3+\alpha \\ 1 & \alpha^2 & \alpha^2+1 & \alpha+1 & \alpha & \alpha^3+\alpha^2  & \alpha^3\end{pmatrix}.$$ This code has minimum distance $d=n-k+1=6$ and can thus correct 1 error and 3 erasures. 
We choose the secret error support $S=\{1,2,4\}$ and the error vector $\be'=(1,\alpha,0,\alpha^2,0,0,0).$
For the message $\bm'=(1,1)$ we get 
$$\by= (1,1) \bG+\be'=(1, \alpha^2, \alpha^2+\alpha, \alpha+1,\alpha^2+\alpha+1,\alpha^2,\alpha).$$

Both $\bG$ and $\by$ are made public. 
Bob wants to send the message $(0,\alpha^2)$ to Alice and  chooses the scrambling $\alpha+1$ and the error vector $\be=(0,0,\alpha,0,0,0,0).$
The cipher is then given by 
\begin{align*}
\bc & =(0,\alpha^2)\bG + (\alpha+1)\by+\be \\ 
&= (\alpha^2+\alpha+1, \alpha^3+\alpha^2+\alpha+1,  \alpha^3+\alpha^2+\alpha+1,\alpha^3+1, 1,\alpha^3+1,0).
\end{align*}
To decrypt, Alice first punctures in the secret positions $\{1,2,4\}$, thus only considering 
$$\bc_{S^C}=(\alpha^3+\alpha^2+\alpha+1, 1,\alpha^3+1,0)$$
and decodes using the punctured Reed-Solomon code 
$$\bG_{S^C}=\begin{pmatrix}
    \alpha+1 & \alpha^2+1 & \alpha^3& \alpha^3+\alpha \\
    \alpha^2+1 & \alpha & \alpha^3+\alpha^2 & \alpha^3
\end{pmatrix},$$
getting the message $(\alpha+1,\alpha^2+\alpha+1)$ and the error vector $(\alpha,0,0,0).$
Due to the construction of the two messages, namely the first position of $\bm$ is zero and the first position of $\bm'$ is one, Alice can read of the first position the scrambling being $\alpha+1$ and thus recovers the message $$\bm=(0,\alpha^2)= (1+\alpha,\alpha^2+\alpha+1)-(\alpha+1,\alpha+1).$$

\end{example}

\begin{exercise}
    \begin{enumerate}
        \item An attacker can guess $\alpha \in \mathbb{F}_q$ and attack the AF system. What is the security level of the above example?
        \item Can we also choose different scramblings for $\by$?
        \item Repeat the example for Gabidulin codes and the rank metric. 
    \end{enumerate}
\end{exercise}

The only requirement for the code $\mathcal{C}$ is thus, that the punctured code can still efficiently decode. 

The original system uses GRS codes, as a punctured GRS code is still a GRS code, and has been attacked in \cite{augotattack}. 

Clearly, this framework is independent of the metric and hence, one could also employ the rank metric. 
In fact, the rank-metric analog of the AF system has been proposed by Faure and Loidreau \cite{FL}, relying the security on the
hardness of reconstructing $p$–polynomials.
Their original system proposes the use of Gabidulin codes and  has been subject to algebraic attacks \cite{FLattack}.

Many repair  attempts \cite{repair, LIGA, IntG, ramesses} have been made, unfortunately all  have been broken in \cite{bombar}. The idea of the attacks is to use list decoding of GRS codes, respectively of Gabidulin codes.

\subsection{GPT Cryptosystem}

The \emph{Gabidulin-Paramonov-Tretjakov} (GPT) cryptosystem was introduced in \cite{gpt} and is based on rank-metric codes. As usual, we pick an $\mathbb{F}_q$-basis of $\mathbb{F}_{q^m}$ and use this to identify elements of $\mathbb{F}_{q^m}$ with vectors in $\mathbb{F}_q^m$.  The system we present is not following the original proposal, which was broken \cite{overbeck}, but an adapted formulation, and as before we present the system as a framework, i.e., without choosing a family of codes for the secret code. 

The GPT system proceeds as follows. 
Alice chooses an $[n,k]$ linear rank-metric code $\mathcal{C}$  over $\mathbb{F}_{q^m}$ with error correction capacity $t$ and generator matrix $\bG.$  For some positive integer $\lambda$, she then chooses $\bS \in \text{GL}_k({q^m}), \bP \in \text{GL}_{n+\lambda}(q)$ and $\bX \in \mathbb{F}_{q^m}^{k \times \lambda}$ of rank $s \leq \lambda$.
She publishes the scrambled matrix $\bG'=\bS [ \bX \mid  \bG ] \bP$ and the target weight $t$.

Bob can then encrypt his message $\bm \in \mathbb{F}_{q^m}^k$, by computing
$$\bc= \bm\bG' + \be,$$
for some randomly chosen error vector $\be \in \mathbb{F}_{q^m}^{n+\lambda}$ with $\text{wt}_R(\be) =t.$
 
To decrypt, Alice can compute 
$$\bc\bP^{-1} = \bm \bS [ \bX \mid \bG ] +\be \bP^{-1}.$$
Since $\text{wt}_R(\be\bP^{-1})=t$, she can apply the decoding algorithm of the code  $\mathcal{C}$ to the last $n$ positions of $\bc\bP^{-1}$ to recover $\bm \bS $ and thus also $\bm.$
 
A systematic description of the GPT system can be found in Table \ref{gpttabular}.

\begin{table}[h!]\small
\caption{GPT Cryptosystem}
\label{gpttabular}
\begin{tabular}{p{5.4cm}p{1cm}p{5cm}}
\hline
\textsf{ALICE} & & \multicolumn{1}{r}{\textsf{BOB}}\\
\hline
\hline
\multicolumn{1}{l}{KEY GENERATION} & & \\
\hline
Choose a generator matrix $\bG \in \mathbb{F}_{q^m}^{k \times n}$ of a rank-metric code of rank distance $d=2t+1$ and a positive integer $\lambda$\\
Choose $\bS \in \text{GL}_k({q^m})$, $\bP \in \text{GL}_{n+\lambda}(q)$ \\
Choose a matrix $\bX \in \mathbb{F}_{q^m}^{k \times \lambda}$ of rank $s \leq \lambda$ and compute $\bG'= \bS[\bX \mid \bG]\bP.$\\ 
The public key is  $\mathcal{P}= (\bG', t)$ and the secret key is $\mathcal{S}= (\bG, \bS, \bX,\bP)$\\
& $\xlongrightarrow{\mathcal{P}}$  &\\
\hline
 & & \multicolumn{1}{r}{ENCRYPTION} \\
 \hline 
 &&\multicolumn{1}{r}{Choose $\be \in \mathbb{F}_{q^m}^{n+\lambda}$ with $wt_R(\be) \leq t$} \\
 &&\multicolumn{1}{r}{Encrypt $\bm  \in \mathbb{F}_{q^m}^k$ as $\bc=\bm\bG'+ \be$}\\
&$\xlongleftarrow{\bc}$&\\
\hline
\multicolumn{1}{l}{DECRYPTION} & & \\
\hline
Compute $\bc' = \bc\bP^{-1}$ and apply the decoding algorithm to the last $n$ positions to recover $\bm'=\bm\bS$ \\
Compute $\bm= \bm' \bS^{-1}$\\
\hline
\end{tabular}
\label{table:GPT}
\end{table}

This framework is closely related to the McEliece framework, as the algebraic code which can be efficiently decoded has to be kept secret and the matrix $\bP$ acts as an isometry. In fact, while for the Hamming metric $\bP$ is chosen  a permutation matrix, which fixes the Hamming weight of a vector, in the rank metric we choose $\bP$ to be a full rank matrix over $\mathbb{F}_q$, which thus fixes the rank weight of a vector over $\mathbb{F}_{q^m}.$

\begin{exercise}
Establish the Niederreiter version of the GPT system using the parity-check matrix. 
\end{exercise}

\begin{example}
We give an example for $n=4$, $m=5$, $k=2$ and $s=\lambda=1$. 
We identify $\mathbb{F}_{32} = \mathbb{F}_2[\alpha]$ with $\alpha^5=\alpha^2+1$ and consider the Gabidulin code with generator matrix
$$ \bG = \begin{pmatrix} 1 & \alpha & \alpha^2 & \alpha^3 \\
1 & \alpha^2 & \alpha^4 & \alpha^3+\alpha\end{pmatrix},$$
which can correct up to $1$ error.
We further need a $\bS \in \text{GL}_2({32})$ and a $\bP \in \text{GL}_5 (2)$ and $\bX$ of rank $s \leq \lambda = 1$, so we take
$$ \bS = \begin{pmatrix} 1 & \alpha \\
0 & 1
\end{pmatrix},$$
and for simplicity
$$ \bP = \begin{pmatrix}
0 & 0 & 1 & 0 & 0 \\
1 & 0 & 0 & 0 & 0 \\
0 & 1 & 0 & 0 & 0 \\
0 & 0 & 0 & 0 & 1 \\
0 & 0 & 0 & 1 & 0
\end{pmatrix},$$
and
$$ \bX = \begin{pmatrix}
1 \\
\alpha^2 +1 
\end{pmatrix}$$
We compute that
$$ \bG' = \bS [\bX \mid \bG] \bP =  \begin{pmatrix} \alpha+1 & \alpha^3+\alpha & \alpha^3 + \alpha +1 & \alpha^4 + \alpha^3 + \alpha^2 & 1 \\
1 & \alpha^2 & \alpha^2+1 & \alpha^3+\alpha & \alpha^4
\end{pmatrix}.$$
The public key is the pair
$$ \mathcal{P} = ( \bG', 1),$$
the secret key is
$$ \mathcal{P} = (\bG, \bS, \bX, \bP).$$
We want to encrypt the message
$$\bm = (\alpha+1, \alpha^2+1).$$
We choose the error vector
$$ \be = ( \alpha^3+1, 0, \alpha^3+1, \alpha^3+1, 0 ), $$
and compute
$$ \bc = \bm \bG' + \be = (\alpha^3+1,\alpha^3+\alpha,\alpha^2+1,\alpha^3+\alpha^2+\alpha+1, \alpha^4+\alpha^3+1) .$$
To decrypt $\bc$, we compute 
$$\bc' = \bc \bP^{-1} =(\alpha^2+1, \alpha^3+1, \alpha^3+\alpha,\alpha^4+\alpha^3+1,\alpha^3+\alpha^2+\alpha+1),$$
and use the decoding algorithm of Gabidulin codes to get 
$$\bm\bS=(\alpha+1,\alpha+1), $$
and by multiplying with 
$$\bS^{-1}= \begin{pmatrix}
1 & \alpha \\ 0 & 1
\end{pmatrix}$$ 
we recover $\bm$.
\end{example}

\section{Code-based Signature Schemes}\label{sec:sign}
We give two approaches of building a code-based signature, one is following the hash-and-sign approach \cite{gpv} of the CFS scheme \cite{cfs}, which can also be adapted to the rank metric and the second one is through  code-based ZK protocols, which can be turned into signature schemes via the Fiat-Shamir transform. 
\medskip

We later discuss their benefits and limitations, but in summary, hash-and-sign schemes often suffer from large public keys and distinguishing attacks, while signature schemes from ZK protocols suffer from large signature sizes. In Section \ref{sec:new} we will then present the novel submission to the additional standardization process of NIST and the respective solutions to these drawbacks.  
\subsection{Hash-and-Sign}\label{sec:hash}

 Hash-and-sign schemes follow directly the usual approach of transforming a public-key encryption scheme into a signature scheme.

In fact, a public key encryption scheme relies on a trapdoor function $f$, which is easy to compute and hard to invert. For the public key encryption scheme one applies $f$ on a message $m$ and gets the cipher $c=f(m)$. In order  to recover the message, an attacker has to invert $f$, which is mathematically a hard problem. However, the constructor with the secret key has access to $f^{-1}.$
\medskip

Similarly, in a signature scheme, one can use the same trapdoor function $f$, or equivalently the hard problem of computing $f^{-1}$. However, only the signer should have access to the secret key and be able to sign in her name, thus,  upon a message $m$ the signer computes  the signature $\sigma=f^{-1}(m)$ and everyone can verify the signature as $f(\sigma)=m$. For an impersonator, however, to find a valid signature for a message is difficult.
\medskip

We present the first such code-based hash-and-sign scheme, CFS  \cite{cfs},   and its rank-metric counterpart RankSign \cite{NISTRankSign}. 

\subsubsection{CFS Scheme}
We present the CFS scheme as framework in  Table \ref{cfstable}.

In the CFS scheme, one starts with a message $\bm$ to sign, and hopes that the hash of this message is the syndrome of a low weight vector, i.e., 
$\mathsf{Hash}(\bm)=\be\bH^\top$ for $\text{wt}_H(\be) \leq t.$

However, not many vectors are syndromes of low weight vectors.

 \begin{exercise}
    Show that in order for any vector to be a syndrome of a vector of weight up to $(d-1)/2$, we require a perfect code. 
\end{exercise}

Since $\mathsf{Hash}(\bm)$ is very likely not a syndrome of a vector of weight up to $t$, one introduces a counter $i.$ That is, one checks whether $\mathsf{Hash}(\bm,i) = \be\bH^\top$ for some $\be$ of weight up to $t$, and if this is not the case one chooses a different $i.$

For certain codes, this requires many iterations, which makes  the signing process slow.

Thus, the authors of \cite{cfs} propose the use of the only family of codes, which is suitable for such an approach, namely high rate Goppa codes. In fact, high rate Goppa codes provide the existence of such error vectors for a non-negligible proportion of syndromes. 

Unfortunately, the use of high rate Goppa codes is not safe, due to the distinguisher in \cite{high}. 
Note that this distinguisher does not break the CFS scheme in general, as it only proves that one of the two problems to which the security of the CFS scheme reduces can be solved in polynomial time. 
\medskip

In the key generation process, one chooses a parity-check matrix $\bH \in \mathbb{F}_2^{(n-k) \times n}$ of a binary code that can efficiently correct $t$ errors. One then hides the parity-check matrix as in the Niederreiter framework, by choosing an $n \times n$ permutation matrix $\bP$  and computing $\bH' = \bH\bP.$ The public key is then given by $\mathcal{P}= (\bH',t)$ and the secret key by $\mathcal{S}= (\bH,\bP)$.
\medskip

In the signing process, given a message $\bm$, one first chooses randomly $i$ and uses the decoding algorithm of $\mathcal{C}$ to find $\be$, such that $\text{wt}_H(\be) \leq t$ and $$\be\bH^\top = \mathsf{Hash}(\bm,i),$$ if possible. The signature is then given by $\sigma= (i, \be\bP).$

In the verification, the verifier checks that $\text{wt}_H(\be\bP) \leq t$ and if $$\be \bP \bH'^\top = \mathsf{Hash}(\bm, i).$$
\medskip

Recall that $\mathsf{Hash}$ is a publicly known hash function. 

        \renewcommand{\arraystretch}{1.5}
\begin{table}[h!]\small
\caption{CFS}
\label{cfstable}
\centering
\begin{tabular}{p{5.4cm}p{1cm}p{5cm}}
\hline
\textsf{PROVER} & & \multicolumn{1}{r}{\textsf{VERIFIER}}\\
\hline
\hline
\multicolumn{1}{l}{KEY GENERATION} & & \\
\hline
Choose a parity-check matrix $\bH \in \mathbb{F}_2^{(n-k) \times n}$ of $\mathcal{C}$, with error correction capacity $t$  & & \\
Choose an $n \times n$ permutation matrix $\bP$   & & \\ Compute  $\bH' = \bH\bP.$ The public key is then given by $\mathcal{P}= (\bH',t)$  & & \\ and the secret key by $\mathcal{S}= (\bH,\bP)$   & & \\
& $\xlongrightarrow{\mathcal{P}}$  &\\
\hline
\multicolumn{1}{l}{SIGNING} \\
 \hline 
Given a message $\bm$, choose random $i$  & & \\ Use the decoding algorithm of $\mathcal{C}$ to find $\be$, with  $\text{wt}_H(\be) \leq t$ and  $\be\bH^\top =  \mathsf{Hash}(\bm, i)$   & & \\ Sign as  $\sigma= (i, \be\bP)$ & & \\
&$\xlongrightarrow{m,s}$&\\
\hline
& & \multicolumn{1}{r}{VERIFICATION}  \\
\hline
& & \multicolumn{1}{r}{\strut \parbox[t]{5cm}{\raggedleft Check if $\text{wt}_H(\be\bP) \leq t$ and if $\be\bP \bH'^\top = \mathsf{Hash}(\bm, i).$\strut}} \\
\hline
\end{tabular}
\label{table:CFS}
\end{table}

\begin{exercise}
Show that $\be\bP \bH'^\top = \mathsf{Hash}(\bm,i).$
\end{exercise}

\begin{remark}

The signing time is inversely related to the proportion of vectors, which are syndromes of error vectors of weight $t\leq \frac{d-1}{2}$ and this proportion scales
badly with the error correction capacity of the code.

\end{remark}

The benefits of the hash-and-sign approach is that the signature is a single vector and thus quite small. 
\medskip

The public key on the other hand, is, as in the McEliece framework, a scrambled secret parity-check matrix, and thus of size $(n-k)k$ bits. 
\medskip

Additionally, the schemes can be vulnerable to distinguishers, i.e., an attacker might retrieve the secret code, as seen in \cite{cfs}.  
\medskip

\begin{example}
    Let us consider also here a small toy example. Let $\mathbb{F}_{8}=\mathbb{F}_2[\alpha]$ and $\alpha^3=\alpha+1$. 
    
    Let us consider the Goppa polynomial $$g(x)=x^2+x+1$$ and the evaluation points $$1,\alpha, \alpha+1,\alpha^2,\alpha^2+\alpha,\alpha^2+1,\alpha^2+\alpha+1.$$ We can compute 
    \begin{align*}
        g(1)^{-1} & =1, \\ 
        g(\alpha)^{-1} &= g(\alpha+1)^{-1}=\alpha^2, \\ 
        g(\alpha^2)^{-1} & =g(\alpha^2+1)^{-1}=\alpha^2+\alpha, \\ 
        g(\alpha^2+\alpha)^{-1} &= g(\alpha^2+\alpha+1)^{-1}=\alpha.
    \end{align*}
    Then, 
    \begin{align*} \tilde{\bH} & = \begin{pmatrix} 1 & 1 & 1 & 1  & 1 &1&1 \\ 1 & \alpha & \alpha+1 & \alpha^2 & \alpha^2+\alpha & \alpha^2+1 & \alpha^2+\alpha+1\end{pmatrix} \text{diag}(1, \alpha^2,\alpha^2, \alpha^2+\alpha, \alpha, \alpha^2+\alpha,\alpha )  \\ & = \begin{pmatrix}
       1 &  \alpha^2 & \alpha^2 &  \alpha^2+\alpha &  \alpha &  \alpha^2+\alpha & \alpha \\
       1 & \alpha+1 & \alpha^2+\alpha+1 & \alpha^2+1 & \alpha^2+\alpha+1 & \alpha+1 & \alpha^2+1
    \end{pmatrix}.
    \end{align*}
    Using the basis $\Gamma= \{1,\alpha,\alpha^2\}$, the parity-check matrix of the Goppa code is then 
    $$\bH= \begin{pmatrix}
     1 & 0 & 0 & 0 & 0  & 0 & 0 \\
    0 & 0 & 0 & 1 & 1  & 1 & 1 \\
     0 & 1 & 1 & 1 & 0  & 1 & 0 \\
      1 & 1 & 1 & 1 & 1 & 1 & 1 \\
       0 & 1 & 1 & 0 & 1  & 1 & 0 \\
        0 & 0 & 1 & 1 & 1  & 0 & 1 \\
    \end{pmatrix}.$$
    The Goppa code $\langle \bH \rangle^\perp$ has minimum distance at least $3$, and can thus correct at least $t=1$ error.
\medskip
    
    The prover chooses the permutation matrix $\bP$, permuting the first two columns and  publishes
      $$\bH'= \begin{pmatrix}
     0 & 1 & 0 & 0 & 0  & 0 & 0 \\
    0 & 0 & 0 & 1 & 1  & 1 & 1 \\
     1 & 0 & 1 & 1 & 0  & 1 & 0 \\
      1 & 1 & 1 & 1 & 1 & 1 & 1 \\
       1 & 0 & 1 & 0 & 1  & 1 & 0 \\
        0 & 0 & 1 & 1 & 1  & 0 & 1 \\
    \end{pmatrix}.$$
    Note that any syndrome of a weight 1 vector is simply given by one column of $\bH$. Thus, there exist 7 possible syndromes.
\medskip
    
    Given a message $\bm$ and a random $i=(1,0,1,1)$, the prover computes the hash of $(\bm,i)$. We assume that the hash function outputs $(1,0,1,0,0,1,0)$. 
    \medskip
    
    Unfortunately, this is not a syndrome of a weight one vector. The prover chooses a different $i$ and gets the hash $(1,0,0,1,0,0)$. Using the syndrome decoder of the Goppa code, the prover finds $$\be=(0,1,0,0,0,0,0)$$ and computes the signature  
    $$\sigma=(i, (1,0,0,0,0,0,0)).$$
\medskip
    
    The verifier checks that $\be\bP$ has indeed weight 1 and computes 
    $$\be\bP\bH'^\top =(1,0,0,0,0,0,0)\begin{pmatrix}
     0 & 1 & 0 & 0 & 0  & 0 & 0 \\
    0 & 0 & 0 & 1 & 1  & 1 & 1 \\
     1 & 0 & 1 & 1 & 0  & 1 & 0 \\
      1 & 1 & 1 & 1 & 1 & 1 & 1 \\
       1 & 0 & 1 & 0 & 1  & 1 & 0 \\
        0 & 0 & 1 & 1 & 1  & 0 & 1 \\
    \end{pmatrix}^\top = (1,0,1,0,0,1,0).$$
    The verifier accepts the signature as 
    $$\mathsf{Hash}(\bm,i)= (1,0,1,0,0,1,0).$$
\end{example}

The random $i$, is usually chosen as a seed, denoted by $\text{seed} \in \{0,1\}^\ell$.

\subsubsection{RankSign}

RankSign \cite{NISTRankSign}, as a framework, is the rank-metric analog of CFS. 
The authors  propose to use augmented LRPC codes over an extension field $\mathbb{F}_{q^m}$ and introduce a mixture of erasures and errors, which can be efficiently decoded. 
\medskip

In the key generation process,  instead of hiding the parity-check matrix $\bH$ of the LRPC code over $\mathbb{F}_{q^m}$ as usual, i.e., using $\bS\bH\bP$, where $\bS \in \text{GL}_{n-k}({q^m})$ and $\bP \in \text{GL}_n(q)$, we first add some random columns to $\bH$. This is similar to the scrambling used in the GPT system.

        \renewcommand{\arraystretch}{1.5}
\begin{table}[h!]\small
\caption{RankSign}
\centering
\begin{tabular}{p{5.4cm}p{1cm}p{5cm}}
\hline
\textsf{PROVER} & & \multicolumn{1}{r}{\textsf{VERIFIER}}\\
\hline
\hline
\multicolumn{1}{l}{KEY GENERATION} & & \\
\hline
Choose $\bS \in \text{GL}_{n-k}({q^m}), \bP \in \text{GL}_{n+t}(q),$  & & \\ Choose $ r, \ell \in \mathbb{N}, \bX \in \mathbb{F}_{q^m}^{(n-k) \times t'}$ & & \\ Choose $ \bH \in \mathbb{F}_{q^m}^{(n-k)\times n}$ a parity-check matrix of a LRPC code & & \\
Compute  $\bH' = \bS (\bX \mid \bH ) \bP$&  & \\
The keys are given by $\mathcal{S}= (\bS, \bP, \bX, \bH),$ & & \\ and $ \mathcal{P}= (\bH', \ell,r)$&  &\\
& $\xlongrightarrow{\mathcal{P}}$  &\\
\hline
\multicolumn{1}{l}{SIGNING} \\
 \hline 
 Choose $\tilde{\be} \in \mathbb{F}_{q^m}^t$ and a message $\bm$ && \\
Choose $ \text{seed} \in \{0,1\}^\ell$ & &  \\  Compute $\bm' =\mathsf{Hash}(\bm \mid \text{seed})$ && \\
Set $\bs' = \bm' (\bS^{-1})^\top - \tilde{\be}\bX^\top$  && \\
Find $\be'$, such that $\text{wt}_R(\be')=r$ and $\be'\bH^\top = \bs' $ & & \\
Set $\be=(\tilde{\be} \mid \be')(\bP^\top)^{-1}$ and $\sigma= (\be, \text{seed})$
&$\xlongrightarrow{\bm,\sigma}$&\\
\hline
& & \multicolumn{1}{r}{VERIFICATION}  \\
\hline
& & \multicolumn{1}{r}{\parbox[t]{5cm}{\raggedleft Check if $\text{wt}_R(\be)= r$ and if $\be\bH'^\top = \mathsf{Hash}(\bm, \text{seed})$\strut}} \\
\hline
\end{tabular}
\label{table:RS}
\end{table}

Let $\bS \in \text{GL}_{n-k}({q^m}), \bP \in \text{GL}_{n+t}(q)$ and $\bX \in \mathbb{F}_{q^m}^{(n-k) \times t'}.$ Typically one sets $t'=t$, but one could also use other choices.

Then, one hides $\bH$ by computing $\bH' = \bS (\bX \mid \bH ) \bP.$
\medskip

While $\bH'$ and some integer $\ell$ are publicly known, the secret key is given by $\bX, \bH, \bS, \bP.$

In the signing process, one first chooses randomly $\tilde{\be} \in \mathbb{F}_{q^m}^t$ and hashes a message $\bm$ and a seed, denoted by $\text{seed} \in \{0,1\}^\ell$   to get $\bm' = \mathsf{Hash}(\bm \mid \text{seed}) \in \mathbb{F}_{q^m}^{n-k}.$

Then one sets a syndrome $$\bs' = \bm' (\bS^{-1})^\top - \tilde{\be}\bX^\top$$ and  tries to syndrome decode this syndrome $\bs'$ using $\bH$.

If one succeeds, that is, there exists a $\be' \in \mathbb{F}_{q^m}^n$ of rank weight $r= t+r'$ and  such that $$\be'\bH^\top = \bs',$$ then one defines $$ \be= (\tilde{\be} \mid \be') (\bP^\top)^{-1}$$ and sets the signature 
$$\sigma= (\be, \text{seed}).$$

If not, this process needs to be repeated until one succeeds.
\medskip

In the verification, the verifier checks that $\text{wt}_R(\be)= r= t+r',$ and if $$\be\bH'^\top = \bm' = \mathsf{Hash}(\bm \mid \text{seed}).$$

\begin{exercise}
Show that $\be\bH'^\top = \bm'.$
\end{exercise}

We want to note here that this signature scheme was later attacked in \cite{attackranksign}. 

\subsection{Code-Based ZK Protocols}\label{sec:ZKID}

As described in Section \ref{sec:fiat}, digital signature schemes can be constructed from a ZK protocol using the Fiat-Shamir transform \cite{fiatshamirtransform}.  In this section, we present two famous ZK protocols for this purpose, namely the scheme by Cayrel, V\'eron and El Yousfi Alaoui (CVE) \cite{cve} and scheme by Aguilar, Gaborit and Schrek (AGS) \cite{ags}. 
\medskip

The CVE scheme \cite{cve} is an improvement of Stern's \cite{SternZK} and V\'eron's \cite{veron} protocols, which are both based on the hardness of decoding a random binary code \cite{berlekamp}. 
The CVE scheme relies on codes over a large finite field. With this choice, the cheating probability for a single round is reduced from $2/3$ of Stern's 3-pass scheme to $\frac{q}{2(q-1)}$. 
\medskip

The idea of the scheme is the following: the secret key is given by a random error vector of weight $t$ and the public key is a parity-check matrix together with the syndrome of this error vector. The challenges are requesting either a response that shows that the error vector has indeed weight $t$ or a response that shows that the error vector solves the parity-check equations. 
\medskip

The scheme is of large interest, as it uses an actual random linear code,  which is possible since no decoding process is required. The security of this scheme, thus, fully relies on the hardness of decoding a random linear code and not on the indistinguishability of a secret code. 
\medskip

Let $\sigma$ be a permutation of $\{1, \ldots, n\}$ and for $\bv \in \left(\mathbb{F}_q^\star\right)^n$ and $\ba \in \mathbb{F}_q^n$ we denote by $$ \sigma_{\bv}(\ba) =\sigma(\bv) \star \sigma(a), $$ where $\star$ denotes the component-wise product.  

\begin{table}[h!]\small
\caption{CVE Scheme}
\centering
\begin{tabular}{p{5.4cm}p{1cm}p{5cm}}
\hline
\textsf{PROVER} & & \multicolumn{1}{r}{\textsf{VERIFIER}}\\
\hline
\hline
\multicolumn{1}{l}{KEY GENERATION} & & \\
\hline
Choose the parameters $q,n,k,t$ and a hash function $\mathsf{Hash}$ \\
Choose $\be \in B_H(t,n,q)$ and a parity-check matrix \\
$\bH \in \mathbb{F}_q^{(n-k) \times n}$. Compute the syndrome \\
$\bs = \be \bH^\top \in \mathbb{F}_q^{n-k}.$ \\
The public key is given by $\mathcal{P}= (\bH, \bs,t)$
&$\xlongrightarrow{\mathcal{P}}$&\\
\hline
 & & \multicolumn{1}{r}{VERIFICATION} \\
 \hline 
Choose $\bu\in \mathbb F_q^n$, a permutation $\sigma$, $\bv\in{(\mathbb{F}_q^{\times})}^n$  &  & \\
Set $c_0 = \mathsf{Hash}\big(\sigma,\bv,\bu\bH^\top\big)$  &  &\\
Set $c_1 = \mathsf{Hash}\big(\sigma_{\bv}(\bu), \sigma_{\bv}(\be)\big)$  &  &\\
 & $\xlongrightarrow{c_0,c_1}$  &\\
 &&\multicolumn{1}{r}{Choose $z\in \mathbb F_q^\star$}\\
&$\xlongleftarrow{z}$&\\
Set $\by =  \sigma_{\bv}(\bu+z\be)$ &&\\
&$\xlongrightarrow{\by}$&\\
& & \multicolumn{1}{r}{Choose $b\in \{0,1\}$}\\
& $\xlongleftarrow{b}$ &\\
If $b=0$, set $r=(\sigma,\bv)$ & & \\
If $b=1$, set $r=\sigma_{\bv}(\be)$ & & \\
&$\xlongrightarrow{r}$&\\
&&\multicolumn{1}{r}{If $b=0$, accept if}\\
&&\multicolumn{1}{r}{$c_0 = \mathsf{Hash}\big(\sigma,\bv,\sigma_{\bv}^{-1}(\by)\bH^\top-z \bs\big)$}\\
&&\multicolumn{1}{c}{or}\\
&&\multicolumn{1}{r}{If $b=1$, accept if $\mathrm{wt}_{\mathrm H}(\sigma_\bv(\be)) = t$ and}\\
&&\multicolumn{1}{r}{$ c_1 = \mathsf{Hash}\big(\by - z\sigma_{\bv}(\be),\sigma_{\bv}(\be)\big)$}\\
\hline
\end{tabular}
\label{table:CVE}
\end{table}

We now  show how the communication cost of this scheme is derived, following the reasoning of \cite{rest}.

In order to represent a vector of length $n$  and Hamming weight $t$ over $\mathbb F_q$, we can either use the full vector, which  requires $n\left\lceil \log_2(q)\right\rceil$ bits, or just consider its support, together with the ordered non-zero entries, resulting in $$t\big(\left\lceil\log_2(n)\right\rceil+\left\lceil\log_2(q-1)\right\rceil\big)$$ bits.
Thus the most convenient choice for a given set of parameters $n$, $t$ and $q$ is  $$\psi(n,q,t) = \min\{n\left\lceil \log_2(q)\right\rceil,t\big(\left\lceil\log_2(n)\right\rceil+\left\lceil\log_2(q-1)\right\rceil\big)\}.$$
Since random objects, such as the monomial transformation, are completely determined by the seed for the pseudo-random generator, they can also be compactly represented as such, whose length is denoted by $l_{\textsf{Seed}}$. Also the length of the hash values will be denoted by $l_\mathsf{Hash}$. 
Using the compression technique for $N$ rounds of the protocol we get the following average communication cost: 
$$l_{\mathsf{Hash}}+N \bigg(\left\lceil\log_2(q-1)\right\rceil+n\left\lceil\log_2(q)\right\rceil+1+l_\mathsf{Hash}+\frac{\psi(n,q,t)+l_{\textsf{Seed}}}{2}\bigg).$$
For the maximal communication cost, we take the maximum size of the response, and thus we obtain
$$l_\mathsf{Hash}+N \bigg(\left\lceil\log_2(q-1)\right\rceil+n\left\lceil\log_2(q)\right\rceil+1+l_\mathsf{Hash}+\max\{\psi(n,q,t)\hspace{1mm},\hspace{1mm}l_{\textsf{Seed}}\}\bigg).$$

Let us fix $t= \left\lfloor (d_H-1)/2\right\rfloor$, for $d_H$ denoting the minimum distance of the Gilbert-Varshamov bound.  
The authors of \cite{cve} have used the analysis due to Peters \cite{peters} to estimate the information set decoding complexity, and have proposed two parameters sets:
\begin{enumerate}
    \item[-] $q=256$, $n=128$, $k=64$, $t=49$, for 87-bits security, having a communication cost of 3.472 kB;
    \item[-] $q=256$, $n=208$, $k=104$, $t=78$, for 128-bits security, having a communication cost of 43.263 kB.
\end{enumerate}

\begin{exercise}
Show  the zero-knowledge property and the completeness property for the CVE scheme.
\end{exercise}

An easy attempt for an impersonator would be to guess the challenge $b$ before sending the commitments.

Thus, the strategy if we guess $b=0$, would be to choose an error vector $\be'$, which satisfies the parity-check equations, that is $$\bs=\be'\bH^\top,$$
and to forget about the weight condition. This can easily be achieved using linear algebra. 
We denote by $s_0$ the strategy for $b=0$, which in detail requires to choose randomly $\bu', \sigma'$ and $\bv'$ according to the scheme and to send the commitments $c_0' = \mathsf{Hash}(\sigma', \bv', \bu'\bH^\top)$ and a random $c_1'$.  When the impersonator received a $z \in \mathbb{F}_q^\star$, the impersonator now computes $\by'$ according to the cheating error vector $\be'$, i.e.,
$$\by' = \sigma'_{\bv'}(\bu'+z\be').$$
The impersonator wins, if the verifier now asks for $b=0$, since the verifier will check 
\begin{align*} 
c_0' & = \mathsf{Hash}(\sigma', \bv', {\sigma'_{\bv'}}^{-1}(\by') \bH^\top - z\bs) \\
& = \mathsf{Hash}(\sigma', \bv', {\sigma'_{\bv'}}^{-1}( \sigma'_{\bv'}(\bu'+z\be') ) \bH^\top - z\bs)  \\ 
& = \mathsf{Hash}(\sigma', \bv', (\bu'+z\be')  \bH^\top - z\bs)  \\ 
& = \mathsf{Hash}(\sigma', \bv', \bu'\bH^\top +z\be' \bH^\top - z\bs)  \\ 
& = \mathsf{Hash}(\sigma', \bv', \bu'\bH^\top +z\bs - z\bs).
\end{align*}
If the verifier asks for $b=1$, the impersonator looses. 

Whereas the strategy if we guess $b=1$, would be to choose an error vector $\be'$, which has the correct weight, i.e., $\text{wt}_H(\be')=t$, but does not satisfy the parity-check equations.
We denote by $s_1$ the strategy for $b=1$, which in detail requires to choose randomly $\bu', \sigma'$ and $\bv'$ according to the scheme and to send the commitments: a random $c_0'$ and $c_1'= \mathsf{Hash}( \sigma'_{\bv'}(\bu'), \sigma'_{\bv'}(\be'))$.  When the impersonator received a $z \in \mathbb{F}_q^\star$, the impersonator now computes $\by'$ according to the cheating error vector $\be'$, i.e.,
$$\by' = \sigma'_{\bv'}(\bu'+z\be').$$
The impersonator wins, if the verifier now asks for $b=1$, since the verifier will check  if $\text{wt}_H(\sigma'_{\bv'}(\be'))=t$ and 
\begin{align*} 
c_1' & = \mathsf{Hash}( \by' - z \sigma'_{\bv'}(\be'), \sigma'_{\bv'}(\be')) \\
& = \mathsf{Hash}( \sigma'_{\bv'}(\bu'+z\be') - z \sigma'_{\bv'}(\be'), \sigma'_{\bv'}(\be')) \\
& = \mathsf{Hash}( \sigma'_{\bv'}(\bu') + \sigma'_{\bv'}(z\be') - z \sigma'_{\bv'}(\be'), \sigma'_{\bv'}(\be')).
\end{align*}
If the verifier asks for $b=0$, the impersonator looses.

With this easy strategy, one would get a cheating probability of $1/2$, which just corresponds to choosing the challenge $b$ correctly. 
However,  by also guessing $z$ correctly one can  improve the above strategy.
\begin{proposition}
The cheating probability of the CVE scheme is  $\frac{q}{2(q-1)}.$
\end{proposition}

\begin{proof}
We modify the easy strategies $s_i$, following \cite{cve}:

Let us denote by $s_0'$ the improved strategy on $s_0$, which works as follows: recall that $\be'$ is chosen such that the parity-check equations are satisfied but not the weight condition. Instead of randomly choosing the commitment $c_1'$, we choose a $z' \in \mathbb{F}_q^\star$ and a second cheating error vector $\tilde{\be}$ of weight $t$, we compute a $\tilde{\by}= \sigma'_{\bv'}(\bu'+z'\be') $ with this guess and  compute $$c_1' = \mathsf{Hash}(\tilde{\by} - z'\tilde{\be}, \tilde{\be}).$$ 
When we receive a $z$ from the verifier, we check if we made the correct choice, that is: if $z=z'$, we send the pre-computed $\tilde{\by}$, and if $z\neq z'$ we compute $\by' =\sigma'_{\bv'}(\bu'+z\be') $. 
If the verifier asks for $b=0$, we use the usual strategy of $s_0$ and will get accepted, as before.
If the verifier asks for $b=1$, we send as answer $\tilde{\be}$. If we have guessed correctly and $z=z'$, we will get accepted also in this case as
\begin{align*} 
c_1' & = \mathsf{Hash}( \tilde{b}\by - z\tilde{\be}, \tilde{\be}) 
\end{align*}
by definition.

Let us denote by $s_1'$ the improved strategy on $s_1$, which works as follows: recall that $\be'$ is chosen having the correct weight. Instead of randomly choosing the commitment $c_0'$, we choose a $z' \in \mathbb{F}_q^\star$ and compute a $\tilde{\by}= \sigma'_{\bv'}(\bu'+z'\be') $ with this guess and  compute  $$c_0'= \mathsf{Hash}(\sigma', \bv', \bu'\bH^\top + z'(\be'\bH^\top - \bs) ).$$
When we receive a $z$ from the verifier, we check if we made the correct choice, that is: if $z=z'$, we send the pre-computed $\tilde{\by}$, and if $z\neq z'$ we compute $\by' =\sigma'_{\bv'}(\bu'+z\be') $. 
If the verifier asks for $b=1$, we use the usual strategy of $s_1$ and will get accepted.
If the verifier asks for $b=0$, we send as answer $(\sigma',\bv')$. If we have guessed correctly and $z=z'$, we will get accepted also in this case as
\begin{align*} 
c_0' & = \mathsf{Hash}(\sigma', \bv', {\sigma'_{\bv'}}^{-1}(\by') \bH^\top - z\bs) \\
& = \mathsf{Hash}(\sigma', \bv', {\sigma'_{\bv'}}^{-1}( \sigma'_{\bv'}(\bu'+z'\be') ) \bH^\top - z\bs)  \\ 
& = \mathsf{Hash}(\sigma', \bv', (\bu'+z'\be')  \bH^\top - z\bs)  \\ 
& = \mathsf{Hash}(\sigma', \bv', \bu'\bH^\top +z'\be' \bH^\top - z\bs)  \\ 
& = \mathsf{Hash}(\sigma', \bv', \bu'\bH^\top +z'\bs - z\bs).
\end{align*}

Thus,  the probability that an impersonator following the strategy $s_i'$ will get accepted is given by
$$P(b=i) + P(b=1-i)\cdot P(z=z') = \frac{1}{2} + \frac{1}{2} \cdot \frac{1}{q-1} = \frac{q}{2(q-1)}, $$
which concludes this proof.
\end{proof}

The second ZK protocol we want to present is the scheme by Aguilar, Gaborit and Schrek \cite{ags}, which we will denote by AGS.  This scheme is constructed upon quasi-cyclic codes over $\mathbb F_2$. 
Let us consider a vector $\ba\in\mathbb F_2^{jk}$ divided into $j$ blocks of $k$ entries each, that is,
\[
\ba=\left(a^{(1)}_1,\ldots,a^{(1)}_{k}, \ldots, a^{(j)}_1,\ldots,a^{(j)}_{k}\right).
\]
Let $\rho_i^{(k)}$ denote a function that performs a block-wise cyclic shift of $\ba$ by $i$ positions, i.e.,
$$\rho_i^{(k)}(\ba)=\left( a^{(1)}_{1-i \mod k},\ldots, a^{(1)}_{k-i \mod k}, \ldots, a^{(j)}_{1-i \mod k},\ldots, a^{(j)}_{k-i \mod k} \right).$$

The idea is similar to that of the CVE scheme, but working with the generator matrix instead. 

The secret key consists of a message and an error vector, while the public key consists of an erroneous codeword and the generator matrix. The challenges either require the proof of the error vector having the correct weight or  of the knowledge of the message.

When performing $N$ rounds, the average communication cost is
$$l_{\mathsf{Hash}} + N\bigg(\left\lceil\log_2(k)\right\rceil+1+2l_{\mathsf{Hash}}+\frac{l_{\textsf{Seed}}+k+n+\psi(n,t,2)}{2}\bigg),$$
while the maximum communication cost is
$$l_{\mathsf{Hash}} + N\bigg(\left\lceil\log_2(k)\right\rceil+1+2l_{\mathsf{Hash}}+\max\{l_{\textsf{Seed}}+k\hspace{1mm},\hspace{1mm}n+\psi(n,t,2)\}\bigg).$$
In \cite{ags}, three parameters sets are proposed:
\begin{enumerate}
    \item[-] $n=698$, $k=349$, $t=70$, for $81$-bits security, having a communication cost of $2.5$ kB; 
    \item[-] $n=1094$, $k=547$, $t=109$, for $128$-bits security, with communication cost of $28$ kB.
\end{enumerate}

\begin{exercise}
Show  the zero-knowledge property and completeness for the AGS scheme.
\end{exercise}
 
 We remark that in a code-based ZK protocol   one does not require a code with an efficient decoding algorithm. Which stands in contrast to the requirements for many of the  code-based public-key encryption schemes. Thus, choosing a random code the security of such schemes is much closer related to the actual NP-hard problem of decoding a random linear code.

\begin{table}[h!]\small
\caption{AGS Scheme}
\centering
\begin{tabular}{p{5.4cm}p{1cm}p{5cm}}
\hline
\textsf{PROVER} & & \multicolumn{1}{r}{\textsf{VERIFIER}}\\
\hline
\hline
\multicolumn{1}{l}{KEY GENERATION} & & \\
\hline
Choose the parameters $n,k, t$ and a hash function $\text{Hash}$ \\
Choose $\bm \in \mathbb{F}_2^k$ and $\be \in B_H(t,n,2)$ and   \\
generator matrix $\bG \in \mathbb{F}_2^{k \times n}$. \\  Compute the erroneous codeword
$\bc = \bm \bG+\be \in \mathbb{F}_2^{n}$ \\
The public key is given by $\mathcal{P}= (\bG,\bc,t)$
&$\xlongrightarrow{\mathcal{P}}$&\\
\hline
 & & \multicolumn{1}{r}{VERIFICATION} \\
 \hline 
Choose $\bu\in \mathbb F_2^k$, a permutation $\sigma$  &  & \\
Set $c_0 = \mathsf{Hash}\big(\sigma\big)$  &  &\\
Set $c_1 = \mathsf{Hash}\big(\sigma (\bu\bG)\big)$ &  &\\
 & $\xlongrightarrow{c_0,c_1}$  &\\
 &&\multicolumn{1}{r}{Choose $z\in \{1, \ldots, k\}$}\\
&$\xlongleftarrow{z}$&\\
Set $c_2 = \mathsf{Hash}\big(\sigma(\bu\bG+ \rho_{z}^{(k)}(\be))\big)$ &&\\
&$\xlongrightarrow{c_2}$&\\
& & \multicolumn{1}{r}{Choose $b\in \{0,1\}$}\\
& $\xlongleftarrow{b}$ &\\
If $b=0$, set $r=(\sigma,\bu+\rho_z^{(k)}(\bm))$ & & \\
If $b=1$, set $r=(\sigma(\bu\bG), \sigma( \rho_{z}^{(k)}(\be)))$ & & \\
&$\xlongrightarrow{r}$&\\
&&\multicolumn{1}{r}{If $b=0$, accept if $c_0 = \mathsf{Hash}\big(\sigma\big)$ and}\\
&&\multicolumn{1}{r}{$c_2= \mathsf{Hash}\big((\bu+\rho^{(k)}_{z}(\bm))\bG+\rho^{(k)}_{z}(\bc)  \big)$}\\
\multicolumn{3}{r}{If $b=1$, accept if $\mathrm{wt}_{\mathrm H}( \rho_{z}^{(k)}(\be)) = t$}\\
&&\multicolumn{1}{r}{and $ c_1=\mathsf{Hash}\big(\sigma(\bu\bG)\big)$ and  }\\
&&\multicolumn{1}{r}{$c_2= \mathsf{Hash}\big(\sigma(\bu\bG)+\sigma( \rho_{z}^{(k)}(\be))\big)$}\\
\hline
\end{tabular}
\label{table:ags}
\end{table}
\newpage

Clearly, using any of  the two code-based ZK protocols presented above and the Fiat-Shamir transform one  immediately gets  a signature scheme.
~

 \newpage
\section{Security Analysis}\label{sec:security}

In the security analysis of a cryptographic scheme we make a difference between two main attack approaches: 
\begin{enumerate}
    \item structural attacks,
    \item non-structural attacks.
\end{enumerate}
A structural attack aims at exploiting the algebraic structure of the cryptographic system.

Whereas a non-structural attack tries to combinatorically recover the message or the secret key without exploiting any algebraic structure.
\medskip

For example the security of the McEliece and Niederreiter type of cryptosystems rely on two assumptions. The first one being   
\begin{center} \emph{The public code is not distinguishable from a random code.} \end{center}
A structural attack would usually aim at exactly this assumption, and try to recover the secret code, if the scrambled public version of it does not behave randomly.

Clearly, structural or algebraic attacks heavily depend on the chosen secret codes for the cryptosystem, if the system depends on an algebraic code that is efficiently decodable, and is not attacking the presented frameworks in general.
\medskip

Assuming that this first assumption is met, however, the   security of most code-based cryptosystems relies also on this second assumption
\begin{center} \emph{Decoding a random linear code is hard/ infeasible.} \end{center}
A non-structural attack on the McEliece cryptosystem would, thus, assume that the public code is in fact random, and rather try to decode this random code. 

\medskip

In general we also speak of attacks in terms of: \emph{key-recovery attacks}, where an attacker tries to recover the secret key (usually structural attacks), and \emph{message-recovery attacks}, where an attacker directly tries to decrypt the cipher without first recovering the secret key.
\medskip

Code-based cryptography is rapidly advancing and new cryptosystems are basing their security on novel problems from algebraic coding theory.
\medskip

In the following we list  the main problems used in cryptography and discuss their hardness.

\subsection{Problems from Coding Theory}

The most prominent problem in algebraic coding theory is the decoding problem:

 \begin{problem}{\textbf{Decoding Problem (DP)}} \label{DP}
Let $\mathbb{F}_q$ be a finite field and $k \leq n$ be positive integers.   Given $\bG \in \mathbb{F}_q^{k\times n}$, $\br \in \mathbb{F}_q^{n}$ and  $t \in \mathbb N$, is there a vector $\bm  \in \mathbb{F}_q^k$ and $\be \in \mathbb{F}_q^n$ of weight less than or equal to $t$ such that $\br=\bm\bG + \be$?
\end{problem}
 Note that the DP formulated through the generator matrix is equivalent to the syndrome decoding problem, which is formulated through the parity-check matrix.

\begin{problem}{\textbf{Syndrome Decoding Problem (SDP)}} \label{prob:SDP}
Let $\mathbb{F}_q$ be a finite field and $k \leq n$ be positive integers.   Given $\bH \in \mathbb{F}_q^{(n-k) \times n}$, $\bs \in \mathbb{F}_q^{n-k}$ and $t \in \mathbb N$, is there a vector $\be \in \mathbb{F}_q^n$ such that $\text{wt}_H(\be)\leq t$ and $\be \bH^\top = \bs$?
\end{problem}

These two problems are also equivalent to the Given Weight Codeword Problem:

\begin{problem}{\textbf{Given   Weight Codeword Problem (GWCP)}}\label{prob:GWCP}
\\ Let $\mathbb{F}_q$ be a finite field and $k \leq n$ be positive integers. Let $k  \leq n$ be positive integers. Given $\bH\in \mathbb{F}_q^{(n-k)\times n}$ and $w\in\mathbb{N}$, is there a vector $\mathbf c \in \mathbb{F}_q^n$ such that  $\text{wt}_H(\mathbf c) =w$  and $ \mathbf c\mathbf H^\top = \mathbf 0_{n-k}$? 
\end{problem}

\begin{theorem}
    The DP, SDP and GWCP are equivalent.
\end{theorem}

\begin{proof}
    Let us start with showing that the DP and SDP are equivalent.  For this we start with an instance of DP, i.e., $\bG,\br,t$. We can then transform this instance to an instance of the SDP. In fact, we can bring $\bG$ into systematic form, that is 
    $$\bG'= \begin{pmatrix} \text{Id}_k & \bA \end{pmatrix}$$ and immediately get a parity-check matrix for the same code 
    $$\bH= \begin{pmatrix}
        -\bA^\top & \text{Id}_{n-k}
    \end{pmatrix}.$$
    We can then multiply $\bH$ to the received vector $\br= \bm\bG + \be$, getting the syndrome 
    $$\bs= \br\bH^\top = \be\bH^\top.$$
    Hence, if we can solve the SDP on the instance $\bH, \bs,t$, thus finding $\be$, we have also solved DP. 

    On the other hand, given an instance of SDP, i.e., $\bH,\bs,t$, we can find an instance of DP. In fact, we can bring $\bH$ into systematic form and read of a generator matrix $\bG$ for the same code. 
    We can now solve $\bx\bH^\top=\bs$ and since this is a linear system of $n-k$ equations in $n$ unknowns, we get $N=q^k$ possible solutions for $\bx_1, \ldots,\bx_N$. Note that for each of the $q^k$ codewords $\bc_1, \ldots, \bc_N$, we have that $\bc_i+\be$ is a possible solution. Thus, each of the $q^k$ solutions $\bx_i$ correspond to some $\bc_i+\be$. Hence, any of the solutions $\bx_i$ can be used as received vector $\br$ and we have recovered an instance of DP, as $\bG,\br,t$. Hence, solving DP, i.e., finding $\be$, also solves the SDP instance.

    Finally, it is enough to show that DP and SDP are also equivalent to  GWCP. 

    Given an instance of DP, i.e., $\bG,\br,t$ we can add $\br$ as a row to the generator matrix, getting $$\bG'=\begin{pmatrix} \bG \\ \br \end{pmatrix}.$$ Note that the code generated by $\bG'$ is also generated by 
    $$\begin{pmatrix}
         \bG \\ \be
    \end{pmatrix}, $$ as $\br= \bm\bG+\be$. The  new code of dimension $k+1$ has now as lowest weight codeword $\be$ of weight $t$. Hence, we can compute the corresponding parity-check matrix $\bH'$ and solving the GWCP  on the instance $\bH',t$ we recover the solution $\be$ to the DP instance. 
    
    On the other hand, given an instance $\bH,w$ of GWCP, we can define an instance of SDP, by taking the same parity-check matrix and setting the syndrome $\bs=\bz$. Thus, a solver for SDP, searching for a weight $w$ vector $\be$ with $\be\bH^\top=\bz$ also solves the GWCP instance.
\end{proof}

These three equivalent problems are the main problems used for code-based cryptography and will thus be the main focus of the survey.
In the next section, we show that the DP,SDP and GWCP are NP-complete \cite{berlekamp,barg}.
\medskip

There are, however, also other hard problems in coding theory. 
Recall from Section \ref{sec:prelim}, that there are several notions of code equivalence in the Hamming metric. In the lightest version, we ask for two codes to be permutation equivalent.

\begin{problem}[Permutation Equivalence Problem (PEP)]

Given $\bG,\bG'\in \mathbb{F}_{q}^{k\times n}$,  find $\varphi   S_n$, such that $\varphi(\langle \bG \rangle)= \langle \bG'\rangle.$
\end{problem}
This problem is clearly contained in the linear equivalence problem. 
\begin{problem}[Linear Equivalence Problem (LEP)]

Given $\bG,\bG'\in \mathbb{F}_{q}^{k\times n}$,  find $\varphi \in(\mathbb{F}_q^\star)^n \rtimes S_n$, such that $\varphi(\langle \bG \rangle)= \langle \bG'\rangle.$
\end{problem}

On the other hand, we can also ask for a subcode-equivalence. 
\begin{problem}[Permuted Kernel Problem (PKP)]
Given $\bG \in \mathbb{F}_q^{k\times n}, \bH' \in \mathbb{F}_q^{(n-k')\times n}$ find a permutation matrix $\bP$ such that $\bH'(\bG\bP)^\top=\bz$.
\end{problem}
 
 This problem has first been introduced by Shamir in \cite{pkp} and was formulated through parity-check matrices, thus the name \emph{permuted kernel}. In \cite{paolopkp} it has been observed, that the formulation of \cite{pkp} is indeed equivalent to the subcode-equivalence problem.

 \begin{problem}[Subcode Equivalence Problem (SEP)]
 Given $\bG \in \mathbb{F}_q^{k \times n}, \bG' \in \mathbb{F}_{q}^{k'\times n}$,  find permutation matrix  $\bP$ such that  $\langle \bG' \rangle \subset \langle \bG\bP \rangle.$
 \end{problem}
 \begin{exercise}
     Show that PKP is equivalent to SEP.
 \end{exercise}

In the following, we will thus only use the subcode equivalence formulation, also for PKP.
\medskip

 There also exists a relaxed version on PKP, which only asks to find a subcode of dimension 1. 
\begin{problem}[Relaxed PKP]
Given $\bG \in \mathbb{F}_{q}^{k\times n}$,  $\bG' \in \mathbb{F}_q^{k' \times n}$, find $\bx \in \mathbb{F}_q^{k}$ and a permutation matrix  $\bP$ such that $\bx\bG\bP \in \langle \bG' \rangle.$
\end{problem}

Since PKP only asks for permutation equivalence it contains PEP and clearly,  PKP contains the Relaxed PKP.
\medskip

The different code equivalence problems have a strong relation to the graph isomorphism problem and live in different complexity classes, which we will exploit in the next section.   
 
\medskip

Clearly, one can also consider the decoding problem or the code equivalence problem in a different metric.

\medskip

Let us start with the Rank-metric analogue of the SDP. 
\begin{problem}[Rank SDP]
Let $\mathbb{F}_{q^m}$ be a finite field and $k \leq n$ be positive integers.   Given $\bH \in \mathbb{F}_{q^m}^{(n-k) \times n}$, $\bs \in \mathbb{F}_{q^m}^{n-k}$ and $t \in \mathbb N$, is there a vector $\be \in \mathbb{F}_{q^m}^n$ such that $\text{wt}_R(\be)\leq t$ and $\be \bH^\top = \bs$?
\end{problem}
Again, Rank SDP is equivalent to Rank DP or Rank GWCP, as the equivalence is independent of the metric. 
In \cite{randomred} the authors provide a randomized reduction from the SDP to Rank SDP. While this gives great evidence of the hardness of  the Rank SDP, it remains one of the largest open problems in code-based cryptography whether Rank SDP is NP-complete or not. 
\newpage

We do get a different problem, however, when considering $\mathbb{F}_q$-linear codes, i.e., matrix codes. 

\begin{problem}[MinRank Problem]
Given $\bG_1, \ldots, \bG_k \in \mathbb{F}_q^{m \times n}$ $t \in \mathbb{N}$ and $\bR \in \mathbb{F}_q^{m \times n},$ find $\bE \in \mathbb{F}_q^{m \times n}$ of rank at most $t$, such that $$\bR= \lambda_1\bG_1+ \cdots + \lambda_k\bG_k+ \bE,$$ for some $\lambda_1, \ldots, \lambda_k \in \mathbb{F}_q.$
\end{problem}

The MinRank problem is simply the DP for $\mathbb{F}_q$-linear codes in the rank metric and clearly equivalent to the respective SDP and GWCP. Note that unlike the Rank SDP, dealing with $\mathbb{F}_{q^m}$-linear codes, the MinRank problem is known to be NP-complete. We will see the proof in the next section and first cover some more hard problems. 

\begin{problem}[Lee SDP]
Let $\mathbb{F}_{p}$ be a prime field and $k \leq n$ be positive integers.   Given $\bH \in \mathbb{F}_{p}^{(n-k) \times n}$, $\bs \in \mathbb{F}_{p}^{n-k}$ and $t \in \mathbb N$, is there a vector $\be \in \mathbb{F}_{p}^n$ such that $\text{wt}_L(\be)\leq t$ and $\be \bH^\top = \bs$?
\end{problem}
The Lee SDP (again equivalent to Lee DP and Lee GWCP) has been proven to be NP-complete in \cite{leeNP}. Thus, marking the Lee metric as a promising alternative for the Hamming metric.

\begin{problem}[Restricted SDP]
Let $\mathbb{F}_{p}$ be a prime field, $g \in \mathbb{F}_p$ have prime order $z$ and define $$\mathbb{E}= \{g^i \mid i \in \{0, \ldots, z-1\}\}.$$ Let $k \leq n$ be positive integers.   Given $\bH \in \mathbb{F}_{p}^{(n-k) \times n}$ and $\bs \in \mathbb{F}_{p}^{n-k}$, is there a vector $\be \in \mathbb{E}^n$ such that   $\be \bH^\top = \bs$?
\end{problem}

The Restricted SDP is not exactly the SDP with a different metric, but rather than asking for $\be$ to have a certain weight, the Restricted SDP asks for all entries of $\be$ to live in a restricted set $\mathbb{E}.$ Hence, we keep the linear condition $\be\bH^\top=\bs$ and exchanged the non-linear constraint $\text{wt}(\be) \leq t$ with $\be \in \mathbb{E}^n$. In the next section, we give a proof on the NP-hardness of the Restricted SDP.

\subsection{NP-completeness}\label{sec:NP}
 In this section, we give the definitions of several complexity classes and the techniques in order to show that a problem belongs to such complexity class. We then show that DP (and thus also SDP and GWCP) are NP-complete. 
We also provide the reduction of PEP to graph isomorphism.

Let us start with a small introduction to complexity theory.

Let $\mathcal{P}$ denote a problem. In order to estimate how hard it is to solve $\mathcal{P}$ we have two main complexity classes.  

\begin{definition}
 P denotes the class of problems that can be solved by a deterministic Turing machine in polynomial time.
\end{definition}
 The concept of deterministic and non-deterministic Turing machines will exceed the scope of this chapter, just note that "can be solved by a deterministic Turing machine in polynomial time" is the same as our usual "can be solved in polynomial time". 
 
 \begin{example}
Given a list  $S$ of $n$ integers and an integer $k$, determine whether there is an integer $s \in S$  such that $s>k$? Clearly, this  can be answered by going through the list and checking for each element whether it is greater than $k$, thus it has running time at most $n$ and this problem is in P.
 \end{example}
 
 \begin{definition}
NP denotes the class of problems that can be solved by a non-deterministic Turing machine in polynomial time.
 \end{definition}
 Thus, in contrary to the popular belief that NP stands for non-polynomial time, it actually stands for non-deterministic polynomial time. The difference is important: all problems in P live inside NP!
 
 To understand NP better, we might use the equivalent definition: 
A problem $\mathcal{P}$ is in NP if and only if one can check that a candidate is a solution to $\mathcal{P}$ in polynomial time. 

The example from before is thus also clearly in NP, since if given a candidate $a$, we can check in polynomial time whether $a \in S$ and whether $a>k.$

There are, however, interesting problems which are in NP, but we do not know whether they are in P.
Let us change the previous example a bit.

\begin{example}
 Given a list $S$  of $n$ integers and an integer $k$, is there a set of integers $T \subseteq S$, such that $\sum\limits_{t \in T} t = k$? 
 Since there are exponentially many subsets of $S$,  there is no known algorithm to solve this problem in polynomial time and thus, we do not know whether it lives in P. But, if given a candidate $T$, we can check in polynomial time if all $t\in T$ are also in $S$ and if $\sum\limits_{t \in T} t = k$, which clearly places this problem inside NP.
\end{example}

The most important complexity class, for us, will be that of NP-hard problems. In order to define this class, we first have to define polynomial-time reductions.

A polynomial-time reduction from $\mathcal{R}$ to $\mathcal{P}$  follows the following steps: \begin{enumerate}
    \item take any instance $I$ of $\mathcal{R}$,
    \item transform $I$ to an instance  $I'$  of $\mathcal{P}$ in polynomial time,
    \item assume that (using an oracle) you can solve $\mathcal{P}$ in the instance $I'$ in polynomial time, getting the solution $s'$,
    \item transform the solution $s'$ in polynomial time to get a solution $s$ of the problem $\mathcal{R}$ in the input $I$. 
\end{enumerate}
The existence of a polynomial-time reduction from $\mathcal{R}$ to $\mathcal{P}$, informally speaking,  means that if we can solve $\mathcal{P}$, we can also solve $\mathcal{R}$ and thus solving $\mathcal{P}$ is at least as hard as solving $\mathcal{R}$.

\begin{definition}
 $\mathcal{P}$ is NP-hard if for every problem $\mathcal{R}$ in NP, there exists a polynomial-time reduction from $\mathcal{R}$ to $\mathcal{P}.$
\end{definition}

Informally speaking this class contains all problems which are at least as hard as the hardest problems in NP.

\begin{example}
 One of the most famous examples for an NP-hard problem is the subset sum problem: given a set of integers $S$, is there a non-empty  subset $T\subseteq S$, such that $\sum\limits_{t\in T}t=0?$
\end{example}

We want to remark here, that NP-hardness is only defined for \emph{decisional} problems, that are problems of the form "decide whether there exists.." and not for \emph{computational/search} problems, that are problems of the form "find a solution..".
However, considering for example the SDP, in its decisional version, it asks whether there exists  error vector $\be$ with certain conditions. If one could solve the computational problem, that is to actually find such an error vector $\be$ in polynomial time, then one would also be able to answer the decisional problem in polynomial time. Thus, not being very rigorous, we call also the computational SDP NP-hard. \\

In order to prove that a problem $\mathcal{P}$ is NP-hard, fortunately we do not have to give a polynomial-time reduction to \emph{every} problem in NP: there are already problems which are known to be NP-hard, thus it is enough to give a polynomial-time reduction from an NP-hard problem to $\mathcal{P}$.

Finally, NP-completeness denotes the intersection of NP-hardness and NP.

\begin{definition}
 A problem $\mathcal{P}$ is NP-complete, if it is NP-hard and in NP.
\end{definition}


Another complexity class is given by AM, respectively MA. In this class live the problems that can be decided through an Arthur-Merlin protocol. (The only difference between AM and MA is whether Arthur or Merlin first sends a message). 

The protocol is similar to the ZK protocol we have seen before, with the prover Merlin and the verifier Arthur. The protocol is a 3 pass protocol and  does not need to have the ZK property. The main difference lies in the power of the two parties:  while Arthur still has polynomial computational power, Merlin has infinite computation power (indeed Merlin is a wizard).

We say that a problem $\mathcal{P}$ can be decided by the AM protocol if Merlin is able to convince Arthur that the answer upon the instance $I$ is ''no''. 
Merlin might be cheating, i.e., the answer to $I$ is actually ''yes'', and we allow for a soundness error of $\leq 1/3.$

\medskip

No NP-hard problem can live in AM, else we have  AM=PH (the polynomial hierarchy) and this implies a collapse of polynomial hierarchy.

\subsubsection{Decoding Problem}

Berlekamp, McEliece and van Tilborg famously proved in \cite{berlekamp}  the NP-completeness of the syndrome decoding problem for the case of binary linear codes equipped with the Hamming metric.
In \cite{barg}, Barg generalized this proof to an arbitrary finite field. Finally, the NP-hardness proof has been generalized to arbitrary finite rings endowed with an additive weight in \cite{leeNP}, thus including famous metrics such as the homogeneous and the Lee metric.

 In this section we provide the proof of NP-completeness for the SDP as in \cite{barg}. \\

\begin{problem}{\textbf{Syndrome Decoding Problem (SDP)}}  
Let $\mathbb{F}_q$ be a finite field and $k \leq n$ be positive integers.   Given $\bH \in \mathbb{F}_q^{(n-k) \times n}$, $\bs \in \mathbb{F}_q^{n-k}$ and $t \in \mathbb N$, is there a vector $\be \in \mathbb{F}_q^n$ such that $\text{wt}_H(\be)\leq t$ and $\be \bH^\top = \bs$?
\end{problem}

Note that the SDP is clearly in NP: given a candidate vector $\be$ we can check in polynomial time if $\text{wt}_H(\be)\leq t$ and if $\be\bH^\top=\bs.$
Thus, we are only left with showing the NP-hardness of the SDP through a polynomial-time reduction. For this, we choose the 3-dimensional matching (3DM) problem, which is a well-known NP-hard problem.
\begin{problem}{\textbf{3-Dimensional Matching (3DM) Problem}}
\\Let $T$ be a finite set and $U \subseteq T \times T \times T$. Given $U,T$, decide if there exists a set $W \subseteq U$ such that $\mid{W} \mid= \mid {T}\mid $ and no two elements of $W$ agree in any coordinate. \end{problem}

\begin{proposition}\label{prop:SDP}
The SDP is NP-complete.
\end{proposition}
For the proof of Proposition \ref{prop:SDP} we follow closely \cite{leeNP}.
\begin{proof}
We prove the NP-completeness by a polynomial-time reduction from the 3DM problem. For this, we start with a random instance of 3DM with $T$ of size $t$, and $U \subseteq T \times T \times T$ of size $u$. Let us denote the elements in $T= \{b_1, \ldots, b_t\}$ and in $U= \{\ba_1, \ldots, \ba_u\}$. From this we build the matrix  $\bH^\top \in \mathbb{F}_q^{u \times 3t}$, as follows: 
\begin{itemize}
    \item for $ j \in \{1, \ldots, t\}$, we set $h_{i,j} = 1$ if $\ba_i[1]= b_j$ and $h_{i,j}=0$ else,
    \item for $ j \in \{t+1, \ldots, 2t\}$, we set $h_{i,j} = 1$ if $\ba_i[2]= b_j$ and $h_{i,j}=0$ else,
    \item for $ j \in \{2t+1, \ldots, 3t\}$, we set $h_{i,j} = 1$ if $\ba_i[3]= b_j$ and $h_{i,j}=0$ else.
\end{itemize}
With this construction, we have that each row of $\bH^\top$ corresponds to an element in $U$, and has weight $3$. 
Let us set the syndrome $\bs$ as the all-one vector of length $3t$.
Assume that we can solve the SDP on the instances $\bH, \bs$ and $t$ in polynomial time. 
Let us consider two cases.

\underline{Case 1:} 
First, assume that the SDP solver returns as answer `yes', i.e., there exists an $\be \in \mathbb{F}_q^u$, of weight less than or equal to $t$ and such that $\be\bH^\top =\bs.$ 
\begin{itemize}
    \item 
We first observe that we must have  $\text{wt}_H(\be)=\mid{\text{supp}_H(\be)}\mid=t$. For this note that each row of $\bH^\top$ adds at most 3 non-zero entries to $\bs$. Therefore, we need to add at least $t$ rows to get $\bs$, i.e., $ \mid{\text{supp}_H(\be)}\mid \geq t$ and hence $\text{wt}_H(\be) \geq t$. As we also have $\text{wt}_H(\be)\leq t$ by hypothesis, this implies that $\text{wt}_H(\be)=\mid{\text{supp}_H(\be)}\mid=t$.  
\item Secondly, we observe that the weight $t$ solution must be a binary vector. For this we note that the matrix $\bH^\top$ has binary entries and has constant row weight three, and since $\mid{\text{supp}_H(\be)}\mid=t$, the supports of the $t$ rows of $\bH^\top$ that sum up to the all-one vector have to be disjoint.  Therefore, we get that the $j$-th equation from the system of equations $ \be \bH^\top =\bs$ is of the form $e_{i} h_{i,j} = 1$ for some $i \in \text{supp}_H(\be)$. Since $h_{i,j}=1$, we have $e_{i} =1$.
\end{itemize}
Recall from above that the rows of $\bH^\top$ correspond to the elements of $U$. The $t$ rows corresponding to the support of $\be$ are now a solution $W$ to the 3DM problem. This follows from the fact that the $t$ rows have disjoint supports and add up to the all-one vector, which implies that each element of $T$ appears exactly once in each coordinate of the elements of $W$.

\underline{Case 2:} 
Now assume that the SDP solver returns as answer `no', i.e., there exists no $\be \in \mathbb{F}_q^u$ of weight at most $t$ such that $\be\bH^\top=\bs.$
This response is now also the correct response for the 3DM problem. In fact, if there exists $W \subseteq U$ of size $t$ such that all coordinates of its elements are distinct, then $t$ rows of $\bH^\top$ should add up to the all one vector, which in turn means the existence of a vector $\be \in \{0,1\}^u$ of weight $t$ such that $\be\bH^\top=\bs.$

Thus, if such a polynomial time solver exists, we can also solve the 3DM problem in polynomial time.
\end{proof}

\begin{example}
Let us consider $T= \{A,B,C,D\}$ and 
$$ U= \{(D,A,B), (C,B,A), (D,A,B),(B,C,D),(C,D,A),(A,D,A),(A,B,C)\}.$$
Then the above construction would yield
$$\bH^\top= \left[\begin{array}{cccc| cccc | cccc}
0 & 0 & 0 & 1 & 1 & 0 & 0 & 0 & 0 & 1 & 0 & 0 \\
0 & 0 & 1 & 0 & 0 & 1 & 0 & 0 & 1 & 0 & 0 & 0 \\
0 & 0 & 0 & 1 & 1 & 0 & 0 & 0 & 0 & 1 & 0 & 0 \\
0 & 1 & 0 & 0 & 0 & 0 & 1 & 0 & 0 & 0 & 0 & 1 \\
0 & 0 & 1 & 0 & 0 & 0 & 0 & 1 & 1 & 0 & 0 & 0 \\
1 & 0 &0 & 0 & 0 & 0 & 0 & 1 & 1 & 0 & 0 & 0 \\
1 & 0 & 0 & 0 & 0 & 1 & 0 & 0 & 0 & 0 & 1 & 0 \\
\end{array}\right].$$
A solution to $\be\bH^\top=(1, \ldots, 1)$ would be $\be=(1,0,0,1,1,0,1)$ which corresponds to $$W= \{(D,A,B), (B,C,D), (C,D,A), (A,B,C)\}. $$
\end{example}

Notice that the very same construction is used also in the   problem of finding codewords with given weight.

\begin{proposition}\label{prop:GWCP}
The GWCP is NP-complete. 
\end{proposition}

\begin{proof}
We again prove the NP-completeness by a reduction from the 3DM problem. To this end, we start with a random instance of 3DM, i.e., $T$ of size $t$, and $U \subseteq T \times T \times T$ of size $u$. Let us denote the elements in $T= \{b_1, \ldots, b_t\}$ and in $U= \{\ba_1, \ldots, \ba_u\}$. At this point, we build the matrix $\overline{\bH}^\top \in \mathbb{F}_q^{u \times 3t}$, like in the proof of Proposition \ref{prop:SDP}. 

Then we construct $\bH^\top \in \mathbb{F}_q^{(3tu+3t+u) \times(3tu+3t)}$ in the following way.
\begin{equation*}
    \bH^\top = \begin{pmatrix}
    \overline{\bH}^\top & \Id_u & \cdots & \Id_u \\
    -\Id_{3t} & \mathbf{0} & \cdots & \mathbf{0} \\
    \mathbf{0} & -\Id_u & & \mathbf{0} \\
    \vdots & & \ddots & \\
    \mathbf{0} & \mathbf{0} & & -\Id_u
    \end{pmatrix},
\end{equation*}
where we have repeated the size-$u$ identity matrix $3t$ times in the first row.
Let us set $w= 3t^2+4tM$ and assume that we can solve the GWCP on the instance given by $\bH, w$ in polynomial time. 
Let us again consider two cases. 

\underline{Case 1:}
In the first case the GWCP solver returns as answer `yes', since there exists a $\mathbf c \in \mathbb{F}_q^{3tu +3t+u}$, of weight equal to $w$, such that $\mathbf c\bH^\top =\mathbf{0}_{3tu+3t}$.
Let us write this $\mathbf c$ as
$$\mathbf c= (\overline{\mathbf c}, \mathbf c_0, \mathbf c_1, \ldots, \mathbf c_{3t}),$$
where $\overline{\mathbf c} \in \mathbb{F}_q^u, \mathbf c_0 \in \mathbb{F}_q^{3t}$ and $\mathbf c_i \in \mathbb{F}_q^u$ for all $i \in \{1, \ldots, 3t\}.$
Then, $\mathbf c\bH^\top = \mathbf{0}_{3tu+3t}$ gives the equations 
\begin{align*}
    \overline{\mathbf c}\overline{\bH}^\top - \mathbf c_0&= \mathbf{0}, \\
    \overline{\mathbf c} - \mathbf c_1 & = \bz, \\
    & \vdots \\
    \overline{\mathbf c} - \mathbf c_{3t} &= \bz.
\end{align*}
Hence, we have that $\mathrm{wt}_H(\overline{\mathbf c}\overline{\bH}^\top)= \mathrm{wt}_H(\mathbf c_0)$ and 
$$\mathrm{wt}_H(\overline{\mathbf c}) = \mathrm{wt}_H(\mathbf c_1) = \cdots = \mathrm{wt}_H(\mathbf c_{3t}).$$ 
Due to the coordinatewise additivity of the weight, we have that 
$$\mathrm{wt}_H(\mathbf c) = \mathrm{wt}_H(\overline{\mathbf c}\overline{\bH}^\top) + (3t+1)\mathrm{wt}_H(\overline{\mathbf c}).$$
Since $\mathrm{wt}_H(\overline{\mathbf c}\overline{\bH}^\top) \leq 3t$, we have that  $\mathrm{wt}_H(\overline{\mathbf c}\overline{\bH}^\top)$ and $\mathrm{wt}_H(\overline{\mathbf c})$ are uniquely determined as the remainder and the quotient, respectively, of the division of $\mathrm{wt}_H(\mathbf c)$ by $3t+1.$
In particular, if $\mathrm{wt}_H(\mathbf c) = 3t^2+4t,$ then we must have $\mathrm{wt}_H(\overline{\mathbf c})=t$ and $\mathrm{wt}_H(\overline{\mathbf c}\overline{\bH}^\top)=3t.$ 
Hence, the first $u$ parts 
of the found solution $\mathbf c$, i.e., $\overline{\mathbf c}$, give a matching for the 3DM in a similar way as in the proof of Proposition \ref{prop:SDP}. For this we first observe that $\overline{\mathbf c}\overline{\bH}^\top$ is a full support vector and it plays the role of the syndrome, i.e., $\overline{\mathbf c} \overline{\bH}^\top = (x_1,\ldots,x_{3t})$, where $x_i \in \mathbb{F}_q^\star$. Now, using the same argument as in the proof of Proposition \ref{prop:SDP}, we note that $\overline{\mathbf c}$ has exactly $t$ non-zero entries, which corresponds to a solution of 3DM. 

\underline{Case 2:}
If the solver returns as answer `no', this is also the correct answer for the 3DM problem. In fact, if there exists a $W \subseteq U$ of size $t$, such that all coordinates of its elements are distinct, then $t$ rows of $\overline\bH^\top$ should add up to the all one vector, which in turn means the existence of a $\mathbf e \in \{0,1\}^u$ of support size $t$ such that $x\be\overline{\bH}^\top=(x,\ldots,x)=: \mathbf c_0$ for any $x \in \mathbb{F}_q^\star$.
And thus, with $\overline{\bc} = x \be$ a solution $\mathbf c$ to the GWCP with the instances constructed as above should exist.

Thus, if such a polynomial time solver  for the GWCP exists, we can also solve the 3DM problem in polynomial time. 
\end{proof}

 We remark that the bounded version of this problem, i.e., deciding if a codeword $\bc$ with $\text{wt}_H(\bc) \leq w$ exists, can be solved by applying the solver of Problem \ref{prob:GWCP} at most $w$ many times.

 The computational versions of Problems \ref{prob:GWCP} and \ref{prob:SDP} are at least as hard as their decisional counterparts. 
Trivially, any operative procedure that returns a vector with the desired properties (when it exists) can be used as a direct solver for the above problems.

Note that the problem on which the McEliece system is based upon is not exactly equivalent to the SDP. In the McEliece system the parameter $t$ is usually bounded by the error correction capacity of the chosen code. Whereas in the SDP, the parameter $t$ can be chosen to be any positive integer. Thus, we are in a more restricted regime than in the SDP. 

\begin{problem}[Bounded SDP]
Let $\mathbb{F}_q$ be a finite field and $k \leq n$ be positive integers. Given $\bH\in \mathbb{F}_q^{(n-k)\times n}, \bs \in \mathbb{F}_q^{n-k}$ and $d\in\mathbb{N}$, such that every set of $d-1$ columns of $\bH$ is linearly independent and $w=\left\lfloor \frac{d-1}{2} \right\rfloor$, is there a vector $\mathbf e \in \mathbb{F}_q^n$ such that  $\text{wt}_H(\mathbf e) \leq w$  and $ \mathbf e\mathbf H^\top = \bs$? 
\end{problem}

 This problem is conjectured to be NP-hard \cite{barg} and in \cite{vardy} it is observed that this problem is not likely to be in NP, since  already verifying that any $d-1$ columns are linearly independent is not possible in polynomial time.

There have been attempts \cite{IKKR}   to transform the McEliece system in such a way that the underlying problem is closer or even exactly equivalent to the SDP, the actual NP-complete problem. 
This proposal has been attacked shortly after in \cite{terry}. However, using a different framework than the McEliece system, this is actually possible, for example by using the quasi-cyclic framework or the AF system.

We also want to remark here, that the following generalization of the GWCP, i.e., Problem \ref{prob:GWCP}, is also NP-complete \cite{vardy}:
\begin{problem}
Let $\mathbb{F}_q$ be a finite field and $k \leq n$ be positive integers. Given $\bH\in \mathbb{F}_q^{(n-k)\times n}$ and $w\in\mathbb{N}$, is there a vector $\mathbf c \in \mathbb{F}_q^n$ such that  $\text{wt}_H(\mathbf c) \leq w$  and $ \mathbf c\mathbf H^\top = \mathbf 0_{n-k}$? 
\end{problem} 
In \cite{vardy} this problem was called the minimum distance problem, since if one could solve the above problem, then by running such solver on $w \in \{1, \ldots, n\}$ until an affirmative answer is found, this would return the minimum distance of a code.

However, this does not mean that finding the minimum distance of a random code is NP-complete. In fact, with the above problem one can prove the NP-hardness of finding the minimum distance, but it is unlikely to be in NP, since in order to check whether a candidate solution $d$ really is the minimum distance of the code, one would need to go through (almost) all codewords. 

\subsubsection{Code Equivalence Problems}

Recall the different code equivalence problems, namely PEP, LEP, PKP and relaxed PKP:

\begin{problem}[Permutation Equivalence Problem (PEP)]

Given $\bG,\bG'\in \mathbb{F}_{q}^{k\times n}$,  find $\varphi \in  S_n$, such that $\varphi(\langle \bG \rangle)= \langle \bG'\rangle.$
\end{problem}

\begin{problem}[Linear Equivalence Problem (LEP)]

Given $\bG,\bG'\in \mathbb{F}_{q}^{k\times n}$,  find $\varphi \in(\mathbb{F}_q^\star)^n \rtimes S_n$, such that $\varphi(\langle \bG \rangle)= \langle \bG'\rangle.$
\end{problem}

\begin{exercise}
    Show that PEP $\subset$ LEP, by showing a reduction from PEP to LEP.
\end{exercise}

\begin{problem}{Permuted Kernel Problem (PKP)}
 Given $\bG \in \mathbb{F}_q^{k \times n}, \bG' \in \mathbb{F}_{q}^{k'\times n}$,  find permutation matrix  $\bP$ such that  $\langle \bG' \rangle \subset \langle \bG\bP \rangle.$
 \end{problem}

 \begin{exercise}
     Show that PEP $\subset$ PKP.
 \end{exercise}
 
\begin{problem}{Relaxed PKP} 
Given $\bG \in \mathbb{F}_{q}^{k\times n}$,  $\bG' \in \mathbb{F}_q^{k' \times n}$, find $\bx \in \mathbb{F}_q^{k}$ and a permutation matrix  $\bP$ such that $\bx\bG\bP \in \langle \bG' \rangle.$
\end{problem}

\begin{exercise}
    Show that Relaxed PKP $\subset$ PKP.
\end{exercise}

A graph $\mathcal{G}$ is usually denoted through its vertices $V$ and edges $E \subset V^2$, i.e., we write $\mathcal{G}=(V,E)$.
We say that $\mathcal{G}=(V,E)$ with $|V|=v, |E|=e$ has \emph{incidence matrix} $\bB \in \mathbb{F}_2^{e \times v}$, if $\bB$ has entries $b_{i,j}$ with 
$$b_{i,j} = \begin{cases} 1 & \text{ if } i=(\ell,j) \in E, \\ 0 & \text{ else.} \end{cases}$$
That is the rows correspond to the edges and the columns to the vertices. Considering the edge $(a,b)$, we set a 1 in the position $a$ and in the position $b.$

Since we consider undirected graphs, the condition $e=( \ell,j)\in E$ should be read as unordered tuple, i.e., also $e=(j,\ell) \in E.$
\begin{example}
    The graph $\mathcal{G}$ with vertex set $V=\{1,2,3,4\}$ and edge set $E=\{ (1,2),(2,3),(3,4)\}$  has incidence matrix
    $$\bB = \begin{pmatrix}
        1 & 1 & 0 & 0 \\ 
        0 & 1 & 1 & 0 \\
        0 & 0 & 1 & 1
    \end{pmatrix}.$$
    Clearly, there are different incidence matrices, depending on the ordering of the edges. 
\end{example}
As mentioned before, the code equivalence problems have a relation to the Graph Isomorphism problem, which states the following.
\begin{problem}[Graph Isomorphism (GI) problem]
    
 Given $\mathcal{G}=(V,E), \mathcal{G}'=(V,E')$, find $f: V \to V$, such that   
 $\{u,v\} \in E \leftrightarrow \{f(u), f(v)\} \in E'.$
\end{problem}

\begin{theorem}
    There exists a reduction from GI to PEP.
\end{theorem}
We follow the proof of \cite{rothGI}.
\begin{proof}
   Let $\mathcal{G}=(V,E)$ and $\mathcal{G}'=(V,E')$ be an instance of GI. 
   Let $\bD$ and $\bD'$ be two incidence matrices for $\mathcal{G},$ respectively $\mathcal{G}'$. We can transform this instance to an instance of PEP, by defining the two generator matrices in $\mathbb{F}_q^{e \times (3e+v)}$
   \begin{align*}
    \bG &= \begin{pmatrix} \text{Id}_e & \text{Id}_e & \text{Id}_e & \bD  \end{pmatrix},  \\ 
      \bG' &= \begin{pmatrix} \text{Id}_e & \text{Id}_e & \text{Id}_e & \bD'  \end{pmatrix}.
   \end{align*}
   Let us consider two cases.  In the first case, the answer to GI is ''yes``, as there exists a $f:V\to V$, such that $\{f(u),f(v)\} \in E'$ for all $ \{u, v\} \in E.$
   Thus, there exists a permutation of $V$ which maps one graph to the other and the two incidence matrices $\bD$ and $\bD'$ are such that 
   $$\bQ \bD \bP = \bD'$$ for some $e \times e$ permutation matrix  $\bQ$ and $v \times v$ permutation matrix $\bP.$
   Clearly, the codes generated by $\bG$ and $\bG'$ are then also permutation equivalent. 

   In the second case, we assume that the two graphs are not isomorphic, hence there exists no permutation on $V$, which maps $\mathcal{G}$ to $\mathcal{G}'.$
   Thus, no $v \times v$ permutation matrix $\bP$ and no $e \times e$ permutation matrix $\bQ$ exists for which $\bQ\bD\bP=\bD'.$

   The two codes generated by $\bG_1$ and $\bG_2$ are only permutation equivalent, if we can find $\bS \in \text{GL}_n(2)$ and $(3e+v) \times (3e+v)$ permutation matrix $\bP$ such that 
   $$\bS\bG \bP=\begin{pmatrix} \bS & \bS & \bS & \bS\bD \end{pmatrix}\bP =\bG'.$$
   Note that  the first $3e$ columns of $\bS\bG$ consist
of all unit vectors of length $e$, each appearing exactly three times.
Hence, the first $3e$ columns of $\bG_2$ are obtained by permuting the first
$3e$ columns of $\bS\bG$ and thus,  we also have
the permutation matrix 
$\bP = \text{diag}(\bS^{-1}, \bS^{-1}, \bS^{-1}, \bT)$, where $\bT$
is a $v \times v$ permutation matrix. Hence, if such $\bS,\bP$ exist, we must have $\bD' = \bS\bD\bT$, which is against the assumption that $\mathcal{G} $ and $\mathcal{G}'$ are not isomorphic. 

\end{proof}
Due to this result, we know that PEP (and thus also LEP) are at least as hard as GI.

  Since PKP is a subcode-equivalence problem it is equivalent to the subgraph isomorphism problem and hence NP-complete \cite{subgraph}.
  However, the hardness of the relaxed version is not known.

  \begin{problem}[Open Problem]
How hard is Relaxed PKP?    
\end{problem}




Using an Arthur-Merlin protocol, we can show that any code-equivalence problem (PEP, LEP or rank-metric) are not NP-hard (as we believe the polynomial hierarchy does not collapse). We follow the proof given in \cite{rothGI}.

\begin{theorem}
    Assuming PH $\neq AM$, code-equivalence is not NP-hard. 
\end{theorem}
\begin{proof}
To show this, we construct the 3-pass Arthur-Merlin protocol. 
Both parties see the instance $(\mathcal{C}_1, \mathcal{C}_2)$ and Merlin wants to convince Arthur, that the two codes are not equivalent. 

Arthur chooses one of the codes, $\mathcal{C}_i$, a random isometry $\varphi$ and computes a generator matrix $\bG'$ for $\varphi(\mathcal{C}_i)$ and sends $\bG'$ to Merlin.

Merlin, with the infinite computational power, can compute which code $\mathcal{C}_i$ Arthur has chosen  and reply with $i$.
If Merlin was honest, then only one of the codes $\mathcal{C}_1,\mathcal{C}_2$ will be equivalent to the sent $\mathcal{C}' =\langle \bG' \rangle.$

If Merlin was cheating and $\mathcal{C}_1$ is equivalent to $\mathcal{C}_2$, then Merlin has two choices and has a success probability of $1/2.$

By repeating this protocol for $t$ rounds, we get a soundness error of $2^{-t}$ that Arthur accepts a cheating Merlin. 
\end{proof}
 

Due to Babai's algorithm \cite{babai}, we know that Graph Isomorphism (GI) takes at most quasi-polynomial time to solve. Thus,   a reduction from PEP to GI, i.e., showing that if we can solve GI then we can also solve PEP,  implies that PEP is easier than GI. In particular, PEP should not be used for cryptography. 
\medskip

The reduction has been proposed in \cite{magali} and has a small drawback: it only works for codes with trivial hull. Since random codes have w.h.p. a trivial hull, we call this a "randomized" reduction, meaning that it will not work for any instance, but it works w.h.p.

\medskip
Before we can give the reduction, let us recall some graph theory. 

A graph $\mathcal{G}$ consists of vertices $V$ and edges $E$ between the vertices, i.e., $E\subset V\times V.$

We will focus on undirected graphs, thus whenever $\{u,v\} \in E$ also $\{v,u\} \in E$ and we label the edges with a weight $w(u,v)$.

\medskip

We say that two weighted graphs $\mathcal{G}=(V,E)$ and $\mathcal{G}'=(V',E')$ are isomorphic, if there exists a bijective map $f:V\to V'$ with 
\begin{enumerate} 
\item $ \{u,v\} \in E \leftrightarrow \{f(u), f(v)\} \in E'$, 
\item $w(u,v)= w(f(u),f(v))$.
\end{enumerate}

Thus, we can focus on $V=V'=\{1,\ldots,n\}$ and maps $f=\sigma \in \mathcal{S}_n$.

\begin{problem}[Weighted Graph Isomorphism Problem]
 Given $\mathcal{G}=(V,E), \mathcal{G}'=(V,E')$, find $\sigma \in \mathcal{S}_n$, such that $ \{u,v\} \in E \leftrightarrow \{\sigma(u), \sigma(v)\} \in E'$ and $w(u,v)= w(\sigma(u),\sigma(v))$.
    
\end{problem}

\begin{definition}
    The \emph{adjacency matrix} of a weighted graph $\mathcal{G}$ is defined as the $n \times n$ matrix $\bA$ with entries $$\bA_{i,j}=\begin{cases} w(i,j) & \text{ if } \{i,j\} \in E, \\ 0 & \text{ else}. \end{cases}$$
\end{definition}
Since we are only interested in undirected graphs, the adjacency matrices are symmetric.

\begin{proposition}
    Two graphs $\mathcal{G}, \mathcal{G}'$ are isomorphic if and only if 
    there exists a permutation matrix $\bP$ such that $\bP^\top \bA \bP= \bA'.$
\end{proposition}

This almost looks like what we need for PEP, except for the fact that in PEP (treating $\bA$ as generator matrix) we also accept $\bS \bA \bP= \bA'$ for any invertible matrix $\bS$, not necessarily of the form $\bP^\top.$

In fact, one can easily make an example of two graphs, where there exists $\bS \in \GL_n(q), \bP \in S_n$ with $\bS \bA \bP = \bA'$ but the two graphs are clearly not isomorphic.

\begin{example}
Let $\mathcal{G}=(V,E)$ with $V=\{1,2,3,4\}$ and $E=\{ w(1,2)=1, w(2,3)=1, w(2,4)=2\}$ and $\mathcal{G}' = (V,E')$ with $E'=\{ w(1,1)=3, w(1,2)=1, w(1,3)=1, w(1,4)=2\}.$
\begin{center} 
\includegraphics[width=7cm]{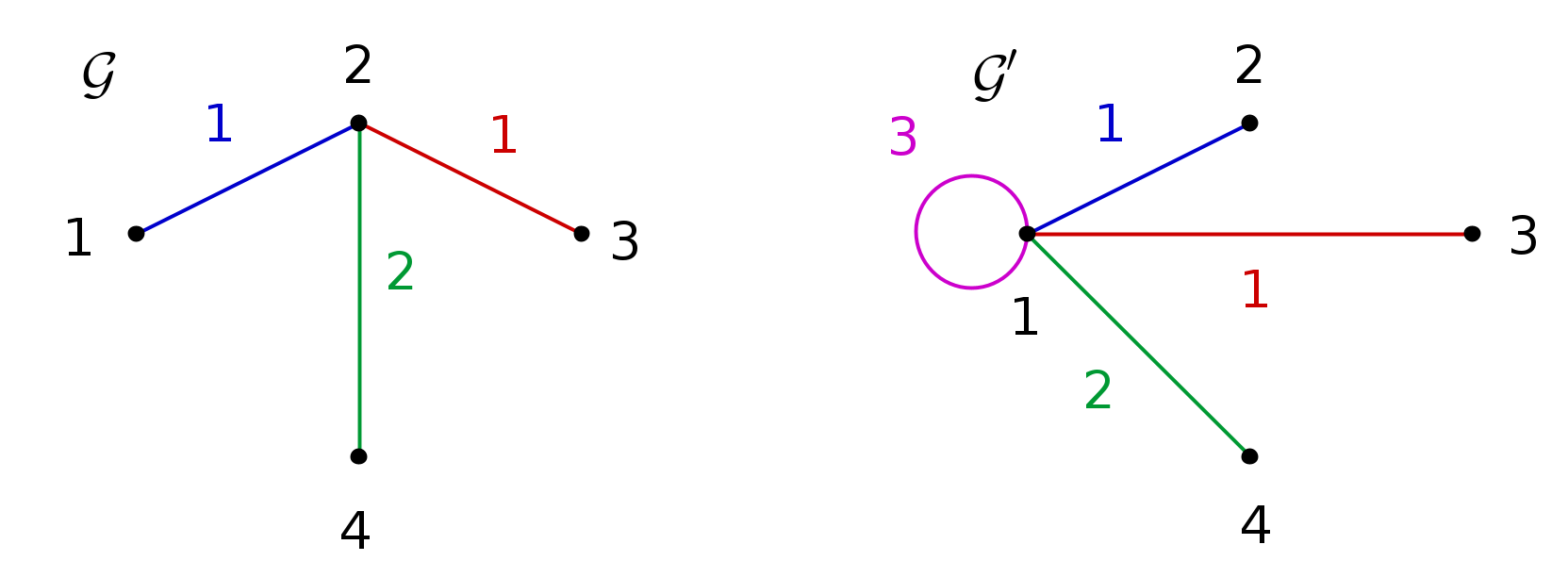}

\end{center}
These graphs can clearly not be isomorphic. However,  their adjacency matrices \begin{align*} \bA= \begin{pmatrix} 0 & 1 & 0 & 0 \\ 1 & 0 & 1 & 2 \\ 0 & 1 & 0 & 0 \\ 0 & 2 & 0  & 0 \end{pmatrix} \ \text{  and  } \ \bA'= \begin{pmatrix} 3 & 1 & 1 & 2 \\ 
          1& 0 & 0 & 0 \\ 
          1 &0 &0 &0 \\
          2 & 0 &0 &0 \end{pmatrix} \end{align*}  generate the codes $\mathcal{C}= \langle \bG \rangle$ with $\bG= \begin{pmatrix}  1 & 0 & 1 & 2 \\ 0&1&0&0\end{pmatrix}$ and $\mathcal{C}' = \langle \bG' \rangle$ with $\bG' = \begin{pmatrix} 1 & 0 &0 &0\\ 0 &1 & 1  &2 \end{pmatrix}$, which are clearly permutation equivalent through the permutation $\sigma =(2134).$
          And there exists $\bS \in \text{GL}_4(5), \bP \in S_4$ (the permutation matrix of $\sigma$) such that $\bS\bA\bP=\bA'$, namely

          $$\bS= \begin{pmatrix} 0& 1 & 0 & 0 \\ 1 & 0 & 0 &0 \\ 3 & 0 & 0 & 2 \\ 0 & 0&3 &0 \end{pmatrix}.$$
\end{example}

Luckily, in order to reduce PEP to GI, we do not have to start with an instance of GI and transform it to an instance of PEP. Instead, we start from codes and transform them to graphs. 
\medskip

Hence, the main question is: given $\bG \in \mathbb{F}_q^{k \times n}$, how to define a symmetric matrix in $\mathbb{F}_q^{n \times n}$, which can act as adjacency matrix?

\medskip

For this \cite{magali} introduced the following:
For $\mathcal{C}$ with trivial hull, we define 
$$\bA = \bG^\top (\bG \bG^\top)^{-1} \bG.$$
Clearly, this matrix can only exist if $\bG\bG^\top$ is invertible, i.e., if the hull of $\mathcal{C}$ is trivial. The matrix $\bA \in \mathbb{F}_q^{n \times n}$ is symmetric, $\langle \bA\rangle = \mathcal{C}$ and, moreover, independent of the choice of $\bG.$

In fact, taking any other generator matrix, $\bS\bG$, we get 
$$(\bS\bG)^\top ( \bS \bG (\bS \bG)^\top)^{-1} \bS \bG= \bG^\top (\bG \bG^\top) \bG.$$

\begin{theorem}
    PEP is easier than weighted GI, for codes with trivial hull.
\end{theorem}

\begin{proof}
    Assume that the codes in the instance of PEP $(\mathcal{C}, \mathcal{C}')$ have trivial hulls. 
    For arbitrary generator matrices $\bG,$ respectively $\bG'$ take
    \begin{align*}
        \bA &= \bG^\top (\bG \bG^\top)^{-1} \bG, \\
         \bA' &= \bG'^\top (\bG' \bG'^\top)^{-1} \bG'.
    \end{align*}
    And define $\mathcal{G},$ respectively $\mathcal{G}'$, to have adjacency matrices  $\bA,\bA'$. 

    We now show that the answer to the constructed weighted GI instance is also the answer to the PEP instance. In fact, $\sigma(\mathcal{G})=\mathcal{G}'$ if and only if $\sigma(\mathcal{C})= \mathcal{C}'.$
    \medskip

    For the first direction, note that for any choice of generator matrices, there exist $\bS \in \GL_k(q)$ with $\bS \bG \bP = \bG'$. However, since $\bA$ is independent of the choice of basis, we can ignore the $\bS$ and get
    $$\bA'= (\bG \bP)^\top ( \bG \bP (\bG \bP)^\top)^{-1} \bG\bP =\bP^\top \bA \bP.$$
    Thus, the two graphs are isomorphic.
    
    The other direction is straightforward, as the $\bP^\top \bA \bP=\bA'$ implies that the two codes $\langle \bA \rangle$ and $\langle \bA' \rangle$ are equivalent. 
\end{proof}


To reduce the LEP to PEP, one can use the \emph{closure} of the code, introduced in \cite{hull}. 

\begin{definition}[Closure of Code] Let $\mathcal{C} \subset \mathbb{F}_q^n$.  The closure of $\mathcal{C}$ is 
$$\widetilde{\mathcal{C}}=\{(\alpha c_i)_{(i,\alpha) \in [1,n] \times \mathbb{F}_q^\star} \mid (c_i)_{i \in [1,n]} \in \mathcal{C}\} \subset \mathbb{F}_q^{n(q-1)}.$$
\end{definition}
Thus, if $\mathcal{C} = \langle \bG \rangle$, with 
$$\bG= \begin{pmatrix} | & & | \\ \bg_1^\top & \cdots & \bg_n^\top \\ | & & | \end{pmatrix}$$ then the generator matrix of the closure is given by 
$$\widetilde{\bG} =\begin{pmatrix}
     | & | & & | & & |& | & & | \\ \bg_1^\top & \alpha \bg_1^\top & \cdots & \alpha^{q-2} \bg_1^\top & \cdots & \bg_n^\top & \alpha \bg_n^\top & \cdots & \alpha^{q-2} \bg_n^\top  \\     | & | & & | & & | & | & & | \\
\end{pmatrix}$$
for some primitive element $\alpha \in \mathbb{F}_{q}$. The closure $\widetilde{\mathcal{C}}$ has still dimension $k$, but now length $n(q-1).$

\begin{proposition}
    If there exists $\varphi \in M_{n,q}$ with $\varphi(\mathcal{C})=\mathcal{C}'$, then there exists $\sigma \in S_n$ with $\sigma(\widetilde{\mathcal{C}})=\widetilde{\mathcal{C}}'.$
\end{proposition}

Hence if $\bG, \bG'$ are such that there exists $\bP \in S_n$ and $\bv \in (\mathbb{F}_q^\star)^n$ with $\bS \bG \bP \text{diag}(\bv) = \bG',$ for some $\bS \in \GL_k(q),$ then $\widetilde{\bS} \widetilde{\bG} \widetilde{\bP} = \widetilde{\bG}',$ where 
$$\widetilde{\bP}= \begin{pmatrix} \bP_1 & & \\ 
& \ddots & \\ 
& & \bP_n \end{pmatrix} \bQ, $$ with $\bP_i \in S_{q-1},$ capture the permutation of $\mathbb{F}_q$, i.e., $\bv_i$ and $\bQ$ is a block permutation matrix, which keeps the blocks of $(q-1)$ columns together and captures $\bP.$
\medskip

This reduction can always be done, however, depending on $q$ it might have different outcomes.
\medskip

For this, let us consider the hull of the closure. Recall, that for most solvers and for the reduction to GI, we want a small or trivial hull.
\medskip

In fact, we can show that  for $q\geq 4$
the closure is weakly self dual, meaning $\widetilde{\mathcal{C}} \subset \widetilde{\mathcal{C}}^\perp.$ Thus, the hull has the largest possible dimension $k$ as $\mathcal{H}(\widetilde{\mathcal{C}})=\widetilde{\mathcal{C}}.$
\medskip

This makes the reduction only interesting for $q<4$. 


\begin{proposition}
    If $q<4,$ then $\widetilde{\mathcal{C}}$ has w.h.p. a trivial hull. If $q \geq 4,$ then $\widetilde{\mathcal{C}}$ is weakly self dual.
\end{proposition}

\begin{proof}
    To understand the hull of the closure, we have to compute 

    $$\bX= \widetilde{\bG}\widetilde{\bG}^\top = \begin{pmatrix}
     | & | & &   | \\ \bg_1^\top  & \alpha \bg_1^\top & \cdots   & \alpha^{q-2} \bg_n^\top  \\     | & | & &   | \\
\end{pmatrix} \begin{pmatrix} - & \bg_1&  - \\ 
- & \alpha \bg_1& - \\ 
& \vdots & \\ 
-&  \alpha^{q-2} \bg_n&  -
\end{pmatrix}.$$

   One can easily check that $$\bX_{i,j}= \sum_{\ell=1}^n \bg_{\ell,i} \bg_{\ell,j} \sum_{\beta \in \mathbb{F}_q^\star} \beta^2.$$

   Now the question becomes, how does the sum of squares behave in $\mathbb{F}_q?$
   If there exists a $\alpha \in \mathbb{F}_q^\star$ with $\alpha^2\neq 1$, then $\beta \mapsto \alpha \beta$ permutes $\mathbb{F}_q^\star,$ hence
   $$\sum_{\beta \in \mathbb{F}_q^\star} \beta^2=\sum_{\beta \in \mathbb{F}_q^\star} (\alpha \beta)^2 = \alpha^2 \sum_{\beta \in \mathbb{F}_q^\star} \beta^2,$$ which implies that $\sum_{\beta \in \mathbb{F}_q^\star} \beta^2=0.$

   To find such $\alpha,$ with $\alpha^2\neq 1$, we need $q\geq 4.$
   Thus, for $q \geq 4$ we have $\bX= \mathbf{0}$ and 
   $$\text{dim}(\mathcal{H}(\widetilde{\mathcal{C}}))= k-\rk(\widetilde{\bG} \widetilde{\bG}^\top)=k.$$
   Hence, for $q\geq 4$, $\widetilde{\mathcal{C}} \subset \widetilde{\mathcal{C}}^\perp, $ i.e., the closure is weakly self dual. 
 \end{proof}

In fact, we can also handle $q=4$, however, not with the classical dual defined through the standard inner product.

\begin{definition}[Hermitian Inner Product]
For $\bx, \by \in \F_q^n$ let us denote by $\langle \bx, \by \rangle_H$ the \emph{Hermitian} inner product, {i.e.}, 
$$\langle \bx, \by \rangle_H = \sum_{i=1}^n x_iy_i^p,$$ where $p=\text{char}(\mathbb{F}_q)$

\end{definition}

Note that the Hermitian inner product is not symmetric!
\medskip

Thus, to define  the  Hermitian dual, we have to fix on which side we place the codewords.

\begin{definition}[Hermitian Dual Code]
Let  $k \leq n$ be positive integers and let  $\mC$ be an $[n,k]$ linear  code over $\F_q$. The \emph{Hermitian dual code}  $\mC^\star$ is  an $[n,n-k]$ linear code over $\F_q$, defined as 
$$\mC^\star = \{ \bx \in \F_q^n \mid \langle \bx, \by \rangle_H = 0 \ \forall \ \by \in \mC \}.$$
\end{definition}

\begin{definition}[Hermitian Parity-Check Matrix]
Let $k \leq n$ be positive integers and let $\mC$ be an $[n,k]$ linear code over $\F_q$ with Hermitian dual code $\mC^\star$. Then, a matrix $\bH^\star \in \F_q^{(n-k) \times n}$ is called a \emph{Hermitian parity-check matrix} of $\mC$, if
 $\bH^\star$ is a generator matrix of $\mC^\star$ and $$\bG((\bH^\star)^p)^\top= \mathbf{0}.$$
\end{definition}

Note that if $\bH$ is the common parity-check matrix of $\mC,$ then $\bH^{1/p}$ is a Hermitian parity-check matrix. 
The Hermitian hull is then defined as $\mathcal{H}^\star(\mathcal{C})= \mathcal{C} \cap \mathcal{C}^\star.$
\medskip

The Hermitian dual and Hermitian hull are still invariants of isometries.

\begin{exercise}
    Let $\mathcal{C} \subset \mathbb{F}_q^n$ be equivalent to $\mathcal{C}'.$ Then $\mathcal{C}^\star$ is equivalent to $(\mathcal{C}')^\star$ (and thus the hulls are as well).
    \emph{Hint:} Use again that $\bS\bG((\bH^\star)^p)^\top=\mathbf{0}$ and $\bG\bP=\bG'.$
\end{exercise}

Having this new definition of hull, we can define 
$$\bA^\star = (\bG^p)^\top (\bG (\bG^p)^\top)\bG,$$ again $\bA^\star$ is independent of the choice of generator matrix, symmetric and exists if $\mathcal{C}$ has trivial Hermitian hull.

\medskip 
Similar to before, for random codes we assume that $\bG$ is chosen uniform at random and thus $\bG (\bG^p)^\top$ has full rank. Thus, random codes have w.h.p. trivial Hermitian hull.
\medskip

The only thing left to check is that the closure of a code in $\mathbb{F}_4$ has w.h.p. trivial Hermitian hull. For this let $\alpha$ be a primitive element in $\mathbb{F}_4$ and consider

    $$\bX= \widetilde{\bG}(\widetilde{\bG}^2)^\top = \begin{pmatrix}
     | & | & |  & & | & | &  | \\ \bg_1^\top  & \alpha \bg_1^\top & \alpha^2 \bg_1^\top & \cdots   & \bg_n^\top  & \alpha \bg_n^\top & \alpha^{2} \bg_n^\top  \\    | & | & |  & & | & | &  | \\
\end{pmatrix} \begin{pmatrix} - & \bg_1^2 &  - \\ 
- & \alpha^2 \bg_1^2  & - \\ - & \alpha \bg_1^2  & - \\ 
& \vdots & \\ 
- & \bg_n^2 &  - \\ 
- & \alpha^2 \bg_n^2  & - \\ - & \alpha \bg_n^2  & -
\end{pmatrix}.$$

One can easily check that
$$\bX_{i,j} = \sum_{\ell=1}^n \bg_{\ell,i} \bg_{\ell,j}^2 (1+\alpha \cdot \alpha^2 + \alpha \cdot \alpha^2) = \sum_{\ell=1}^n \bg_{\ell,i} \bg_{\ell,j}^2.$$ Thus, assuming $\bG$ is random, the matrix $\bX$ is random as well and has w.h.p. full rank.

Even though LEP is not NP-hard, it is considered to be quantum-secure, as only exponential cost solvers (classical and quantum) are known. Hence, it is a promising candidate for post-quantum cryptography.

We can also consider code equivalence for $\mathbb{F}_q$-linear codes endowed with the rank metric.
\begin{problem}[Matrix Code Equivalence (MCE) Problem]
    Given $\bG_1, \ldots, \bG_k \in \mathbb{F}_q^{m \times n}$ and $\bG_1', \ldots, \bG_k' \in \mathbb{F}_q^{m \times n}$. Find  $\bA \in \text{GL}_m(q), \bB\in \text{GL}_n(q)$, such that for all $\bC \in \langle \bG_1, \ldots, \bG_k \rangle$ we have $\bA\bC\bB= \bC'$ for some $\bC' \in \langle \bG_1', \ldots, \bG_k'\rangle.$
\end{problem}

Similar to LEP is not NP-hard but believed to be quantum-secure. In fact, there exists a polynomial time reduction from the Hamming code equivalence problem in \cite{debris}. A nice summary on MCE can be found in \cite{krijn}.

\subsubsection{Rank SDP}

In \cite{randomred}, the authors provide a randomized reduction from the SDP to the Rank SDP. A randomized reduction is a polynomial time reduction, which only works with high probability. 

\begin{proposition}
    There exists a randomized reduction from SDP to Rank SDP.
\end{proposition}

\begin{proof}
Instead of using the SDP, we use the equivalent GWCP, in both metrics. 
    We start with an instance of GWCP in the Hamming metric, namely $\bG \in \mathbb{F}_q^{k \times n}$ and $t.$
    Note that the Rank GWCP is only defined over extension fields, $\mathbb{F}_{q^m}.$ Thus, we consider $\alpha \in \mathbb{F}_{q^m}^n$ a vector with $\mathbb{F}_q$-linearly independent entries. 

    \begin{exercise}
        Show that for any $\bx \in \mathbb{F}_q^n$ of $\text{wt}_H(\bx)=t$ the componentwise product $\bx \star \alpha \in \mathbb{F}_{q^m}^n$ has rank weight $\text{wt}_R(\bx \star \alpha)=t.$
    \end{exercise}
    For the code $\mathcal{C} =\langle \bG \rangle \subset \mathbb{F}_q^n$, we define $\mathcal{C}' = \langle \{\alpha \star \bc \mid \bc \in \mathcal{C}\} \rangle \subset \mathbb{F}_{q^m}^n$.  Let $\bG'$ be a generator matrix of $\mathcal{C}'$. If the answer to the Hamming GWCP is yes, that is: there exists a  $\bc \in \mathcal{C}$ of Hamming weight $t$, then there also exists $\bc \star \alpha$ in $\mathcal{C}'$ of rank weight $t.$ However, if there was no $\bc \in \mathcal{C}$ of Hamming weight $t$, note that there might still be a codeword $\bc' \in \mathcal{C}'$ of rank weight $t$, which is not of the form $\bc \star \alpha$. In fact, since we are now over the extension field, we have generated many more codewords than simply those of the form $\bc \star \alpha$. 

    With high probability, (details can be found in \cite{randomred}), the only codewords of rank weight $t$ are of the form $\bc \star \alpha$ and thus the reduction works.
\end{proof}

\begin{example}
    Let us consider $$\bG= \begin{pmatrix} 1 & 0 & 1 \\ 0 & 1 & 1 \end{pmatrix} $$ which generates the code $\mathcal{C} \subset \mathbb{F}_2^3$. If we let $t=1$, then clearly there is no codeword in $\mathcal{C}$ of Hamming weight 1. However, for $\mathbb{F}_8=\mathbb{F}_2[\alpha]$ and $\alpha^3=\alpha+1$ and the code $$ \mathcal{C}'=\langle \bc \star (1, \alpha, \alpha^2+\alpha+1) \mid \bc \in \mathcal{C}\}\rangle$$ we do have a codeword of rank weight 1, for example
    $$\alpha (1, \alpha, \alpha^2+\alpha+1) \star (1,0,1) + (1, \alpha, \alpha^2+\alpha+1) \star (0,1,1) = (\alpha, \alpha, \alpha).$$
\end{example}

It remains one of the  largest open problems in code-based cryptography, whether there exists a polynomial time reduction, which always works. That is
\begin{problem}[Open Problem]
    Is the Rank SDP NP-hard?
\end{problem}
As opposed to the Rank SDP, considering $\mathbb{F}_{q^m}$-linear codes endowed with the rank metric, for $\mathbb{F}_q$-linear codes the SDP is known to be NP-hard. 
\begin{theorem}
    The MinRank Problem is NP-complete.
\end{theorem}
\begin{proof}
    We use a polynomial time reduction from the Hamming DP. 
    \begin{exercise}
        Let $\bx \in \mathbb{F}_q^n$ have Hamming weight $t$. Show that $\text{diag}(\bx)\in \mathbb{F}_q^{n \times n}$ has rank weight $t$.
    \end{exercise}
    We start with an instance $\bG =\begin{pmatrix} \bg_1 \\ \vdots \\ \bg_k \end{pmatrix} \in\mathbb{F}_q^{k\times n}, \br \in \mathbb{F}_q^n$ and $t \in \mathbb{N}.$
    We transform the instance to a MinRank instance as $$\bG_1=\text{diag}(\bg_1), \ldots, \bG_k=\text{diag}(\bg_k) \in \mathbb{F}_q^{n\times n}$$ and $$\bR=\text{diag}(\br) \in \mathbb{F}_q^{n \times n}.$$
    Let us first assume that the Hamming DP instance has ``yes'' as a solution, i.e., there exists a $\be \in \mathbb{F}_q^n$ of Hamming weight $t$ such that $\br-\be \in \langle \bG \rangle$. In other words, $$\br-\be = \lambda_1 \bg_1 + \cdots +\lambda_k \bg_k$$ for some $\lambda_i \in \mathbb{F}_q$. Then the MinRank instance also has a solution ``yes''. In fact, there exists $\bE= \text{diag}(\be)$ of rank weight $t$ such that
    $$\bR-\bE=\lambda_1 \bG_1 + \cdots  + \lambda_k \bG_k,$$ for the same $\lambda_i \in \mathbb{F}_q.$
    On the other hand, if the Hamming DP instance has ``no'' as a solution, i.e., there is no $\be \in \mathbb{F}_q$ of Hamming weight $t$, such that $\br -\be = \lambda_1 \bg_1 + \cdots + \lambda_k\bg_k$, then the MinRank instance also gives ``no'' as a solution. In fact, assume by contradiction, a $\bE  \in \mathbb{F}_q^{n \times n}$ exists of rank weight $t$, such that 
    $$\bR-\bE= \lambda_1 \bG_1 + \cdots + \lambda_k \bG_k$$ for some $\lambda_i \in \mathbb{F}_q.$ Thus, if we denote by $g_i^j$ the $i$th entry of $\bg_j,$ then
    $$\begin{pmatrix} r_1 &   \cdots & 0\\
    &  \ddots & \\ 0& \cdots &  r_n \end{pmatrix} -\bE = \lambda_1 \begin{pmatrix} g_1^1 & \cdots & 0  \\
    & \ddots & \\ 0& \cdots  &  g_n^1 \end{pmatrix} + \cdots + \lambda_k \begin{pmatrix} g_1^k &   \cdots & 0 \\  
    &  \ddots & \\ 0&\cdots  &  g_n^k \end{pmatrix}.   $$ Hence,
    $$\bE = \begin{pmatrix} r_1- \sum_{i=1}^k \lambda_i g_1^i &   \cdots & 0 \\ 
    
    &  \ddots & \\ 
     0& \cdots  &  r_n -\sum_{i=1}^k \lambda_i g_n^i \end{pmatrix}.$$
     Hence $\bE$ is again a diagonal matrix and we can denote its diagonal by $\be.$
     In order for $\bE$ to have rank weight $t$, we need that $t$ many entries of $\be$ are non-zero. Hence there exists a $\be \in \mathbb{F}_q^n$ of Hamming weight $t$, such that $e_j=r_j - \sum_{i=1}^k \lambda_i g_j^i$, i.e., $\be=\br-\sum_{i=1}^k \lambda_i \bg_i,$ which is a contradiction. 
\end{proof}

The natural question arises, why one cannot prove the NP-hardness of Rank SDP using the MinRank problem. 
In fact, starting with an instance of Rank SDP, i.e., an $\mathbb{F}_{q^m}$-linear code, one can always define the corresponding $\mathbb{F}_q$-linear code. However, for the polynomial time reduction from MinRank to Rank SDP, the other direction is needed. That is, starting with an instance of MinRank, transforming it to an instance of Rank SDP - and this already fails, as not all $\mathbb{F}_q$-linear codes can be lifted to an $\mathbb{F}_{q^m}$-linear code.

\subsubsection{Lee SDP}

\begin{problem}[Lee SDP]
Let $\mathbb{F}_{p}$ be a prime field and $k \leq n$ be positive integers.   Given $\bH \in \mathbb{F}_{p}^{(n-k) \times n}$, $\bs \in \mathbb{F}_{p}^{n-k}$ and $t \in \mathbb N$, is there a vector $\be \in \mathbb{F}_{p}^n$ such that $\text{wt}_L(\be)\leq t$ and $\be \bH^\top = \bs$?
\end{problem}
The Lee SDP (again equivalent to Lee DP and Lee GWCP) has been proven to be NP-complete in \cite{leeNP}. Since the proof follows exactly in the same manner as the reduction for Hamming SDP, we leave it as an exercise.
\begin{exercise}
    Show that Lee SDP is NP-complete using a reduction from 3DM.
\end{exercise}

\subsubsection{Restricted SDP}

The Restricted Syndrome Decoding Problem (R-SDP), first introduced in \cite{rest}, reads as follows.
\begin{problem}[Restricted SDP]
Given $g \in \mathbb{F}_p^*$ of prime order $z$, $\mathbf{H} \in \mathbb{F}_p^{(n-k)\times n}$, 
$\mathbf{s} \in \mathbb{F}_p^{n-k}$, and $\mathbb{E} = \{g^i \mid i \in \{1, \ldots, z\} \} \subset \mathbb{F}_p^*$, decide if there exists $\mathbf{e} \in \mathbb{E}^n$ such that $\mathbf{e} \mathbf{H}^\top = \mathbf{s}$.
\end{problem}
The Restricted SDP is strongly related to other well-known hard problems. For example, when $z=p-1$, the Restricted SDP is close to the classical SDP; if $z=1$, the Restricted SDP is similar to the Subset Sum Problem (SSP) over finite fields.
Consequently, it is unsurprising that the R-SDP is NP-complete for any choice of $\mathbb{E}$.

\begin{theorem}\label{thm:np}
The Restricted SDP is NP-complete.
\end{theorem}
The proof is again similar to the reduction provided for the SDP.  
\begin{proof}
    Recall the NP-hard 3-Dimensional Matching (3DM) problem, where one is given the instance $T=\{b_1, \ldots, b_t\},$ with $|T|=t, U \subset T \times T \times T$ and $|U|=u$ and asks whether there exists a $W \subset U$ with $|W|=t$ and no two words in $W$ coincide in any position.
    \medskip

    Recall that the original SDP has a reduction from 3DM, through the following construction: let $\bH \in \mathbb{F}_p^{ (3t) \times u} $  be the incidence matrix, i.e., each column of $\bH$ corresponds to a word in $U$ and the rows correspond to $T\times T \times T$, thus  the rows $\{1,\ldots, t\}$ correspond to the first position of the word $\bu$, the rows $\{t+1, \ldots, 2t\}$ correspond to the second position of $\bu$ and the rows $\{2t+1, \ldots, 3t\}$ correspond to the third position of $\bu$. More formally,  let $T=\{b_1, \ldots, b_t\}$, $U=\{\ba_1, \ldots, \ba_u\}$ and 
    \begin{itemize}
    \item for $ j \in \{1, \ldots, t\}$, we set $h_{i,j} = 1$ if $\ba_i[1]= b_j$ and $h_{i,j}=0$ else,
     \item for $ j \in \{t+1, \ldots, 2t\}$, we set $h_{i,j} = 1$ if $\ba_i[2]= b_j$ and $h_{i,j}=0$ else,
     \item for $ j \in \{2t+1, \ldots, 3t\}$, we set $h_{i,j} = 1$ if $\ba_i[3]= b_j$ and $h_{i,j}=0$ else.
 \end{itemize}
 We also set $\bs \in \mathbb{F}_p^{3t}$ be the all one vector.

 From the original reduction, we know that any solution $\be \in \mathbb{F}_p^u$ with $\bH\be^\top=\bs^\top$ has weight $t$ and its support corresponds to the solution $W$. That is the columns of $\bH$ indexed by the support of $\be$ are the $t$ words in $W$.
 \medskip

 The polynomial reduction from 3DM to R-SDP uses this construction as well.
  Let $T$ of size $t$ and $U\subset T \times T \times T$ of size $u$ be an instance of 3DM.
  Let $\bH \in \mathbb{F}_p^{(3t) \times u}$ be the incidence matrix and let
  $$\widetilde{\bH}= \begin{pmatrix} \bH & -g \star \bH \\ \text{Id}_u & \text{Id}_u \end{pmatrix} \in \mathbb{F}_p^{(3t+u) \times 2u}$$ be a parity-check matrix. Let us consider the syndrome $(\bs,\bs') \in \mathbb{F}_p^{3t+u}$ with $\bs=(1-g^2, \ldots, 1-g^2) \in \mathbb{F}_p^{3t}$ and $\bs'=(1+g, \ldots, 1+g) \in \mathbb{F}_p^u.$
  Thus, the instance of R-SDP given by $\widetilde{\bH}$ and $(\bs,\bs')$  is asking for $(\be,\be') \in \mathbb{E}^{2u}$ such that $$(\be,\be')\widetilde{\bH}^\top=(\bs,\bs'),$$ where $\mathbb{E}= \{g^i \mid i \in \{0, \ldots, z-1\}\}.$ By assumption of R-SDP, we use a $g$ of order $2<z<q-1.$

  We consider two cases. 
  \begin{enumerate}
      \item Assume that the R-SDP solver returns ``yes'', i.e., there exists $\be,\be' \in \mathbb{E}^u$ such that $(\be,\be')\widetilde{\bH}^\top=(\bs,\bs').$
      Hence, \begin{align*}
          \bH\be^\top -g \star \bH \be'^\top & = (1-g^2, \ldots, 1-g^2)^\top, \\ 
          \be + \be' &= (1+g, \ldots, 1+g).
      \end{align*}
      Hence, for each $i \in \{1, \ldots, u\}$ we have  $e_i + e_i'=1+g$. 
      Let us assume (we later show that this hypothesis is not needed, but it facilitates the proof) that the only elements in $\mathbb{E}$ that add to $1+g$ is $1$ and $g.$

      Hence, whenever $e_i=1,$ we must have $e_i'=g$ and whenever $e_i=g$, we must have $e_i'=1.$ Thus, we split $\be= \be_1 + \be_g$ and $\be'=\be_1' + \be_g'$ where $\be_1,\be_1' \in \{0,1\}^u, \be_g,\be_g' \in \{0,g\}^u $ and  $$\text{supp}(\be_1)= S = \text{supp}(\be_g')$$ and $$\text{supp}(\be_1')=S^C=\text{supp}(\be_g).$$ 
      From this also follows that $$\be_g= g\star \be_1'$$ and $$\be_g' = g \star \be_1.$$
      \medskip

      The first parity-check equation can now be reformulated as
      \begin{align*}
          & \bH\be^\top -g \star \bH \be'^\top  \\ = & \bH\be_1^\top -g \star \bH \be_g'^\top + \bH \be_g^\top -g \star \bH\be_1'^\top \\ 
          = & \bH \be_1^\top - g^2 \star \bH \be_1^\top + g \star \bH \be_1'^\top -g \star \bH \be_1'^\top \\ 
          = & (1-g^2) \star \bH \be_1^\top \\ 
          = & (1-g^2, \ldots, 1-g^2) = \bs',
      \end{align*}
      thus, $\bH\be_1^\top = (1, \ldots, 1)$ is such that $\text{supp}(\be_1)$ corresponds to a solution $W$ of 3DM, as in the classical reduction. 

      \item Assume that the R-SDP solver returns ``no'', i.e., there exists no $\be,\be' \in \mathbb{E}^u$ such that 
      $(\be,\be')\widetilde{\bH}^\top=(\bs,\bs').$
      Let us assume by contradiction, that the 3DM has a solution $W.$ We can then define $S$ to be the indices of words in $U$ belonging to the solution $W$. Let us define $\be_1,\be_1' \in \{0,1\}^u, \be_g,\be_g' \in \{0,g\}^u $ with  $\text{supp}(\be_1)= S = \text{supp}(\be_g')$ and $\text{supp}(\be_1')=S^C=\text{supp}(\be_g)$. 
      From this also follows that $\be_g= g\star \be_1'$ and $\be_g' = g \star \be_1.$
      Then the vector $(\be_1+\be_g, \be_1' + \be_g') \in \mathbb{E}^{2u}$ is a solution to the R-SDP, as in case 1, which gives the desired contradiction, to the R-SDP solver returning ``no''. 
  \end{enumerate}

  Note that the hypothesis, that only 1 and $g$ in $\mathbb{E}$ add up to $1+g$ is not necessary. For this assume that there exists $g^i,g^j \in \mathbb{E}$, with $0 \neq i <j<z$ such that $g^i+g^j=1+g.$ 
  Thus, the splitting of $\be$ and $\be'$ is a bit more complicated:
  \begin{align*}
      \be &= \be_1 + \be_g + \be_i + \be_j, \\ 
      \be' &= \be_1' + \be_g' + \be_i' + \be_j',
  \end{align*}
  where $\be_1, \be_1' \in \{0,1\}^u$,$ \be_g,\be_g' \in \{0,g\}^u$,$ \be_i,\be_i' \in \{0,g^i\}^u$,$ \be_j, \be_j' \in \{0,g^j\}^u$ with \begin{align*} 
  \text{supp}(\be_1) &=S_1 = \text{supp}(\be_g'), \\  \text{supp}(\be_g) &=S_1' = \text{supp}(\be_1'), \\ 
   \text{supp}(\be_i) &=S_i = \text{supp}(\be_j'), \\ 
    \text{supp}(\be_j) &=S_i' = \text{supp}(\be_i'),
  \end{align*} and the supports $S_1,S_1',S_i,S_i'$ are distinct and partition $\{1,\ldots,u\}.$
  Again it follows that 
  \begin{align*}
      \be_g &= g \star \be_1', \\ \be_g' &= g \star \be_1, \\ 
      \be_j &= g^{j-i} \star \be_i', \\ 
      \be_j' &= g^{j-i} \star \be_i. 
  \end{align*}
  Thus, rewriting the first parity-check equation, we get 
  \begin{align*}
      & \bH\be^\top -g \star \bH\be'^\top \\
      = & \bH\be_1^\top + \bH\be_g^\top + \bH\be_i^\top + \bH\be_j^\top \\ 
      & - g \star \bH \be_1'^\top -g \star \bH \be_g'^\top - g \star \bH \be_i'^\top - g \star \bH \be_j'^\top \\ 
      = &  \bH\be_1^\top +g \star  \bH\be_1'^\top + \bH\be_i^\top + g^{j-i} \star \bH\be_i'^\top \\ 
      & - g \star \bH \be_1'^\top -g^2 \star \bH \be_1^\top - g \star \bH \be_i'^\top - g^{j-i+1} \star \bH \be_i^\top \\ 
      = & (1-g^2) \star \bH \be_1^\top +(1-g^{j-i+1}) \star \bH \be_i^\top + (g^{j-i}-g) \star \bH \be_i'^\top \\ = & (1-g^2, \ldots, 1-g^2) =\bs'.
  \end{align*}
  Since $\be_1, \be_i, \be_i'$ all have different supports, the only way to get $1-g^2$ in each entry, is to have $\be_i=\be_i'=0$. In fact, any other sum leads to a contradiction:
  \begin{itemize}
      \item If $(1-g^2) + (1-g^{j-i+1})=1-g^2$ then $1=g^{j-i+1}$ and hence $j=i-1$ which contradicts $j>i$.
      \item If $(1-g^2)+(g^{j-i}-g)=1-g^2$ then $g^{j-i}=g$ and hence $j-i=1$. However, as then $g^j+g^i=g^i(1+g) = 1+g$, it follows that $g^i=1$, which contradicts $i \neq 0.$
      \item If $(1-g^2)+(1-g^{j-i+})+(g^{j-i}-g)=1-g^2$, then $1+g^{j-i}=g^{j-i+1}+g=g(1+g^{j-i})$ and thus $g=1$, which contradicts $\mathbb{E}\neq \mathbb{F}_q^\star.$
      \item If $(1-g^{j-i+1})+(g^{j-i}-g)=1-g^2$, then $g^{j-i}-g^{j-i+1}=g-g^2$ and hence $g^{j-i}(1-g)=g(1-g)$ and thus $j-i=1$, which is a contradiction again as in the second case. 
  \end{itemize}
\end{proof}
 
 \subsection{Information Set Decoding}\label{sec:ISD}

In this section we cover the main approach to solve the SDP, namely information set decoding (ISD) algorithms. For this we will follow closely \cite{weger}.

In the McEliece and the Niederreiter framework the secret code is usually endowed with a particular algebraic structure to guarantee the existence of an efficient decoding algorithm and is then hidden from the public to appear as a random code. In different frameworks, such as the quasi-cyclic framework, the secret is actually purely the error vector and the algebraic code is made public. In both cases an adversary has to solve the NP-complete problem of decoding a random linear code. \\

An adversary would hence use the best generic decoding algorithm for random linear codes.  
Two main methods are known until today for decoding random linear codes: ISD and the generalized
birthday algorithm (GBA). ISD algorithms are more efficient if the decoding problem has only a small number of solutions, whereas GBA is more efficient when there are many solutions. Also other ideas such as statistical decoding \cite{statistical}, gradient decoding \cite{gradient} and supercode decoding \cite{super} have been proposed but fail to outperform ISD algorithms.\\

 ISD algorithms are  an important aspect of code-based cryptography since  
they predict the key size achieving a given security level. ISD algorithms should not be considered as attacks in the classical sense, as they are not breaking a code-based cryptosystem, instead they  determine the choice of parameters for a given
security level.

Due to the duality of the decoding problem and the SDP also ISD algorithms can be formulated through the generator matrix or the parity-check matrix. Throughout this survey, we will stick to the parity-check matrix formulation.\\

The first ISD algorithm was proposed in 1962 by Prange \cite{prange} and interestingly,  all improvements display the same structure: choose an information set,   use Gaussian elimination to bring the parity-check matrix in a standard form, assuming a certain weight distribution on the error vector, we can go through smaller parts of the error vector and check if the parity-check equations  are satisfied. The assumed weight distribution of the error vector thus constitutes the  main part of an ISD algorithm.

In an ISD algorithm we fix a weight distribution and go through all information sets to find an error vector of this weight distribution. This is in contrast to `brute-force attacks' where one fixes an information set and goes through all weight distributions of the error vector. 
In fact, due to this, ISD algorithms are in general not deterministic, since there are instances for which there exists no information set where the error vector has the sought after  weight distribution.
Clearly, a brute-force algorithm requires much more binary operations than an ISD algorithm, thus, in practice we only consider ISD algorithms.

For this section we will need to recall some notation: let $S \subseteq \{1, \ldots, n\}$ be a set of size $s$, then for a vector $\bx \in \mathbb{F}_q^n$ we denote by $\bx_S$ the vector of length $s$ consisting of the entries of $\bx$ indexed by $S$. Whereas, for a matrix $\bA \in \mathbb{F}_q^{k \times n}$, we denote by $\bA_S$ the matrix consisting of the columns of $\bA$ indexed by $S.$ For a set $S$ we denote by $S^C$ its complement. For $S \subseteq \{1, \ldots, n\}$ of size $s$ we denote by $\mathbb{F}_q^n(S)$ the vectors in $\mathbb{F}_q^n$ having support in $S$. The projection of $\bx \in \mathbb{F}_q^n(S)$ to $\mathbb{F}_q^s$ is then canonical and denoted by $\pi_S(\bx)$. On the other hand,  we denote by $\sigma_S(\bx)$ the canonical embedding of a vector $\bx \in \mathbb{F}_q^s$ to $\mathbb{F}_q^n(S).$

\subsubsection{General Algorithm}

We are given a parity-check matrix $\bH \in \F_q^{(n-k) \times n}$ of a code $\mC$, a positive integer $t$ and a syndrome $\bs \in \F_q^{n-k}$, such that there exists a vector $\be \in \F_q^n$ of Hamming weight less than or equal to $t$ with syndrome $\bs$, {i.e.}, $\be\bH^\top = \bs$. The aim of the algorithm is to find such a  vector $\be$.

\begin{itemize}
\item[1.] Find an information set $I \subset \{1, \ldots, n\}$ of size $k$ for $\mC$. 
\item[2.] Bring $\bH$ into the systematic form corresponding to $I$, i.e., find an invertible matrix $\bU \in \F_q^{(n-k) \times (n-k)}$, such that  $(\bU\bH)_I = \bA$, for some $\bA \in \F_q^{(n-k) \times k}$ and $(\bU\bH)_{I^C} = \Id_{n-k}$.
\item[3.] Go through all error vectors $\be \in \F_q^n$ having the assumed weight distribution  (and in particular having Hamming weight $t$).
\item[4.] Check if the parity-check equations, {i.e.}, $\be \bH^\top \bU^\top = \bs \bU^\top$ are satisfied. 
\item[5.] If they are satisfied, output $\be$, if not start over with a new choice of $I$. 
\end{itemize}

Since the iteration above has to be repeated several times, the cost of such algorithm is given by the cost of one iteration times the number of required iterations. 

Clearly, the average number of iterations required is given as the reciprocal of the success probability of one iteration and this probability is completely determined by the assumed weight distribution.

\subsubsection{Overview Algorithms}

The first ISD algorithm was proposed in 1962 by Prange \cite{prange} and is sometimes referred to as plain ISD. In this algorithm Prange makes use of an information set of a code, that in fact contains all the necessary information to decode, in a clever way.  For this we have to assume that there is an information set where the error vector has weight 0 (thus all $t$ errors are outside of this information set). One now only has to bring the parity-check matrix into systematic form according to this information set, which has a polynomial cost, this is called an iteration of the algorithm. However, one has to find such an information set first. This is done by trial and error, which results in a large number of iterations. 
Indeed, the assumption that no errors happen in the information set is not very likely and thus   the success probability of one iteration is  very low. 
\medskip
 
All the improvements that have been suggested to Prange’s simplest form of ISD (see for example \cite{chabanne, canteautsendrier, chabaud, dumer, kruk, leon, vantilburg}) assume a  more likely weight distribution of the error vector, which results in a higher cost
of one iteration but give overall a smaller cost, since less iterations have to be performed.
\medskip

 The improvements split into two directions: the first direction is following the idea of Lee and Brickell \cite{leebrickell} where 
they ask for $v$ errors in the information set and $t - v$ outside.
The second direction is Dumer’s approach \cite{dumer}, which is asking for $v$ errors in $k + \ell$ bits, which
are containing an information set, and $t - v$ in the remaining $n-k - \ell$ bits. Clearly, the second direction includes the first direction by setting $\ell=0$.

Following the first direction, Leon \cite{leon} generalizes Lee-Brickell's algorithm by introducing a  set of size $\ell$ outside the information set called zero-window, where no errors happen.  
In 1988, Stern \cite{stern} adapted the algorithm by Leon and proposed
to partition the information set into two sets and ask for $v$ errors in each part and $t-2v$ errors outside the information set (and outside the zero-window).  In 2010, with the rise of code-based cryptography over a general finite field $\mathbb{F}_q$, Peters generalized these algorithms to $\F_q$ \cite{peters}.
 \medskip
 
In 2011, Bernstein, Lange and Peters proposed the ball-collision algorithm \cite{ballcoll},
where   they reintroduce errors in the zero-window. In fact, they partition the zero-window into two sets and ask for $w$ errors in both and hence for $t -2v -2w$
errors outside.  This algorithm and its speed-up techniques were then generalized to $\F_q$ by Interlando, Khathuria, Rohrer, Rosenthal and Weger in \cite{ballcollFq}. 
In 2016, Hirose \cite{hirose} generalized the nearest neighbor algorithm over $\F_q$ and applied it to
the generalized Stern algorithm.
\medskip

An illustration of these algorithms is given in Figure \ref{figure1}, where we  assume for simplicity that the information set is in the first $k$ positions and the zero-window is in the adjacent $\ell$ positions.

\begin{center}
\begin{figure}[h!]
\includegraphics[width=12cm]{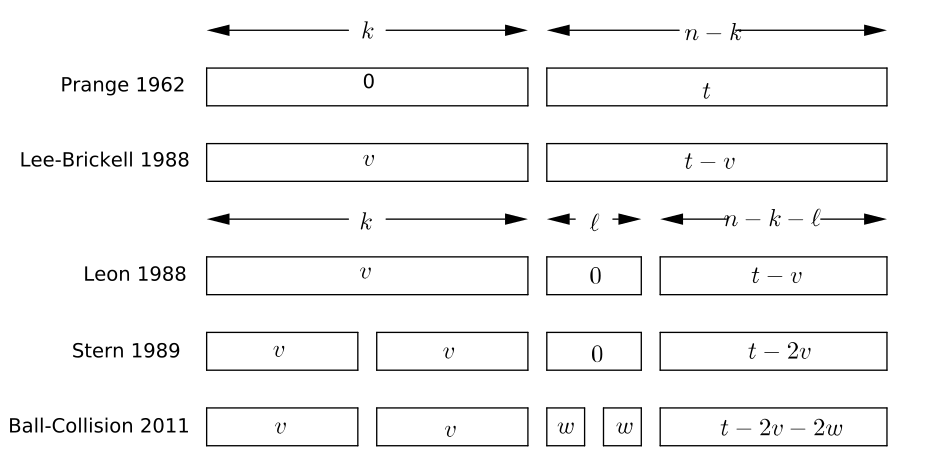}
\caption{Overview of algorithms following the splitting of Lee-Brickell, adapted from \cite{ballcoll}.}\label{figure1}
\end{figure}
\end{center}

The second direction has resulted in many improvements, for example in 2009 Finiasz and Sendrier \cite{sendrier} have built two intersecting subsets of the $k + \ell$ bits, which contain an information set, and ask for
$v$ disjoint errors in both sets and $t -2v$ in the remaining $n- k - \ell$ bits. Niebuhr, Persichetti, Cayrel,
Bulygin and Buchmann  \cite{niebuhr} in 2010 improved the performance of ISD algorithms over $\F_q$ based on the
idea of Finiasz and Sendrier. 

\medskip

In 2011, May, Meurer and Thomae \cite{mmt} proposed the use of 
the representation technique introduced by Howgrave-Graham and Joux \cite{howgrave} for the subset sum problem. Further improvements have been proposed by  Becker,
Joux, May and Meurer \cite{bjmm} in 2012 by introducing overlapping supports. We will refer to this algorithm   as BJMM. 
In 2015,
May-Ozerov \cite{MO} used the nearest neighbor algorithm to improve  BJMM  and finally in
2017, the nearest neighbor algorithm over $\F_q$ was applied to the generalized BJMM algorithm by Gueye,
Klamti and Hirose \cite{klamtihir}.

\medskip

These new approaches do not use set partitions of the support but rather a sum partition of the weight. An illustration of these algorithms is given in Figure \ref{figure2}, where we again assume that the $k + \ell$ bits containing an information set are in the beginning. The overlapping sets are denoted by $X_1$ and $X_2$ and their intersection of size $2 \alpha(k+\ell)$ is in blue. The amount of errors within the intersection is denoted by $\delta$.

\begin{center}
\begin{figure}[h!]
\includegraphics[width=12.5cm]{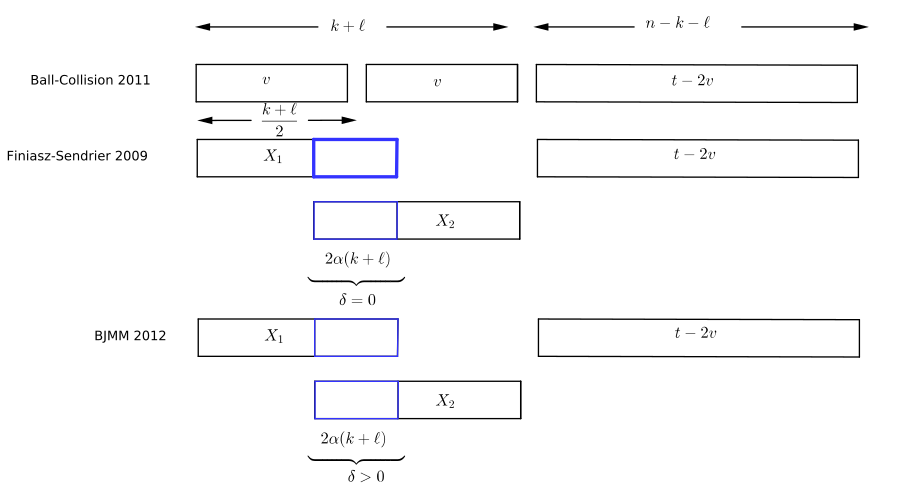}
\caption{Overview of the weight splitting in the different algorithms.}\label{figure2}
\end{figure}
\end{center}

A very introductory reading on ISD algorithms is in the thesis of Weger \cite{weger}, which we also follow closely  and for binary ISD algorithms, a very informative reading is the thesis of Meurer \cite{meurer}.\\

It is important to remark (see \cite{meurer}) that the BJMM algorithm, even if having the smallest complexity until today,
comes with a different cost: memory. In order to achieve a complexity of 128 bits, BJMM needs about
$10^9$ terabytes of memory. In fact, Meurer observed that if one restricts the memory to $2^{40}$ (which is a reasonable restriction), BJMM and
the ball-collision algorithm are performing almost the same. 

\medskip
 
 What is the possible impact on the cost of ISD algorithms when using a capable quantum computer? In \cite{bernpq} the authors  expect that quantum computers result in a square root speed up for ISD algorithms, since Grover's search algorithm \cite{grover1, grover2}  needs only $O(\sqrt{N})$ operations to find an element in a set of size $N$, instead of $O(N)$ many. Thus, intuitively, the search of an information set will become faster and thus the number of iterations needed in an ISD algorithm will decrease. 
 \medskip
 
Since all the improvements upon Prange's algorithm were only focusing on decreasing this number of iterations, the speed up for these algorithms will be smaller, than for the original algorithm by Prange. 
Hence the authors predict that  on a capable quantum computer Prange's algorithm will result as the fastest. 

  \newpage
\subsubsection{Techniques}

In   the following we introduce some speed-up techniques for ISD algorithms, mostly introduced in \cite{ballcoll} over $\F_2$ and later generalized to $\F_q$ in \cite{ballcollFq}. \\

First of all, we want to fix the cost that we consider throughout this chapter of one addition and  one multiplication over $\F_q$, {i.e.}, we assume that one addition over $\F_q$ costs $\lcq$ binary operations and one multiplication costs $\lcq^2$ binary operations. The cost of the multiplication is clearly not using the fastest algorithm known but will be good enough for our purposes. Also for the cost of multiplying two matrices we will always stick to a broad estimate given by school book long multiplication, {i.e.}, multiplying $\bA \bB$, where $\bA \in \F_q^{k \times n}$ and $\bB \in \F_q^{n \times r}$ will cost $nkr \left( \lcq + \lcq^2 \right)$ binary operations.

\paragraph{Number of Iterations}

One of the main parts in the cost of an information set decoding algorithm is the \emph{average number of iterations} needed. This number depends on the success probability of one iteration. 
In turn, the success probability is completely given by the assumed weight distribution of the error vector.
Since in one iteration we consider a fixed information set, the success probability of an iteration is given by the fraction of how many vectors  there are with the assumed weight distribution, divided by how many vectors  there are in general with the target weight $t$. 

\begin{example}
For example, we are looking for $\be \in \F_q^n$ of Hamming weight $t$, and we assume that the error vector has no errors inside an information set $I$, and thus all $t$ errors appear in $I^C$ of size $n-k$. Since there are $\binom{n-k}{t} (q-1)^t$ many vectors having support of size $t$ in a set of size $n-k$  and the total number of vectors of support $t$ in a set of size $n$ is given by $\binom{n}{t}(q-1)^t$, we have that the success probability of one iteration is given by 
$$\binom{n-k}{t} \binom{n}{t}^{-1},$$
and hence the number of iterations needed on average is given by 
$$\binom{n-k}{t}^{-1} \binom{n}{t}.$$ 
\end{example}

\paragraph{Early Abort}

In some of the algorithms we have to perform a computation %
and the algorithm only proceeds if the result of this computation satisfies a certain condition. In our case, the condition is that the weight of the resulting vector does not exceed a target weight. 

We  thus compute one entry of the result and check the weight of this entry, before proceeding to the next entry. As soon as the weight of the partially computed vector is above the target weight, we can stop the computation, hence the name \emph{early abort}.

\begin{example}
 To provide an example also for this technique, assume that we have to compute $\bx \bA$, for  $\bx \in \F_q^k$ of Hamming weight $t$ and $\bA \in \F_q^{k \times n}$. Usually computing $\bx \bA$ would cost
 $nt \left( \lcq^2 +  \lcq\right)$ binary operations.

However, assuming our algorithm only proceeds if $\wtH(\bx\bA ) =w$, we can use the method of early abort, {i.e.}, computing one entry of the resulting vector and checking its weight simultaneously. 
For this we assume that the resulting vector is uniformly distributed. Since we are over $\F_q$, the probability that an entry adds to the weight of the full vector is given by $\frac{q-1}{q}$. Hence we can expect that after computing $\frac{q}{q-1}w$ entries the resulting vector should have reached the weight $w$, and after computing $\frac{q}{q-1}(w+1)$ entries we should have exceeded the target weight $w$ and can abort. 
Since computing only one entry of the resulting vector  costs $ t \left( \lcq^2 +  \lcq \right)$ binary operations, the cost of this step is  given by 
$$\frac{q}{q-1}(w+1) t \left(  \lcq^2 +  \lcq\right)$$ binary operations, instead of  
$$nt \left(  \lcq^2 + \lcq\right).$$
Clearly, this is a speed up, whenever $ \frac{q}{q-1}(w+1) < n$.
 \end{example}

\paragraph{Number of Collisions}
 
In some algorithms we want to check if a certain condition is verified and only then we would proceed. This condition depends on two vectors $\bx$ and $\by$ living in some sets. $S$, respectively $T$. 
Hence the algorithm would go through all the vectors $\bx \in S$ 
and then through all the vectors $\by \in T$ 
in their respective sets and check if the condition is satisfied for a fixed pair $(\bx, \by)$. If this is the case, such a pair is called a \emph{collision}. 
The subsequent steps of the algorithm would be performed on all the collisions, thus multiplying the cost of these steps with the size of the set of all $(\bx, \by)$, i.e., $\mid S \mid \mid T \mid$.

Instead, we can compute the \emph{average number of collisions} we can expect on average. 

\begin{example}
Let us also give an example for this technique; assume that we only proceed whenever $$\bx + \by = \bs,$$ for a fixed   $\bs \in \F_q^{k}$ and for all $\bx \in \F_q^k$ of Hamming weight $v$ and all $\by \in \F_q^k$ of Hamming weight $w$. 
To verify this condition we have to go through all possible $\bx$ and $\by$, thus costing
$$\binom{k}{v} \binom{k}{w} (q-1)^{v+w}\min\{k,v+w\} \log_2(q)$$ binary operations. As a subsequent step one would compute for all such $(\bx,\by)$ the vector  $\bA\bx -\bB\by$, for some fixed $\bA \in \F_q^{n \times k}$ and $\bB \in \F_q^{n \times k}$. Usually one would do this for all elements in $S=\{(\bx, \by) \mid \bx, \by \in \F_q^k, \wtH(\bx) = v, \wtH(\by) =w\}$, giving this step a cost of 
$$\binom{k}{v}\binom{k}{w} (q-1)^{v+w} \min\{k,v+w\}n \left( \log_2(q) + \log_2(q)^2 \right).$$
 
However, we only have to perform the subsequent steps as many times as on average we expect a collision, {i.e.}, a pair $(\bx,\by)$ such that $\bx+\by=\bs$.
Assuming a uniform distribution, this amount is given by $$\frac{\mid S \mid}{q^n} = \frac{\binom{k}{v}\binom{k}{w}(q-1)^{v+w}}{q^n}< \binom{k}{v}\binom{k}{w}(q-1)^{v+w-n}.$$ 
Thus computing $\bA\bx -\bB\by$ for all $(\bx, \by) \in S$ costs on average 
$$\binom{k}{v}\binom{k}{w}(q-1)^{v+w-n} \min\{k,v+w\} n \left( \log_2(q) + \log_2(q)^2\right)$$ binary operations, which is clearly less than the previous cost.
\end{example}

\paragraph{Intermediate Sums}

In some algorithms we have to do a certain computation for all vectors in a certain set. The idea of \emph{intermediate sums} is to do this computation in the easiest case and to use  the resulting vector to compute the results for harder cases. This will become clear with an example.

\begin{example}
 Let $\bA \in \F_2^{k \times n}$ and assume that we want to compute $\bx \bA$ for all $\bx \in \F_2^k$ of Hamming weight $t$. This would usually cost 
$$nt \binom{k}{t}  $$ binary operations.

Using the concept of intermediate sums helps to speed up this computation: we first compute $\bx \bA$ for all $\bx \in \F_2^k$ of Hamming weight 1, thus just outputting the rows of $\bA$ which is for free. As a next step, we compute $\bx \bA$ for all $\bx \in \F_2^k$ of Hamming weight 2, which is the same as adding two rows of $\bA$ and hence costs $\binom{k}{2} n$ binary operations. As a next step, we compute $\bx \bA$ for all $\bx \in \F_2^k$ of Hamming weight 3. This is the same as adding one row of $\bA$ to one of the already computed vectors from the previous step, thus this costs $\binom{k}{3} n$ binary operations. If we proceed in this way, until we compute $\bx \bA$ for all $\bx \in \F_2^k$ of Hamming weight $t$, this step costs 
$$nL(k,t)$$
binary operations, where  
$$L(k,t) = \sum_{i=2}^t \binom{k}{i}.$$
This is a speed up to the previous cost, since  
$$n \sum_{i=2}^t \binom{k}{i} = n \left( \binom{k}{2} + \cdots + \binom{k}{t} \right) < n t \binom{k}{t}.$$

When generalizing this result to $\F_q$, computing $\bx \bA$ for all $\bx \in \F_q^k$ of Hamming weight 1 does not come for free anymore. Instead we have to compute  $\bA \cdot \lambda$ for all $\lambda \in \F_q^\star$ which costs $kn \lcq^2$ binary operations. 
Further, if we want to compute $\bx \bA$ for all $\bx \in \F_q^k$ of Hamming weight 2, we have to add two multiples of rows of $\bA$. While there are  still $\binom{k}{2} $ many rows, we now  have $(q-1)^2$ multiples. Thus, this step costs  $\binom{k}{2} (q-1)^2 n \lcq$ binary operations.
Proceeding in this way, the cost of computing $\bx \bA$ for all $\bx \in \F_q^k$ of Hamming weight $t$, is given by 
$$L_q(k,t)n \lcq + kn \lcq^2$$
binary operations, where  
$$L_q(k,t) = \sum_{i=2}^t \binom{k}{i}(q-1)^i.$$
Which is clearly less than the previous cost of 
$$\binom{k}{t}(q-1)^t  n t \left( \lcq^2 + \lcq \right)$$ binary operations.

\end{example}

\subsubsection{Prange's Algorithm}

In Prange's algorithm we assume that there exists an information set $I$ that is disjoint to the support of the error vector $\text{supp}(\be)$, {i.e.},
$$I \cap \text{supp}(\be) = \emptyset.$$ Of course, such an assumption comes with a probability whose reciprocal defines how many iterations are needed on average if the algorithm ends. Note that   Prange's algorithm is  not deterministic, {i.e.}, there are  instances which Prange's algorithm can not solve. For an easy example, one can just take an instance where $\wtH(\be)= t > n-k = \mid I^C \mid$. For a more elaborate example, which also allows unique decoding, assume that we have a parity-check matrix, which is such that each information set includes the first position. Then an error vector with non-zero entry in the first position could never be found through Prange's algorithm. \\

To illustrate the algorithm, let us assume that the information set is $I=\{1, \ldots, k\}$, and let us denote by $J= I^C$.  To bring the parity-check matrix $\bH \in \F_q^{(n-k) \times n}$  into systematic form, we multiply by an invertible matrix $\bU \in \F_q^{(n-k) \times (n-k)}$. Since we assume that no errors occur in the information set, we have that $\be = (\bz_k, \be_J)$ with $\wtH(\be_J)=t$. 
We are in the following situation:
\begin{align*}
\be\bH^\top\bU^\top = \begin{pmatrix}
\bz_k  & \be_J 
\end{pmatrix}  \begin{pmatrix}
\bA^\top \\ \Id_{n-k}
\end{pmatrix} =  \bs\bU^\top,
\end{align*}
for $\bA \in \F_q^{(n-k) \times k}$.

It follows that $\be_J  = \bs\bU^\top$ and hence we are only left with checking the weight of $\bs\bU^\top$.

We will now give the algorithm of Prange in its full generality, i.e., we are not restricting to the choice of $I$ and $J$ that we made before for simplicity.

\begin{algorithm}[h!]
\caption{Prange's Algorithm over  $\F_{q}$ in the Hamming metric}\label{prangeFq}
\begin{flushleft}
Input: $\bH \in \F_{q}^{(n-k) \times n}$, $\bs \in\F_{q}^{n-k}$, $t \in \N$.\\ 
Output: $\be \in\F_{q}^n$ with $\be\bH^\top=\bs $ and $\wtH(\be)=t$.
\end{flushleft}
\begin{algorithmic}[1]
\State Choose an information set $I\subset \{1, ...,n\}$  of size $k$ and define $J= I^C$.
\State Compute $\bU\in\F_{q}^{(n-k) \times (n-k)}$, such that $$(\bU\bH)_{I}=  \bA    \    \text{ and} \ (\bU\bH)_J=  
	 \Id_{n-k},  
 $$	 where $\bA\in\F_{q}^{(n-k) \times k}$.
\State Compute $\bs' = \bs\bU^\top$.
		\If{$\wtH(\bs')=t$}
		\State Return $\be $ such that $\be_I = \bz_k$ and $\be_J =  \bs'$.
				\EndIf
		\State  Start over with Step 1 and a new selection of $I$.
\end{algorithmic}
\end{algorithm}

\begin{theorem}\label{prangeFqcost}
Prange's algorithm over $\F_q$ requires on average
$$ \binom{n-k}{t}^{-1} \binom{n}{t}   (n-k)^2(n+1) \left(\lcq+\lcq^2\right)  $$
binary operations.
\end{theorem}

\begin{proof}
One iteration of  Algorithm \ref{prangeFq}   only consists of bringing $\bH$ into systematic form and applying the same row operations on the syndrome; thus, the cost can be assumed equal to that of computing  $\bU \begin{pmatrix} \bH & \bs^\top \end{pmatrix} $, i.e., $$(n-k)^2(n+1)(\lcq+\lcq^2) $$
binary operations. 

The success probability is given by having chosen the correct  weight distribution of $\be$.  
 In this case, we require that no errors happen in the chosen information set, hence  the probability is given by
\begin{equation*}
\binom{n-k}{t} \binom{n}{t}^{-1}.    
\end{equation*}

 Then, the estimated overall cost of Prange's ISD algorithm   over $\F_q$ is given as in the claim. 
\end{proof}

Let us consider an example for Prange's algorithm. 

\begin{example}
Over $\mathbb{F}_5$, we are given 
$$\bH= \begin{pmatrix}
3 & 2 & 1 & 4 & 3 & 0 & 4 & 4 & 3 & 4 \\
2 & 3 & 4 & 0 & 1 & 2 & 3 & 2 & 4 & 2 \\
3 & 0 & 3 & 1 & 4 & 0 & 2 & 2 & 0 & 0 \\
2 & 3 & 0 & 2 & 3 & 1 & 4 & 4 & 3 & 0 \\
0 & 2 & 3 & 0 & 2 & 0 & 3 & 4 & 2 & 4 \\
2 & 3 & 4 & 0 & 2 & 2 & 0 & 0 & 1 & 2 \\
\end{pmatrix},$$ $\bs=(2,4,0,2,0,4)$ and $t=2$.
We start by choosing an information set, since $I_1=\{1,2,3,4\}$ is not an information set, our first choice might be $I_2=\{1,2,3,5\}.$
As a next step we compute $\bU$ to get $\bH$ into systematic form. 
For this information set we have that 
$$\bU_2\bH = \begin{pmatrix}
3 & 4 & 1 & 1 & 0 & 0 & 0 & 0 & 0 & 0 \\
0 & 3 & 3 & 0 & 4 & 1 & 0 & 0 & 0 & 0 \\
4 & 4 & 2 & 0 & 4 & 0 & 1 & 0 & 0 & 0 \\
1 & 4 & 4 & 0 & 3 & 0 & 0 & 1 & 0 & 0 \\
2 & 0 & 2 & 0 & 2 & 0 & 0 & 0 & 1 & 0 \\
0 & 1 & 3 & 0 & 1 & 0 & 0 & 0 & 0 & 1 \\
\end{pmatrix}.$$
We apply the same on the syndrome, getting 
$$\bs_2' = \bs\bU_2^\top =(3,2,4,3,4,1),$$ which is now unfortunately not of Hamming weight 2. Thus, we have to choose another information set. 
This procedure repeats until the chosen information set succeeds. For example for $I= \{7,8,9,10\}.$ In fact, if we now compute the systematic form we get
$$\bU\bH = \begin{pmatrix}
1 & 0 & 0 & 0 & 0 & 0 & 4 & 0 & 0 & 4 \\
0 & 1 & 0 & 0 & 0 & 0 & 1 & 1 & 0 & 3 \\
0 & 0 & 1 & 0 & 0 & 0 & 4 & 2 & 1 & 1 \\
0 & 0 & 0 & 1 & 0 & 0 & 0 & 4 & 4 & 0 \\
0 & 0 & 0 & 0 & 1 & 0 & 2 & 3 & 2 & 0 \\
0 & 0 & 0 & 0 & 0 & 1 & 2 & 4 & 4 & 3 \\
\end{pmatrix}$$ 
and $\bs'=\bs\bU^\top= (2,0,0,4,0,0)$,  which has Hamming weight 2. Thus, $$\be=(\bs', \mathbf{0})=(2,0,0,4,0,0,0,0,0,0).$$
\end{example}

\subsubsection{Stern's Algorithm}

Stern's algorithm \cite{stern} is one of the most used ISD algorithms, as it is considered one of the fastest algorithms on a classical computer.   
In this algorithm we use the idea of Lee-Brickell and allow errors inside the information set and in addition we   partition the information set into two sets and ask for $v$ errors in both of them. Further, we also use the idea of Leon \cite{leon} to have a \emph{zero-window}  of size $\ell$ outside the information set, where no errors happen.  
\medskip

 Stern's   algorithm is given  in Algorithm \ref{sternFq}. But first we explain the algorithm and illustrate it.

The steps are the usual: we first choose an information set and then bring the parity-check matrix into systematic form according to this information set. We partition the information set into two sets and define the sets $S$ and $T$, where $S$ takes care of all vectors living in one partition and $T$ takes care of all vectors   living in the other partition. We can now check whether two of such fixed vectors give us the wanted error vector.
 
 \medskip
 
To illustrate the algorithm, we assume that the information set is $I=\{1, \ldots, k\}$ and that the zero-window is  $Z=   \{k+1, \ldots, k + \ell\}$. Further, let us define $J= (I \cup Z)^C = \{k+\ell+1, \ldots, n\}$.   
We again denote by $\bU$ the matrix that brings the parity-check matrix into systematic form and write the error vector partitioned into the information set part $I$, the zero-window part $Z$ and the remaining part $J$, as $\be = (\be_I, \bz_\ell,  \be_J)$, with $\wtH(\be_I)=2v$ and $\wtH(\be_J)=t-2v$. 
Thus, we get the following:
\begin{align*}
\be\bH^\top\bU^\top = \begin{pmatrix}
\be_I  & \bz_\ell  & \be_J 
\end{pmatrix}  \begin{pmatrix}
\bA^\top  & \bB^\top \\ \Id_{\ell } & \bz_{\ell \times (n-k-\ell)} \\ \bz_{(n-k-\ell) \times \ell} & \Id_{n-k-\ell}
\end{pmatrix} =  \begin{pmatrix} \bs_1 & \bs_2 \end{pmatrix} = \bs \bU^\top,
\end{align*}
where $\bA \in \F_q^{\ell \times k}$ and $\bB \in \F_q^{(n-k-\ell) \times k}$.

From this we get the following two conditions

\begin{align}
\be_I\bA^\top &= \bs_1, \label{sternF2cond1} \\
\be_I\bB^\top + \be_J &= \bs_2. \label{sternF2cond2}
\end{align}

We partition the information set $I$ into the sets $X$ and $Y$, for the sake of simplicity, assume that $k$ is even and $m=k/2$. Assume that  $X= \{1, \ldots, m\}$ and $Y= \{m+1, \ldots, k\}$. Hence, we can write $\be_I = (\be_X, \be_Y)$, and Condition \eqref{sternF2cond1} becomes 
\begin{equation}\label{forST} 
\sigma_X(\be_X) \bA ^\top = \bs_1 - \sigma_Y(\be_Y) \bA^\top.
\end{equation}

Observe that the $\sigma_X$ is needed, as $\be_X$ has length $m$ but we want to multiply it to $\bA^\top \in \F_q^{k \times \ell}$. In the algorithm we will not use the embedding $\sigma_X$ but rather $\F_q^k(X)$, thus $\be_X$ will have length $k$, but only support in $X$.
\medskip

In the algorithm, we  define a set $S$ that contains all  vectors of the form  $\sigma_X(\be_X) \bA^\top$, i.e., of the left side of \eqref{forST} and a set $T$ that contains all vectors  of the form $\bs_1 - \sigma_Y(\be_Y) \bA^\top$, i.e., of the right side of \eqref{forST}. Whenever a vector in $S$ and a vector in $T$ coincide, we call such a pair a collision. 
\medskip

 For each collision we define $\be_J$ such that Condition \eqref{sternF2cond2} is satisfied, i.e.,
$$\be_J =  \bs_2 - \be_I\bB^\top$$
and if the weight of $\be_J$ is the remaining $t-2v$, we have found the sought-after error vector.

 We  now give the algorithm of Stern in its full generality, i.e., we are not restricting to the choice of $I,J$ and $Z$, that we made before for illustrating the algorithm.

\begin{algorithm}[h!]
\caption{Stern's Algorithm over  $\F_{q}$ in the Hamming metric}\label{sternFq}
\begin{flushleft}
Input: $\bH \in \F_{q}^{(n-k) \times n}$, $\bs \in\F_{q}^{n-k}$, $t \in \N$, $k= m_1+m_2, \ell < n-k$ and $v < \min\{m_1,m_2, \lfloor \frac{t}{2} \rfloor\}$.\\ 
Output: $\be \in\F_{q}^n$ with $\be\bH^\top=\bs $ and $\wtH(\be)=t$.
\end{flushleft}
\begin{algorithmic}[1]
\State Choose an information set $I\subset \{1, ...,n\}$  of size $k$ and choose a zero-window $Z \subset   I^C$ of size $\ell$, and define $J=   (I \cup Z)^C$.
\State Partition $I$ into $X$ of size $m_1$ and $Y$ of size $m_2=k-m_1$.
\State Compute $\bU\in\F_{q}^{(n-k) \times (n-k)}$, such that \begin{align*} (\bU\bH) _{I}=  \begin{pmatrix} \bA \\ \bB \end{pmatrix}, \   (\bU\bH) _{Z}=  \begin{pmatrix} \Id_\ell \\ \bz_{ (n-k-\ell)\times \ell } \end{pmatrix} \   \text{ and} \ (\bU\bH)_J=  
	 \begin{pmatrix} \bz_{\ell \times (n-k-\ell)} \\ \Id_{n-k-\ell} \end{pmatrix},  
\end{align*}	 where $\bA\in\F_{q}^{\ell \times k}$ and $\bB \in \F_q^{(n-k-\ell) \times k}$.
\State Compute $\bs\bU^\top = \begin{pmatrix}
\bs_1 & \bs_2 
\end{pmatrix}$, where $\bs_1 \in \F_q^\ell$ and $\bs_2 \in \F_q^{n-k-\ell}$.
\State Compute the set $S$
$$S= \{( \be_X\bA^\top, \be_X) \mid \be_X \in \F_q^k(X), \wtH(\be_X)=v\}.$$
\State Compute the set $T$
$$T= \{( \bs_1-  \be_Y\bA^\top, \be_Y) \mid \be_Y \in \F_q^k(Y), \wtH(\be_Y)=v\}.$$
\For{$(\ba, \be_X) \in S$} \For{$(\ba,\be_Y) \in T$}
		\If{$\wtH(\bs_2 -(\be_X+\be_Y)\bB^\top)=t-2v$}
		\State Return $\be $ such that $\be_I = \be_X+ \be_Y$, $\be_Z = \bz_\ell$ and $\be_J =  \bs_2 - (\be_X+\be_Y)\bB^\top$.
				\EndIf
				\EndFor
				\EndFor
		\State  Start over with Step 1 and a new selection of $I$.
\end{algorithmic}
\end{algorithm}

\begin{theorem}\label{sternFqcost}
Stern's algorithm over $\F_q$ requires on average
\begin{align*} 
& \binom{m_1}{v}^{-1} \binom{m_2}{v}^{-1}   \binom{n-k-\ell}{t-2v}^{-1} \binom{n}{t} \\ & \cdot  \left( (n-k)^2(n+1) \left( \lcq + \lcq^2 \right)  +(m_1+m_2)\ell \lcq^2 \right. \\ 
 & +    \ell \left( L_q(m_1,v) +L_q(m_2,v) + \binom{m_2}{v}(q-1)^v \right)\lcq    \\
 & + \frac{\binom{m_1}{v}\binom{m_2}{v}(q-1)^{2v}}{q^\ell} \min\left\{n-k-\ell, \frac{q}{q-1}(t-2v+1)\right\}  \\ & \cdot \left. 2v\left( \lcq^2 +  \lcq \right) \right)
\end{align*}
binary operations.
\end{theorem}

\begin{proof}
As in Prange's algorithm, as a first step we bring $\bH$ into systematic form and  apply the same row operations on the syndrome;  a broad estimate for the cost is given by $$(n-k)^2(n+1)\left(\lcq + \lcq^2\right) $$ binary operations. 

To compute  the set $S$, we can use the technique of intermediate sums.  We want to compute $\be_X \bA^\top$ for all $\be_X \in \F_q^k(X)$ of Hamming weight $v$. Using intermediate sums, this costs
$$L_q(m_1,v)\ell \lcq +m_1\ell \lcq^2$$
binary operations.

Similarly, we can build set $T$: we want to compute $\bs_1 -  \be_Y\bA^\top$, for all $\be_Y \in \F_q^k(Y)$ of Hamming weight $v$. Using intermediate sums, this costs
$$L_q(m_2,v)\ell \lcq + m_2\ell \lcq^2 + \binom{m_2}{v}(q-1)^v \ell \lcq$$
binary operations. Note that the $L_q(m_2,v)\ell \lcq +m_2 \ell \lcq^2$ part comes from computing $ \be_Y\bA^\top$, whereas the $   \binom{m_2}{v}(q-1)^v \ell \lcq $ part comes from subtracting from each of the vectors $\be_Y\bA^\top$ the vector $\bs_1$.
\medskip

In the remaining steps   we go through all $(\ba,\be_X) \in S$ and all $(\ba,\be_Y) \in T$, thus usually the cost of these steps should be multiplied by the size of $S \times T$. However, since the algorithm first checks for a collision, we can use instead of  $\mid S \mid \mid T \mid$ the number of collisions we expect on average. 

More precisely: 
since $S$ consists of all $\be_X \in \F_q^k(X)$ of Hamming weight $v$, $S$ is of size $\binom{m_1}{v}(q-1)^v$ and similarly $T$ is of size $\binom{m_2}{v}(q-1)^v$.

The resulting vectors $\be_X \bA^\top$, respectively, $\bs_1 -  \be_Y\bA^\top$ live in $\F_q^\ell$, and we assume that they are uniformly distributed. Hence, we have to check on average
$$\frac{\binom{m_1}{v} \binom{m_2}{v}(q-1)^{2v}}{q^{\ell}}$$ many collisions. 

For each collision we have to compute $$\bs_2 - (\be_X+\be_Y)\bB^\top.$$   
Since the algorithm only proceeds if the weight of $$\bs_2 - (\be_X+\be_Y)\bB^\top$$ is $t-2v$, we can use the concept of early abort.

Computing one entry of the vector $\bs_2 - (\be_X+\be_Y)\bB^\top$ costs $$2v \left( \lcq^2 +  \lcq \right)$$ binary operations. 
Thus,  we get that this step costs on average
$$ \frac{q}{q-1}(t-2v+1) 2v\left( \lcq^2 + \lcq\right)$$ binary operations. 
\medskip

Finally, the success probability is given by having chosen the correct  weight distribution of $\be$; this is exactly the same as over $\F_2$ and given by 
\begin{equation*}
\binom{m_1}{v}\binom{m_2}{v}\binom{n-k-\ell}{t-2v} \binom{n}{t}^{-1}.    
\end{equation*}
Thus, we can conclude.
\end{proof}

Note that we usually set in Stern's algorithm  the parameter $m_1 = \lfloor \frac{k}{2} \rfloor$. Hence assuming that $k$ is even we get a nicer formula for the cost, being

\begin{align*} 
& \binom{k/2}{v}^{-2}  \binom{n-k-\ell}{t-2v}^{-1} \binom{n}{t} \left(   \left( \lcq + \lcq^2 \right)  \right.  \\
\cdot & \left( (n-k)^2(n+1) + \binom{k/2}{v}^2(q-1)^{2v-\ell}  \min\left\{n-k-\ell, \frac{q}{q-1}(t-2v+1)\right\} 2v  \right)\\
  + & \left. k\ell \lcq^2  +    \ell \left( 2L_q(k/2,v) + \binom{k/2}{v}(q-1)^v \right)\lcq  \right)
\end{align*}
binary operations.

 \subsubsection{BJMM Algorithm}
 
 In what follows we cover the BJMM algorithm proposed in \cite{bjmm}, this is considered to be the fastest algorithm over the binary, for this reason we will stick to the binary case also for this paragraph.
 \medskip
 
In the previous ISD algorithms one always represented  the entries of the error vector as $0=0+0$ and $1=1+0=0+1$, that is one was looking for a set partition of the support.  The novel idea of the algorithm is to use also the other representations, i.e., $0=0+0=1+1$. Thus, the search space for the smaller error vector parts become larger but the probability to find the correct error becomes larger as well.

The idea of the BJMM algorithm is to write a vector $\be$ of some length $n$ and weight $v$ as  $\be=\be_1+\be_2$, where $\be_1$ and $\be_2$ are both of length $n$ and of weight $v/2 +\varepsilon,$ thus we are asking for an overlap in $\varepsilon$ positions, which will cancel out. 
\medskip

The first part of all algorithms, which belong to the second direction of improvements, is to perform a partial Gaussian elimination (PGE)  step, that is for some positive integer $\ell \leq n-k$ one wants to find an invertible matrix $\bU \in \mathbb{F}_2^{(n-k) \times (n-k)}$, such that (after some permutation of the columns)
$$\bU\bH= \begin{pmatrix}
\text{Id}_{n-k-\ell} & \bA \\ \mathbf{0} & \bB
\end{pmatrix},$$
where $\bA \in \mathbb{F}_2^{(n-k-\ell) \times (k+\ell)}$ and $\bB \in \mathbb{F}_2^{\ell \times (k+\ell)}$. Hence we are looking for $\be=(\be_1, \be_2)$, with $\be_1 \in \mathbb{F}_2^{n-k-\ell}$ of weight $t-v$ and $\be_2 \in \mathbb{F}_2^{k+\ell},$  of weight $v$. For the parity-check equations, we also split the new syndrome $\bs\bU^\top =(\bs_1, \bs_2)$ with $\bs_1 \in \mathbb{F}_2^{n-k-\ell}$ and $\bs_2 \in \mathbb{F}_2^\ell,$ that is we want to solve
$$\bU\bH\be^\top = \begin{pmatrix}
\text{Id}_{n-k-\ell} & \bA \\ \mathbf{0} & \bB
\end{pmatrix} \begin{pmatrix}
\be_1^\top \\ \be_2^\top
\end{pmatrix} = \begin{pmatrix}
\bs_1^\top \\ \bs_2^\top
\end{pmatrix}.$$ The parity-check equations can thus be written as
\begin{align*}
    \be_1^\top + \bA\be_2^\top &= \bs_1^\top, \\ \bB\be_2^\top &= \bs_2^\top.
\end{align*}
The idea of the algorithms using PGE is to solve now the second equation, i.e., to search for $\be_2$ of length $k+\ell$ and weight $v$ such that $\be_2\bB^\top=\bs_2$ and then to define $\be_1= \bs_1-\be_2\bA^\top$ and to check if this has then the remaining weight $t-v.$
\medskip

Note that this is now a smaller instance of a syndrome decoding problem, for which we want to find a list of solutions.
The success probability of such a splitting of $\be$ is then given be 
$$\binom{k+\ell}{v}\binom{n-k-\ell}{t-v}\binom{n}{t}^{-1}.$$

An important part of such algorithms is how to merge two lists of parts of the error vector together.  For this we consider two lists $\mathcal{L}_1, \mathcal{L}_2$, a positive integer $u<k$, which denotes the number of positions on which one merges, a target vector $\mathbf{t} \in \mathbb{F}_2^u$ and a target weight $w.$ For a vector $\bx$, let us denote by $\bx_{\mid u}$ the vector consisting of the first $u$ entries of $\bx.$

\begin{algorithm}[h!]
\caption{Merge}\label{algo:merge}
\begin{flushleft}
Input: The input lists $\mathcal{L}_1, \mathcal{L}_2$, the positive integers $0<u<k$ and  $0 \leq v \leq n$, the matrix  $\bB \in \mathbb{F}_2^{k \times (k+\ell)}$ and the target $\mathbf{t} \in \mathbb{F}_2^{u}$. \\ 
Output: $\mathcal{L} = \mathcal{L}_1 \bowtie \mathcal{L}_2$.
\end{flushleft}
\begin{algorithmic}[1]
\State Lexicographically sort $\mathcal L_1$  and $\mathcal{L}_2$ according to  $(\bB \bx_i^\top)_{\mid u}$, respectively $(\bB\by_j)_{\mid u} +\mathbf{t}$ for $\bx_i \in \mathcal L_1$ and $\by_j \in \mathcal{L}_2$. 
\For{$(\bx_i,\by_j) \in \mathcal{L}_1\times \mathcal{L}_2$ with $(\bB \bx_i^\top)_{\mid u} =  (\bB \by_j^\top)_{\mid u} + \mathbf{t}$}
    \If{$\wtH(\bx_i+\by_j) = w$}
        \State $\mathcal L = \mathcal L \cup \{\bx_i + \by_j\}$.
    \EndIf
\EndFor
\State Return $\mathcal{L}.$
\end{algorithmic}
\end{algorithm}

\begin{lemma} \label{lemma:merge-cost}
The average cost of the merge algorithm (Algorithm \ref{algo:merge}) is given by
\begin{align*}
    (L_1 + L_2)u(k+\ell)  +L_1\log(L_1)\\ + L_2\log_2(L_2) +(k+\ell) \left(  L_1 \cdot L_2 2^{-u}\right),
\end{align*} where $L_i = |\mathcal L_i |$ for $i = 1,2$. 
\end{lemma}
\begin{exercise}
Prove Lemma \ref{lemma:merge-cost}.
\end{exercise}
The algorithm will use this merging process three times.

For the internal parameter $v$ (which can be optimized), we also choose the positive integers $\varepsilon_1, \varepsilon_2$ (also up to optimization), and define
\begin{align*}
    v_1 &= v/2 + \varepsilon_1, \\
    v_2 &= v_1/2 + \varepsilon_2.
\end{align*}
We start with creating the two base lists $\mathcal{B}_{1}$ and $\mathcal{B}_{2}$,  which depend on a partition $P_1,P_2$ of $\{1, \ldots, k+\ell\}$, of same size, i.e., $\frac{k+\ell}{2}:$
$$\mathcal{B}_i = \{ \bx \in \mathbb{F}_2^{k+\ell}(P_i) \mid \wtH(\bx) = v_2/2 \}.$$
These lists have size $$B= \binom{(k+\ell)/2}{v_2/2}.$$ 

We now choose $\mathbf{t}_1^{(1)} \in \mathbb{F}_2^{u_1}$, which determines 
$\mathbf{t}_2^{(1)}= (\mathbf{s}_2)_{\mid u_1} + \mathbf{t}_1^{(1)}.$ We also choose $\mathbf{t}_1^{(2)}, \mathbf{t}_3^{(2)} \in \mathbb{F}_2^{u_2}$, which define
\begin{align*}
    \mathbf{t}_2^{(2)} &= (\mathbf{t}_1^{(1)})_{\mid u_2} + \mathbf{t}_1^{(2)},  \\
    \mathbf{t}_4^{(2)} &= (\mathbf{t}_2^{(1)})_{\mid u_2} + \mathbf{t}_3^{(2)}.
\end{align*}

Then, for a positive integer $u_2$ and the four target vectors $\mathbf{t}_i^{(2)}$, for $i \in \{1, \ldots, 4\}$ we perform the first four merges using Algorithm \ref{algo:merge} to get $\mathcal{L}_i^{(2)}= \mathcal{B}_1 \bowtie \mathcal{B}_2$ on $u_2$ positions, weight $v_2$ and target vector $\mathbf{t}_i^{(2)}$ for $i \in \{1, \ldots, 4\}.$
The lists $\mathcal{L}_i^{(2)}$ are expected to be of size $L_2=\binom{k+\ell}{v_2}2^{-u_2}$.

With the four new lists we then perform another two merges yielding
$$\mathcal{L}_i^{(1)} = \mathcal{L}_{2i-1}^{(2)} \bowtie \mathcal{L}_{2i}^{(2)}$$ on $u_1$ positions, with weight $v_1$ and target vectors $\mathbf{t}_i^{(1)}$ for $i \in \{1,2\}.$  These lists are expected to be of size $L_1=\binom{k+\ell}{v_1}2^{-u_1}.$  

As a last step we then merge the two new lists to get the final list
$$\mathcal{L} = \mathcal{L}_1^{(1)} \bowtie \mathcal{L}_2^{(1)}$$ on $\ell$ positions, with weight $v$ and target vector $\bs_2.$ 
The final list is expected to be of size $L=\binom{k+\ell}{v}2^{-\ell}. $

One important aspect of such algorithms is the following 
\begin{center} \emph{We have to make sure that at least one representation of the solution lives in each list.}\end{center}
This can either be done by employing the probability of this happening in the success probability, thus increasing the number of iterations or by choosing $u$, the number of positions on which one merges in such a way that we can expect that at least one representation lives in the lists.

In \cite{bjmm} the authors chose the second option: observe that the number of tuples $(\be_1^{(1)}, \be_2^{(1)}) \in \mathcal{L}_1^{(1)} \times \mathcal{L}_2^{(1)}$ that represent a single solution $\be_2 \in \mathcal{L}$ is given by 
$$U_1 = \binom{v}{v/2} \binom{k+\ell-v}{\varepsilon_1}.$$ Hence choosing $u_1 = \log_2(U_1)$ ensures that $L \geq 1.$ Similarly, since we also represent $\be_i^{(1)}$ as sum of two overlapping vectors $(\be_{2i-1}^{(2)}, \be_{2i}^{(2)})$, we have that for each $\be_i^{(1)}$  we have approximately 
$$ U_2= \binom{v_1}{v_1/2} \binom{k+\ell-v_1}{\varepsilon_2}$$ many representations. Thus, we can choose $u_2 = \log_2(U_2).$

\begin{proposition}
Algorithm \ref{algo:bjmm} has an average cost of 
 \begin{align*}
 & \binom{n}{t} \binom{n-k-\ell}{t-v}^{-1}\binom{k+\ell}{v}^{-1}  \cdot \left[ (n-k-\ell)^2(n+1)\right. \\ 
     &   +   4  (2Bu_2(k+\ell) +2B\log(B)  +(k+\ell) B^2 2^{-u_2}) \\
     & + 2(2L_2u_1(k+\ell)+2L_2\log(L_2)  +(k+\ell) L_2^2 2^{-u_1}) \\
     & + (2L_1 \ell (k+\ell) +2 L_1 \log(L_1)  +(k+\ell) L_1^2 2^{-\ell})  \\
     & \left. + \binom{k+\ell}{v}2^{-\ell} 2(t-v+1)v \right] 
 \end{align*}
 binary operations.
\end{proposition}

\begin{algorithm}[h!]
\caption{BJMM}\label{algo:bjmm}
\begin{flushleft}
Input:  $0\leq \ell\leq n-k$, $0 \leq u_2 \leq u_1 \leq \ell$, $  \varepsilon_1, \varepsilon_2$,    $ t,v<t, \bH \in \mathbb{F}_2^{(n-k) \times n}$ and $\bs \in \mathbb{F}_2^{n-k}$. \\ 
Output: $\be \in \mathbb{F}_2^{n}$ with $\text{wt}_H(\be)= t$ and $\bH\be^\top = \bs^\top.$ 
\end{flushleft}
\begin{algorithmic}[1]
\State Choose an $n \times n$ permutation matrix $\bP$.
\State Find $\bU  \in \mathbb{F}_2^{(n-k) \times (n-k)}$, such that 
$$\bU\bH\bP = \begin{pmatrix}
 \Id_{n-k-\ell} & \bA \\ \mathbf{0} & \bB
\end{pmatrix}, $$ where $\bA \in \mathbb{F}_2^{(n-k-\ell) \times (k+\ell)}$ and $\bB \in \mathbb{F}_2^{\ell \times (k+\ell)}.$
\State Compute $\bU\bs^\top = \begin{pmatrix}
\bs_1^\top \\ \bs_2^\top
\end{pmatrix},$ where $\bs_1 \in \mathbb{F}_2^{n-k-\ell}, \bs_2 \in \mathbb{F}_2^{\ell}.$
\State Choose partitions $P_1, P_2$ of $\{1, \ldots, k+\ell\}$ of size $(k+\ell)/2.$
\State Set $$\mathcal{B}_j = \left\{ \bx \in \mathbb{F}_2^{k+ \ell}(P_j) \mid \wtH(\bx)= v_2/2 \right\}$$ for  $j \in \{1,2\}.$
\State Choose $\mathbf{t}_1^{(1)} \in \mathbb{F}_2^{u_1}$, set $\mathbf{t}_2^{(1)}=(\mathbf{s}_2)_{\mid u_1} + \mathbf{t}_1^{(1)}$
\State Choose $\mathbf{t}_1^{(2)}, \mathbf{t}_3^{(2)} \in \mathbb{F}_2^{u_2}$, set $\mathbf{t}_2^{(2)}=(\mathbf{t}_1^{(1)})_{\mid u_2} + \mathbf{t}_1^{(2)}$ and $\mathbf{t}_4^{(2)}=(\mathbf{t}_2^{(1)})_{\mid u_2} + \mathbf{t}_3^{(2)}$
 \For{ $ i \in \{1, \ldots, 4\}$}
        \State Compute $\mathcal{L}_i^{(2)} =  \mathcal{B}_1 \bowtie \mathcal{B}_1$ using Algorithm \ref{algo:merge} on $u_2$ positions to get weight $v_2$ and target vectors $\mathbf{t}_i^{(2)}$.
    \EndFor
    \For{ $i \in \{1,2\}$}
   \State Compute $\mathcal{L}_i^{(1)} =  \mathcal{L}_{2i-1}^{(2)} \bowtie \mathcal{L}_{2i}^{(2)}$ using Algorithm \ref{algo:merge} on $u_1$ positions to get weight $v_1$ and target vectors $\mathbf{t}_i^{(1)}$.
    \EndFor
    \State Compute $\mathcal{L} =  \mathcal{L}_{1}^{(1)} \bowtie \mathcal{L}_{2}^{(1)}$ using Algorithm \ref{algo:merge} on $\ell$ positions to get weight $v$ and target vector $\bs_2$.
    \For{$\be_2 \in \mathcal{L}$}
    \If{$\wtH(\bs_1-\be_2\bA^\top)=t-v$}
    \State Set $\be=(\be_1, \be_2).$
    \EndIf
    \EndFor
        \State Return $\bP\be.$
        \State Else start over at step 1.
 \end{algorithmic}
\end{algorithm}
\newpage
 ~
 \newpage

\subsubsection{Generalized Birthday Decoding Algorithms}

In the syndrome decoding problem (SDP)  we are given a parity-check matrix $\bH \in \mathbb{F}_q^{(n-k) \times n}$, a syndrome $\bs \in \mathbb{F}_q^{n-k}$ and a weight $t \in \mathbb{N}$ and want to find an error vector $\be \in \mathbb{F}_q^n$, such that $\bs=\be\bH^\top$ and $\text{wt}_H(\be)=t.$
\medskip

The first step of a generalized birthday algorithm (GBA) decoder is the partial Gaussian elimination step, i.e., for some positive integer $\ell \leq n-k$ we bring the parity-check matrix into the form 
$$\bH'= \begin{pmatrix} \text{Id}_{n-k-\ell} & \bA \\ 0 & \bB \end{pmatrix},$$ up to permutation of columns. 
We recall from the BJMM algorithm, that this leaves us with solving 
the smaller SDP instance: 
find $\be_2 \in \mathbb{F}_q^{k+\ell}$ of Hamming weight $v \leq t$, such that 
$$\be_2\bB^\top =\bs_2,$$ for $\bs_2 \in \mathbb{F}_q^\ell$ and $\bB \in \mathbb{F}_q^{\ell \times (k+\ell)}.$
\medskip
  
This second step is usually performed using Wagner's algorithm on $a$ levels.

By abuse of notation, we write for the rest $\be\bB^\top = \bs$, instead of $\be_2\bB^\top=\bs_2$. In a Lee-Brickell approach, one would now go through all possible $\be \in \mathbb{F}_q^{k+\ell}$ of weight $v$ and check if they satisfy the parity-check equations. The idea of GBA is to split the vector $\be$ further. Let us start with GBA on one level, that is 
$$\be=(\be_1, \be_2) $$ with $\be_i \in \mathbb{F}_q^{(k+\ell)/2}$ of weight $v/2$, for $i \in \{1,2\}.$
Hence we define $\bB= \begin{pmatrix} \bB_1 & \bB_2 \end{pmatrix},$ with $\bB_i \in \mathbb{F}_q^{\ell \times (k+\ell)/2}$, for $i \in \{1,2\}$ and split the syndrome $\bs= \bs_1+\bs_2$. We hence want that $$\be_1\bB_1^\top + \be_2\bB_2^\top = \bs = \bs_1+\bs_2.$$

For this we define two lists 
\begin{align*}
    \mathcal{L}_1 &= \{(\be_1, \be_1\bB_1^\top-\bs_1) \mid \be_1 \in \mathbb{F}_q^{(k+\ell)/2}, \text{wt}_H(\be_1) = v/2\}, \\ \mathcal{L}_2 &= \{(\be_2, \be_2\bB_2^\top-\bs_2) \mid \be_2 \in \mathbb{F}_q^{(k+\ell)/2}, \text{wt}_H(\be_2) = v/2\}.
\end{align*}
We are then looking for an element  $$((\be_1, \bx_1), (\be_2,\bx_2)) \in \mathcal{L}_1 \times \mathcal{L}_2,$$ such that $\bx_1 + \bx_2=0$, which will then imply that  $$\be_1\bB_1^\top + \be_2\bB_2^\top = \bs = \bs_1+\bs_2.$$
\medskip

This idea can be generalized to $a$ levels, thus splitting 
$$\be= (\be_1^{(1)}, \ldots, \be_{2^a}^{(1)}),$$ where $\be_i^{(1)} \in \mathbb{F}_q^{(k+\ell)/2^a}$ of weight $v/(2^a)$ and writing $$\bB= \begin{pmatrix} \bB_1 & \cdots & \bB_{2^a} \end{pmatrix},$$ where $\bB_i \in \mathbb{F}_q^{\ell \times (k+\ell)/2^a}$ and splitting $\bs= \bs_1 + \cdots +\bs_{2^a}$.  For this we will need the merging positions $0 \leq u_1 \leq \cdots \leq u_a = \ell$. One first constructs the base lists
\begin{align*}
    \mathcal{L}_j^{(1)} = \{(\be_j^{(1)},  \be_j^{(1)}\bB_j^\top- \bs_j) \mid \be_j^{(1)} \in \mathbb{F}_q^{(k+\ell)/2^a}, \text{wt}_H(\be_j^{(1)})= v/2^a\},
\end{align*} for $j \in \{1, \ldots, 2^a\}$ and then performs $a$ merges: in the $i$-th merge we are given a parameter $0 \leq u_i\leq v$ and we want to merge $$\mathcal{L}_j^{(i+1)}= \mathcal{L}_{2j-1}^{(i)} \bowtie_{u_i}  \mathcal{L}_{2j}^{(i)}.$$  

For this let us define the merge $\mathcal{L}= \mathcal{L}_1 \bowtie_u \mathcal{L}_2$ first formally. Given $\mathcal{L}_i = \{(\be_i, \bx_i)\},$ for $i \in \{1,2\}$ and $u$
$$\mathcal{L}_1 \bowtie_u \mathcal{L}_2 = \{((\be_1,\be_2), \bx_1+\bx_2) \mid \bx_1 +\bx_2 =_u \mathbf{0} \}, $$ where $\ba=_u \bb$, denotes that $\ba$ and $\bb$ are equal on the first $u$ positions.  The merging process follows the following algorithm
\begin{enumerate}
    \item Lexicographically order the elements $(\be_i, \bx_i) \in \mathcal{L}_i$ for $i \in \{1, 2\}$ according to the first $u$ positions,
    \item Search for a collision, i.e., $\bx_1+\bx_2 =_u \mathbf{0}$ and if found insert the corresponding $((\be_1,\be_2),\bx_1+\bx_2)$ in $\mathcal{L}.$
\end{enumerate}

The general idea of GBA is  that we will not use the probability that we can split $\be$ into $(\be_1, \ldots,  \be_{2^a})$ each having weight $v/2^a$, but rather we want that the merging process of  will produce a solution with high probability. The average size of $\mathcal{L}$ is given by  $$ L= \mid \mathcal{L}_1 \bowtie_u \mathcal{L}_2 \mid = \frac{\mid \mathcal{L}_1 \mid \mid \mathcal{L}_2 \mid}{q^u},$$ and thus, whenever  $L \geq 1$ we can be assured that this algorithm returns (on average) a solution $\be.$ 

This is only possible for large weights $v$. If we are in this case, there exists a further improvement on the algorithm, where one does not take the whole lists $\mathcal{L}_i^{(1)}$ but only $2^b$ many such elements, and thus the algorithm works as long as $\frac{2^{2b}}{q^u} \geq 1.$

Stern's ISD algorithm is a special case of Wagner's algorithm on one level, where $\ell=0$ and $\bs_1=0.$ However, in Stern's algorithm one employs the probability of splitting the error vector into $(\be_1,\be_2)$, rather than asking for  $$\frac{\mid \mathcal{L}_1 \mid \mid \mathcal{L}_2 \mid}{q^\ell} \geq 1.$$

The idea of GBA or more precisely of Wagner's approach was used in famous ISD papers such as BJMM and MMT, where 3 levels turned out to be an optimal choice.

\subsubsection{Asymptotic Cost}
An important aspect of ISD algorithms (apart from the cost) is their asymptotic cost. 
The idea of the asymptotic cost is that we are interested in the exponent  $e(R, q)$ such that for large $n$ the cost
of the algorithm is given by $q^{(e(R,q)+o(1))n}$.  This is crucial in order to compare different algorithms.

 We consider codes of large length $n$, and consider the dimension and the error correction capacity as functions in $n$, i.e., $k, t : \mathbb{N} \to \mathbb{N}$. For these we define \begin{align*}
    \lim\limits_{n \to \infty} t(n)/n &=T, \\
    \lim\limits_{n \to \infty} k(n)/n &=R.
\end{align*}  
If $c(n,k,t,q)$ denotes the cost of an algorithm, for example Prange's algorithm, then we are now interested in 
$$C(q,R,T) = \lim\limits_{n \to \infty} \frac{1}{n} \log_q(c(n,k,t,q)).$$

For this we often use Stirlings formula, that is
$$ \lim\limits_{n \to \infty} \frac{1}{n} \log_q\binom{(\alpha + o(1))n}{(\beta + o(1))n} = \alpha \log_q(\alpha) - \beta \log_q(\beta) - (\alpha-\beta)\log_q(\alpha-\beta).$$

One of the most important aspects in computing the asymptotic cost, is  that random codes attain the asymptotic Gilbert-Varshamov bound with high probability, thus we are allowed to choose a relative minimum distance $\delta$ such that $R=1-H_q(\delta).$

\begin{example}
The asymptotic cost of Prange's algorithm is easily computed as 
\begin{align*} \lim\limits_{n \to \infty} \frac{1}{n} &  \log_q \left( \binom{n-k}{t}^{-1} \binom{n}{t}\right) = \\ 
&-(1-T)\log_q(1-T) - (1-R) \log_q(1-R)  +(1-R-T)\log_q(1-R-T). 
\end{align*}
\end{example}

\begin{exercise}
Prove that the asymptotic cost of Prange is equal to 
$$H_q(T) - (1-R)H_q(T/(1-R)).$$
\end{exercise}

For the more sophisticated algorithms such as Stern and BJMM, we will also have internal parameters, such as $\ell,v$, which will be chosen optimal, i.e.,  giving the smallest cost.

Note that we assume half-distance decoding, i.e., $T= \delta/2$, thus $C(q,R,\delta/2)= e(R,q)$ and then compute the largest value of $e(R^\star,q)$ by taking $$R^\star= \text{argmax}_{0 < R <1} e(R,q).$$

With the asymptotic cost, we can now compare different ISD algorithms. For this, we will restrict ourselves to the binary case, since we  presented the BJMM algorithm only over the binary.
In the  following table BJMM refers to the algorithm presented in \cite{bjmm}, MMT to \cite{mmt}, BCD to the algorithm from \cite{ballcoll} and Stern and Prange refer to the algorithms of \cite{stern}, respectively \cite{prange}. 
\begin{table}[h]
 \begin{center}
 \begin{tabular}{|c|c|}
 \hline 
 Algorithm  & $e(R^*,2)$ \\\hline
 BJMM & 0.1019 \\
 MMT & 0.115 \\
 BCD & 0.1163 \\
 Stern & 0.1166 \\
 Prange & 0.1208 \\
  \hline
  \end{tabular}\caption{Asymptotic cost of different ISD algorithms over the binary}
\end{center}  
\end{table}

\subsubsection{Rank-metric ISD Algorithms}

Finally, we want to conclude this section on ISD algorithms explaining the idea of rank-metric ISD algorithms. 

For this we first recall that the Hamming support of an error vector $\be \in \mathbb{F}_{q^m}^n$ is defined as 
$$ \text{supp}_H(\be)= \{ i \in \{1, \ldots, n\} \mid \be_i \neq 0\}. $$
The Hamming weight of $\be$ is then given by the size of the Hamming support, i.e.,
$$\text{wt}_H(\be) = \mid \text{supp}_H(\be) \mid \leq n.$$
If we would want to go through all error vectors of a given Hamming weight $t$, there are 
$$\binom{n}{t}(q^m-1)^t$$ many choices.
This concept changes when we move to the rank-metric. The rank support of an error vector $\be \in \mathbb{F}_{q^m}^n$ is usually defined as the $\mathbb{F}_q$-vector space spanned by the entries of $\be:$
$$\text{supp}(\be) = \langle \be_1, \ldots, \be_n \rangle_{\mathbb{F}_q}.$$
The rank weight of $\be$ is then defined as the $\mathbb{F}_q$-dimension of the rank support, i.e.,
$$\text{wt}_R(\be) = \dim_{\mathbb{F}_q}(\text{supp}(\be)).$$
If we want to go through all vectors of a given rank weight $t$, there are
$$ \gb{m}{t}_q = \prod\limits_{i=0}^{t-1} \frac{q^m -q^i}{q^t-q^i}  \sim q^{(m-t)t}$$ many choices. 
Thus, it is quite clear, that to look for an error vector in the rank metric poses a more costly problem than its Hamming metric counterpart.

However, depending whether $m$ or $n$ are smaller, we could also consider the row or column support. 

\begin{example}
Let us consider $\be =(1, \alpha) \in \mathbb{F}_8^2$, where $\mathbb{F}_8  = \mathbb{F}_2[\alpha]$ with $\alpha^3=\alpha+1$ and the basis $\Gamma=\{1, \alpha,\alpha^2\}$. Then $\be =\bc \bR$, where $\bc=(1,\alpha)$ and $\bR=\begin{pmatrix} 1 & 0 \\ 0 & 1 \end{pmatrix}.$ Thus, the column support of $\be$ is given by $$\text{supp}_C(\be)=\langle \Gamma(\bc)^\top\rangle = \langle (1,0,0),(0,1,0)\rangle \subset \mathbb{F}_2^3 $$ of dimension 3. Whereas the row support of $\be$ is given by 
$$\text{supp}_R(\be)=\langle \bR \rangle = \langle (1,0),(0,1) \rangle \subset \mathbb{F}_2^2.$$

Note that the column and row support can also be read of  
$$\Gamma(\be) = \begin{pmatrix}
1 & 0  \\ 0 & 1  \\ 0 & 0 
\end{pmatrix} $$
as
$$\text{supp}_R(\be)=\text{rowsp}(\Gamma(\be)) \subset \mathbb{F}_q^n$$ and 
$$\text{supp}_C(\be)=\text{colsp}(\Gamma(\be)) \subset \mathbb{F}_q^m.$$
\end{example}
Thus, 
\begin{enumerate}
    \item if $m \leq n$, we consider the column support of $\be$. In this case we  have $\gb{m}{t}_q$ vector spaces to go through. 
    \item If $n \leq m$, we row support of $\be$. In this case we have $\gb{n}{t}_q$ many vector spaces.
\end{enumerate}

In the following we give only the ideas of the combinatorial and algebraic algorithms to solve the rank SDP. 
First observe that we can write $\be = \beta \bE$, where $\beta=(\beta_1, \ldots, \beta_t)$ is a basis of the support of the error vector $\be$ and $\bE \in \mathbb{F}_q^{t \times n}.$

The first proposed rank ISD algorithm \cite{chabaudstern} performs a basis enumeration. That is, we want to enumerate all possible choices for $\beta$. Since if we know $\beta$, then solving $\beta\bE\bH^\top = \bs$ has quadratic complexity. This attack has approximately a complexity of $q^{tm}$ operations. 

The second proposed rank ISD algorithm \cite{johannson} enumerates all possible matrices $\bE$ instead, resulting in a cost of approximately $q^{(t-1)(k+1)}$ operations. These approaches are called combinatorial attacks, as they solve the rank SDP through enumerations.

In \cite{grs} the authors give a Prange-like rank metric ISD algorithm. The algorithm is usually called GRS, as abbreviation for the authors Gaborit, Ruatta, Schrek, not to be confused with generalized Reed-Solomon codes.
One first chooses whether to guess the row or column support of $\be$, depending whether $n\leq m$, or $m\leq n.$
Let us first assume  that $m \leq n$ and hence we guess the column support.

Recalling that $\be=\bc \bR,$ if we know a basis of the column support $\{\gamma_1, \ldots, \gamma_t\}$ with $\gamma_i \in \mathbb{F}_q^m$, such that $\Gamma(c_i)=\gamma_i$, we can write for each $i \in \{1, \ldots,n\}$
$$e_i=\sum_{j=1}^t c_j r_{i,j}.$$
And over $\mathbb{F}_q$
$$\Gamma(e_i)= \sum_{j=1}^t \gamma_j r_{i,j}.$$
Thus, we have $nt$ unknowns $r_{i,j}$ and from $\bs=\be\bH^\top$ we have $m(n-k)$ equations.

\begin{example}
    Let us consider $\mathbb{F}_8=\mathbb{F}_2[\alpha]$ with $\alpha^3=\alpha+1$ and basis $\Gamma=\{1,\alpha,\alpha^2\}.$
    We are given the parity-check matrix $$\bH= \begin{pmatrix} 1 & 0 & 1 & \alpha^2 \\ 0 & 1 & \alpha & 1 \end{pmatrix} $$ and the syndrome $\bs=(\alpha^2,\alpha+1)$ ant $t=1.$

    We guess the column support of $\be$ to be $\langle (1,1,0) \subset \mathbb{F}_2^3$, this corresponds to $\bc=(\alpha+1)$. Hence $$e_i=\begin{pmatrix} 1 \\ 1 \\ 0 \end{pmatrix} r_i.$$
    We consider the 2 syndrome equations 
    \begin{align*}
        e_1 + e_3 + \alpha^2 e_4 & = s_1 = \alpha^2 \\ 
        e_2 + \alpha e_3 + e_4 &= s_2 = \alpha+1.    \end{align*}
        In order to write these equations over $\mathbb{F}_2$ we observe that 
        $\alpha_2 e_4 = \alpha^2(\alpha+1) r_4 =(\alpha^2+\alpha+1)r_4$ and $\alpha e_3 = \alpha (\alpha+1)r_3= (\alpha^2+\alpha)r_3$. 
        Hence we get the linear system of equations
        $$\begin{pmatrix} 1 & 0 & 1 & 1 \\ 1 & 0 & 1 & 1 \\ 0 & 0 & 0 & 1 \\ 0 & 1 & 0 & 1 \\ 0 & 1 & 1 & 1 \\ 0 & 0 & 1 & 0 \end{pmatrix} \begin{pmatrix} r_1 \\ r_2 \\ r_3 \\ r_4 \end{pmatrix} = \begin{pmatrix} 0 \\ 0 \\ 1 \\ 1 \\ 1 \\ 0 \end{pmatrix}. $$
        After solving the system, we get the unique solution $r_1=1,r_2=0,r_3=0,r_4=1$ and recompute $\be=\bc\bR= (\alpha+1,0,0,\alpha+1)$, which indeed has rank weight 1.
\end{example}
\begin{exercise}
    Perform the same example but guess the column support to be $(1,0,0).$
\end{exercise}

If we know the row support $\{\br_1, \ldots, \br_t\}$ for $\supp_R(\be) \subset \mathbb{F}_q^n$, i.e., the rows of $\bR$, then we can write for each $i \in \{1, \ldots, n\}$
$$e_i=\sum_{j=1}^t c_j r_{i,j},$$ and using the basis $\Gamma$ of $\mathbb{F}_{q^m}$ over $\mathbb{F}_q$ we can write
$$\Gamma(e_i)= \sum_{j=1}^t \Gamma(c_j) r_{i,j}.$$
Thus, over $\mathbb{F}_q$ we have $mt$ unknowns and $m(n-k)$ equations. 

Let us use a neat trick for the next example: in order to bring the parity-check equations to the base field, we need to know what to do with a multiplication.
Let $\Gamma =\{\gamma_1, \ldots, \gamma_m\}$ be a basis of $\mathbb{F}_{q^m}$ over $\mathbb{F}_q$. The multiplication with $a \in \mathbb{F}_{q^m}$ is given by 
\begin{align*}
    m_a: \mathbb{F}_{q^m} &\to \mathbb{F}_{q^m} \\
    x & \mapsto x a.
\end{align*}

This map can be extended to $\mathbb{F}_q$ as
\begin{align*} \bM_a: \mathbb{F}_{q^m} & \to \mathbb{F}_q^m \\ 
x & \mapsto \bM_a \Gamma(x), \end{align*}
where 
$\bM_a \in \mathbb{F}_q^{m \times m}$ is defined through having the columns $\Gamma(a \gamma_1), \ldots, \Gamma( a \gamma_m).$

\begin{example}

    Let us consider $\mathbb{F}-8=\mathbb{F}_2[\alpha]$ with $\alpha^3=\alpha+1$ and the basis $\Gamma=\{1,\alpha,\alpha^2\}.$ 
    Multiplication with $\alpha^2$ is given by the matrix 
    $$\bM_{\alpha^2}=\begin{pmatrix} 0 & 1 & 0 \\ 0 & 1 & 1 \\ 1 & 0 & 1 \end{pmatrix}.$$
    Then for any $x \in \mathbb{F}_8$, we get that $\Gamma(\alpha^2 x)= \bM_{\alpha^2} \Gamma(x).$
\end{example}
\begin{algorithm}[h!]
\caption{GRS Algorithm}\label{algo:grs}
\begin{flushleft}
Input:  $ \bH \in \mathbb{F}_{q^m}^{(n-k) \times n},$ $\bs \in \mathbb{F}_{q^m}^{n-k}$ and $t\leq r\leq n-k$. \\ 
Output: $\be \in \mathbb{F}_{q^m}^{n}$ with $\text{wt}_R(\be)= t$ and $\bH\be^\top = \bs^\top.$ 
\end{flushleft}
\begin{algorithmic}[1]
\State Choose random subspace $\mathcal{S}= \langle \bs_1, \ldots, \bs_r \rangle \subset \mathbb{F}_q^n$ of dimension $r.$ 
			\State Write the error vector in terms of the basis $\bs_1, \ldots, \bs_t$ as $e_i = \sum_{j=1}^{r} e_{ij}\bs_j$, with unknowns $ e_{ij}\in\mathbb{F}_q$.
			\State Solve the linear system of equations (over $\mathbb{F}_q$) implied by $\be\bH^\top = \bs$ to obtain the $e_{ij}$.
			\If{$\text{wt}_R(\be)\leq t$}
   \State Return $\be$. 
   \EndIf 
    \State Else, go to Step 1.

 \end{algorithmic}
\end{algorithm}

The cost of the GRS algorithm is only given by guessing a subspace $\mathcal{S} \subset \mathbb{F}_q^n$ of dimension $r$, which contains $\text{supp}_R(\be).$

Thus the success probability of one iteration is given by 
$$P = \frac{ | \{\mathcal{S} \subset \mathbb{F}_q^n \mid \text{dim}(\mathcal{S})=r, \text{supp}_R(\be) \subset \mathcal{S}\}|}{\mathcal{S} \subset \mathbb{F}_q^n \mid \text{dim}(\mathcal{S})=r\}|} = \gb{n-t}{r-t}_q\gb{n}{r}_q^{-1}.$$
All the other steps, namely writing $\be$ in terms of the basis of $\mathcal{S}$ and solving the linear system of equations can be done in polynomial time.  

Thus,    the GRS algorithm costs $$\gb{n}{r}_q \gb{n-t}{r-t}^{-1}_q \sim q^{(n-r)t}.$$
In order to get an overdetermined system and thus a candidate solution for $\be$, we only require to have more equations than unknowns. Since there are $rn$ many unknowns $e_{ij}$, and we have $m(n-k)$ equations over $\mathbb{F}_q$, this forces us to choose $r\leq n-k$.

\begin{proposition}
    The GRS algorithm has an asymptotic cost of  $$\gb{n}{t}_q \gb{n-k}{t}^{-1}_q \sim q^{kt}.$$
\end{proposition}
\begin{example}
  
       Let us consider $\mathbb{F}_8=\mathbb{F}_2[\alpha]$ with $\alpha^3=\alpha+1$ and basis $\Gamma=\{1,\alpha,\alpha^2\}.$
    We are given the parity-check matrix $$\bH= \begin{pmatrix} 1 & 0 & 1 & \alpha^2 \\ 0 & 1 & \alpha & 1 \end{pmatrix}, $$ the syndrome $\bs=(\alpha^2,\alpha+1)$ and $t=1.$

    We guess the row support of $\be$ to be $\langle (1,0,0,1) \subset \mathbb{F}_2^4$. Hence $e_1=e_4=c$ and $e_2=e_3=0.$
    We consider the 2 syndrome equations 
    \begin{align*}
        e_1 + e_3 + \alpha^2 e_4 & e_1+\alpha^2e_4 = s_1 = \alpha^2 \\ 
        e_2 + \alpha e_3 + e_4 &e_4 = s_2 = \alpha+1.    \end{align*}
        Using $\bM_{\alpha^2}$, we can write the equations as 
        
        \begin{align*}
\begin{pmatrix}c_0 \\ c_1 \\ c_2 \end{pmatrix} + \begin{pmatrix} c_1\\ c_1+c_2 \\ c_0+c_2\end{pmatrix} = \begin{pmatrix} 0 \\ 0 \\ 1 \end{pmatrix}, \quad \quad \text{ and }
 \quad \quad  \begin{pmatrix}c_0 \\ c_1 \\ c_2 \end{pmatrix}  = \begin{pmatrix} 1 \\ 1 \\ 0 \end{pmatrix}.       
        \end{align*} 
        From here we can already solve the system and get $\bc=(\alpha+1). $ We recompute $\be=\bc\bR= (\alpha+1,0,0,\alpha+1)$, which indeed has rank weight 1.
\end{example}
\begin{exercise}
    Perform the same example but guess the row support to be $(1,1,0,0).$
\end{exercise}

    We say that the GRS algorithm is the rank-metric analog of Prange, as it searches for $\mathcal{S}$ of dimension $n-k$ with $\supp(\be) \subset \mathcal{S}$. While Prange's algorithm in the Hamming metric searches for $I^C$ of size $n-k$ with $\supp_H(\be) \subset I^C$.

    \medskip

    Indeed, while Prange's algorithm in the Hamming metric has the cost 
    $$\binom{n}{t} \binom{n-k}{t}^{-1},$$ the rank-metric analog has the cost
    $$\gb{n}{t}_q \gb{n-k}{t}_q^{-1}.$$

The algebraic approach aims at translating the notion of the rank metric into an algebraic setting. For example via linearized polynomials: in \cite{grs} and \cite{aght} it was observed that for $\be \in \mathbb{F}_{q^m}^n$ there exists a linearized polynomial of $q$-degree $t$ of the form 
$$f(x) = \sum\limits_{i=0}^t f_i x^{q^i}$$
annihilating the error vector, i.e., $f(\be_i) = 0$ for all $i \in \{1, \ldots, n\}.$ This algorithm works well for small choices of $t$, giving an approximate cost  \cite{grs} of 
$$ \mathcal{O}\left( (n-k)^3q^{t \lceil \frac{(k+1)m}{n}\rceil-n}\right).$$

Recently,  a new benchmark for the complexity of the rank SDP has been achieved by the paper \cite{minrank}, which solves the rank SDP using the well studied MinRank problem from multivariate cryptography. This might be one of the major reasons why NIST did not choose to finalize any of the code-based cryptosystem based on the rank metric, although they were achieving much lower public key sizes; this area of code-based cryptography needs further research before we can deem it secure.

\subsubsection{Attacks on other Code-Based Problems}

 we have seen that ISD is the fastest algorithm to solve the Decoding Problem, the Syndrome Decoding Problem or the Given Weight Codeword Problem, whether we use the Hamming or the rank metric. This stays true also for the Lee metric \cite{leeNP} or restricted errors \cite{restisd,crosspaper}, clearly, adapted to the considered metrics.

 When considering code-equivalence problems,  one could expect other algorithms to be faster. 
However, also in this case the fastest known algorithms rely on ISD \cite{LESSattack}.
In fact, we have seen in Section \ref{sec:prelim}, that two equivalent codes $\mathcal{C}$ and $\mathcal{C}'$ have the same weight enumerator 
$$W_i(\mathcal{C})= |\{ \bc \in \mathcal{C} \mid \text{wt}(\bc) = i \}| = W_i(\mathcal{C}').$$
Thus, the main algorithm to solve the code equivalence problem asks to find some low weight codewords in $\mathcal{C}$ and $\mathcal{C}'$ using ISD, ordering them as 
$$S=\{bc_1, \ldots, \bc_N\},$$ respectively 

$$S'=\{\bc_1', \ldots, \bc_N'\}$$
and then searching for an isometry that maps $S$ to $S'$.
Recall, that a code $\mathcal{C} \subseteq \mathbb{F}_q^n$ of dimension $k$ has on average 
$$| B(q,n,r)| q^{k-n}$$
many codewords of weight $r,$ where $B(q,n,r)$ denotes the balls of radius $r$ in the respective metric.

If we search for codewords of very small weight, we thus get smaller sets $S,S'$ and it becomes easier to find an isometry between the two sets. However, searching for a small weight increases the cost of the ISD algorithm to find them. On the other hand, when searching for a moderate weight $r$, the ISD algorithm has a small cost, but due to the large size of $S,S'$ it becomes harder to find an isometry.

\subsection{Algebraic Attacks}\label{sec:attack}

In this section, we present some techniques which are used for algebraic attacks on certain code-based cryptosystems. Most famously, is the square code attack, which is in general a distinguisher attack.
\emph{Distinguishers} a priori want to show that the public code is in fact not behaving randomly but like an algebraically structured code.  Distinguishers can then further imply a strategy on how to recover the structure of the secret code, e.g. the evaluation points of a GRS code, or be used directly in a message recovery. 

\begin{definition}
Let $v=(v_1, \ldots, v_n), w=(w_1, \ldots, w_n) \in \mathbb{F}_q^n$ be two vectors. The \emph{Schur product} $v*w$ of $v$ and $w$ is the coordinatewise product of $v$ and $w$, i.e.,
$$ v*w := ( v_1 w_1, \ldots, v_n w_n ).$$
\end{definition}

With this definition we can also define the Schur product of two linear codes.

\begin{definition}
Let $\mC_1,\mC_2 \subset \mathbb{F}_q^n$ be two linear codes. The Schur product of $\mC_1$ and $\mC_2$ is defined as the $\mathbb{F}_q$-span generated by the Schur product of all combinations of elements, i.e.,
$$ \mC_1 * \mC_2 := \langle \{ \bc_1 * \bc_2 \: | \: \bc_1 \in \mC_1, \: \bc_2 \in \mC_2 \} \rangle \subset \mathbb{F}_q^n.$$
For a linear code $\mC \subset \mathbb{F}_q^n$, we call $\mC * \mC$ the \emph{square code} of $\mC$ and denote it with $\mC^{(2)}$.
\end{definition}

Clearly for any code $\mC \subseteq \mathbb{F}_q^n$ of dimension $k$, we have that
$$ \dim(\mC^{(2)}) \leq \min\left\{ \frac{k(k+1)}{2}, n \right\}.$$

However, for codes which  have a lot of algebraic structure, this square code dimension might be much smaller. 

\begin{proposition}\label{prop:sqGRS}
 Let $k \leq n \leq q$ be positive integers. Then,
$$\dim(\text{GRS}_{n,k}(\alpha,\beta)) = \min\{2k-1, n\}.$$
\end{proposition}

\begin{exercise}
Prove Proposition \ref{prop:sqGRS}.
\end{exercise}

 Whereas for a random linear code of dimension $k$, the expected dimension of its square code is typically quadratic in the dimension $k$:

\begin{theorem}[\text{\cite[Theorem 2.3]{cascudo}}]
For a random linear code $\mC$ over $\mathbb{F}_q$ of dimension $k$ and length $n$, we have with high probability that
$$\dim(\mC^{(2)})= \min\left\{\binom{k+1}{2},n\right\}. $$
\end{theorem}

This clearly provides a distinguisher between random codes and algebraically structured codes. Let us list some of the codes, which suffer from such a distinguisher
 
\begin{enumerate}
    \item GRS codes: Proposition \ref{prop:sqGRS},
    \item low-codimensional subcodes of GRS codes: \cite{wiesche},
    \item Reed-Muller codes: \cite{borodin},
    \item Polar codes: \cite{vlad},
    \item some Goppa codes: \cite{wildattack},
    \item high rate alternant codes: \cite{high},
    \item algebraic geometry codes \cite{attackag, attackag2}.
\end{enumerate}
  Note that square code attacks often need to be performed on a modified version of the public code, for example 
 \begin{enumerate}
     \item the sum of two GRS codes: \cite{distinguisher, distinguisher2},
     \item GRS codes with additional random entries: \cite{rlceattack},
     \item expanded GRS codes: \cite{matthieu}.
 \end{enumerate}
 McEliece proposed to use classical binary Goppa codes  as secret codes in \cite{mceliece}, and no algebraic attack on this system has been developed. Thus, they are considered to be reasonably secure and were chosen as the finalists for the NIST standardization process \cite{NISTMcEliece}. 

Recall that Goppa codes are heavily connected to GRS codes:
let us consider a GRS code over $\mathbb{F}_{q^m}$ and some $1 \leq \lambda \leq m$. The code $\mC$ which contains all codewords of the GRS code living in a fixed $\lambda$-dimensional  $\mathbb{F}_q$-vector subspace of $\mathbb{F}_{q^m}$ is called a \emph{subspace subcode} of a GRS code.
\begin{itemize}
    \item If we choose $\lambda=m$ we get a GRS code, which provides very low key sizes for the McEliece cryptosystem due to their large error correction capacity  and only considering ISD attacks. They are however insecure due to the square code attack.
    \item If we choose $\lambda =1$ we get a Goppa code, which  suffers from very large key sizes due to their small correction capacity, but they are deemed to be secure against algebraic attacks. 
\end{itemize}

The proposal \cite{kk} and also \cite{ssrs} propose to use a different $\lambda$ in the McEliece system, trying to find a balance between the two extreme points and profiting from both advantages: smaller key sizes than Goppa codes would provide and thwarting the vulnerability of GRS codes. But also this suggestion has been attacked for $\lambda \geq m/2$ by the square code attack in \cite{matthieu}:

\begin{center}
    \includegraphics[width=6cm]{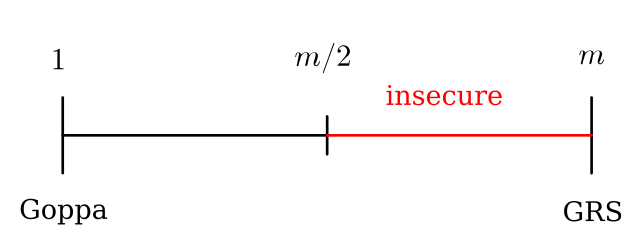}
\end{center}

 Let us summarize this in  Table \ref{sqcode}.

 \begin{table}[h!]
 \begin{center}
 \begin{tabular}{|c|c|}
 \hline 
 Code $\mC$  & $\dim\left(\mC^{(2)}\right)$\\\hline \hline
 Random Code &  $ \min\left\{ \frac{k(k+1)}{2}, n \right\} $ (with high probability) \\ \hline
 RS Code  & $\min\{2k-1,n\}$ \\ \hline 
 Binary Goppa Codes & $\min\left\{ \frac{k(k+1)}{2} - \frac{mr}{2}\left(2r\log_2(r)-r-1\right),n \right\} $\\
 $[n,k=n-mr]$ & (with high probability) \\ \hline
 Expanded GRS Code & $\min\left\{ \mathcal{O}(mk^2), n \right\}$ \\
 $[mn,mk]$ & (with high probability) \\
  \hline
  \end{tabular}\caption{Square code dimension of different codes}\label{sqcode}
\end{center}  
\end{table}

Note that for the rank-metric based cryptosystems a similar distinguisher exists for the rank analogues of the Reed-Solomon codes, namely the Gabidulin codes:  
these attacks all stem from the original attack of Overbeck \cite{overbeck} on the proposal \cite{gpt} to use Gabidulin codes in the GPT framework, but also includes the attack of \cite{AL} on its generalization \cite{loidreau, ggpt}. They main tool here is that instead of taking the square code, one performs the Frobenius map on the code.

Let us consider an extension field $\mathbb{F}_{q^m}$ of the base field $\mathbb{F}_q$. We denote by $[i]$ the $i$th Frobenius power, $q^i$. The Frobenius map can be applied to a matrix or a vector by doing so coordinatewise, i.e., for a matrix $\bM \in \mathbb{F}_{q^m}^{k \times n}$ with entries $(m_{j,\ell})$ we denote by $\bM^{[i]}$ the matrix with entries $(m_{j,\ell}^{[i]})$.

\begin{definition}
Let $\bM \in \mathbb{F}_{q^m}^{k \times n}$ and $\ell \in \mathbb{N}$, then we define the operator $\Lambda_\ell$ as
\begin{align*}
    \Lambda_\ell: \mathbb{F}_{q^m}^{k \times n} &\to \mathbb{F}_{q^m}^{(\ell+1)k \times n}, \\
    \bM &\mapsto \Lambda_\ell(\bM) = \begin{pmatrix}
    \bM \\ \bM^{[1]}, \\ \vdots \\ \bM^{[\ell]} 
    \end{pmatrix}.
\end{align*}
\end{definition}
The Frobenius attack now considers the rowspan of this new matrix.

\begin{proposition}[\cite{overbeck}, Lemma 5.1]\label{prop:frob}
If $\bM$ is the generator matrix of an $[n,k]$ Gabidulin code and $\ell \leq n-k-1$, then the subvector space spanned by the rows of $\Lambda_\ell(\bM)$ is an $[n,k+\ell]$ Gabidulin code.
\end{proposition}
Note that this is similar to Proposition \ref{prop:sqGRS}, where one shows that the square code of a GRS code is again a GRS code. And as the square code dimension of a GRS code is $2k-1$, in this case the dimension of the rowspace of the Frobenius of a Gabidulin code is $k+\ell.$

However, for a random code $\mC$, the Frobenius of this code should have dimension of order $k\ell.$
\begin{theorem}[\cite{loid}]
Let $\bM \in \mathbb{F}_{q^m}^{k \times n}$ be a random matrix of full column rank over $\mathbb{F}_q$. Then $\Lambda_\ell(\bM)$ has rank
$$\min\{(\ell+1)k, n\},$$ with probability at least $1-4q^{-m}.$
\end{theorem}

The Frobenius map can thus distinguish between a   Gabidulin code and a random code. 

\subsection{Other Attacks}\label{sec:other}

We want to note here, that there exist also several other attacks on code-based cryptosystems, such as: side-channel attacks and chosen-ciphertext attacks. 
Since these attacks are less mathematically involved, we will just quickly cover them  and refer interested readers to \cite{cayrel}.
\medskip

\emph{Side-channel attacks} try to get information from the implementation of the cryptosystem, which includes timing information, power consumption and many more. 
Thus, side-channel attacks complement the algebraic and non-structural attacks we have discussed before by considering also the physical security of the cryptosystem. 
\medskip

There have been many side-channel attacks on the McEliece cryptosystem (see for example \cite{side1, side2, side3, side4, side5}) which aim for example at the timing/reaction attacks based on the error weight or recover the error weight using a  simple power analysis on the syndrome
computation. 
\medskip

Note that recently the information gained through side-channel attacks was used in ISD algorithms in \cite{julian}.
\medskip

Another line of attacks is the \emph{chosen-ciphertext attack } (CCA): in a chosen-ciphertext attack we consider the scenario in which the attacker has the ability to choose ciphertexts $c_i$ and to view their corresponding decryptions, i.e., the messages $m_i$. In this scenario we might speak of an oracle that is queried with ciphertexts. The aim of the attacker is to gain the secret key or to get as much information as possible on the attacked system. 

\medskip

In an \emph{adaptive chosen-ciphertext attack} (CCA2) the attacker wants to distinguish a target ciphertext without consulting the oracle on this target.  Thus, the attacker may query the oracle on many ciphertext but the target one. This means that the new ciphertexts are created based on responses (being the corresponding messages) received previously. 
\medskip

In this context we also speak of ciphertext indistinguishability, meaning that an attacker can not distinguish   ciphertexts based on the message they encrypt. We have two main definitions:
\begin{enumerate}
    \item \emph{ Indistinguishability under chosen-plaintext attack }  (IND-CPA),
    \item \emph{ Indistinguishability under adaptive chosen-ciphertext attack }  (IND-CCA2).
\end{enumerate}

These are usually defined over a game, which is played between an attacker and a \emph{challenger}, where we assume that we have a public-key encryption scheme with a secret key $\mathcal{S}$ and a publicly known public key $\mathcal{P}$. 

\medskip

For IND-CPA, the attacker and the challenger are playing the following game.
\begin{enumerate}
    \item The attacker sends two distinct messages $m_1,m_2$ to the challenger.
    \item The challenger selects one of the messages $m_i$ and sends the \emph{challenge} $c_i$, which is the encrypted message $m_i.$
    \item The attacker tries to guess $i.$
\end{enumerate}
We say that a system is \emph{IND-CPA secure} if an attacker has only a negligible advantage over randomly guessing $i.$
\newpage

For IND-CCA2, the attacker and the challenger are playing the following game.
\begin{enumerate}
    \item The attacker sends two distinct messages $m_1,m_2$ to the challenger.
    \item The challenger selects one of the messages $m_i$ and sends the \emph{challenge} $c_i$, which is the encrypted message $m_i.$
    \item The attacker may query a decryption oracle on any cipher but the target cipher $c_i.$
        \item The attacker tries to guess $i.$
\end{enumerate}
\medskip
We say that a system is \emph{IND-CCA2 secure} if an attacker has only a negligible advantage over randomly guessing $i.$
\medskip

Let us consider the McEliece framework from Section \ref{sec:mcframework}.
\medskip

The IND-CPA security for this framework translates  as:
the challenger preforms the key generation, getting the secret key $\bG\in \mathbb{F}_q^{k \times n}$ and sends the public key $\bG'\in \mathbb{F}_q^{k \times n}$ to the attacker. 
The attacker chooses two messages $\bm_1, \bm_2 \in \mathbb{F}_q^k$ and sends them to the challenger. The challenger now chooses $b \in \{1,2\}$ and encrypts $\bm_b$ as
$$\bc= \bm_b \bG'+\be,$$ for some random error vector of Hamming weight $t$. 
The challenger sends $\bc$ back to the attacker. The attacker tries to figure out whether $\bm_1$ or $\bm_2$ was encrypted.
\medskip

\begin{proposition}
    The classic McEliece framework is not IND-CPA secure.
\end{proposition}
\begin{proof}
    The attacker can easily recover which message was encrypted by computing 
    \begin{align*} 
    \bc_1 &= \bm_1\bG', \\ 
    \bc_2 &= \bm_2\bG',
    \end{align*}
    and testing whether
    the received $\bc$ has distance $t$ from one of the codewords. Indeed, if $\bm_1$ was encrypted, then $$\bc-\bc_1=\bm_1\bG' -\bm_1\bG' +\be=\be$$ has weight $t$, whereas 
    $$\bc-\bc_2=\bm_1\bG'-\bm_2\bG'+\be= (\bm1-\bm_2)\bG'+\be$$ has weight larger than $t$, as any codeword (thus also $(\bm_1-\bm_2)\bG'$) has weight at least $2t+1$ and adding $\be$, we can decrease the weight to at least $t+1.$
\end{proof}

\begin{exercise}
    Show that the Niederreiter framework is not IND-CPA secure.
\end{exercise}

An easy fix for this issue is called \emph{random padding.}
Instead of choosing the message $\bm \in \mathbb{F}_q^k$, we only choose a part of the message, say $\bm' \in \mathbb{F}_q^{\ell}$ and choose the remaining $k-\ell$ position at random, called $\br$. 
\newpage

\begin{proposition}
    The McEliece framework using random padding is IND-CPA secure.
\end{proposition}
\begin{proof}
    The attacker has now chosen $\bm_1, \bm_2 \in \mathbb{F}_q^\ell$ and sends them to the challenger. 
    Assume the challenger encrypts $\bm_1$ as 
    $$\bc= (\bm_1, \br) \bG' + \be,$$ and sends this back to the attacker. Let us split the public generator matrix into $\bA \in \mathbb{F}_q^{\ell \times n}$ and $\bB \in \mathbb{F}_q^{(k-\ell) \times n},$ hence the ciphertext is 
    $$\bc=\bm_1\bA +\br \bB +\be.$$
    The attacker can now compute
    \begin{align*}
        \bc_1 &= \bm_1\bA, \\
        \bc_2 &= \bm_2\bA.
    \end{align*}
    Taking these away from the received ciphertext, the attacker gets
    \begin{align*} \bc-\bc_1& =\br\bB+\be,\\ 
    \bc-\bc_2&= \br\bB + (\bm_1-\bm_2)\bA +\be. \end{align*}
    However, the only way to recover $\br$ or $\be$ is to solve the SDP.
\end{proof}

Note that in \cite{kobara} the authors gave conversions of the McEliece system to achieve CCA2 security. 
\medskip
 
For digital signature schemes, we have a similar notion to CCA and CPA, called Existential UnForgeability under Chosen Message Attack (EUF-CMA).

The new game works as follows.
\begin{enumerate}

\item     The challenger generates a secret key $\mathcal{S}$ and a public key $\mathcal{P}$ and sends $\mathcal{P}$ to the attacker.
   \item  The attacker chooses  messages $m_1, \ldots, m_N$ and sends them to the challenger. 
   \item The challenger   generates the signatures $(\sigma_1,\ldots, \sigma_N)$ and sends them to the attacker. 
   \item The attacker wins, if the attacker is able to generate a valid signature $\sigma$ for some message $m \neq m_i$.  
   \end{enumerate}

The signature scheme is called EUF-CMA secure if no (efficient) adversary has a non-negligible advantage in winning the game. 
Note that EUF-CMA security, thus, also asks for signatures to behave indistinguishably from some random distribution. 

\newpage

\section{Historical Overview}\label{sec:overview}
 
 There have been many proposals especially for the McEliece framework. We will here only list a small choice of them, which we hope represent well the major difficulties in proposing new code-based cryptosystems.

McEliece proposed to use binary Goppa codes for his framework, and while the initially proposed parameters are now broken with information set decoding \cite{def}, algebraic attacks are only known for specific parameter sets of Goppa codes \cite{wildattack,high}. In fact, for most parameter sets, there is no algebraic property of binary Goppa codes known which distinguishes them from a random code. The drawback of binary Goppa codes, however, is that they can only correct a small amount of errors, leading to  large  generator matrices for cryptosystems to reach a fixed security level, resulting in large key sizes. 

Other proposals have tried to avoid this problem by using other classes of algebraic codes. Several proposals are based on GRS codes, since these codes have the largest possible error correction capability, but were ultimately broken: Sidelnikov-Shestakov proposed an attack \cite{sidelnikovattack} which recovers parameters for the Niederreiter scheme \cite{niederreiter}, where GRS codes were originally proposed. 

Attempts to avoid this weakness \cite{berger, bbcrs, bbcrs2, bbcrs3, bolkema, weight2, niederreiter, kk, trs} were often unsuccessful, as GRS codes can be distinguished from random codes with the help of the square code \cite{wiesche, distinguisher, distinguisher2, matthieu, attacktrs}, since the square code of a  GRS code has a very low dimension. 

Other proposals have been made using non-binary Goppa codes \cite{wild}, algebraic geometry
codes \cite{ag}, LDPC and MDPC codes \cite{ldpc,mo00p,mdpc}, Reed-Muller codes \cite{rm} and convolutional
codes \cite{conv}, but most of them were unsuccessful in hiding the structure of the private code
\cite{attackag, wildattack, convattack, rmattack, qcattack}. \\

\begin{table}[h!]
\begin{center}
\begin{tabular}{|c|c|c|}
\hline 
Code   & proposed in & attack  \\\hline \hline
Goppa & \cite{mceliece, NISTMcEliece} & \\ \hline
Wild Goppa  & \cite{wild} & \cite{wildattack} \\ \hline 
 Interleaved Goppa & \cite{IG} & \\ \hline
GRS & \cite{niederreiter} & \cite{sidelnikovattack} \\ \hline
 Twisted RS & \cite{trs} & \cite{attacktrs} \\ \hline 
low-codimensional subcodes of GRS  & \cite{berger} & \cite{wiesche} \\ \hline
Sum of GRS & \cite{bbcrs, weight2} & \cite{distinguisher} \\ \hline 
Expanded GRS & \cite{kk} & \cite{matthieu} \\ \hline 
Subspace Subcodes of GRS & \cite{ssrs} & \cite{matthieu} \\ \hline
GRS and random columns & \cite{NISTRLCE, wiesch}  & \cite{rlceattack} \\  \hline
$(U, U+V)$ RS & \cite{uuv} & \\ \hline 
Reed-Muller & \cite{rm} & \cite{borodin, rmattack} \\ \hline
Polar  & \cite{polar} & \cite{vlad} \\ \hline 
Algebraic geometry  & \cite{ag} & \cite{attackag, attackag2} \\ \hline
LDPC  & \cite{ldpc,mo00p} & \cite{qcattack} \\ \hline 
MDPC & \cite{mdpc} & \cite{qcattack} \\ \hline
Convolutional & \cite{conv} & \cite{convattack} \\
 \hline
 Ordinary concatenated & \cite{concat1} & \cite{concatattack} \\ \hline
 Generalized concatenated & \cite{concat2} & \\ \hline
 \end{tabular}\caption{Proposals for the McEliece Framework}
\end{center}  
\end{table}

The first rank-metric code based cryptosystem called GPT  was proposed  in 1991 by   Gabidulin,   Paramonov and  Tretjakov  \cite{gpt}. The authors suggest the use of Gabidulin codes, which can be seen as the rank-metric analog of GRS codes. Similar to the distinguisher on GRS codes, namely the square code attack, also Gabidulin codes suffer from a distinguisher by Overbeck \cite{overbeck} using the Frobenius map. 
The GPT system was then generalized in \cite{ggpt}, but still suffers from an extended Frobenius distinguisher \cite{AL}. Since this proposal some authors have tried to fix this security issue by tweaking the Gabidulin code \cite{gab, tgab}. Other rank-metric systems include \cite{pierre,gab3,acht}. \\
\begin{table}[h!]
\begin{center}
\begin{tabular}{|c|c|c|}
\hline 
Code   & proposed in & attack  \\\hline \hline
Gabidulin & \cite{gpt,ggpt} & \cite{overbeck,AL} \\ \hline
Subspace subcodes of Gabidulin & \cite{gab} & \\ \hline
Twisted Gabidulin & \cite{tgab} &   \\ \hline
 \end{tabular}\caption{Proposals for the GPT framework}
\end{center}  
\end{table}

\begin{table}[h!]
\begin{center}
\begin{tabular}{|c|c|c|}
\hline 
Code   & proposed in & attack  \\\hline \hline
GRS (list decoding) & \cite{augot} & \cite{augotattack} \\ \hline 
Gabidulin (list decoding) & \cite{FL} & \cite{FLattack} \\ \hline
Interleaved Gabidulin & \cite{repair, LIGA, IntG} & \cite{bombar}  \\ \hline
Gabidulin & \cite{ramesses} & \cite{bombar} \\ \hline 
 \end{tabular}\caption{Proposals for the AF framework}
\end{center}  
\end{table}

Next, we want to list some of the most important proposals for code-based signature schemes. The first code-based signature scheme was proposed in 2001 by Courtois, Finiasz and Sendrier (CFS) \cite{cfs}. Again this can be considered as a framework, but the code suggested by the authors was a high rate Goppa code, for which, unfortunately, a distinguisher exists \cite{high}. 
Another way to approach this problem is to relax the weight condition on the error vector. This idea has been followed in \cite{ldgm} where low-density generator matrices were proposed, in  \cite{convsign}, where convolutional codes were suggested, and in \cite{NISTpqsigRM}, where they use Reed-Muller codes. The proposals \cite{ldgm, convsign} have been attacked in \cite{ldgmattack, convsignattack} respectively. 

Also notable are the signature schemes in \cite{KKS1, KKS2, BMS, GS}, which can at most be considered as one-time signatures due to the attack in \cite{COV, OT}. 

In \cite{SURF} the authors propose binary $(U,U+V)$ codes in a signature scheme and the security relies on the problem of finding a closest codeword. However, the hull of such a code is typically much larger than
for a random linear code of the same length and dimension.
Thus, this proposal has been attacked in \cite{SURFattack}.
This problem has later been solved by the authors of Wave  \cite{wave},
by using generalized $(U,U+V)$ codes over the ternary and basing the security on the farthest codeword problem. In addition, Wave provides a proof of the preimage sampleable property (first introduced in \cite{gentry}), which thwarts all attacks trying to exploit the knowledge of signatures.

In \cite{SHMWW} the authors propose a  code-based signature scheme from the Lyubashevsky framework, which was then broken in \cite{karan}. 

Also the code-equivalence problem has been used for a code-based signature scheme in \cite{less}, which was attacked in \cite{LESSattack}. The $\mathsf{LESS}$ signature scheme resolved the vulnerability in \cite{lessNIST}.

A one-time signature scheme from quasi-cyclic codes has been proposed in \cite{persign}. Also this proposal has been attacked in \cite{paolo}.

The signature scheme RaCoSS \cite{NISTRaCoSS} submitted to NIST standardization process is similar to the hash-and-sign approach of CFS but depending on some Bernoulli distributed vector. This proposal has been broken (either see \cite{racossattack} or the comment section on the NIST website\footnote{\url{https://csrc.nist.gov/CSRC/media/Projects/Post-Quantum-Cryptography/documents/round-1/official-comments/RaCoSS-official-comment.pdf}}).  

Finally, the signature scheme pqsigRM \cite{NISTpqsigRM} is an adaption of the broken CFS scheme \cite{cfs}, where the authors propose the use of Reed-Muller codes instead of  Goppa codes, this proposal has also been cryptanalyzed\footnote{\url{https://csrc.nist.gov/CSRC/media/Projects/Post-Quantum-Cryptography/documents/round-1/official-comments/pqsigRM-official-comment.pdf}}.  \\

In the rank metric, one of the most notable signature schemes is that of RankSign \cite{NISTRankSign}, which has been attacked in \cite{attackranksign}. 
Other rank-metric signature schemes include Durandal \cite{durandal}, which is in the Lyubashevsky framework and MURAVE \cite{murave}. 
Note that, even though Durandal has an EUF-CMA security proof, it has recently been broken \cite{durandalattack}.  \\

Due to the Fiat-Shamir transform, we also  include code-based ZK protocols here, although the proposals until now all suffer from large signature sizes. The ZK protocols usually use random codes, thus we will often  not specify a particular proposed code.

 The first code-based ZK protocol was proposed by Stern in 1993 \cite{SternZK} and recently after also by V\'eron \cite{veron}. In this survey we have covered two improvements on their idea, namely CVE \cite{cve} and AGS \cite{ags}. 
 
 In a recent paper \cite{rest} the authors propose to use restricted error vectors in CVE, which leads to smaller signature sizes.
 
 Another approach to reduce the signature sizes is the quasi-cyclic version of Stern's ZK protocol, proposed in \cite{qcstern}.  
 
 Also rank-metric ZK protocols have been proposed in the recent paper \cite{bellini}, with the aim of turning it into a fully fledged rank-metric signature scheme. 

\section{Submissions to NIST}\label{sec:nist}

In 2016 the National Institute of Standards and Technology (NIST) started a competition to establish post-quantum cryptographic standards for public-key cryptography and signature schemes. Initially, 82 proposals were submitted of which 69 could participate in the first round. 19 of these submissions were based on coding theory.
\medskip

In 2020, the third round was announced. Of the initial candidates, 9 public-key systems and 6 signature schemes still remain in this round. Three of the 9 public-key cryptosystems are code-based, one of them being Classic McEliece \cite{NISTMcEliece}, a Niederreiter-based adaption of the initial McEliece cryptosystem. 
\medskip

The other two candidates put effort on avoiding the drawback of large public-key sizes. BIKE \cite{NISTBike} achieves this by combining circulant matrices with MDPC codes, whereas HQC \cite{NISTHQC} is a proposal based on the quasi-cyclic scheme, which does not require using the algebraic structure of the error-correcting code.
\medskip

In this section, we will study these candidates in depth, for this we  provide tables summarizing the submissions that were eliminated in round 1, round 2 and finally the finalists of round 3. 
\medskip
 
Table \ref{round1} contains all public-key encryption and key-encapsulation mechanism candidates, which were eliminated in round one. All candidates use the Hamming metric (HM) or the rank metric (RM). Key sizes will be given in kilobytes, pk denotes the public key and sk the secret key.
\medskip

Due to space limitations, we will sometimes abbreviate the McEliece framework with MF, the Niederreiter framework with NF,   the framework of Alekhnovich by AF, the quasi-cyclic framework by QCF and finally a Diffie-Hellman approach by DH. 
\medskip

In addition to acronyms that were already introduced, we also abbreviate quasi-cyclic (QC), Ideal Code (IC) and double-circulant (DC).

The given key sizes are for the parameter sets that were proposed for $128$ bits of security (however, some proposals contained multiple suggestions for parameter sets for this security level).
\medskip

All data is taken from the supporting documentations of the NIST proposals BIG QUAKE \cite{NISTBigQuake}, DAGS \cite{NISTDags}, Edon-K \cite{NISTEdonK}, LAKE \cite{NISTLake}, LEDAkem \cite{NISTLEDAkem}, LEDApkc \cite{NISTLEDApkc}, Lepton \cite{NISTLepton}, LOCKER \cite{NISTLOCKER}, McNie \cite{NISTMcNie}, Ouroboros-R \cite{NISTOuroborosR}, QC-MDPC KEM \cite{NISTQCMDPC}, Ramstake \cite{NISTRamstake} and RLCE-KEM \cite{NISTRLCE}.\\
 
 \begin{table}[h!]
\begin{center}
\begin{tabular}{|c|c|c|c|c|c| }
\hline
Candidate & Framework & Code & Metric & Pk Size &  Reason for Drop Out\\ 
\hline
BIG QUAKE &  NF & QC Goppa & HM & $25-103$ & large key sizes \\
DAGS &  MF & dyadic GS & HM &$8.1$ & broken \cite{DAGSattack} \\
Edon K &  MF & binary Goppa & HM & $2.6$ &  broken \cite{EdonKattack} \\
LAKE &  NF & IC, DC, LRPC & RM & $0.4$ & merged (ROLLO) \\
LEDAkem &  NF & QC LDPC & HM & $3.5 - 6.4$ &  merged (LedAcrypt) \\
LEDApkc &  MF & QC LDPC & HM & $3.5 - 6.4$ &  merged (LedAcrypt) \\
Lepton &  AF & BCH & HM & $1.0$ &  cryptanalysis \\
LOCKER &  NF & IC, DC, LRPC & RM  &$0.7$ &  merged (ROLLO) \\
McNie &  MF/NF& QC LRPC & RM & $0.3 - 0.5$ &  broken \cite{McNieattack1} \cite{McNieattack2}\\
Ouroboros-R &  QCF & DC LRPC & RM & $1.2$ & merged (ROLLO)\\
QC-MDPC KEM &  MF & QC MDPC & HM & $1.2 -2.6$ & N/A\\
Ramstake & DH & RS & HM & $26.4$ & broken \cite{Ramstakeattack}   \\
RLCE-KEM &  MF & GRS & HM & $118-188 $ &   broken \cite{rlceattack}\\
\hline
\end{tabular}\caption{Code-based PKE/KEM submissions to NIST, eliminated in round 1}\label{round1}
\end{center}
\end{table}

The reason for the drop out of BIG QUAKE was mainly discussed at CBC 2019\footnote{\url{https://drive.google.com/file/d/1nruEobwdeJbtwouJssbjZCK0WQiBN7rW/view}}, and is due to the large key sizes of the proposal, as it is "still worse than completely unstructured lattice KEM."
The reason for Lepton's drop out, is a security issue that can be found in the comment section of the NIST website\footnote{url{https://csrc.nist.gov/CSRC/media/Projects/Post-Quantum-Cryptography/documents/round-1/official-comments/Lepton-official-comment.pdf}}.
 
 \newpage
In Table \ref{sigprop} we list  all code-based signature schemes that were eliminated during round one, which in every case was due to cryptanalysis.
\medskip

The table contains their signature sizes, public key sizes, secret key sizes (all in kilobytes) and the recommended number of rounds necessary to ensure verification with a very high probability. 
\medskip

For this a security level of 128-bit is fixed in the respective scheme. The signature size of pqsigRM is taken from \cite{pqsigmod}, all other data is taken from the supporting documentations pqsigRM \cite{NISTpqsigRM}, RaCoSS \cite{NISTRaCoSS} and RankSign \cite{NISTRankSign}.\\

\begin{table}[h!]
\begin{center}
\begin{tabular}{|c|c|c|c|c|c|c|c| }
\hline
Candidate & Signature Size & Pk Size & Sk Size & Rounds \\
\hline
pqsigRM & $0.5$ & $262$ & $138$ & $100$ \\
RaCoSS &  $0.3$ & $169$ & $100$ &  $100$ \\
RankSign &  $1.4 -1.5$& $10$ & $1.4 - 1.5$ &  N/A \\
\hline
\end{tabular}\caption{Code-based signature submissions to NIST, eliminated in round 1}\label{sigprop}
\end{center}
\end{table}
Table \ref{r2} contains all PKE/KEM candidates that were eliminated during round two. There are no code-based signature schemes that made it to round two or further.
\medskip

All data is taken from the supporting documentations of LEDAcrypt \cite{NISTLEDAcrypt}, NTS-KEM \cite{NISTNTS}, ROLLO \cite{NISTRollo} and RQC \cite{NISTRQC}.\\

\begin{table}[h!]
\begin{center}
\begin{tabular}{|c|c|c|c|c|c| }
\hline
Candidate & Framework & Code & Metric & Pk Size &  Reason for Drop Out\\ 
\hline
LEDAcrypt &  McE/N & QC LDPC & HM & $1.4 - 2.7$ & broken \cite{LEDAcryptattack} \\
NTS-KEM & Niederreiter & binary Goppa & HM & $319$ & merged (Classic McE) \\
ROLLO & Niederreiter & IC, LRPC & RM & $0.7$ & cryptanalysis \cite{ROLLOattack}\\
RQC &  Quasi-Cyclic & IC, Gabidulin & RM & $1.8$ & N/A\\
\hline
\end{tabular}\caption{Code-based PKE/KEM submissions to NIST, eliminated in round 2}\label{r2}
\end{center}
\end{table}

Finally, there are three candidates that made it to the final round, round three. Classic McEliece, as main candidate, and BIKE and HQC as alternative candidates.

As before, the public key (pk) size is given in kilobytes, data is taken from the proposed parameters for the $128$-bit security level.

\begin{table}[h!]
\begin{center}
\begin{tabular}{|c|c|c|c|c| }
\hline
Candidate & Framework & Code & Metric & pk size \\ 
\hline
Classic McEliece &  Niederreiter & binary Goppa & Hamming & $261$ \\
BIKE &  Niederreiter & MDPC & Hamming & $1.5$ \\
HQC &  Quasi-Cyclic & decodable code of choice, QC & Hamming & $2.2$ \\
\hline
\end{tabular}\caption{Final round code-based PKE submissions to NIST}
\end{center}
\end{table}

\newpage
\subsection{Round 4 Candidates: Classic McEliece, BIKE and HQC}\label{sec:finalist}
In this section, we present the three code-based proposals Classic McEliece, BIKE and HQC, which are in the fourth round of the NIST standardization call from 2016. For each one, we give a mathematical description and the proposed parameters.

\subsubsection{Classic McEliece}

The NIST submission Classic McEliece uses the Niederreiter framework (Section \ref{sec:niedframework}) with binary Goppa codes (Definition \ref{def:goppa}) as secret codes. This subsection is based on the round 3 submission \cite{NISTMcEliece}.

Let us start with the description of the scheme.  Let $m$ be a positive integer, $q=2^m$, $n \leq q$ and $t \geq 2$ be positive integers such that $mt < n$ and set $k = n - mt$.

Further, pick a monic irreducible polynomial $f(z) \in \mathbb{F}_2 [z]$ of degree $m$ and identify $\mathbb{F}_{q}$ with $\mathbb{F}_2[z] / f(z)$. Note that under this identification, every element in $\mathbb{F}_{2^m}$ can be written as $$u_0 + u_1 z + \ldots + u_{m-1}z^{m-1}$$ for a unique vector $(u_0, u_1, \ldots, u_{m-1}) \in \mathbb{F}_2^m$.

With these preliminaries set, we can describe the public-key encryption scheme:
\begin{itemize}
    \item \textbf{Key Generation:}   \begin{enumerate}
        \item Generate a random monic irreducible polynomial $g(x) \in \mathbb{F}_q[x]$ of degree $t$ and $n$ random distinct elements $\alpha_1, \ldots, \alpha_n \in \mathbb{F}_q$.
        \item Compute a parity-check matrix $\tilde{\bH} = \{ \tilde{h}_{ij} \}_{ij}$ of the binary Gopppa code with parameters $(g, \alpha_1, \ldots, \alpha_n)$ by computing $\tilde{h}_{ij} = \alpha_j^{i-1} / g(\alpha_j)$.
        \item Apply an invertible matrix to $\tilde{\bH}$ and permute the columns of this matrix to get a matrix in systematic form $\bH=( \text{Id}_{n-k} | \bT)$.
        
        Denote with $(\alpha'_1, \ldots, \alpha_n')$ the $n$-tuple obtained by applying the same permutation to $(\alpha_1, \ldots, \alpha_n)$. 
        
        Note that $( \text{Id}_{n-k} | \bT)$ is a parity-check matrix of the Goppa code defined by $(g, \alpha'_1, \ldots, \alpha'_n)$.
    \end{enumerate}
    \item \textbf{Private Key:} The private key is the $(n+1)$-tuple $\Gamma' = (g, \alpha'_1, \ldots, \alpha'_n)$.
    \item \textbf{Public Key:} The public key is the $(n-k) \times (n-k)$ matrix $\bT$ and the number $t$.
    \item \textbf{Encryption:} Encode the message as weight $t$ vector $\be \in \mathbb{F}_2^n$ and compute $$\bc_0 = \bH\be^\top \in \mathbb{F}_2^{n-k}.$$
    \item \textbf{Decryption:} Extend $\bc_0$ to $\bv=(\bc_0^\top, 0, \ldots, 0) \in \mathbb{F}_2^n$. The parameters $\Gamma'$ of the private key define a Goppa code, so we can use a decoding algorithm for Goppa codes to find a codeword $\bc$ with distance $\leq t$ to $\bv$ (if it exists). 
    \medskip
    
    We then recover $\be$ as $\be=\bv+\bc$ and check that it indeed satisfies $\bH\be^\top = \bc_0$ and is of weight $t$.
\end{itemize}

\begin{remark}
The decryption works for the following reason: we have that $\bH = (\text{Id}_{n-k} | \bT)$, so $$\bH\bv^\top = \text{Id}_{n-k} \bc_0 = \bc_0.$$ Thus,  it follows that $$\bH(\bv+\be)^\top = 0, $$ and  $\bc=\bv+\be$ is a codeword of the Goppa code defined by $\Gamma'$. 
\medskip

Since this code has minimum distance at least $ 2t+1$, we get that $\bv+\be$ is also the unique codeword of distance up to $ t$ from $\bv$, so we may recover the error vector as $\be=\bv+\bc$.
\end{remark}

\subsubsection{Proposed Parameters for Classic McEliece}

We give an overview of the proposed parameter sets, input and output sizes for the  expected security levels. Level 1 corresponds to 128 bits, level 3 corresponds to 192 bits and level 5 corresponds to 256 bits of security. The key sizes and ciphertext size are given in bytes.

\begin{table}[h!]
\begin{center}
\begin{tabular}{|c|c|c|c|c|c|c|c|}
\hline
Parameter set & $m$ & $n$ & $t$ & Public key & Private key & Ciphertext & Security level\\ 
\hline
mceliece348864 & $12$ & $3488$ & $64$& $261120$ & $6492$ & 128 & 1 \\
mceliece460896 & $13$& $4608$ & $96$& $524160$ & $13608$ & 188 & 3 \\
mceliece6688128 &$13$ & $6688$& $128$& $1044992$ & $13932$ & $240$ & 5  \\
mceliece6960119 & $13$& $6960$& $119$& $1047319$ & $13948$ & 226 & 5 \\
mceliece8192128 &$13$& $8192$ & $128$& $1357824$ & $14120$ & 240 & 5 \\
\hline
\end{tabular}\caption{Parameters for Classic McEliece}
\end{center}
\end{table}

The Classic McEliece submission is considered the main candidate for standardization by NIST. It is clearly based on the original proposal  of McEliece \cite{mceliece} and thus a rather conservative choice by NIST. 
The main advantage of Classic McEliece is thus its well studied security, as there are no known algebraic attacks on the original proposal of McEliece since 1978, but it still suffers from the same disadvantage, i.e.,   the large size of its public keys.

\subsubsection{BIKE}

The NIST submission Bit Flipping Key Encapsulation (BIKE) combines circulant matrices with the idea of moderate density parity-check matrices (Definition \ref{def:mdpc}). The usage of circulant matrices keeps key sizes small  while using moderate density parity-check matrices allows efficient decoding with a Bit-Flipping algorithm. We follow the NIST round 3 submission  \cite{NISTBike} and give a ring-theoretic description of the system. Note however that BIKE can also be fully described with matrices.

Let $r$ be prime number such that $2$ is primitive modulo $r$, i.e., $2$ generates the multiplicative group $\mathbb{Z}/ r\mathbb{Z}^\star$. The parameter $r$ denotes the block size, from which we obtain the code length $n=2r$. We further pick an even row weight $w \approx \sqrt{n}$ such that $w/2$ is odd and an error weight $t \approx \sqrt{n}$. 

We then set $R:= \mathbb{F}_2[x]/(x^r -1)$. Any element $a \in R$ can be represented as polynomials of degree less or equal than $r-1$ and can uniquely be written as linear combination of the form 
$$a=\sum_{i=0}^{r-1} a_i x^i,$$
where $a_i \in \mathbb{F}_2$ for all $i \in \{ 0, 1, \ldots, r-1 \}$. 
\medskip

This gives us a natural notion of the weight of $a$, which we denote with $\wt(a)$, i.e.,
$$ \wt(a)= | \{ i \in \{ 0,1, \ldots, r-1 \} \mid  a_i \neq 0 \}|. $$

\begin{remark}
The choice of $r$ ensures that the irreducible factors of $x^r -1$ are $x-1$ and $x^{r-1} + x^{r-2} + \cdots + 1$ (see Exercise \ref{factors}). As a consequence of this, an element $a \in R$ is invertible if and only if $\wt(a)$ is odd and $\wt(a) \neq r$.
\end{remark}

\begin{itemize}
    \item \textbf{Key Generation:} Pick a pair $(h_0, h_1) \in R^2$ such that $\wt(h_0) = \wt(h_1) = w/2$. Then compute $h = h_1 h_0^{-1} \in R$.
    \item \textbf{Private Key:} The private key is the pair $(h_0, h_1)$.
    \item \textbf{Public Key:} The public key is the element $h \in R$ and the integer $t$.
    \item \textbf{Encryption:} The message gets encoded as error $(e_0, e_1) \in R^2$ such that $\wt(e_0) + \wt(e_1) = t$ and then encrypted as $s= e_0 + e_1h$.
    \item \textbf{Decryption:} We compute $sh_0 = e_0 h_0 + e_1 h_1$. Since $h_0$ and $h_1$ are of moderate density, this can be decoded efficiently with a Bit-Flipping algorithm to recover the pair $(e_0, e_1)$.
\end{itemize}

\begin{remark}
The difficulty of attacking BIKE lies in finding an element $\tilde{h} \in R$ of at most moderately high weight, such that $h \tilde{h}$ is also of at most moderately high weight.
\end{remark}

\begin{remark}
BIKE can also be described with matrices: for $$a=\sum_{i=0}^{r-1} a_i x^i \in R$$ and $$b=\sum_{i=0}^{r-1} b_i x^i,$$ we are considering the code with parity-check matrix
$$ \bH = \left(\begin{array}{ccccc|ccccc}
    a_0 & a_1 & \cdots&a_{r-2} & a_{r-1} & b_0 & b_1 & \cdots& b_{r-2} & b_{r-1} \\
    a_{r-1}& a_0 & \cdots & a_{r-3} & a_{r-2} & b_{r-1}& b_0 & \cdots & b_{r-3} & b_{r-2}\\
    \vdots       &    & \ddots    &   & \vdots & \vdots       &    & \ddots    &   & \vdots \\
    a_2       & a_3   & \cdots    & a_0  & a_1 & b_2       & b_3   & \cdots    & b_0  & b_1 \\
    a_1& a_2 & \cdots  & a_{r-1} & a_0 & b_1& b_2 & \cdots  & b_{r-1} & b_0 \end{array}\right).$$

    \medskip

In this case, the errors $e_0 = \sum_{i=0}^{r-1} e_{0,i}x^i$ and $e_1 = \sum_{i=0}^{r-1} e_{1,i}x^i$ may be viewed as vectors $$\tilde{\be}_j=( e_{j,0}, e_{j,r-1}, e_{j, r-2}, \ldots, e_{j,1})$$ for all $j \in \{1,2 \}$. We then compute syndromes by $$\bH (\tilde{\be}_1 \: | \: \tilde{\be}_2 )^\top.$$
\end{remark}

\begin{exercise}\label{factors}
Let $r$ be a prime such that $2$ generates $\mathbb{Z}/r\mathbb{Z}^\star$. Show that the irreducible factors of $x^r -1 \in \mathbb{F}_2[x]$ are $x-1 $ and $x^{r-1} + x^{r-2} + \cdots 1$. You may use the following steps:
\begin{enumerate}
    \item Let $p(x)$ be a monic irreducible factor of $x^{r-1} + x^{r-2} + \cdots +1$ and $\alpha$ a root of $p(x)$ in the algebraic closure. Show that $r$ is the smallest positive integer such that $\alpha^r = 1$.
    \item Justify that the roots of $p(x)$ are the elements of the set $\left\{ \alpha^{(2^n)} \: | \: n \in \mathbb{N}_{\geq 1} \right \}$.
    \item Show that $\left\{ \alpha^{(2^n)} \: | \: n \in \mathbb{N}_{\geq 1} \right \}$ contains exactly $r-1$ elements and conclude that $p(x)=x^{r-1} + x^{r-2} + \cdots + 1$.
\end{enumerate}
\end{exercise}

\subsubsection{Proposed Parameters for BIKE}

We now present the proposed parameters for three levels of security, where again level 1 is 128 bits of security, level 3 is 192 bits, and level 5 is 256 bits of security. We also include an estimate for the decoding failure rate (DFR) and key and ciphertext sizes in bytes.

\begin{table}[h!]
\begin{center}
\begin{tabular}{ | c | c | c | c | c | c | c | c | }
\hline
Security & $r$ & $w$ & $t$ & Private key & Public key & Ciphertext & DFR\\ 
\hline
 Level 1 & $12323$ & $142$ & $134$ & $281$ & $1541$ & $1573$ & $2^{-128}$ \\
 Level 3 & $24659$ & $206$ & $199$ & $419$ & $3083$ & $3115$ & $2^{-192}$ \\  
 Level 5 & $40973$ & $274$ & $264$ & $580$ & $5122$ & $5154$ & $2^{-256}$ \\
 \hline
\end{tabular}\caption{Parameters for  BIKE }
\end{center}
\end{table}

It can be seen that BIKE has small public key sizes, which is a big advantage over the other systems.

\subsubsection{HQC}

The submission Hamming Quasi-Cyclic (HQC) is based on the quasi-cyclic framework (see Section \ref{sec:quasicyclic}) and uses a combination of a decodable code of choice and circulant matrices. 
\medskip

The third round proposal suggests to use concatenated Reed-Muller and Reed-Solomon codes (Definitions \ref{def:concat}, \ref{def:reedmuller}, \ref{def:reedsolomon}), in the initial NIST submission \cite[Section 1.6]{hqcround1} a tensor product code of a BCH and a repetition code was proposed. An important feature of HQC is the fact that the used codes are not secret. 
\medskip

We follow the NIST submission \cite{NISTHQC} for the detailed description.

\medskip

Let $n$ be such that $(x^n -1)/(x-1)$ is irreducible over $\mathbb{F}_2$. We pick a positive integer $k<n$ and an $[n,k]$ linear code $\mathcal{C}$ with an efficient decoding algorithm, whose error correcting capacity is given by $t$. We are further given error weights $w$, $w_r$ and $w_e$, all in the range of $\frac{\sqrt{n}}{2}$. 
We set $R:= \mathbb{F}_2[x]/(x^n-1)$. Recall that any element $a \in R$ can be written as 
$$a=a_{n-1}x^{n-1} + a_{n-2}x^{n-2} + \ldots +a_0$$ 
for unique $a_0, a_1, \ldots, a_{n-1} \in \mathbb{F}_2$. For such an element we denote its Hamming weight as
$$ \wt_H(a) = | \{ i \in \{ 0, 1, \ldots, n-1 \} \mid a_i \neq 0 \}|.$$
Note also that we can identify a vector  $ \ba =( a_0, a_1, \ldots, a_{n-1}) \in \mathbb{F}_2^n$ with the element $a=\sum_{i=0}^{n-1} a_i x^i \in R$ and vice versa. In the following description any bold letter, e.g. $\bu$, refers to the associated vector in $\mathbb{F}_2^n$ of an element in $R$, e.g. $u \in R.$

\begin{itemize}
    \item \textbf{Key Generation:}  Given the parameters $(n,k, t, w, w_e, w_r)$, choose a generator matrix $\bG$ of the code $\mathcal{C}$ and generate a random $h \in R$.  
    \item \textbf{Private Key:} The private key is a randomly generated pair $(y,z) \in R^2$ such that $\wt_H(y) = \wt_H(z) = w$.
    \item \textbf{Public Key:} We compute $s=y+hz \in R$. The public key is given by  $(\bG,h,s, t)$.
    \item \textbf{Encryption:} We randomly generate an element $e \in R$ such that $\wt(e) = w_e$ and a pair $(r_1, r_2) \in R^2$ such that $\wt_H(r_1) = \wt_H(r_2) = w_r$. 
    
    Let $\bm \in \mathbb{F}_2^k$ be the message, which gets encrypted as the pair $\bc=(\bu,\bv) \in R^2$, where $u=r_1 + hr_2$ and $\bv=\bm\bG + \bs\br_2 +\be$. 
    \item \textbf{Decryption:} As mentioned in the quasi-cyclic framework, we compute that 
    $$\bv-\bu\mathbf{z} =  \bm\bG + (\by\br_2 - \br_1\mathbf{z} +\be).$$
    The term $\by\br_2 - \br_1\mathbf{z} +\be$ has Hamming  weight $\leq t$ with high probability (this follows non-trivially from the choice of the parameters). If this is the case, we can use the decoding algorithm of $\mathcal{C}$ to recover the message $\bm$.
\end{itemize}

\subsubsection{Proposed Parameters for HQC}

The following table contains the proposed parameters for HQC together with an upper estimate on the decoding failure rate (DFR) and ciphertext size and key sizes.  The key and ciphertext sizes are given in bytes and as before, security levels 1,3 and 5 correspond to $128$-bit, $192$-bit and $256$-bit security respectively.

\begin{table}[h!]
\begin{center}
\begin{tabular}{ | c | c | c | c | c | c | c | c | c | c |}
\hline
Security & $n$ & $w$ & $w_r = w_e$ & Public key & Private key & Ciphertext & DFR \\ 
\hline
 Level 1 &  $17669$ & $66$ & $75$& $2249$& $40$ & $4481$ &$2^{-128}$ \\
 Level 3 &  $35851$ & $100$ & $114$ & $4522$ & $40$ & $9026$ & $2^{-192}$ \\
 Level 5 &  $57637$ & $131$ & $149$ & $7245$ & $40$ & $14469$ & $2^{-256}$ \\
 \hline
\end{tabular}\caption{Parameters for HQC}
\end{center}
\end{table}

The advantages of HQC are its efficient implementation and its small key sizes. However, HQC suffers from a low encryption rate.

\subsection{Code-Based Signature Schemes}\label{sec:new}

In 2023, NIST has opened an additional standardization call for post-quantum signature schemes. Out of  the  50 submitted schemes, 40 have been found complete and proper and have been published as official round 1 candidates. 
\medskip

Among the 40 schemes, we find
\begin{itemize}
    \item[12] multivariate schemes,
    \item[7] lattice-based schemes,
    \item[4] symmetric schemes,
    \item[1] isogeny-based scheme,
    \item[5] schemes that have been grouped as ``other'',
    \item[11] code-based schemes.
    
\end{itemize}
Within the first 2 months, 11 of the schemes have been attacked. At the moment of this writing, we have 29 surviving schemes, out of which we find
\begin{itemize}
    \item[9] multivariate schemes,
    \item[5] lattice-based schemes,
    \item[4] symmetric schemes,
    \item[1] isogeny-based scheme,
    \item[1] scheme that has been grouped as ``other'',
    \item[9] code-based schemes.
    
\end{itemize}
The interested reader can compare the 29 survivors on \begin{center} 

\url{https://pqshield.github.io/nist-sigs-zoo/}
\end{center}

\medskip

In the following, we will only consider the 11 submitted code-based signatures. 

\newpage

Recall the three different approaches to construct a signature scheme, with the benefits and limitations:
\begin{table}[h!]
    \centering
    \begin{tabular}{|c|c|c|} \hline 
       \multicolumn{3}{|c|}{\textbf{Hash-and-Sign}}   \\ \hline 
         Needs & Limitations & Advantages \\ 
        \hline 
    Trapdoor & Large public keys & Small signatures \\ 
    Secret code & Slow signing & \\ 
    \hline 
    \multicolumn{3}{|c|}{\textbf{ZK Protocol and Fiat-Shamir Transform}}   \\ \hline 
        Needs & Limitations & Advantages \\ 
        \hline 
    Hard problem & Large signatures & Small public keys \\ 
        \hline 
        \multicolumn{3}{|c|}{\textbf{ZK Protocol and MPCitH}}   \\ \hline 
        Needs & Limitations & Advantages \\ 
        \hline 
    Hard problem & Slow signing & Small signatures \\ 
    $(N-1)$-private MPC & Slow verifying & Small public keys\\ 
    \hline 
    \end{tabular}
    \caption{Comparison of the different techniques to construct a code-based signature scheme.}
    \label{tab:comptech}
\end{table}

\subsubsection{Hash-and-sign schemes}
Let us start with the three code-based hash-and-sign schemes.

\begin{table}[h]
    \centering
    \begin{tabular}{|c|c|c|c|} \hline 
        \textbf{Trapdoor} & \textbf{Secret Code} & \textbf{Scheme} & \textbf{Comment} \\ 
        \hline 
    Lee SDP & Quasi-cyclic code & $\mathsf{FuLeeca}$ & Broken \\ 
        \hline 
     
        SDP & Reed-Muller code & Enhanced pqsigRM  & Broken\\ 
        \hline 
   SDP & $(U,U+V)$-code & $\mathsf{WAVE}$ & Large public keys \\ 
        \hline 
    \end{tabular}
    \caption{Hash-and-sign schemes submitted to the additional call of NIST for signature schemes.}
    \label{tab:compHS}
\end{table}

\begin{enumerate}
\item $\mathsf{FuLeeca}$ \\ 

$\mathsf{FuLeeca}$ \cite{FuLeecaNIST} is the first cryptosystem based on the Lee metric.  It uses a secret quasi-cyclic code with low Lee weight generators $\ba, \bb$, i.e., $\text{wt}_L(\ba,\bb)=w_{key}$, defining the two circulant matrices $\bA, \bB$ which give the secret generator matrix $$\bG= \begin{pmatrix} \bA & \bB \end{pmatrix}.$$

\medskip

The public generator matrix is given by $\bG$ in systematic form,   $$\bG'= \begin{pmatrix} \text{Id}_k & \bT\end{pmatrix},$$ for $$\bT=\bA^{-1}\bB.$$ Clearly, it is enough to publish one row of $\bT.$
\medskip

In order to sign a message $\bm$, the signer hashes  $\bm$ getting $\bc=\mathsf{Hash}(\bm)$ and iteratively searches for a small $\bx$, such that $\bv=\bx\bG$ satisfies two conditions 
\begin{enumerate}
    \item $\text{wt}_L(\bv) \in [w_{sig}-2w_{key},w_{sig}]$,
    \item $\text{LMP}(\bv,\bc) > \lambda+64$.
\end{enumerate}
The first assumption ensures that an impersonator has to solve the Lee SDP in order to forge a signature, and the second conditions binds the message to the signature. On a high level, the hash of the message should have many signs matching with the codeword.    By setting their LMP larger than $\lambda$, one ensures that an impersonator has to go through $2^{\lambda}$ randomly chosen $\bv$ before finding enough signs matching.  
Since the codeword $\bv=(\by,\by\bT),$
he signature is then given by $\by.$

A verifier first recovers $\bv=(\by,\by\bT)$ checks exactly these two conditions
\begin{enumerate}
    \item $\text{wt}_L(\bv) \in [w_{sig}-2w_{key},w_{sig}]$,
    \item $\text{LMP}(\bv,\bc) > \lambda+64$,
\end{enumerate} in order to accept the signature $\by.$

The signature scheme shines with very small public key and signature sizes, one of the only code-based schemes to achieve both. 

\begin{table}[]
    \centering
    \begin{tabular}{c|c|c|c|c}
        Level & Public key size & Signature size & Signing time & Verification time  \\ \hline 
        
        I & 1.3 & 1.1 & 1803 & 1.4 \\ 
        III & 1.9 & 1.6 & 2139 & 2.5 \\
        V & 2.6 & 2.1 & 11805 & 3.8 
    \end{tabular}
    \caption{Performance of $\mathsf{FuLeeca}$. Sizes are in kilobytes and timings in MCycles. }
    \label{tab:fuleeca}
\end{table}

Unfortunately, the scheme was broken by van Woerden and H\"ormann. The attack makes use of the following facts:
\begin{itemize} 
\item The $\bx$ used to get the codeword $\bv=\bx\bG \mod p$ is chosen so small, there is no modular reduction necessary.  That is $\bv=\bx\bG$ also over $\mathbb{Z}.$ This allows the attackers to directly use the integer lattice $L(\bG).$
\item The quasi-cyclic structure of the code allows the attackers further to only search for a solution in one part, i.e., $\bG=\begin{pmatrix}
    \bA & \bB
\end{pmatrix}$ and it is enough to work with $L(\bA).$
\item Finally, using BKZ \cite{bkz}, the attacker can find short Euclidean vectors in $L(\bA)$. Usually, one would expect exponentially many such short vectors and only very few of those are also of small Lee weight. However, the chosen instances of $\mathsf{FuLeeca}$ allow for this attack to work fast. 
\end{itemize}
\medskip

\item {Enhanced pqsigRM}\\

 This proposals \cite{enhpqsigrmNIST} follows closely the original idea of CFS using a modified Reed-Muller code.

 Thus, the secret code is given by a Reed-Muller code having parity-check matrix $\bH$ and the public code is a scrambled parity-check matrix $\bH'=\bH\bP.$
 Upon a message $\bm$, one hashes the messages $\mathsf{Hash}(\bm)$ and hopes that it is the syndrome of a low weight vector $\be$, i.e.,
 $\be\bH^\top =\mathsf{Hash}(\bm).$
 In this case, one sends $\be\bP$ as signature.
 The verifier can easily check that  $\be\bP \bP^\top \bH^\top =\mathsf{Hash}(\bm).$

 Note that a scrambled Reed-Muller code can be distinguished and the secret code can be recovered using the attack \cite{rmattack}.
 Thus, Enhanced pqsigRM proposes a modified Reed-Muller code. Recall from Section \ref{sec:prelim}, that Reed-Muller codes are$(U,U+V)$ codes. The original attack makes use of the fact that the hull of such a code, i.e., $\mathcal{C} \cap \mathcal{C}^\perp$ only consists of $(U,U)$-codewords, which helps to reveal the secret code.  To avoid this, the
proposed code is designed so that $\text{dim}(U^\perp \cap V)$ is large.

\begin{table}[]
    \centering
    \begin{tabular}{c|c|c|c|c}
        Level & Public key size & Signature size & Signing time & Verification time  \\ \hline 
        
        I & 2000 & 1.03 & 2.2 & 0.2  
    \end{tabular}
    \caption{Performance of Enhanced pqsigRM. Sizes are in kilobytes and timings in MCycles. }
    \label{tab:enhanced}
\end{table}

Nevertheless,  Enhanced pqsigRM has been broken by  Debris-Alazard,  Loisel and  Vasseur again exploiting the $(U,U+V)$ structure to recover the secret code.
\medskip

\item $\mathsf{WAVE}$ \\ 

$\mathsf{WAVE}$ \cite{WAVENIST} is a hash-and-sign scheme, whose trapdoor is based on permuted generalized $(U,U+ V )$-codes. Unlike most code-based schemes, $\mathsf{WAVE}$ does not rely on finding small weight codewords, but rather large weight codewords. In fact, until the Hamming weight $\frac{q-1}{q}(n-k)$ it is hard to find low weight codewords and similarly after the Hamming weight $k+\frac{q-1}{q}(n-k)$ it is again hard to find large weight codewords.
\medskip

Again a signer starts with a secret generalized $(U,U+V)$ code and scrambles it to publish the parity-check matrix $\bH'=\bH\bP.$
\medskip

Upon a message $\bm$ the signer computes the hash $\mathsf{Hash}(\bm)$ and hopes that it is the syndrome of a \emph{large} weight vector, i.e., $\mathsf{Hash}(\bm)=\be\bH\top$. 
In order to find such large weight $\be$, $\mathsf{WAVE}$ makes use of the secret generalized $(U,U+V)$ code and performing ISD in the $V$ part.

\medskip

In this case, the signer sends the signature $\be\bP$. A verifier can then easily check that $\mathsf{Hash}(\bm)=\be\bP\bP^\top\bH^\top.$

\begin{table}[]
    \centering
    \begin{tabular}{c|c|c|c|c}
        Level & Public key size & Signature size & Signing time & Verification time  \\ \hline 
        
        I & 3677 & 0.8 & 1160 &205   \\ 
        III & 7867 & 1.2 & 3507  &  464 \\ 
         V & 13632 & 1.6 &  7936 &  813  
    \end{tabular}
    \caption{Performance of $\mathsf{WAVE}$. Sizes are in kilobytes and timings in MCycles. }
    \label{tab:wave}
\end{table}

The main advantage of $\mathsf{WAVE}$ is in its security, in fact a large amount of work has been performed using rejection sampling and smartly choosing the distribution, such that  the preimage sampleable property is achieved, which thwarts all attacks trying to exploit the knowledge of signatures.

As limitations, $\mathsf{WAVE}$ has quite large public key sizes in the range of 3 MB. 

\end{enumerate}

\subsubsection{ZK Protocols and Fiat-Shamir Transform}

In the additional call 3 code-based signature schemes using ZK protocols have been submitted, namely $\mathsf{CROSS}$ based on restricted errors, $\mathsf{LESS}$ based on LEP and $\mathsf{MEDS}$ based on MCE. 

Let us start with the three code-based hash-and-sign schemes.

\begin{table}[h]
    \centering
    \begin{tabular}{|c|c|c|} \hline 
        \textbf{Hard Problem}   & \textbf{Scheme} & \textbf{Comment} \\ 
        \hline 
  Restricted SDP  & $\mathsf{CROSS}$ &  \\ 
        \hline 
     
       LEP & $\mathsf{LESS}$  & Large total size\\ 
        \hline 
  MCE  &   $\mathsf{MEDS}$ & Large total size \\ 
        \hline 
    \end{tabular}
    \caption{Signatures from ZK protocols submitted to the additional call of NIST for signature schemes.}
    \label{tab:compzk}
\end{table}

\begin{enumerate}
    \item $\mathsf{CROSS}$\\

   The signature scheme  $\mathsf{CROSS}$ \cite{CROSSNIST} uses an adapted version of the code-based ZK protocol CVE (see Section \ref{sec:ZKID}). However, instead of using SDP and thus $\sigma$ a linear isometry in the Hamming metric, $\mathsf{CROSS}$ relies on the Restricted SDP. This allows not only to represent vectors $\be \in \mathbb{E}^n$ using only the exponents $\ell(\be) \in \mathbb{F}_z^n$, thus having size $n \lceil \log_2(z)\rceil$, but also the maps that act transitively on $\mathbb{E}^n$ are given by componentwise multiplication with vectors in $\mathbb{E}^n$. 
   \medskip

   $\mathsf{CROSS}$ makes use of several techniques to compress sizes, such as Merkle trees and  and weighted challenge vectors, $\bb \in \{0,1\}^t$. In fact, seeing that one of the responses (where $b_i=1$) has a much smaller size to send than the other, in order to reduce the signature size one would sample challenge vectors $\bb$ of large weight. Note that this information could potentially be used by an attacker. Thus, $\mathsf{CROSS}$ adapted the forgery attack \cite{kales} in order to choose the weight $w$ of $\bb$ and the number of rounds $t$, in a secure way.

\begin{table}[]
    \centering
    \begin{tabular}{c|c|c|c|c|c}
       Variant &  Level & Public key size & Signature size & Signing time & Verification time  \\ \hline 
        
      R-SDP-f &   I & 0.06 & 19 & 1.28 & 0.78   \\ 
        R-SDP-b &   I & 0.06 & 12 & 2.38 & 1.44   \\ 
          R-SDP-s &   I & 0.06 & 10 & 8.96 & 5.84   \\
              R-SDP($G)$-f &   I & 0.03 & 12 & 0.94 & 0.55   \\ 
        R-SDP($G)$-b &   I & 0.03 & 9.2 & 1.85 & 1.09   \\ 
         R-SDP($G)$-s &   I & 0.03 & 7.9 & 6.54 & 3.96   \\
          
        R-SDP-f &   III & 0.09 & 42 & 2.75 & 1.69   \\ 
        R-SDP-b &   III & 0.09 & 28 & 4.97 &2.89  \\ 
          R-SDP-s &   III & 0.09 & 23 & 12.2 & 6.8   \\
              R-SDP($G)$-f &   III & 0.06 & 27 & 2.04 & 1.21   \\ 
        R-SDP($G)$-b &   III & 0.06 & 23 & 2.63 & 1.53   \\ 
         R-SDP($G)$-s &   III & 0.06 & 18 & 9.67 & 5.61   \\

          R-SDP-f &   V & 0.12 & 76 & 4.93 & 3.04   \\ 
        R-SDP-b &   V & 0.12 & 51 & 8.26 & 5   \\ 
          R-SDP-s &   V & 0.12 & 43 & 15.69 & 9.37   \\
              R-SDP($G)$-f &   V & 0.07 & 48 & 3.93 & 2.32   \\ 
        R-SDP($G)$-b &   V & 0.07 & 40 & 4.99 & 2.96   \\ 
         R-SDP($G)$-s &   V & 0.07 & 32 & 14.12 & 7.73   \\
    \end{tabular}
    \caption{Performance of $\mathsf{CROSS}$. Sizes are in kilobytes and timings in MCycles. }
    \label{tab:cross}
\end{table}

\medskip

$\mathsf{CROSS}$ provides several variants, one relying on Restricted SDP, denoted by R-SDP, one relying on Restricted SDP in a subgroup $G$, denotes by R-SDP($G$). The ``f'' variant stands for \emph{fast}, the ``b'' variant provides a \emph{balanced} solution and the ``s'' variant provides a \emph{small} solution. 

\newpage

\item $\mathsf{LESS}$ \\ 

$\mathsf{LESS}$ \cite{lessNIST} is a code-based signature scheme based on LEP and using a ZK protocol with the Fiat-Shamir transform. 

On a high level, the idea of $\mathsf{LESS}$ is as follows.   A prover publishes $\bG \in \mathbb{F}_q^{k \times n}$ chosen at random and chooses a secret permutation matrix  $\bP$ and a $\bv \in (\mathbb{F}_q^\star)^n$ at random. The prover computes and publishes $\bG'=\bG\bP\text{diag}(\bv)$, while the monomial transformation $\bP \text{diag}(\bv)$ is kept secret.
In order to prove knowledge of the monomial transformation, the prover also computes the commitment $\bG''= \bG\bP'\text{diag}(\bv')$ for some permutation matrix $\bP'$ and $\bv' \in (\mathbb{F}_q^\star)^n.$ The prover can thus easily provide the monomial transformation from $\bG$ to $\bG''$ (being $\bP'\text{diag}(\bv')$) or the linear isometry from $\bG'$ to $\bG''$ (being $\bP^{-1}\text{diag}(\bv)^{-1}\bP' \text{diag}(\bv')$) without revealing any information on the secret monomial from $\bG$ to $\bG'$ (being $\bP\text{diag}(\bv)$).

Clearly such ZK protocol comes with a cheating probability of 1/2. $\mathsf{LESS}$ decreases the cheating probability by using multiple public keys. In more details, one chooses several monomial transformations $\bQ_1, \ldots, \bQ_N$ and publishes $\bG\bQ_1, \ldots, \bG\bQ_N$. The verifier now chooses from which $\bG\bQ_i$ the monomial transformation to $\bG''$ should be revealed, thus increasing the challenge space to $N+1$ and the cheating probability to $\frac{1}{N+1}$.

Since the $\bG$ was chosen at random it is enough to send a seed as public key. 
A drawback that comes with $\mathsf{LESS}$ is that the commitments and the responses are structured matrices, thus needing a lot of bits to be sent. 

$\mathsf{LESS}$ also makes use of several compression techniques such as seed trees and weighted challenges.

\begin{table}[]
    \centering
    \begin{tabular}{c|c|c|c|c|c}
       Variant &  Level & Public key size & Signature size & Signing time & Verification time  \\ \hline 
        
      $\mathsf{LESS}$-1b &   I & 13.7 & 8.4 &  878.7 & 890.8    \\ 
       $\mathsf{LESS}$-1i &   I & 41.1 & 6.1 & 876.6  & 883.6   \\ 
       $\mathsf{LESS}$-1s &   I & 95.9 & 5.2 & 703.6  &   714.7 \\
      $\mathsf{LESS}$-3b &   III & 34.5 & 18.4 & 7224 & 7315   \\ 
       $\mathsf{LESS}$-3s &   III & 68.9 & 14.1 &  8527 & 8608   \\ 
    $\mathsf{LESS}$-5b&   V & 64.6 & 32.5& 33787  &    34014 \\
$\mathsf{LESS}$-5s&   V & 129& 26.1 &  22621&   22703 \\
  
    \end{tabular}
    \caption{Performance of $\mathsf{LESS}$. Sizes are in kilobytes and timings in MCycles. }
    \label{tab:less}
\end{table}

Also $\mathsf{LESS}$ provides several variants: a \emph{balanced} configuration, denoted with ``b'', where public
key and signature are roughly of the same size, and a \emph{small} configuration, denotes with ``s'', providing a small signature at the cost of larger public keys. Finally, for level I also an \emph{intermediate} configuration, denoted with ``i'' is given.

Note that at the moment of this writing, the contributors of $\mathsf{LESS}$ suggested a novel approach to shorten the signatures. In \cite{lessnew} the authors propose to use canonical forms of matrices, this corresponds to a short representative of a certain equivalence class. As a first step, the monomial transformations are split as 
$(\bP,\bv, \bP',\bv')$ for $\bP$ a $k\times k$ permutation matrix, $\bP'$ a $(n-k) \times (n-k)$ permutation matrix and $\bv \in (\mathbb{F}_q^\star)^k, \bv' \in (\mathbb{F}_q^\star)^{n-k}$, thus getting $\bG$ and $\bG'$ are monomially equivalent if there exist $(\bP,\bv,\bP',\bv')$ such that 
$$\bG= \bS \bG' \begin{pmatrix} \bP \text{diag}(\bv) & \bz \\ \bz & \bP'\text{diag}(\bv')\end{pmatrix},$$ for some $\bS \in \text{GL}_k(q).$

Since one sends the generator matrices in systematic form, this allows the authors to  restrict the monomial transformation to the redundant $k \times (n-k)$ part and only send $\bP^{-1}\text{diag}(\bv)^{-1}\bP' \text{diag}(\bv').$

The resulting sizes are much smaller now, as shown in Table \ref{tab:Lessnew}.

\begin{table}[]
    \centering
    \begin{tabular}{c|c|c|c}
    Variant &  Level & Public key size & Signature size    \\ \hline 
        
         $\mathsf{LESS}$-1c &   I & 13.9 & 2.4 \\
      $\mathsf{LESS}$-1f &   I & 41.8& 1.8 \\
     $\mathsf{LESS}$-3c &   III & 35 &5.6 \\
     $\mathsf{LESS}$-3f &   III & 105.2& 4.4 \\

   $\mathsf{LESS}$-5c &   V & 65.8 &  10 \\
$\mathsf{LESS}$-5f &   V & 197.3& 7.8\\
  
    \end{tabular}
    \caption{New sizes of $\mathsf{LESS}$. Sizes are in kilobytes. }
    \label{tab:Lessnew}
\end{table}

   \medskip 
   \item $\mathsf{MEDS}$ \\

   $\mathsf{MEDS}$ \cite{MEDSNIST} uses the same strategy as $\mathsf{LESS},$ but adapted to matrix codes and the rank metric.
   
   A prover publishes $\bG_1, \ldots, \bG_k \in \mathbb{F}_q^{m \times n}$ chosen at random and chooses the secret matrices  $\bA \in \text{GL}_m(q), \bB \in \text{GL}_n(q)$   at random. The prover computes and publishes $\bG_i'=\bA \bG_i\bB$, while the rank-metric isometry $(\bA,\bB)$ is kept secret.
In order to prove knowledge of the monomial transformation, the prover also computes the commitment $\bG_i''= \bA' \bG_i\bB'$ for some $\bA' \in \text{GL}_m(q), \bB \in \text{GL}_n(q)$. The prover can thus easily provide the   transformation from $\mathcal{C}= \langle \bG_1, \ldots, \bG_k \rangle$ to $\mathcal{C}'= \langle \bG_1'',\ldots, \bG_k''\rangle$ (being $\bA', \bB'$) or the   isometry from $\mathcal{C}'= \langle \bG_1', \ldots, \bG_k\rangle$ to $\mathcal{C}''=\langle \bG_1'', \ldots, \bG_k''\rangle$ (being $\bA' \bA^{-1}, \bB^{-1}\bB'$) without revealing any information on the secret isometry from $\mathcal{C}$ to $\mathcal{C}'$ (being $\bA,\bB$).

The signature scheme $\mathsf{MEDS}$ also makes use of the same compression techniques as $\mathsf{LESS},$ namely seed trees, fixed weight challenges and multiple public keys.

Similar to $\mathsf{LESS},$ also $\mathsf{MEDS}$ results in quite total sizes, being the size of the signature added to the size of the public key.

\begin{table}[]
    \centering
    \begin{tabular}{c|c|c|c|c|c}
       Variant &  Level & Public key size & Signature size & Signing time & Verification time  \\ \hline 
        
      $\mathsf{MEDS}$-9923 &   I & 9.9 & 9.8 &  518 & 515.6   \\ 
      $\mathsf{MEDS}$-13220 &   I & 13.2 & 12.98 & 88.9 &   87.48 \\ 
       $\mathsf{MEDS}$-41711 &   III & 41.7 & 41 &  1467& 1462 \\
     $\mathsf{MEDS}$-55604 &   III & 55.6 &  54.7 & 387.3  &  380.7 \\ 
   $\mathsf{MEDS}$-134180 &   V & 134.2 & 132.6 & 1629.9 & 1612.6 \\ 
     $\mathsf{MEDS}$-167717&   V & 167.7 & 165.5& 961.8 & 938.9 \\

    \end{tabular}
    \caption{Performance of $\mathsf{MEDS}$. Sizes are in kilobytes and timings in MCycles. }
    \label{tab:meds}
\end{table}
The $\mathsf{MEDS}$ proposal also gives two different parameter sets for each security level, one being tuned for small signatures, and the other for fast signing and verifying. 

\end{enumerate}

\subsubsection{ZK Protocols and MPCitH}

In the additional round for post-quantum signature schemes, one finds 5 code-based schemes which are using ZK protocols and using the MPCitH technique.

\begin{table}[h]
    \centering
    \begin{tabular}{|c|c|c|} \hline 
        \textbf{Hard Problem} & \textbf{MPC}  & \textbf{Scheme}   \\ 
        \hline 
SDP  & hypercube, threshold & SDitH  \\ 
        \hline 
     
       Rank SDP & hypercube additive &  RYDE    \\ 
        \hline 
  Relaxed PKP  &  BG splitting &   PERK    \\ \hline  
  MinRank & additive hypercube, linearized polynomials & MIRA  \\ \hline 
  MinRank & Kipnis-Shamir modeling & MiRitH  \\ 
        \hline 
    \end{tabular}
    \caption{Signatures from ZK protocols  and MPCitH technique submitted to the additional call of NIST for signature schemes.}
    \label{tab:compmpc}
\end{table}

\begin{enumerate}
    \item SDitH: \\

The SDitH signature scheme \cite{SDitHNIST}  relies on the SDP and an MPC protocol which efficiently checks whether
a given shared input corresponds to the solution of a SDP  instance. The used MPC protocol is called \emph{hypercube technique} \cite{hypercube} and instead traditional additive sharings, SDitH uses low-threshold linear secret
sharings to exploit their error-correcting feature, called \emph{threshold
approach} \cite{threshold}.
\medskip

First of all, recall that due to the systematic form of a parity-check matrix $$\bH= \begin{pmatrix} \bA & \text{Id}_{n-k} \end{pmatrix},$$ any syndrome $$\bs= (\be,\be')\bH^\top= \be\bA^\top+\be'.$$ Thus, it is enough to use $\be$ for the secret sharing. 
\medskip

Let $\mathbb{F}_q=\{f_1, \ldots, f_q\}$ . The   MPC
protocol is based on four polynomials, $$S(x),P(x),Q(x),F(x),$$ defined as
\begin{itemize}
    \item $S(x) \in \mathbb{F}_q[x]$ of degree up to $n-1$ such that $S(f_i)=e_i$,
    \item $Q(x) \in \mathbb{F}_q[x]$ of degree $t= \text{wt}_H(\be',\be)$ such that $Q(x)= \prod_{i \in \text{supp}(\be',\be)} (x-f_i)$,
    \item $F(x) \in \mathbb{F}_q[x]$ of degree $n$ such that $F(x)= \prod_{i=1}^n (x-f_i)$,
    \item $P(x) \in \mathbb{F}_q[x]$ of degree up to $t-1$, such that $P(x)= \frac{S(x)Q(x)}{F(x)}.$
\end{itemize}
The correctness of the SDP solution amounts to verifying the relation:
$$S(x)Q(x)=P(x)F(x).$$
 
 While $F(x)$ is made public, the prover wants to convince the verifier of the knowledge of $P(x),Q(x)$, such that $S(f_i)Q(f_i))=P(f_i)F(f_i))=0$ for all $i \in \{1,\ldots,n\}.$

 The soundness of the MPC protocol is based on the fact that $\text{wt}_H(\be',\be)=t$ is equivalent to the existence of $P(x),Q(x)$ of degree up to $t-1$, respectively $t$, such that $S(x)Q(x)=P(x)F(x).$ The parties thus get as shares $(\be, P(x),Q(x))$, locally compute $(\be',\be)$ and $S(x)$ by Lagrange interpolation and verify that $S(x)Q(x)=P(x)F(x).$

\begin{table}[]
    \centering
    \begin{tabular}{c|c|c|c|c|c}
       Variant &  Level & Public key size & Signature size & Signing time & Verification time  \\ \hline 
        
      SDitH-gf256-L1-hyp &   I & 0.1 & 8.2 &  13.4 & 12.5   \\ 
      SDitH-gf251-L1-hyp & I & 0.1 & 8.2 & 22.1 & 21.2\\
       SDitH-gf256-L1-thr &I & 0.1& 10.1& 5.1 & 1.6 \\ 
       SDitH-gf251-L1-thr & I & 0.1 & 10.1 & 4.4 & 0.6 \\
             SDitH-gf256-L3-hyp&   III & 0.2 & 19.1 & 30.5& 27.7  \\
          SDitH-gf251-L3-hyp & III & 0.2 &19.1  & 51.1 & 49 \\
          SDitH-gf256-L3-thr & III & 0.2&24.9 & 14.8 & 4.9\\
          SDitH-gf251-L3-thr & III & 0.2 & 24.9& 11.7 & 1.5 \\ 
                        SDitH-gf256-L5-hyp &   V & 0.2& 33.4& 59.2 & 54.4 \\ 
          SDitH-gf251-L5-hyp&   V &  0.2 & 33.4 & 94.8 & 91.3 \\
          SDitH-gf256-L5-thr & V &  0.2&43.9& 30.5 & 10.2 \\ 
          SDitH-gf251-L5-thr & V & 0.2 &43.9 & 23.9 & 3.2 \\ 
    \end{tabular}
    \caption{Performance of SDitH. Sizes are in kilobytes and timings in MCycles. }
    \label{tab:sdith}
\end{table}  

SDitH provides for each security level 4 parameter sets, two for the hypercube approach and two for the threshold approach. There is a clear trade-off between the two variants, as the hypercube approach achieves smaller signatures, while the threshold approach is faster.

 \medskip
 \item RYDE: \\ 

 RYDE \cite{RYDENIST} is based on the Rank SDP and using the $(\ell,N)$-threshold linear secret sharing
scheme as MPC protocol. For this a secret $s$ is split into $N$ shares $[[s]]= (s_1, \ldots, s_N)$, such that the secret can be recovered from any $\ell+1$ shares $s_i.$ 

RYDE uses an additive $(N,N)$-threshold linear secret sharing scheme, as explained in Section \ref{sec:mpc}, that is the shares of $s$ are given by $$(r_1, \ldots, r_{N-1}, s-\sum_{i=1}^{N-1} r_i),$$ for some random $r_i.$

In more details, the MPC protocol works as follows. We are given a parity-check matrix $$\bH=\begin{pmatrix} \bA & \text{Id}_{n-k} \end{pmatrix} \in \mathbb{F}_{q^m}^{(n-k) \times n},$$ a syndrome $\bs \in \mathbb{F}_{q^m}^{n-k}$ and a weight $t$. Let $(\be,\be')$ be a solution to the Rank SDP instance. Each party is then given a share of  $[[\be]]$. Let $\mathcal{S}$ be the error support of $(\be, \be'),$ i.e., $\mathcal{S}=\langle e_1, \ldots, e_n\rangle$ of $\mathbb{F}_q$- dimension $t$. Then $\mathcal{S}$ has an annihilator polynomial $$f(x)= \prod_{s \in \mathcal{S}}(x-s).$$ The parties also take the following as shares $\bb,\ba,c$, where $\bb \in \mathbb{F}_{q^m}^t$ is a vector containing the coefficients of 
$$L=\sum_{i=1}^t b_i(x^{q^i}-x),$$
 $\ba \in \mathbb{F}_{q^{m \eta}}^t$ is randomly sampled and $c= - \langle \bb,\ba\rangle.$
 The parties now proceed as
 \begin{enumerate}
     \item sample at random $(\gamma_1, \ldots, \gamma_n,\varepsilon) \in \mathbb{F}_{m \cdot \eta}^{n+1}$,
     \item locally compute $\be' = \bs-\bA\be$,
     \item locally compute $z=-\sum_{j=1}^n \gamma_j (e_j^{q^t}-e_j)$,
     \item locally compute $w_i= \sum_{j=1}^n \gamma_j(e_j^{q^i}-e_j)$ for all $i \in \{1, \ldots, t-1\}$,
     \item locally compute and open $\alpha = \varepsilon \mathbf w + \ba$,
     \item locally compute and open $v=\varepsilon z - \langle \alpha,\bb\rangle -c$,
     \item and they accept if $v=0.$
 \end{enumerate}

 RYDE is able to achieve smaller signatures than SDitH, however at the cost of a slower signing and verifying process. In Table \ref{tab:ryde}, we can see the two parameter sets for each security level, one denoted by ``F'' for a \emph{fast} version and one denoted by ``S'' for a \emph{small} version.

\begin{table}[]
    \centering
    \begin{tabular}{c|c|c|c|c|c}
       Variant &  Level & Public key size & Signature size & Signing time & Verification time  \\ \hline 
        
      RYDE-128F &   I & 0.09 & 7.4 &  5.4& 4.4   \\ 
     RYDE-128S & I & 0.09 & 6 & 23.4 & 20.1 \\
    RYDE-192F &   III & 0.13 & 16.4 & 12.2 & 10.7  \\
          RYDE-192S & III & 0.13 & 13 & 49.6& 44.8\\ 
          RYDE-256 &   V & 0.2 & 29.1& 26 & 22.7 \\ 
          RYDE-256 & V & 0.2 &22.8 &105.5 & 94.9\\ 
    \end{tabular}
    \caption{Performance of RYDE. Sizes are in kilobytes and timings in MCycles. }
    \label{tab:ryde}
\end{table}

\medskip
\item PERK: \\

PERK \cite{PERKNIST} is based on the relaxed PKP, that is, one publishes a parity-check matrix $\bH \in \mathbb{F}_q^{(n-k)\times n}$, a vector $\be \in \mathbb{F}_q^n$ and a permuted syndrome $\bs \in \mathbb{F}_q^{n-k},$ i.e., there exists some $\sigma \in S_n$ such that $\bH\sigma(\be)^\top=\bs^\top.$ Hence, the secret is given by the permutation $\sigma.$

PERK is based on the BG ZK protocol introduced in \cite{bg}, and employs a simple MPC protocol. 

The BG protocol works as follows: one samples randomly permutations $\sigma_2, \ldots, \sigma_N \in S_n,$ and vectors $\bv_2, \ldots, \bv_N \in \mathbb{F}_q^n$. One computes the commitments $c_i$ from the hashes of the used seeds to generate $\sigma_i,\bv_i.$

One then computes the permutation $\sigma_1= \sigma_2^{-1} \circ \cdots \circ \sigma_N^{-1} \circ \sigma$ and samples a random $\bv_1 \in \mathbb{F}_q^n.$ The commitment $c_1$ is given by the hash of $\sigma_1,$ and the seed for $\bv_1$. One then computes $$\bv= \bv_N+\sum_{i=1}^{N-1} \sigma_N \circ \cdots \circ \sigma_{i+1}(\bv_i)$$ and the commitment $c$ which is the hash of the syndrome $\bv\bH^\top.$

The first challenge of the verifier is some $\beta \in \mathbb{F}_q,$ with this the prover computes $\widetilde{\be}_0=\beta\be$ and for all $i \in \{1,\ldots, N\}$ the vectors $\widetilde{\be}_i= \sigma_i(\widetilde{\be}_{i-1})+\bv_i.$
The first response is given by the hash of all the $\widetilde{\be}_i.$ The verifier can then challenge any $i \in \{1, \ldots, N\}$ and the prover responds with $c_i, \widetilde{\be}_i$ and in the case $i=1$ also with $\sigma_1.$

The employed MPC protocol asks $N$ parties  to perform the BG steps $i \in \{1,\ldots, N\}$
\begin{itemize}
    \item if $i\neq 1$ sample random $(\sigma_i,\bv_i) \in S_n \times \mathbb{F}_q^n$ and compute the commitment $c_i=\mathsf{Hash}(\sigma_i,\bv_i)$ (actually of their seeds),
    \item if $i=1$ sample random $\bv_1 \in \mathbb{F}_q^n$ and compute $\sigma_1$ as usual, i.e., $\sigma_1= \sigma_2^{-1} \circ \cdots \circ \sigma_N^{-1}\circ \sigma$ and the commitment $c_1=\mathsf{Hash}(\sigma_1,\bv_1)$,
    \item upon the challenge $\beta$ one sets $\widetilde{\be}_0)=\beta \be$ and each party computes $\widetilde{\be}_i=\sigma_i(\widetilde{\be}_{i-1})+\bv_i.$
    \item The verifier has to recompute the commitments and $\widetilde{\be}_i$ for each $i \in \{1, \ldots,N\}$.
\end{itemize}
We will denote this MPC protocol as ``BG splitting''.

\begin{table}[]
    \centering
    \begin{tabular}{c|c|c|c|c|c}
       Variant &  Level & Public key size & Signature size & Signing time & Verification time  \\ \hline 
        
    PERK-I-fast3 &   I &0.15 & 8.35 &  7.6 & 5.3   \\ 
    PERK-I-fast5 & I & 0.24& 8.03 & 7.2 & 5.1\\
    PERK-I-short3 & I & 0.15& 6.56 &  39 & 27   \\ 
    PERK-I-short5 & I &0.24 & 6.06 & 36 & 25 \\
    PERK-III-fast3 &   III &0.23 &  18.8&  16 & 13   \\ 
    PERK-III-fast5 & III &0.37 & 18 & 15 & 12 \\
    PERK-III-short3 &   III & 0.23& 15 &  82 & 65\\ 
    PERK-III-short5 & III & 0.37&13.8  & 77 & 60 \\
      PERK-V-fast3 &   V & 0.31& 33.3 &  36 & 28   \\ 
    PERK-V-fast5 & V &0.51 & 31.7 & 34 & 26 \\
    PERK-V-short3 &   V & 0.31& 26.4 &  185 & 143   \\ 
    PERK-V-short5 & V & 0.51&  24.2& 171 & 131 \\
    \end{tabular}
    \caption{Performance of PERK. Sizes are in kilobytes and timings in MCycles. }
    \label{tab:perk}
\end{table}  
\medskip 

\item MIRA: \\

MIRA \cite{MIRANIST} is based on the MinRank problem, i.e., the decoding problem for Matrix codes endowed with the rank metric. The MPC protocol used in MIRA is an additive sharing. That is for a secret $s$, the shares are $(r_1, \ldots, r_{N-1}, s- \sum_{i=1}^{N-1} r_i)$, for some random $r_i.$ 

The MPC protocol is similar to the one in RYDE; we have the generating matrices $\bG_1, \ldots, \bG_k \in \mathbb{F}_q^{m \times n},$ one chooses a secret $\bx \in \mathbb{F}_q^k$ and publishes $\bE$ of rank $t$ and $\bR= \bE- \sum_{i=1}^k \bG_i x_i.$

Each party received $\bx \in \mathbb{F}_q^k$ and the coefficients $b_i \in \mathbb{F}_{q^m}$ of the annihilating polynomial $$L(x)= \sum_{i=1}^t b_i x^{q^i},$$ a random $\ba \in \mathbb{F}_{q^{m \eta}}^t$ and $c= -\langle \ba,\bb \rangle.$
The parties proceed as follows
\begin{enumerate}
    \item sample random $(\gamma_1, \ldots, \gamma_n,\varepsilon) \in \mathbb{F}_{q^{m \eta}}^{n+1},$
    \item compute $\bE= \bR+ \sum_{i=1}^k x_i \bG_i$,
    \item set $e_i \in \mathbb{F}_{q^m}$ associated to the $i$th column of $\bE$, that is for some basis $\Gamma$ of $\mathbb{F}_{q^m}$ over $\mathbb{F}_q$ compute $e_i = \Gamma^{-1}(\bE_{\{i\}}),$
    \item compute $z= -\sum_{j=1}^n \gamma_j e_j^{q^t}$,
    \item compute $w_i= \sum_{j=1}^n \gamma_j e^{q^i}$ for all $i \in \{1, \ldots, t\}$,
    \item open the shares to compute $\alpha=\varepsilon \mathbf w + \ba,$
    \item open the shares to compute $v= \varepsilon z - \langle \alpha,\bb \rangle -c$,
    \item and accept if $v=0.$
\end{enumerate}

 MIRA has two parameter sets for each security level, given in Table \ref{tab:mira}. One parameter set is denoted by ``F'' for a \emph{fast} version and one denoted by ``S'' for a \emph{small} version. Compared to RYDE, which uses the same MPC protocol but is based on the rank decoding problem for $\mathbb{F}_{q^m}$-linear codes instead of $\mathbb{F}_q$-linear codes, we can observe that MIRA is able to achieve slightly smaller signature sizes than RYDE, however at the cost of a much slower signing and verification process.

\begin{table}[]
    \centering
    \begin{tabular}{c|c|c|c|c|c}
       Variant &  Level & Public key size & Signature size & Signing time & Verification time  \\ \hline 
        
      MIRA-128F &   I & 0.09 & 7.4 &  37.4 & 36.7  \\ 
     MIRA-128S & I & 0.09 & 5.6 & 46.8 & 43.9 \\
    MIRA-192F &   III & 0.12 & 15.5 & 107.2 & 107  \\
     MIRA-192S & III & 0.12 & 11.8 & 119.7 & 116.2 \\ 
    MIRA-256F &   V & 0.15 & 27.7 & 322.3 & 323.2 \\ 
    MIRA-256S & V & 0.15 &20.8 & 337.7 & 331.4\\ 
    \end{tabular}
    \caption{Performance of MIRA. Sizes are in kilobytes and timings in MCycles. }
    \label{tab:mira}
\end{table}  

\medskip
\item MiRitH: \\

Also MiRitH \cite{MiRitHNIST}  is based on the MinRank problem and uses an MPC protocol. However, MiRitH uses a Kipnis-Shamir \cite{kipnis} modeling, instead of the linearized polynomials used in MIRA. This leads to faster verification and singing.

Recall that in MinRank,
the generating matrices $\bG_1, \ldots, \bG_k \in \mathbb{F}_q^{m \times n},$  a received matrix $\bR \in \mathbb{F}_q^{m \times n}$ are made public, and the task is to find $\bx \in \mathbb{F}_q^k$ such that $\bE= \bR- \sum_{i=1}^k \bG_i x_i$ has rank at most $t$. 

 The Kipnis-Shamir modeling is based on the following fact, if there exists a vector $\bx \in \mathbb{F}_q^k$ and a matrix $\mathbf K\in \mathbb{F}_q^{t\times (n-t)},$ such that 
\begin{equation}\label{ks} (\bR-\sum_{i=1}^k x_i \bG_i)\begin{pmatrix} \bS \\ \mathbf K \end{pmatrix}=\bz,\end{equation} for some invertible $\bS \in \mathbb{F}_q^{(n-t) \times (n-t)}$ then $\bx$ is a solution to the MinRank instance $\bR, \bG_1, \ldots, \bG_k.$
Thus, if we write $\bR=\begin{pmatrix}\bR' & \bR''\end{pmatrix}$ and  $\bG_i=\begin{pmatrix} \bG_i' & \bG_i'' \end{pmatrix},$ for each $i \in \{1, \ldots,k\}$  then we can transform Equation \eqref{ks} to
$$\bR' - \sum_{i=1}^k x_i \bG_i' = \left( \bR''-\sum_{i=1}^k x_i \bG_i''\right) \mathbf K.$$
Thus, let us write $\bR_x= \bR-\sum_{i=1}^k x_i \bG_i$ and as before $\bR_x= \begin{pmatrix}
    \bR_x' & \bR_x''
\end{pmatrix}$, hence the Kipnis-Shamir modeling amounts to showing that $\bR_x'= \bR_x'' \mathbf K.$

Thus, each party gets the additive sharings $[[\bx]]$ and $[[\mathbf K]]$ and $[[\bA]]$ for a random $\bA \in \mathbb{F}_q^{s \times t}$ and $[[\bC]]$ for $ \bC= \bA\mathbf K.$
The parties then proceed as follows
\begin{enumerate}
    \item locally compute sharings $[[\bR_x']]$ and $[[\bR_x'']]$,
    \item sample a random matrix $\bX \in \mathbb{F}_q^{s \times m}$,
    \item locally compute $$[[\mathbf Y]]= \bX [[\bR_x'' ]] + [[\bA]]$$ and open the sharings, so each party gets $\mathbf Y$,
    \item locally compute $$[[\mathbf V]]=\mathbf Y\mathbf K-\bX[[\bR_x']]-\bC$$ and open the sharings, so that each party gets $\mathbf V$,
    \item accept if $\mathbf V=\bz.$
\end{enumerate}

\begin{table}[]
    \centering
    \begin{tabular}{c|c|c|c|c|c}
       Variant &  Level & Public key size & Signature size & Signing time & Verification time  \\ \hline 
        
      MiRitH-Iaf &   I & 0.13 & 7.7 &  4.8 & 4.5  \\ 
     MiRitH-Ias & I & 0.13 & 5.7 & 42.9 & 42.7 \\
    MiRitH-Ibf &   I & 0.14 & 8.8 & 6.4 & 5.9  \\ 
     MiRitH-Ibs & I & 0.14 & 6.3 & 51.5 & 51.8 \\
      MiRitH-IIIaf &   III & 0.2 & 16.7 &  11.2&10.4  \\ 
     MiRitH-IIIas & III & 0.2 & 12.4 & 94.5 & 94.2 \\
      MiRitH-IIIbf &   III & 0.2 & 17.9 &  13.3 &12.3  \\ 
     MiRitH-IIIbs & III & 0.2 & 13.1 & 112.2 & 112\\
      MiRitH-Vaf &   V & 0.25 & 29.6 &  23.9 & 22.2   \\ 
     MiRitH-Vas & V & 0.25 & 21.8 & 196.7 & 194.6 \\
     MiRitH-Vbf &   V & 0.27 & 32 &  28.3 & 26.3   \\ 
     MiRitH-Vbs & V & 0.27 & 23.1 & 241.6 & 241 \\
    \end{tabular}
    \caption{Performance of MiRitH. Sizes are in kilobytes and timings in MCycles. }
    \label{tab:mirith}
\end{table}  
MiRitH presents four parameter sets for each security level, two denoted with ``a'', and two denotes with ``b'', where the ``b'' variant achieves a greater security level to leave some margins for possible further improvements on solving PKP. The parameter sets denoted by ``f'' are a \emph{fast}  variant, while the ``s'' denotes the \emph{small} variant.  We can see a clear difference in the timings compared to MIRA. 

Another variant of MiRitH is using the hypercube technique, which allows to get even shorter signatures. While the hypercube variant presents several parameter sets for short signatures, we chose only the shortest variant.

\begin{table}[]
    \centering
    \begin{tabular}{c|c|c|c|c|c}
       Variant &  Level & Public key size & Signature size & Signing time & Verification time  \\ \hline 
        
      MiRitH-hyper-Iaf &   I & 0.13 & 6.2 &  4.1 & 3.4  \\ 
     MiRitH-hyper-Ias & I & 0.13 & 3.9 & 3122 & 3066 \\
    MiRitH-hyper-Ibf &   I & 0.14 & 6.7 & 5.3 & 4.4  \\ 
     MiRitH-hyper-Ibs & I & 0.14 & 4.1 & 3184 &3156 \\
      MiRitH-hyper-IIIaf &   III & 0.21 & 13.4 &  9 & 8.2   \\ 
     MiRitH-hyper-IIIas & III & 0.21 & 8.7 & 5149 & 5120 \\
      MiRitH-hyper-IIIbf &   III & 0.21 & 13.8 &  10.2 &9.1  \\ 
     MiRitH-hyper-IIIbs & III & 0.21 & 8.8 & 5278 & 5250\\
      MiRitH-hyper-Vaf &   V & 0.25 & 23.9 &  17.4 & 14.8   \\ 
     MiRitH-hyper-Vas & V & 0.25 & 15.1 & 9730 & 9800 \\
     MiRitH-hyper-Vbf &   V & 0.27 & 25 &  21.2 & 18.2   \\ 
     MiRitH-hyper-Vbs & V & 0.27 & 15.4 & 9767& 9811 \\
    \end{tabular}
    \caption{Performance of MiRitH using Hypercube. Sizes are in kilobytes and timings in MCycles. }
    \label{tab:mirith2}
\end{table}  

\end{enumerate}
 \begin{remark}
     Note that all the reported timings are from the respective documentations and based on different implementations. For signature sizes, we have taken the average sizes. 
 \end{remark}

\newpage
\section{Conclusion}
In this book chapter, we presented a comprehensive collection of code-based cryptography, concerning its history and most famous schemes, until the latest advances, especially in signature schemes.

There are several open question within this research area, prominent ones include 
\begin{itemize}
    \item Is the Rank Syndrome Decoding Problem NP-hard?
    \item Can we distinguish classical Goppa codes?
    \item How to improve the code-equivalence solvers?
    \item How to construct an efficient and secure hash-and-sign scheme?
\end{itemize}
.. and many more. 

We hope that this book chapter helps young researchers to get into code-based cryptography, so that we can advance in these open question together. 

Any comments, typos or additions can be sent to \url{violetta.weger@tum.de} and we will update the ArXiv version regularly. 
       
\section*{Acknowledgement}
The authors would like to thank Jean-Pierre Tillich, Nicolas Sendrier and Thomas Debris-Alazard for fruitful discussions. The authors would also like to thank Giovanni Tognolini and the anonymous reviewers for pointing out some of the typos.
\newline
Violetta Weger is  supported by the Swiss National Science Foundation grant number 195290 and  by the European Union's Horizon 2020 research and innovation programme under the Marie Sk\l{}odowska-Curie grant agreement no. 899987. 
\newline
Niklas Gassner and Joachim Rosenthal are supported by armasuisse Science and Technology (Project Nr.: CYD C-2020010). 

 \clearpage

\renewcommand{\bibsection}{\section*{Bibliography}}
\bibliographystyle{plain}
\bibliography{references.bib}

\end{document}